\documentclass[11pt,twoside]{ThesisStyle}

\newcommand\anytesi{2012}

\usepackage{tikz}
\usetikzlibrary{arrows}

\usepackage{multirow}
\usepackage{amsmath,amssymb}             
\usepackage[latin1]{inputenc}
\usepackage[T1]{fontenc}
\usepackage[left=1.5in,right=1.3in,top=1.1in,bottom=1.1in,includefoot,includehead,headheight=13.6pt]{geometry}

\usepackage[nottoc, notlof, notlot]{tocbibind}
\usepackage{minitoc}

\usepackage{dsfont}

\usepackage{amsthm}

\usepackage{aecompl}

\usepackage{ifpdf}

\ifpdf
   \DeclareGraphicsExtensions{.jpg,.pdf}
   \usepackage[hyperindex=true]{hyperref}
\else
  \DeclareGraphicsExtensions{.ps,.eps}
  \usepackage[dvipdfm,hyperindex=true]{hyperref}
\fi

\usepackage{color}
\definecolor{linkcol}{rgb}{0,0,0}
\definecolor{citecol}{rgb}{0,0,0}

\hypersetup
{
bookmarksopen=true,
pdftitle="RFID Systems",
pdfauthor="Rolando Trujillo Rasua",
pdfsubject="Creation of scalable protocolos for RFID systems",
pdfmenubar=true, 
pdfhighlight=/O, 
colorlinks=true, 
pdfpagemode=None, 
pdfpagelayout=SinglePage, 
pdffitwindow=true, 
linkcolor=linkcol, 
citecolor=citecol, 
urlcolor=linkcol 
}

\usepackage{rotating}                    
\usepackage{fancyhdr}                    

\let\headruleORIG\headrule
\renewcommand{\headrule}{\color{black} \headruleORIG}

\usepackage{colortbl}
\arrayrulecolor{black}

\fancypagestyle{plain}{
    \fancyhead{}
    \fancyfoot{}
    
}

\makeatletter

\def\cleardoublepage{\clearpage\if@twoside \ifodd\c@page\else%
  \hbox{}%
  \thispagestyle{empty}
  \newpage%
  \if@twocolumn\hbox{}\newpage\fi\fi\fi}

\makeatother

{%

\hrulefill
\vspace*{0.5cm}%
\end{minipage}
}

\let\minitocORIG\minitoc
\renewcommand{\minitoc}{\begingroup

\definecolor{tmp}{rgb}{0,0,0.4} 

\hypersetup{linkcolor=tmp}

\minitocORIG \endgroup \vspace{1.5em}}

\usepackage{multirow}
\usepackage{slashbox}

\usepackage{algorithm}
\usepackage{algorithmic}

{ \begin{list}%
	{$\bullet$}%
	{\setlength{\labelwidth}{25pt}%
	 \setlength{\leftmargin}{30pt}%
	 \setlength{\itemsep}{\parsep}}}%
{ \end{list} }

\newtheorem{definition}{Definition}
\renewcommand{\epsilon}{\varepsilon}

\newenvironment{protocol_message}%
{\noindent\ignorespaces}%
{\par\noindent%
\ignorespacesafterend}

\newtheorem{proposition}{Proposition}
\newtheorem{lemma}{Lemma}
\newtheorem{corollary}{Corollary}
\newtheorem{theorem}{Theorem}
\newtheorem{remark}{Remark}
\newtheorem{example}{Example}
\newtheorem{case}{Case}

\usepackage{colortbl}

\usepackage{multirow}
\usepackage{booktabs}
\usepackage{ctable}

\usepackage{amssymb}
\usepackage{algorithm}
\usepackage{algorithmic}
\usepackage{amsthm}
\usepackage{amsmath}
\usepackage{url}

\usepackage{multirow}
\usepackage{subfigure}

\usepackage{longtable}
\usepackage{array}
\usepackage{booktabs}

\usepackage{color}

\begin{document}

\begin{titlepage}
\begin{center}
\noindent {\LARGE \textbf{Universitat Rovira i Virgili}} \\
\vspace*{0.3cm}
\noindent {\large \textbf{Department of}} \\
\vspace*{0.3cm}
\noindent {\large \textbf{Computer Engineering and Mathematics}} \\
\vspace*{0.5cm}
\noindent \Huge \textbf{Ph.D. Dissertation} \\
\vspace*{0.5cm}
\noindent \Large \textbf{\textsc{Privacy in RFID and mobile objects}}\\
\vspace*{0.4cm}
\noindent \large {Author:\\}
\noindent \LARGE Rolando \textsc{Trujillo-Rasua} \\
\vspace*{0.8cm}
\noindent \Large {Thesis Advisors:} \\
\noindent \Large Dr. Agusti \textsc{Solanas}\\
\noindent \Large Dr. Josep \textsc{Domingo-Ferrer}\\
\vspace*{2.0cm}
\noindent \large Dissertation submitted to the Department of Computer \\
\noindent \large Engineering and Mathematics in partial fulfillment of the  \\
\noindent \large requirements of the degree of Doctor of Philosophy \\
\noindent \large in Computer Science \\
\vspace*{0.5cm}

\end{center}

\newpage


\thispagestyle{empty}

\begin{center}

$ $\\$ $\\$ $\\$ $\\$ $\\$ $\\$ $\\$ $\\$ $\\$ $\\$ $\\$ $\\

{\large

\copyright~ Copyright \anytesi~by Rolando Trujillo-Rasua\\

All Rights Reserved

}

\end{center}

\newpage


\thispagestyle{empty}

\begin{center}

{\normalsize

\begin{minipage}{10cm}

I certify that I have read this dissertation and that in

my opinion it is fully adequate, in scope and quality, as

a dissertation for the degree of Doctor of Philosophy

in Computer Science.

\vspace{3cm}

\begin{flushright}

\begin{minipage}{7cm}

\hrule

\begin{center}

\vspace{0.1cm}

Dr. Agusti Solanas\\

(Advisor)

\end{center}

\end{minipage}

\end{flushright}

\vspace{0.5cm}

I certify that I have read this dissertation and that in

my opinion it is fully adequate, in scope and quality, as

a dissertation for the degree of Doctor of Philosophy

in Computer Science.

\vspace{3cm}

\begin{flushright}

\begin{minipage}{7cm}

\hrule

\begin{center}

\vspace{0.1cm}

Dr. Josep Domingo-Ferrer\\

(Advisor)

\end{center}

\end{minipage}

\end{flushright}

\vspace{0.5cm}

Approved by the University Committee on Graduate Studies:

\vspace{3cm}

\begin{flushright}

\begin{minipage}{7cm}

\hrule

$ $

\end{minipage}

\end{flushright}

\end{minipage}

}

\end{center}

\newpage

\thispagestyle{empty}


\end{titlepage}
\sloppy

\titlepage

\dominitoc

\pagenumbering{roman}

 \cleardoublepage

\section*{Resum}

Els sistemes RFID permeten la identificaci\'{o} r\`{a}pida i autom\`{a}tica d'etiquetes RFID a trav\'{e}s d'un canal de comunicaci\'{o} sense fils. Aquestes etiquetes s\'{o}n dispositius amb cert poder de c\`{o}mput i amb capacitat d'emmagatzematge de informaci\'{o}. Es per aix\`{o} que els objectes que porten una etiqueta RFID adherida permeten la lectura d'una quantitat rica i variada de dades que els descriuen i caracteritzen, com per exemple un codi \'{u}nic d'identificaci\'{o}, el nom, el model o la data d'expiraci\'{o}. A m\'{e}s, aquesta informaci\'{o} pot ser llegida sense la necessitat d'un contacte visual entre el lector i l'etiqueta, la qual cosa agilita considerablement els processos d'inventariat, identificaci\'{o} o control autom\`{a}tic.

Perquè l'\'{u}s de la tecnologia RFID es generalitzi amb \`{e}xit, es convenient complir amb diversos objectius: efici\`{e}ncia, seguretat i protecci\'{o} de la privadesa. No obstant aix\`{o}, el disseny de protocols d'identificaci\'{o} segurs, privats i escalables \'{e}s un repte dif\'{i}cil d'abordar ateses les restriccions computacionals de les etiquetes RFID i la seva naturalesa sense fils. Es per aix\'{o} que, en la present tesi, partim de protocols d'identificaci\'{o} segurs i privats, i mostrem com es pot aconseguir escalabilitat mitjançant una arquitectura distribu\"{i}da i col·laborativa. D'aquesta manera, la seguretat i la privadesa s'aconsegueixen mitjançant el propi protocol d'identificaci\'{o}, mentre que l'escalabilitat s'aconsegueix per mitj\`{a} de nous protocols col·laboratius que consideren la posici\'{o} espacial i temporal de les etiquetes RFID.

Independentment dels aven\c{c}os en protocols d'identificaci\'{o} sense fils, existeixen atacs que poden superar amb \`{e}xit qualsevol d'aquests protocols sense necessitat de con\`{e}ixer o descobrir claus secretes v\'{a}lides, ni de trobar vulnerabilitats a les seves implantacions criptogr\`{a}fiques. La idea d'aquests atacs, coneguts com atacs de ``repetidor'', consisteix en crear inadvertidament un pont de comunicaci\'{o} entre una etiqueta leg\'{i}tima i un lector leg\'{i}tim. D'aquesta manera, l'adversari utilitza els drets de l'etiqueta leg\'{i}tima per superar el protocol d'autentificaci\'{o} utilitzat pel lector. En aquesta tesi proposem un nou protocol que, a m\'{e}s d'autentificaci\'{o}, realitza una revisi\'{o} de la dist\`{a}ncia a la qual es troben el lector i l'etiqueta. Aquests tipus de protocols es coneixen com a protocols de fitaci\'{o} de dist\`{a}ncia, els quals no impedeixen aquests tipus d'atacs, per\`{o} s\'{i} que poden frustrar-los amb alta probabilitat.

Per \'{u}ltim, afrontem els problemes de privadesa associats amb la publicaci\'{o} de informaci\'{o} recollida a trav\'{e}s de sistemes RFID. En concret, ens concentrem en dades de mobilitat, que tamb\'{e} poden ser proporcionades per altres sistemes \`{a}mpliament utilitzats tals com el sistema de posicionament global (GPS) i el sistema global de comunicacions m\`{o}bils. La nostra soluci\'{o} es basa en la coneguda noci\'{o} de k-anonimat, obtingut mitjançant permutaci\'{o} i microagregaci\'{o}. Per a aquesta finalitat, definim una nova funci\'{o} de dist\`{a}ncia entre traject\`{o}ries amb la qual desenvolupen dos m\`{e}todes diferents d'anonimitzaci\'{o} de traject\`{o}ries.

\section*{Resumen}

Los sistemas RFID permiten la identificaci\'{o}n r\'{a}pida y autom\'{a}tica de etiquetas RFID a trav\'{e}s de un canal de comunicaci\'{o}n inal\'{a}mbrico.  Dichas etiquetas son dispositivos con cierto poder de c\'{o}mputo y capacidad de almacenamiento de informaci\'{o}n. Es por ello que los objetos que contienen una etiqueta RFID adherida permiten la lectura de una cantidad rica y variada de datos que los describen y caracterizan, por ejemplo,  un c\'{o}digo \'{u}nico de identificaci\'{o}n, el nombre, el modelo o la fecha de expiraci\'{o}n. Adem\'{a}s, esta informaci\'{o}n puede ser le\'{i}da sin la necesidad de un contacto visual entre el lector y la etiqueta, lo cual agiliza considerablemente los procesos de inventariado, identificaci\'{o}n, o control autom\'{a}tico.

Para que el uso de la tecnolog\'{i}a RFID se generalice con \'{e}xito, es conveniente cumplir con varios objetivos: eficiencia, seguridad y protecci\'{o}n de la privacidad. Sin embargo,  el diseño de protocolos de identificaci\'{o}n seguros, privados, y escalables es un reto dif\'{i}cil de abordar dada las restricciones computacionales de las etiquetas RFID y su naturaleza inal\'{a}mbrica.  Es por ello que, en la presente tesis, partimos de protocolos de identificaci\'{o}n seguros y privados, y mostramos c\'{o}mo se puede lograr escalabilidad mediante una arquitectura distribuida y colaborativa. De este modo,  la seguridad y la privacidad se alcanzan mediante el propio protocolo de identificaci\'{o}n, mientras que la escalabilidad se logra por medio de novedosos m\'{e}todos colaborativos que consideran la posici\'{o}n espacial y temporal de las etiquetas RFID.

Independientemente de los avances en protocolos inal\'{a}mbricos de identificaci\'{o}n, existen ataques que pueden superar exitosamente cualquiera de estos protocolos sin necesidad de conocer o descubrir claves secretas v\'{a}lidas ni de encontrar vulnerabilidades en sus implementaciones criptogr\'{a}ficas.  La idea de estos ataques,  conocidos como ataques de ``repetidor'', consiste en crear inadvertidamente un puente de comunicaci\'{o}n entre una etiqueta leg\'{i}tima y un lector  leg\'{i}timo. De este modo, el adversario usa los derechos de la etiqueta leg\'{i}tima para pasar el protocolo de autenticaci\'{o}n usado por el lector. En esta tesis proponemos un nuevo protocolo que adem\'{a}s de autenticaci\'{o}n realiza un chequeo de la distancia a la cual se encuentran el lector y la etiqueta.  Este tipo de protocolos se conocen como protocolos de acotaci\'{o}n de distancia, los cuales no impiden este tipo de ataques, pero s\'{i} pueden frustrarlos con alta probabilidad.

Por \'{u}ltimo, afrontamos los problemas de privacidad asociados con la publicaci\'{o}n de informaci\'{o}n recogida a trav\'{e}s de sistemas RFID. En particular, nos concentramos en datos de movilidad que tambi\'{e}n pueden ser proporcionados por otros sistemas ampliamente usados tales como el sistema de posicionamiento global (GPS) y el sistema global de comunicaciones m\'{o}viles. Nuestra soluci\'{o}n se basa en la conocida noci\'{o}n de k-anonimato, alcanzada mediante permutaciones y microagregaci\'{o}n. Para este fin, definimos una novedosa funci\'{o}n de distancia entre trayectorias con la cual desarrollamos dos m\'{e}todos diferentes de anonimizaci\'{o}n de trayectorias.

\section*{Abstract}

Radio Frequency Identification (RFID) is a technology aimed at efficiently identifying and tracking goods and assets. Such identification may be performed without requiring line-of-sight alignment or physical contact between the RFID tag and the RFID reader, whilst tracking is naturally achieved due to the short interrogation field of RFID readers.
That is why the reduction in price of the RFID tags has been accompanied with an increasing attention paid to this technology. However, since tags are resource-constrained
devices sending identification data wirelessly, designing secure and private RFID identification protocols is a challenging task. This scenario is even more complex when  scalability must be met by those protocols.

Assuming the existence of a lightweight, secure, private and scalable RFID identification protocol, there exist other concerns surrounding the RFID technology. Some of them arise from the technology itself, such as \emph{distance checking}, but others are related to the potential of RFID systems to gather huge amount of tracking data. Publishing and mining such moving objects data is essential to improve efficiency of supervisory control, assets management and localisation, transportation, etc. However, obvious privacy threats arise if an individual can be linked with some of those published trajectories.

The present dissertation contributes to the design of algorithms and protocols aimed at dealing with the issues explained above. First, we propose a set of protocols and heuristics based on a distributed architecture that improve the efficiency of the
identification process without compromising privacy or security. Moreover, we present a novel distance-bounding protocol based on graphs that is extremely low-resource consuming. Finally, we present two trajectory anonymisation methods aimed at preserving the individuals' privacy when their trajectories are released.

\begingroup

\definecolor{tmp}{rgb}{0,0,0.4} 

\hypersetup{linkcolor=tmp}

\tableofcontents

\endgroup


\mainmatter

\chapter{Introduction}
\label{chap:1}

\emph{This chapter introduces the issues we are facing in this dissertation. In addition, it briefly describes the solutions we propose to tackle those issues. Finally, the structure and organisation of the present thesis are outlined.}

\minitoc

\section{Motivation} \label{sec:motivation}

Radio frequency identification (RFID) allows the simultaneous identification of multiple RFID tags. The identification process is performed over a wireless channel without requiring line-of-sight alignment or physical contact between the RFID tags and the RFID reader. These features together 
with others like low deployment costs, being flexible and manageable,
computational power, etc., are causing the RFID technology to be 
preferred to traditional options (\textit{e.g.} barcodes systems). Indeed, nowadays several RFID systems are massively deployed worldwide, namely for \textit{assets tracking} (\textit{e.g.} Air Canada decided to use this
technology to control their food trolleys so as to reduce more than \$2 million
in unexplained losses~\cite{RFIDJournal-2003-March}), \textit{manufacturing}
(\textit{e.g.} Boeing uses RFID to track parts as they arrive, and as they move
from one shop to another within their facilities, thus reducing errors and the
need for people to look for parts~\cite{RFIDJournal-2003-Sept}), \textit{supply
chain management} (\textit{e.g.} Paramount farms, the largest producer of
pistachio in the US, receives 50\% of its production from a network of about
400 partners; the shipments are processed by using RFID that reduces processing
times to up to 60\%~\cite{Violino-2004-RFIDJournal}), \textit{retailing}
(\textit{e.g.} Walmart started to explore the RFID technology in 2003 and
devoted at least three billion dollars to implement
it~\cite{Haley-2003-InternetNews}), and for other applications such as
\textit{payments}, \textit{security} and \textit{access control}.

The RFID technology also has the capability of naturally collecting trajectories of moving objects. Differently to other positioning systems like the GPS, RFID systems do not continuously track a moving object. Instead, RFID readers located at different waypoints create trajectories by identifying tags passing through those waypoints. In this sense, the RFID technology can be considered a simple and low-cost tracking system where complex trilateration and precisely timing of signals are not required. This type of \emph{coarse-grained} tracking is particularly useful to improve the quality of diagnostic processes and business decisions in the healthcare industry~\cite{jones:015001}, to monitor animal behaviour~\cite{10.1109/ICCNT.2010.40}, to enhance bike races~\cite{Tsai-master}, to provide location based services in indoor environments~\cite{10.1109/SUTC.2006.36}, etc.

However, all those potential benefits have been partially overshadowed by important security and privacy threats. RFID tags are resource-constrained devices that respond to any reader interrogation through an insecure channel. This means that both the data stored in the tags' memory and the data transmitted to readers cannot be protected by cryptographically strong primitives and/or large key sizes. Instead, lightweight cryptography requiring no more than $3000$ logic gates should be used~\cite{DBLP:conf/ifip6-8/Peris-LopezHER06}. In this scenario, the privacy of tag bearers could be seriously compromised by disclosing the individual's locations or other sensitive information contained in the RFID tag's memory. Moreover, other security risks like tag impersonation and counterfeiting increase due to the lack of randomness and computational power in the tag's side.

In recent years, several efforts have been made on designing secure and private RFID identification protocols. Among those protocols, the Randomised Hash-Lock Scheme~\cite{JuelsW-2007-percom} has been widely accepted due to its strong privacy and security properties, and its low computational
requirements in the tag's side, \textit{i.e.} it only needs a pseudo-random numbers generator and a one-way hash function. However, this protocol is not scalable. This is particularly problematic if we consider that millions or billions of tags should be managed by typical RFID applications (\textit{e.g.} for supply chain management or inventory control). That is why many other RFID identification protocols have been proposed aimed at being secure and private, yet scalable. Nevertheless, none of them has achieved those three goals at the same time~\cite{DBLP:journals/comcom/AlomairP10}, especially when strong privacy definitions must be met. In the present dissertation, we mainly focus on designing collaborative algorithms that improve the scalability of the Randomised Hash-Lock Scheme~\cite{JuelsW-2007-percom}. The algorithms are collaborative in the sense that several readers deployed in the system exchange information in order to efficiently identify RFID tags. By doing so, our proposals are able to improve the scalability, being as private and secure as the Randomised Hash-Lock Scheme~\cite{JuelsW-2007-percom}, though.

Regardless of the improvements on designing identification/authentication protocols, Desmedt, Goutier and Bengio~\cite{DBLP:conf/crypto/DesmedtGB87} presented in CRYPTO'87 a novel attack called \emph{mafia fraud} that defeated any authentication protocol. In this attack, an adversary succeeds by simply relying the messages between a legitimate reader and a legitimate tag. Initially, the mafia fraud attack was thought to be rather unrealistic because the legitimate prover should take part on the execution of the protocol. However, the RFID technology clearly opens the door to this type of attack since RFID tags answer to any reader's interrogation without any awareness or agreement of their holders. Other types of frauds are also applicable to RFID systems. 
Among them, the \emph{distance fraud} attack~\cite{188361}, in which a dishonest prover claims to be closer to the verifier than he really is, is very important. Both mafia and distance frauds may be accomplished despite of the authentication protocol used by tags and readers. This means that even assuming secure, private, and scalable RFID identification/authentication protocols at the application layer, there exists the need for designing protocols that thwart the mafia and distance fraud attacks. Among the countermeasures against these attacks, \emph{distance bounding} protocols based on the measurement of the round trip time of exchanged messages~\cite{springerlink:10.1007/BF00196726, Beth:1990:ITS:646755.705208} are considered the most suitable for RFID systems. We contribute by designing a novel distance bounding protocol based on graphs that is highly resilient to mafia and distance fraud attacks. Our protocol also deals with other RFID systems's requirements such as efficiency and low memory consumption.

Seemingly, an increasing number of articles are being written on RFID security and privacy areas, namely
ultralightweight protocols, distance-bounding protocols, privacy-preserving lightweight protocols, cryptographically secure pseudo-random numbers generators, cryptosystems based on elliptic curves, RFID privacy models, zero-knowledge authentication protocols for RFID systems, ownership transfer protocols, among others~\cite{RFID_lounge}. All these efforts contribute to the establishment of a technology that may help to do business as much as other revolutionary technologies like internet do. This means that, in the near future, billions of RFID tags will send information to thousands of RFID readers so as to enrich our interaction with
the environment and make our processes more efficient and resilient. Therefore, gathering huge databases of trajectories by using the potential of the RFID technology to track moving objects will be a reality.

Analysing this kind of databases can lead to useful and previously unknown knowledge. However, even when tracking is performed by legitimate parties, the privacy of individuals may be affected by the publication of such databases of trajectories. Simple de-identification realised by removing identifying attributes is insufficient to protect the privacy of individuals. Just knowing some locations visited by an individual can help an adversary to identify the individual's trajectory in the published database. In this context, privacy preservation means that no sensitive location ought to be linkable to an individual. The privacy threat grows if such a trajectory is linked with sensitive personal data like, price of products,  name of drugs, etc, which are usually stored in the tags' memory.

These privacy issues motivate our last research line in this thesis. We note that the boom of the RFID technology strongly promotes the design of privacy-preserving trajectory anonymisation methods. In this sense, we finally focus on mitigating the privacy issues that may arise from the publication of databases of trajectories, rather than on providing 
security and/or privacy to the RFID technology. In particular, we propose a novel distance measure for trajectories which naturally considers both spatial and temporal aspects of trajectories, is computable in polynomial
time, and can cluster trajectories not defined over the same time span. Mainly based on this metric, we propose two methods for trajectory anonymisation which yield anonymised trajectories formed by fully accurate true original locations.

\section{Contributions}

The main contributions of this dissertation are the following:

\begin{enumerate}
    \item \textbf{Efficient RFID identification protocol by means of collaborative readers:} Designing secure, private, and scalable, RFID identification protocols in a multiple tags to one reader scenario is a challenge. However, in scenarios where multiple readers are deployed in the system, scalability may be achieved without compromising privacy or security. In particular, we consider a scenario where readers should continuously monitor moving tags in the system. Under such an assumption, we propose a scheme that has been proven to be efficient in terms of both server and network overhead.
    \item \textbf{Predictive protocol for the scalable identification of RFID tags:} We improve the state-of-the-art of RFID identification schemes based on collaborative readers by proposing a protocol that predicts future and past locations of RFID tags. By doing so, readers are aware of which tags may be identified at some slot of time. Therefore, the identification process is considerably improved.
    \item \textbf{A new distance-bounding protocol:} RFID systems are specially susceptible to mafia and distance frauds. Among the countermeasures to thwart these attacks, distance-bounding protocols are considered the most suitable solutions for RFID systems. We propose a novel distance-bounding protocol only requiring a single hash computation and a linear amount of memory in the tag's side. Despite those limitations, our proposal is highly resilient to both mafia and distance frauds.
    \item \textbf{Privacy-preserving publication of trajectories:} It is hard to say how much personal information and tracking data may be collected by RFID readers in the near future. Nevertheless, trajectories are massively gathered by GPS and/or GSM technologies, and apparently the RFID technology is strongly supporting this massive collection of moving objects data. We focus on designing trajectory anonymisation algorithms that may work over trajectories not defined over the same time span. In particular, we propose two algorithms based on microaggregation and permutation aimed at achieving trajectory $k$-anonymity and location $k$-diversity. Both algorithms are based on a novel distance measure that effectively considers both dimensions: space and time.
\end{enumerate}

\section{Organisation}

This thesis is organised as follows:

\begin{itemize}
    \item Chapter~\ref{chap:2} provides an overview of RFID systems and describes some challenges that the RFID technology should address in order to be successfully deployed worldwide. Also, the controversy between privacy, security, and scalability in RFID systems
    is discussed in detail. Other ramifications of RFID systems are 
    introduced as well, namely distance checking and trajectory anonymisation.
    \item Chapter~\ref{chap:4} presents our first contributions to the scalability issue of RFID identification protocols. In particular, it describes a protocol based on collaborative readers that outperforms previous proposals in terms of both number of cryptographic operations and bandwidth usage.
    \item Chapter~\ref{chap:4_5} introduces the concept of RFID identification protocols based on location prediction. This new proposal is also based on collaborative readers, but it is able to improve the identification process by predicting the next reader that should identify a tag.
    \item Chapter~\ref{chap:5} is devoted to the description of a novel distance-bounding protocol based on graphs. The goal of this proposal is to reduce memory requirements while still achieving high security properties regarding both distance and mafia fraud. To do so, the concept of distance bounding protocols based on graphs is introduced and defined.
    \item Chapter~\ref{chap:7} presents our contributions to the anonymisation of moving objects data. In particular, two anonymisation methods releasing trajectories that contain true original locations are proposed. Both methods are able to effectively deal with trajectories not defined over the same time span thanks to a novel distance measure presented in this chapter.
    \item Finally, Chapter~\ref{chap:8} summarises our contributions and describes possible future research lines.
\end{itemize}

\chapter{Background}
\label{chap:2}

\emph{This chapter briefly describes RFID systems, from the very beginning of the technology to the most recent applications and challenges. Among all the challenges, it focuses on the security, privacy, and scalability issues of RFID systems, distance checking, and the anonymisation of mobility data collected by either the RFID or the GPS technologies. In addition, the main contributions aimed at facing all those challenges are reviewed.}

\minitoc

\section{A brief history of the RFID technology}

The first RFID system dates from the Second World War~\cite{Rieback:2006:ERS:1115690.1115749}. In those days, radar technology was used to detect approaching aircrafts by sending pulses of radio energy and receiving the echoes generated by those aircrafts. However, visual contact was the only way to identify an incoming plane as enemy or allied. The Germans solved this problem by rolling their aircrafts in response to a signal from the ground radar station so as to change the radar reflection's polarisation and thus, creating a distinctive blip on the radars. In military terms, this is considered a huge advantage over previous radar systems. Actually, some people believe that this ingenious German military strategy could have helped the US army to prevent the attack on Pearl Harbor.

Later, the British army introduced a more sophisticated system called Identify Friend or Foe (IFF). Closer to current RFID systems, each plane was equipped with a transponder that modulated back the radar signal, thereby allowing identification
of that aircraft as friendly. Due to its simplicity and resiliency, this technology is still being used by the aviation industry to keep airplanes tracked. However, a \emph{not} friendly aircraft should be treated with care since there is no proof of it being an enemy. Apparently, this inconvenience has been the cause of unfortunate accidents (\emph{e.g.} in 1983, the Soviet Union army shot down a Korean civilian airplane that was confused with a spy plane~\cite{korean_crash}. Similarly, an Iranian civilian plane was shot down in 1988 by the United States army,  and all 290 passengers and crew were killed, including 66 children~\cite{iran_crash}).

As advances in radio frequency communications systems and low-cost embedded computers continued through the 1950s, 1960s, and 1970s, several technologies related to radio waves were developed (\emph{e.g.} the Electronic Article Surveillance application (EAS) designed to prevent shoplifting from retail stores). Nevertheless, the first patent for a passive, read-write RFID tag, was received by Mario Cardullo in 1973. This is considered the first true ancestor of modern RFID as it was a passive radio transponder with memory. Since then, RFID systems hardly seem recognisable. Modern RFID tags may be similar in size to a grain of rice; may have computational capabilities, Read Only Memory (ROM), Electrically Erasable Programmable Read-Only Memory (EEPROM); may be active in the sense of using batteries rather than RFID readers' power, etc.

Consequently, over the years, the number of solutions based on RFID has considerably grown. In fact, RFID systems are nowadays more related with business than with the military industry. In the 1980s and 1990s, RFID applications emerged in transport, access control, animal identification, tracking nuclear material and trucks and electronic toll collection~\cite{citeulike:755357}. This trend is increasing  exponentially in the 21st century due to tags's price reduction~\cite{five_cents} and RFID standardisation~\cite{epcglobal}; over 33 billion RFID tags were produced in 2010\footnote{According to a study of In-Stat (\url{http://www.in-stat.com}) - \url{http://www.instat.com/press.asp?Sku=IN0502115WT&ID=1545}} and 2.31 and 2.88 billion tags were sold in 2010 and 2011, respectively\footnote{According to an extensive research by IDTechEx (\url{http://www.idtechex.com/research/reports/rfid_forecasts_players_and_opportunities_2011_2021_000250.asp})}.


\subsection{RFID systems} \label{chap:rfid}


An RFID system is supposed to identify and track objects by using radio waves. Similar to other identification systems such as barcodes, fingerprints or eyes' iris, the reader (RFID reader) reads from some source of identification data (RFID tag). Then, the identification data are usually processed by a data processing subsystem or server. However, RFID systems outstand from other identification systems because they may be nearly as cheap
as barcode systems, use a wireless channel like GPS or GSM, and have some computational capabilities like magnetic cards. That is why more and more attention has been paid to this technology in recent years.

In technical terms, an RFID system consists of three key elements:

\begin{itemize}
  \item The RFID tag, or \textit{transponder}, that contains information and identification data.
  \item The RFID reader, or \textit{transceiver}, that queries transponders for
  information stored on them. This information can range from static
  identification numbers to user or sensory data.
  \item The \textit{data processing subsystem or server}, which processes the data obtained
  from readers.
\end{itemize}

Intuitively, all objects to be identified shall be physically tagged with RFID tags. Then, RFID readers should be
strategically distributed to interrogate tags where their data are required (\emph{e.g.} a bicycle race timing system needs to place, at least, a reader at the start line and another one at the finish line). Other properties, namely readers' interrogation field size, computation capabilities, and memory size of tags, vary from application to application.

\subsubsection{RFID tags}

Typical \textit{transponders} (\textit{trasnmitters/responders}) or RFID tags (see examples in Figure~\ref{fig:tag_pictures}),
consist of integrated circuits connected to an antenna~\cite{finkenzeller99}.
The memory element serves as writable and non-writable data storage, which can range between few bytes up to several kilobytes. Tags can be designed to be
read-only, write-once, read-many, or fully rewritable. Therefore, tag programming can take place at the
manufacturing level or at the application level. A tag can obtain power from the signal
received from the reader, or it can have its own internal source
of power. The way tags get their power generally defines their
category:

\begin{itemize}
  \item \textbf{Passive tags} use power provided by the reader by means of electromagnetic waves. The lack of an onboard power supply means that the device can be quite small and cheap.
\item \textbf{Semi-passive tags} use a battery to run
the microchip's circuitry but communicate by harvesting
power from the reader signal.
\item \textbf{Active tags} have their own internal power source, usually a battery, which is
used to power the outgoing signal.
\end{itemize}

\begin{figure}[!t]
\centering
 \subfigure{
   \includegraphics[width=1.45in]
   {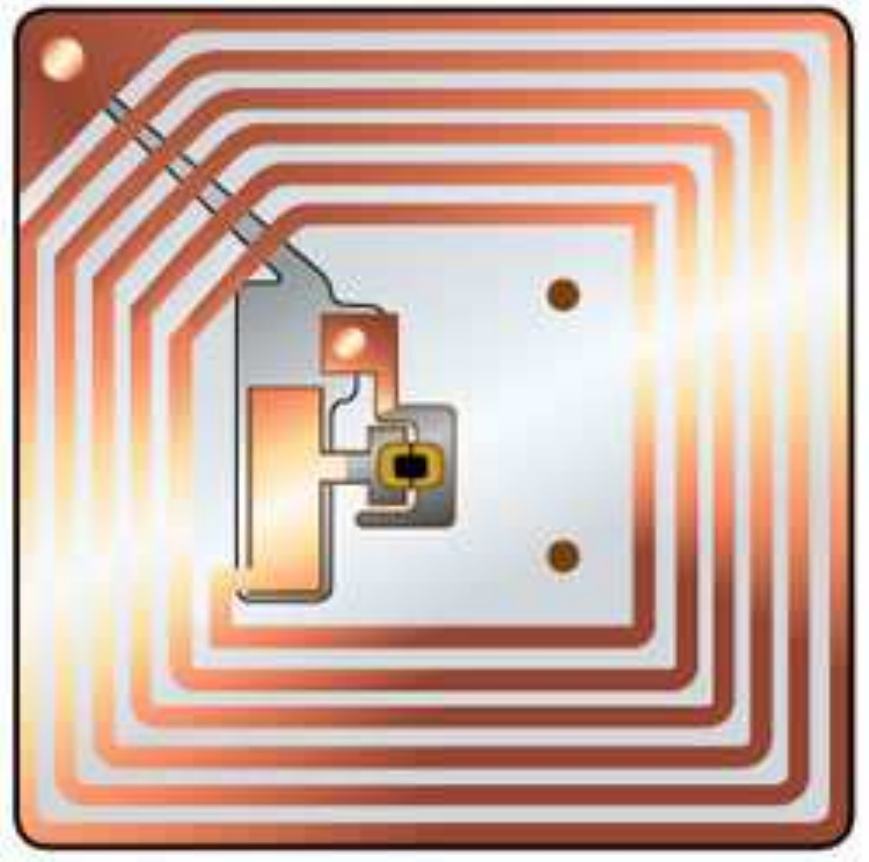}
 }
 \subfigure{
   \includegraphics[width=1.45in]
   {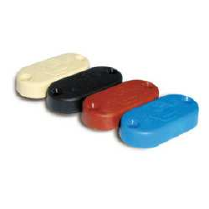}
 }
 \subfigure{
   \includegraphics[width=1.45in]
   {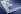}
 }
 \subfigure{
   \includegraphics[width=1.45in]
   {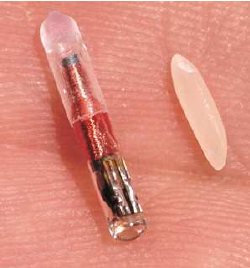}
 }
 \subfigure{
   \includegraphics[width=1.45in]
   {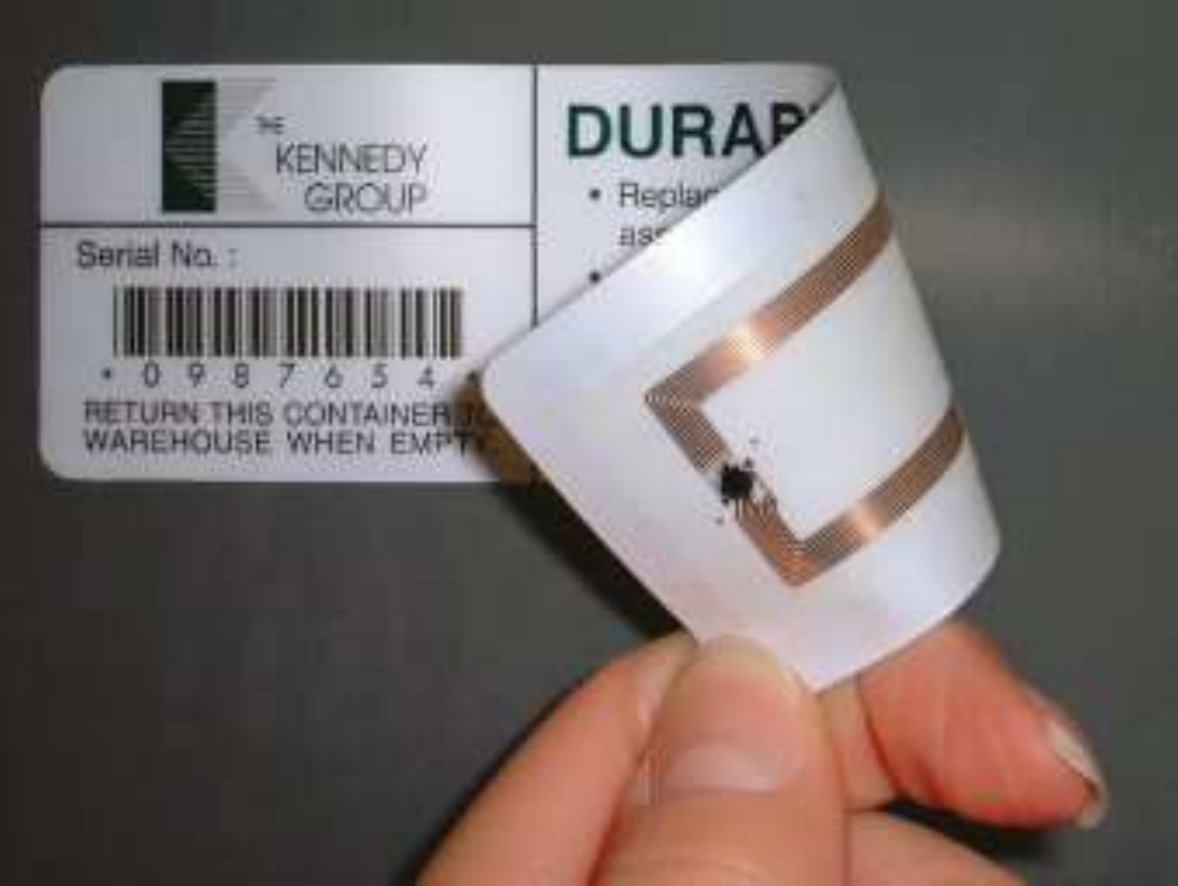}
 }
 \subfigure{
   \includegraphics[width=1.45in]
   {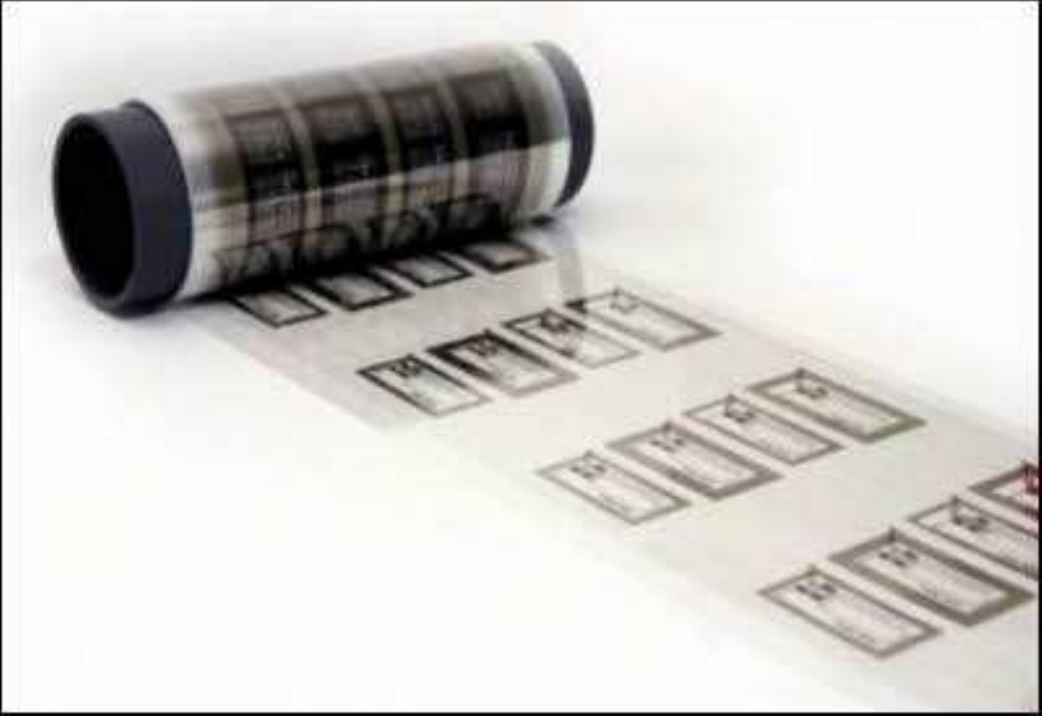}
 }
\caption{Pictures of some types of RFID tags}
\label{fig:tag_pictures}
\end{figure}

RFID tags may also be classified according to their processing power. A \emph{dumb} tag has no significant processing power, while \emph{smart} tags have on-board processors able to perform cryptographic operations~\cite{Adam20074}. Dumb tags are considered the heart of RFID systems. Manufacturers and retailers claim that reducing tags' cost is indispensable for the success of RFID systems. In some cases, sending a unique identifier would not necessarily be a problem. For instance, Canada, Israel, Japan, Belgium and The Netherlands, among other countries, require owners of pets to implant an RFID tag on their pets. These tags contain information that allow a fast and efficient localisation of pet owners in case of loss of their pets. Since those tags have a really short identification range and pets normally do not have enemies aimed at counterfeiting their identities, dumb tags could be the most practical option for this type of application. In turn, smart tags are used in those applications requiring some level of security and/or privacy, namely for e-passports, supply chain management or access control.

Considering that RFID systems rely on radio waves, tags operate in a well-defined frequency. There are four main bands: low frequency (LF), high frequency (HF), ultra-high frequency (UHF), and
microwave. The exact frequency varies depending on the application and the regulations in different countries. The frequency bands and the most common RFID system frequencies are listed in Table~\ref{tab:frequency}.

\begin{table}[!h]
\begin{center}
	\begin{tabular}{|p{3.5cm}|p{2cm}|p{6.5cm}|}
	\hline
	 \textbf{Frequency Band} &  \textbf{Operating Range} & \textbf{Applications}\\
	\hline
	 125kHz to 134kHz (LF) &  $\approx$ 0.5 Meters & Access control and Animal identification\\
	\hline
	 13.56MHz (HF) &  $\approx$ 1 Meters & Library books and Smart cards\\
	\hline
	 860MHz to 930MHz (UHF) &  $\approx$ 3 Meters &  Logistic and Parking access\\
	\hline
	 2.4GHz (Microwave) &  $\approx$ 10 Meters & Electronic toll collection and Airline baggage tracking\\
	\hline
	\end{tabular}
\end{center}

\caption{RFID frequency bands and characteristics.}\label{tab:frequency}

\end{table}

\subsubsection{RFID readers}

Typical \textit{transceivers} or RFID readers consist of a radio frequency
module, a control unit and a coupling element to interrogate RFID tags
via radio frequency communication. Readers may issue two types of challenge:
multicast and unicast. Multicast challenges are addressed to all tags in the
range of a reader whereas unicast challenges are addressed to specific tags.
In order to keep readers as simple as possible, they have, in general, an interface that allows them to forward the received data to a \textit{data processing subsystem, back-end database or server}. By doing so, readers delegate most of the computational effort to other computationally more powerful devices.

\subsubsection{Data processing subsystem}

The \emph{data processing subsystem} or \emph{server} is used to overcome the computational limitations of tags and readers. On the one hand, tags may not be able to store in their memory all the information required by readers. Thus, this information is usually stored in indexed databases. On the other hand, aimed at reducing the cost of RFID readers, cryptographic functions or processing data algorithms should rely on a data processing subsystem or server. It should be remarked that a secure connection between readers and back-end databases is generally assumed; anyway, secure communication between two computationally powerful devices does not belong to the problems addressed by the RFID technology. Refer to
Figure~\ref{fig:BasicScheme} for a graphical representation of the RFID components
and their basic relations/connections.

\begin{figure}[!t]
\centering
\includegraphics[scale=0.50,angle=270]{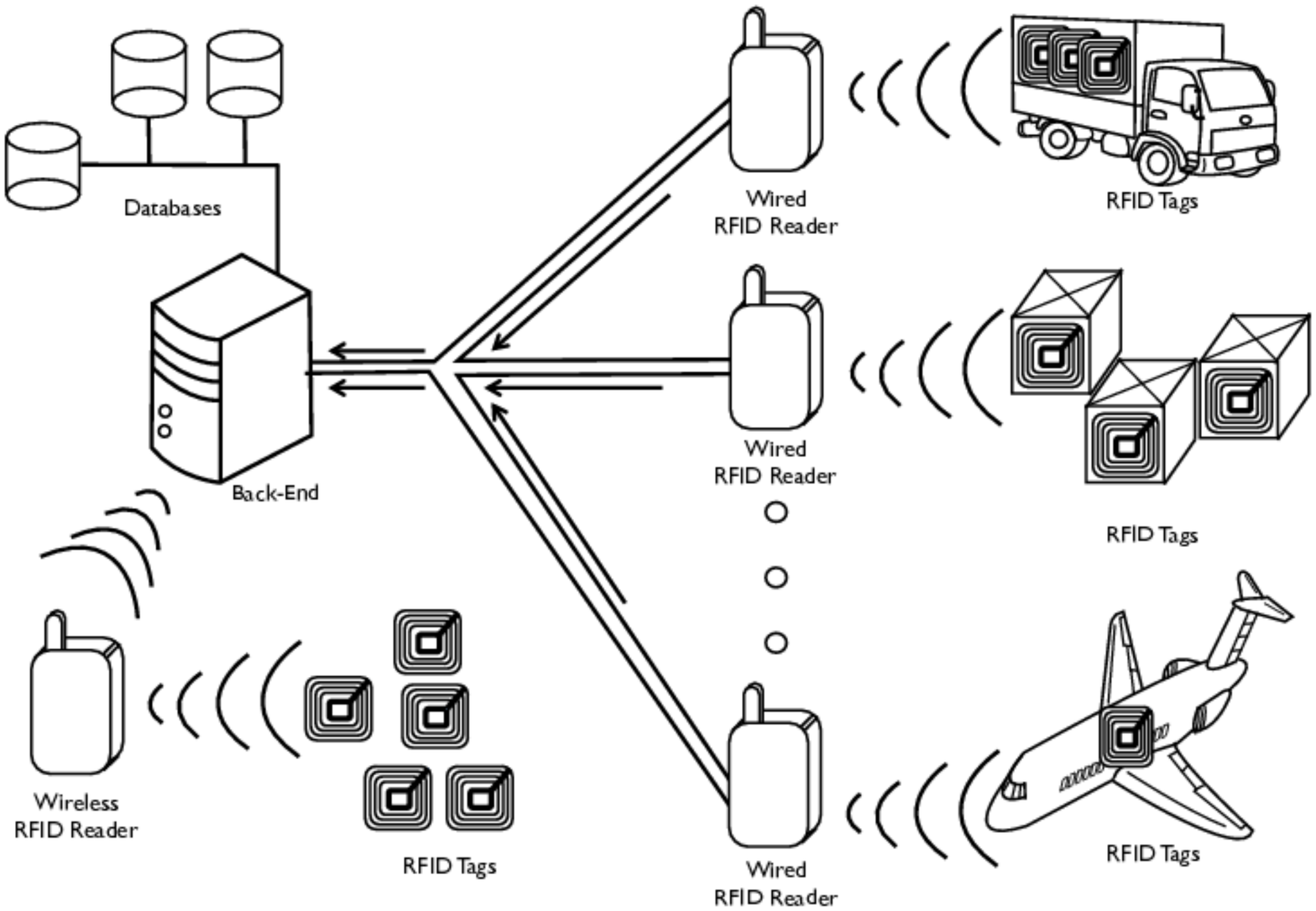}
\caption{Basic components of an RFID system. From left to right: a back-end,
RFID readers, and RFID tags. The back-end uses databases to store identification
information. RFID readers are used to query RFID tags (that can take a variety
of embodiments), retrieve their information, and forward it to the back-end
through a wireless or wired channel.}
\label{fig:BasicScheme}
\end{figure}

\subsection{RFID standards}

Nowadays, most technologies we use are governed by standards. Basically, these define the minimum requirements of some technology in order to achieve interoperability, which is particularly important in RFID systems. To illustrate the need for interoperability in the RFID technology it is important to understand the problems of supply chains. We may say that a supply chain management begins in a mine or a farm and it ends on a recycling or garbage plant~\cite{2283}. In between, the initial material is modified or processed from stage to stage, it may change hands from one owner to another, etc. In this globalised world, such material or item, presumably attached to an RFID tag, could travel around the world more than most people in their whole life (\emph{e.g.} from manufacturers to warehouses, from warehouses to points of sale, from points of sale to retailers, from retailers to customers, from customers to customers or second-hand retailers, etc). This means that RFID tags should be correctly read by everyone and everywhere, in the present and in the future, and without any restricted access or implementation, \emph{i.e.} RFID systems should be interoperable.

Continuing the work of Auto-ID Labs~\cite{auto_id_labs}, EPCGlobal is leading the development of industry-driven standards for the Electronic Product Code (EPC) to support the use of RFID systems~\cite{epcglobal}. Their task has been to specify frequencies, coupling methods, types of keying and modulation, information storage capacity, and modes of interoperability~\cite{28}. Table~\ref{tab:tag_classes} shows the classification of RFID tags according to the EPCGlobal organisation.

\begin{table}[!h]
  \begin{tabular}{l l }
    \hline
    \textbf{Class} & \textbf{Description} \\
    \hline
    Class 0 & Passive, read-only.\\
    \hline
    Class 0+ & Passive, write-once but using class 0 protocols.\\
    \hline
    Class I & Passive, write-once.\\
    \hline
    Class II & Passive, write-once with extras such as encryption.\\
    \hline
    Class III & Rewritable, semi-passive, integrated sensors.\\
    \hline
    Class IV & Rewritable, active, may communicate with other active tags.\\
    \hline
    Class V & Rewritable, active, can power and read other tags.\\
    \hline
  \end{tabular}
  \caption{EPC global tag classes.}\label{tab:tag_classes}
\end{table}

On the other hand, the International Organisation for Standardisation (ISO) has also created standards for RFID. Initially, there was some conflict between EPCGlobal and ISO specially due to the air interface protocol. At EPCGlobal, the ISO UHF protocol was thought to be too complex and costly. That is why they developed their own UHF protocol. Finally, in 2004, EPCglobal developed a second-generation protocol (Gen 2) aimed at creating a single, global standard that would be closer to the ISO standards and lastly accepted by ISO. Undoubtedly, this new generation has been the cornerstone of a massive deployment and global adoption of the RFID technology.

\subsection{Applications}

RFID technology has been characterised by its growing popularity. Consequently, a large and diverse number of RFID solutions are being used by more and more business companies. Not surprisingly, national governments have also noticed the benefits of RFID systems in their ordinary tasks, namely for passport control and document tracking. Therefore, it is hard to say exactly how many RFID systems are already deployed worldwide. However, it is clear that these systems are becoming more popular with each passing day.

\subsubsection{Identification}

Since the very beginning of the RFID technology during the Second World War, identification was its primary goal. Nowadays, the scenario is not so different; animal identification, inventory systems, human implants for identification of patients and drug control, are just a few examples of identification by radio frequency. Indeed, other RFID's features like tracking implicitly identify their targets, otherwise it would not be possible to track them.

\subsubsection{Tracking}

There exist several scenarios in which RFID systems are the most suitable for tracking (\emph{e.g.} indoor environments or for animal surveillance). Also, in comparison with other tracking systems like GPS or GSM, the RFID technology is considered much less costly. That is why tracking, together with identification, is considered one of the primary goals of RFID systems.

For tracking, tags operating at high frequency are usually required because they have a larger reading range. Those types of tags are used for tracking
in libraries or bookstores, pallet tracking, building access control,
  airline baggage tracking, and apparel and pharmaceutical items tracking (\emph{e.g.} in February 2008, the Emirates airline started a trial of RFID baggage tracking at London and Dubai airports~\cite{webster08}; in May 2007, Bear River Supply began to utilise ultrahigh-frequency
  identification tags to help monitor their agricultural equipment~\cite{bacheldo08}).

\subsubsection{Healthcare}

The healthcare industry has been heavily investing in RFID. The healthcare supply
chain, prevention of drug counterfeiting or patient safety, are just some examples of critical processes monitored by RFID. By doing so, patients of a hospital in England might avoid exposure to diseases caused by
infected equipment that was not properly tracked and classified~\cite{healthcare1}. Furthermore, discarded drug packaging will not be reusable by companies attempting to sell counterfeit pharmaceuticals, as noted by Colombian pharmacy chain Medicarte~\cite{healthcare2}. Indeed, it is expected that investments in RFID technology by the healthcare industry rise from $90\$$ million in 2006 to $2.1\$$ billion in 2016~\cite{healthcare3}.

\subsubsection{Electronic passports}

Electronic passports (e-passports) or passports with an embedded RFID tag have been introduced in many countries, including Malaysia (1998), New Zealand (2005), Belgium, The Netherlands (2005),
Norway (November 2005), Ireland (2006), Japan (2006), Pakistan,
Germany, Portugal, Poland (2006), Spain (August 2006), The United Kingdom,
Australia and the United States (2007), Serbia (July 2008) and Republic of Korea
(August 2008). Contrary to most RFID applications, RFID tags on passports are a sort of smart card rather than a low-cost tag. They are able to execute computationally complex public key cryptosystems with large key size and are tamper-proof. Also, plenty of information may be stored on the tag's memory such as, name, birthdate, biometric information, photo, etc. Such information may be contrasted with the information available on paper, thereby reducing the risk of passport forgery and fraud. However, several weaknesses have been found in e-passports~\cite{Nithyanand-2009-eprint}. Especially disturbing are those that allow an adversary to effectively clone an e-passport so that a reader cannot distinguish between a legitime and a cloned passport~\cite{Hlavac:2009:KAR:1617722.1617736,Blundo:2008:ISN:1462455.1462487}. Anyway, governments claim that cloning is not a big problem as the electronic information must match the physical characteristics of users. Furthermore, a third generation of e-passports has been released in which, to the best of our knowledge, no cloning attack has been reported yet.

\subsubsection{Transportation payments}

Public transportation payments with RFID cards is probably one of the first perceptible contacts we have with this technology. By this solution neither we need to have coins in our pockets nor the bus drivers need to regularly manage and change cash. Hence, the bus drivers' workload decreases, thereby reducing the risk of a traffic accident due to distractions and increasing the compliance with the schedule.

Conforming to the Calypso\footnote{Calypso is the international
  electronic ticketing standard for microprocessor contactless smartcards. It
  ensures multi-sources of compatible products, and makes possible the
  inter-operability between several transport operators in the same area.} (RFID)
  international standard, several countries in Europe and America use RFID passes for public transport systems. In Asia, in particulary Hong Kong, other type of RFID cards, called \emph{Octopus Cards}, are also being used for transport systems. Those cards also have grown to be similar to a cash card, and they may still be used in vending machines, fast-food restaurants and supermarkets. Many other public transport payments based on RFID exist (\emph{e.g.} the Moscow Metro
  pay system or the bike sharing system that prevents bicycle theft in Barcelona).

\subsection{RFID challenges}

There are many challenges associated with the deployment of RFID systems (\emph{e.g.} false or missing reads due to radio wave corruption, scalability, security and privacy, antenna design, deployment cost, among others). However, there are other challenges that may not seem so obvious. The introduction of a new order of things might create a maelstrom of uncertainty. For many industries, RFID deployment will change their business process, forcing new investments on personal training, infrastructure, testing, etc. For example, the McCarran International Airport in Las Vegas had to invest around $125, 000, 000 \$$ to RFID-enable its baggage-tracking system~\cite{Mary-2005-RFIDJournal}. Therefore, companies are requiring to carefully evaluate the economic viability of what may represent a big initial investment of money.

On the other hand, although an RFID system provides plenty of data essential to control and understand business processes, applications like supply chain management or real-time tracking may generate such a huge volume of information that could not be handled by traditional transactional databases (\emph{e.g.} it is predicted that WalMart will generate over $7$ terabytes of operational RFID data per day~\cite{Palmer-2011}). Therefore, software architectures and back-end databases should be rethought for the collection, correlation, filtering, and cleansing of RFID data.

Strongly related with technical and deployment details, three different challenges, namely security, privacy and, scalability, are the main subject of discussion in this dissertation. Almost every object is likely to be attached to an RFID tag. Therefore, billions of tags will need to be managed efficiently and in a scalable way. On the other hand, due to the wireless nature and the computational constraints of RFID tags, guaranteeing the security of tags' data and the privacy of tags' bearers is a challenging task. The privacy threats grow if we consider all personal data surrounding the huge amount of information collected from tags. If such data are not properly treated, sensitive information might be disclosed without the awareness of RFID's users. This means that the need for efficient and scalable privacy-preserving methods for microdata and trajectories increases with the massive deployment of RFID solutions.

\section{Security, privacy and scalability issues in RFID identification protocols}

The rapid proliferation of RFID solutions strongly supports the vision of
ubiquitous computing, in which tags interacting with readers throughout our
everyday life improve our experiences with the environment. Consequently,
in most applications, readers must be able to identify one or several tags among
a set of millions or billions. This scenario characterises an important property that an RFID identification
protocol should meet: scalability.

As for most identification systems, being secure and private are another two mandatory properties of RFID systems. These two features are even more relevant in the RFID context due to the insecure and easily accessible communication channel between tags and readers. Generally speaking, security means that data stored in a tag's memory should be accessed only by authorised parties and that impersonating or counterfeiting a tag may be achieved just with a negligible probability. On the other hand, privacy-preservation may be defined as the ability of tags to generate uncorrelated identification messages.

\subsection{Security}

RFID systems are subject to plenty of attacks, from attacks operating on the physical layer to attacks exploiting weaknesses on those protocols executed at the application layer~\cite{DBLP:journals/isf/MitrokotsaRT10, DBLP:conf/iwrt/MitrokotsaRT08, springerlink:10.1007/11599548_7}. Physical attacks may be as simple as wrapping an RFID tag in aluminum foil, which potentially causes denial of services (DoS) because readers will be not able to communicate with such tag. Other physical attacks are more sophisticated (\emph{e.g.} jamming attacks that permanently damage radio devices
or side-channel attacks that obtain information from the physical implementation of cryptosystems). However, in the present dissertation, we focus on adversaries aimed at breaking the identification/authentication schemes by using theoretical weaknesses of such algorithms. To do so, we assume that the adversary can observe, block, modify, and inject messages in the communication between a tag and a reader. Furthermore, as tags are not tamper-resistant, we assume an adversary able to clone and tamper with any RFID tag.

The most relevant attack to RFID systems is the so-called \emph{spoofing} or \emph{impersonation} attack. In this attack, an adversary is able to clone a tag without physically replicating it. By doing so, the adversary gains the privileges of such tag, which is considered an important security threat for almost every RFID applications. The worst situation occurs when the adversary is able to break the cryptosystem used during the authentication process (\emph{total break}), \emph{i.e.} the adversary gains knowledge of the authentication protocols and the secrets. In other cases, the adversary does not even need to spend too much time breaking the cryptographic protocol. Instead, the adversary could impersonate a tag by replaying and/or manipulating some tag's responses recorded from past transactions (\emph{forgery}). Although these attacks have been successfully
thwarted by lightweight and symmetric key cryptography suitable for low-cost RFID tags~\cite{WeisSRE-2003-spc}, there still exist open issues when privacy and scalability must be also considered.

\subsection{Privacy}

There exist two main privacy concerns in RFID systems: information leakage and traceability. Information leakage is potentially dangerous because tags may reveal sensitive information about products (\emph{e.g.} the name of drugs or the price of expensive products). Such data may be used for quick, easy, and low-cost profiling of individuals, or even for industrial espionage. The basic idea to prevent information leakage in RFID systems is to move all the tags' data to one or several servers. By doing so, only authorised parties may retrieve those data when required. However, this may not prevent traceability. For instance, a tag sending its unique identifier does not reveal trivial information about the object to which is attached, but it is traceable. To thwart traceability, readers and tags should exchange \emph{fresh} information at each identification so as to make the response of two different tags \emph{indistinguishable}.

The challenge is that indistinguishability is an application-dependent concept where the abilities of adversaries, tag's owners, physical constraints, etc, must be taken into account in order to provide a fair privacy definition for RFID systems. This is why different privacy models for RFID have been defined~\cite{Ng:2008:RPM:1462455.1462478, springerlink:10.1007/978-3-540-88313-5-18, Avoine-2005-techrep, Vaudenay:2007:PMR:1781454.1781461, Le07forward-securerfid}. Among them, we recall two well-known notions of privacy proposed by Avoine in~\cite{Avoine-2005-techrep}:

\begin{definition}[Universal untraceability]\label{def:universal}
Universal untraceability is achieved when any pair of tag's responses, separated by a successfully identification with a legitimate reader, cannot be correlated with high confidence by an adversary.
\end{definition}

\begin{definition}[Existential untraceability]\label{def:existential}
Existential untraceability is achieved when any pair of tag's responses cannot be correlated with high confidence by an adversary.
\end{definition}

Intuitively, existential untraceability is stronger than universal untraceability. Note that the latter ensures privacy against passive adversaries only. That is why protocols that achieve universal untraceability are usually referred as \emph{passively private}, whilst those protocols achieving existential untraceability are referred as \emph{actively private}~\cite{DBLP:journals/comcom/AlomairP10}.

There exist other notions of privacy in RFID systems like forward and backward untraceability. Both notions rely on the fact that RFID tags are not tamper-resistant and thus, an adversary may be able to get full access to the internal state of a tag. Informally, backward and forward untraceability ensure that revealing the internal state of a tag cannot help an adversary to identify previous or future transactions of such tag. However, both are beyond the scope of this dissertation (interested readers may refer to~\cite{Berbain:2009:EFP:1653662.1653669, springerlink:10.1007/s11277-010-0001-0, 4144831, Song:2008:RAP:1352533.1352556, Cai:2009:AIR:1514274.1514282, 4815081, springerlink:10.1007/s11277-010-0001-0, Berbain:2009:EFP:1653662.1653669, Le07forward-securerfid}).

\subsection{Why is secure, private, and scalable identification hard?}

 \raggedbottom

An RFID identification system where tags send their unique identifier in plain text to readers is scalable (\emph{e.g.} EPC Radio Frequency Identity Protocols Class-1 Generation-2 UHF	RFID~\cite{epcglobal, class_1_gen_2}). In some way, this is how the barcode systems work. However, from the security point of view, it would be quite easy to counterfeit an RFID tag just by building a device able to replay the unique identifier of this tag, which was previously eavesdropped or maybe read from the tag's embodiment. Some RFID manufacturers argue that, in any case, counterfeiting an RFID tag is much more difficult than counterfeiting a barcode label. Actually, this is true for RFID tags intended as the replacement of barcode labels. But RFID systems aimed at being an active part of the future of pervasive computing need anti-counterfeiting measures; they need entity authentication~\cite{Menezes:1996:HAC:548089}.

Public key cryptography (PKC) is known to achieve private and scalable authentication. Actually, most of the applications we use nowadays are protected by PKCs, namely secure remote login (SSH), digital signatures, internet key exchange (IKE), digital cash and secure transparent voting. That is why so many efforts have been devoted to designing and implementing asymmetric cryptosystems
such as Elliptic Curve Cryptosystem (ECC)~\cite{Batina07public-keycryptography} and N-th Degree Truncated Polynomial Ring (NTRU)~\cite{springerlink:10.1007/3-540-36563-X_9}, in RFID tags. However, they require a high number of logic gates to be implemented in tags. Therefore, the price of tags
will increase drastically to accommodate those cryptosystems in RFID tags. Even though other proposals~\cite{springerlink:10.1007/11967668_24} reduce the number of required logic gates by performing some pre-computation and storing partial results in tags, they increase the memory usage to around 1700 bits. It is an open problem whether public key cryptography will be suitable for low-cost RFID tags. Therefore, most RFID identification protocols are based on symmetric key cryptography rather than on public key cryptography.

However, symmetric key cryptography does not satisfy all the requirements of RFID systems because, in general, it
draws scalability problems on the server's side. Since tags are not tamper-resistant, each tag must contain a unique and private key with which encrypt its response. The paradox is that, in order to determine the tag's identity the server needs to decrypt the message using the tag's key, but retrieving the tag's key is only possible when the server knows the tag's identity. Consequently, the server should perform an exhaustive search looking for the proper key to decrypt the message. Several protocols overcome this scalability problem by performing an update phase after the identification process (stateful protocols)~\cite{NingLY-2011-cje, AlomairLP-2010-jcs, Li:2007:PRC:1229285.1229318, 4912773, springerlink:10.1007/11601494_13, ChatmonVB-2006-techrep, Ateniese:2005:URT:1102120.1102134, Karthikeyan:2005:RSW:1102219.1102229}. This means that the tag and the reader should share and synchronously update the next identification message. However, it has been shown~\cite{YuSMS-2009-esorics} that this synchronisation process must be carefully designed in order to resist Denial of Services attacks. On the other hand,
 the update phase is not only useless against active adversaries, but also is rather inefficient on the tag's side. Note that a write operation in a tag may take roughly $16.7$ milliseconds, while a read operation just needs
 around $0.007$ milliseconds~\cite{5678437}. Indeed, writing in the tag's memory is a time consuming operation; roughly five times more time consuming than a classical AES-128 encryption, which needs around $2.8$ milliseconds.

In conclusion, three main techniques have been proposed for the private identification of RFID tags: (i) Public key cryptosystems that are secure, private, and scalable, but not suitable for low-cost tags. (ii)  Stateful protocols which are scalable, but they are not strong against active attacks and may be less efficient than other symmetric key protocols due to the writing operation in the memory of the tags. (iii) Stateless protocols based on symmetric key cryptography that could be lightweight, private, and secure, but not scalable. As a result, designing lightweight, secure, private, yet scalable, RFID identification protocols still is a challenging task.

\subsection{Advances in RFID identification protocols}

As stated in~\cite{Chien-2007-ieeetdsc}, a tag can be classified according to the operations it supports. \emph{High-cost} tags are those that support on-board conventional cryptography like symmetric encryption and public key cryptography. In turn, \emph{simple} tags are also considered high-cost tags, but they only support random number generators and one-way hash functions. Likewise, low-cost tags can be classified as \emph{lightweight} tags or \emph{ultralightweight} tags. Both are able to compute simple bitwise operations like XOR, AND, OR, etc, but the former also support a random number generator (RNG) and simple functions like a cyclic redundancy code (CRC) checksum. Undoubtedly, low-cost and simple tags, intended as the replacement of the barcode labels, represent the greatest challenge in terms of security and privacy preservation.

\subsubsection{Identification protocols for low-cost tags}

 \raggedbottom
 
Several efforts have been made in order to achieve some level of security in low-cost tags~\cite{Wong:2006:CAR:1143147.1143152, citeulike:1141061, Burmester:2008:SEG:1788857.1788887}. One of the first proposals in this direction is due to Duc {\em et al.}~\cite{citeulike:1141061}. They designed a protocol where messages are encrypted using a CRC-16 function and randomised by an updating key process. Although this protocol is not resilient to desynchronisation attacks~\cite{Chien:2007:MAP:1222669.1222985}, its main weakness lies in the linearity of the CRC-16 function. Indeed, Burmester and Medeiros~\cite{Burmester:2008:SEG:1788857.1788887} show how to successfully implement an impersonation attack by eavesdropping only one session of the protocol. In the same article, four protocols with different levels of privacy are proposed. The first encrypts messages using the RNG defined in the EPCGlobal2 standard seeded with a key shared by the tag and the reader. The security of this protocol is based on the statistical behaviour of the RNG. However, the key is of size 16 bits; therefore, an exhaustive search on all possible $2^{16}$ key values can be enough to recover the key. On the other hand, EPCGen2 does not specify any protection of the RNG against the \emph{related-key} attack, in which it is possible to find a correlation between a sequence of outputs given by the RNG defined in the EPCGlobal2~\cite{Burmester:2008:SEG:1788857.1788887}. Therefore, an adversary could be able to disambiguate tags that respond with pseudonyms drawn from a EPCGlobal2 RNG complaint. The solution proposed in~\cite{Burmester:2008:SEG:1788857.1788887} is to build a Pseudo Random Function (PRF) from a RNG~\cite{DBLP:journals/jacm/GoldreichGM86}. The new PRF has an input size of 32 bits and it is defined by recursively executing a RNG 16 times.


Other proposals do not even consider tags generating random numbers. In such protocols, the randomness on the tag's side is provided by readers. It should be noted that, those protocols either do not provide anonymity or ensure privacy by updating tags' internal state. In 2006, the first three ultralightweight protocols were proposed: M$^{\mbox{2}}$AP~\cite{DBLP:conf/uic/Peris-LopezCER06}, EMAP~\cite{Peris-lopez06emap:an} and LMAP~\cite{PerisHER-2006-rfidsec}. Although it was a step forward on the security of low-cost tags, and many other ultralightweight protocols~\cite{Burmester:2008:SEG:1788857.1788887, Peris-Lopez:2009:AUC:1530270.1530277, Chien-2007-ieeetdsc} have been proposed so far, all of them have been proven to be insecure~\cite{YehL-2010-iasl, AvoineCM-2010-rfidsec, HernandezEPQ-2009-wcc, Phan-2008-ieeetdsc, 4159809, Li07securityanalysis}. According to Peris-Lopez {\em et al.}~\cite{Peris-Lopez:2009:AUC:1530270.1530277}, the main weakness of these protocols is the use of triangular functions like AND, XOR, etc. This problem was detected by Chien~\cite{Chien-2007-ieeetdsc} and he incorporated a left rotation operation (which is non triangular) to his proposal named SASI. Nevertheless, the SASI protocol has other weaknesses that can be found in~\cite{AvoineCM-2010-rfidsec, HernandezEPQ-2009-wcc, Phan-2008-ieeetdsc, Peris-Lopez:2009:AUC:1530270.1530277}. As a consequence, Peris-Lopez {\em et al.} designed the Gossamer protocol~\cite{Peris-Lopez:2009:AUC:1530270.1530277} aimed at being more secure than previous ultralightweight protocols, though more computationally expensive. In a similar line, Juels proposed a protocol~\cite{Juels03minimalistcryptography} where each tag has a list of one-time pads that together with the tag's keys identify the tag. The protocol is minimalist in the sense that involves only low-cost operations like: rudimentary memory management, string comparisons, and a basic XOR. However, the security of this protocol depends on the size of a list that should be updated at each session.

A completely different approach to the security of RFID systems was proposed by Juels in~\cite{springerlink:10.1007/11535218}. Juels adopts the human-to-computer authentication protocol designed
by Hopper and Blum (HB)~\cite{Hopper01asecure}, and shows it can be practical for low-cost
pervasive devices. HB is a probabilistic protocol consisting of several rounds (around 128 according to Hopper and Blum~\cite{Hopper01asecure}). In each round, the verifier sends $n$ bits of challenge and the prover responds correctly to each bit with a probability greater than $1/2$. At the end of the protocol, the verifier decides whether the prover gave a sufficient number of correct bits of response. The HB protocol can be also considered an ultralightweight protocol. Later, Juels modified slightly the HB protocol proposing a new protocol~\cite{springerlink:10.1007/11535218} (HB$^+$) that claims to be secure against active adversaries. Although HB$^+$ is suitable for EPC-Gen 2 tags, it has
a high false rejection rate (for 80 bits of security, the false rejection rate is estimated
at $44\%$~\cite{DBLP:conf/eurocrypt/GilbertRS08, cryptoeprint:2005:237}). Furthermore, the communication overheads increase linearly with the security parameter, which may be chosen around the $80$ bits~\cite{Fossorier06anovel}. Another protocol based on an NP-complete problem was proposed by Castelluccia and Soos~\cite{CastelluscciaS-2007-rfidsec} at RFIDSec'07. Their protocol (ProbIP) is based on the hardness of the boolean satisfiability problem (SAT), which is proven to be in the NP class of complexity. However, the protocol is neither private nor secure as was shown in~\cite{Ouafi:2008:PRR:1788494.1788513}.

A recent proposal \cite{Wu:2009:HSR:1574927.1574959} based on the hardness of the \emph{noisy polynomial interpolation} problem aims to be private and scalable. However, this protocol presents some shortcomings: (i) it has been shown in~\cite{Bleichenbacher:2000:NPI:1756169.1756175} that the \emph{noisy polynomial interpolation} problem can be easier than expected, (ii) the server needs to solve $mb$ polynomials of degree $k$ where $m$, $b$, and $k$, are predefined security parameters and, considering that typical values for m and b are 16 and 8 respectively \cite{Wu:2009:HSR:1574927.1574959}, then, the server needs to solve around 128 polynomials which can be considered still too heavy, (iii) tags should be designed with about 10,000 gates more than regular tags capable of hashing computation.

\subsubsection{Identification protocols for simple tags}

In general, lightweight RFID identification protocols use a pseudo random number generator and a symmetric key cryptographic primitive on the tag's side (\emph{e.g.} a hash function or some lightweight block cipher). In this vein, the first AES implementation devoted to RFID tags was proposed in 2004~\cite{FeldhoferDW-2004-ches, FeldhoferWR-2005-ieeproceedings} (other cryptographic primitives may be found in~\cite{4253019}). However, these implementations require no less than 3400 logic gates, which is beyond the capabilities of extremely constrained devices such as RFID tags~\cite{DBLP:conf/ifip6-8/Peris-LopezHER06}. That is why several other attempts have been made in order to find cryptographic primitives designed specifically for low-cost RFID tags. To the best of our knowledge, one of the most relevant block ciphers dedicated to RFID tags is PRESENT~\cite{Rolfes:2008:UIS:1430796.1430803, BogdanovKLPPRSV-2007-ches}. Surprisingly, PRESENT only requires between 1000 and 1600 logic gates depending on its variants. Undoubtedly, this improvement supports the rapid proliferation of RFID identification protocols based on symmetric key cryptography. Nevertheless, even considering that those protocols are suitable for very resource-constrained RFID tags, they should face the scalability issues inherent to these type of proposals. Next, we discuss some of those proposals according to their time complexity on the server's side.

\subsubsection{Lightweight protocols with linear time complexity}

The Improved Randomised
Hash-lock Scheme~\cite{JuelsW-2007-percom, Juels:2009:DSP:1609956.1609963} is a popular RFID identification protocol due to its strong
privacy and security properties at a low cost; it only uses a pseudo-random function
generator and a hash function on the tag's side. This protocol works as follows.
The reader sends a random nonce $r_1$ to the tag. Upon reception, the tag
generates a nonce $r_2$ and computes the response $r = h(r_1, r_2, ID)$ where $ID$
is its identifier and $h(...)$ is a one-way hash function. Finally, the reader
receives both the response $r$ and the nonce $r_2$, with which it performs an
exhaustive search on its database looking for an identifier $ID_i$ such that
$r = h(r_1, r_2, ID_i)$. Figure~\ref{fig:irhl} shows a detailed description of this protocol.
Although the improved randomised hash-lock scheme~\cite{JuelsW-2007-percom, Juels:2009:DSP:1609956.1609963} is a private,
and secure, RFID authentication protocol, it is not acceptable when a large
number of tags should be managed (\emph{e.g.} manufacturing processes). Note that this protocol performs an exhaustive search in the database in order to identify a tag. That is why several other protocols based on hash functions have been proposed in order to reduce this linear time complexity.

\begin{figure}[tb]
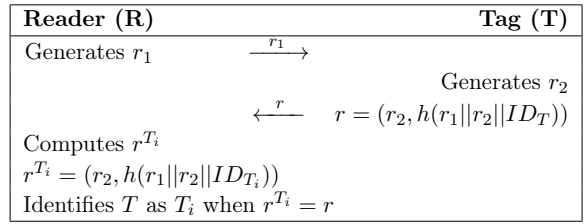


\centering
\scalebox{0.8}{
\begin{tabular}{|ccr|}
\hline
\textbf{Reader (R)}&&\textbf{Tag (T)}\\
\hline
Generates $r_1$ & $\hspace*{1cm}$ $\xrightarrow{\ \ r_1 \ \ }$&\\
& &	Generates $r_2$ \\
& $\hspace*{1cm}$ $\xleftarrow{\ \ r \ \ }$& $r=(r_2, h(r_1||r_2||ID_T))$\\
\multicolumn{3}{|l|}{Computes $r^{T_i}$}\\
\multicolumn{3}{|l|}{$r^{T_i}=(r_2, h(r_1||r_2||ID_{T_i}))$}\\
\multicolumn{3}{|l|}{Identifies $T$ as $T_i$ when $r^{T_i}=r$}\\
\hline
\end{tabular}
}\caption{Scheme of the Improved Randomised Hash-locks Protocol}\label{fig:irhl}
\end{figure}

Another hash-based protocol was proposed by Ohkubo {\em et al.}~\cite{Koutarou03cryptographicapproach}. This protocol uses hash-chains in order to guarantee forward secrecy. However, a high memory consumption
and a linear time complexity $O(N)$ are its main drawbacks. Although Avoine
{\em et al.}~\cite{DBLP:conf/sacrypt/AvoineDO05, Avoine05ascalable} reduced this time complexity to $O(N^{2/3})$ by applying a time-memory
trade-off, their improvement demands even more memory on the database than the original protocol~\cite{Koutarou03cryptographicapproach}. Another protocol based on hash-chains and resistant to denial-of-service attacks (DoS) is proposed in~\cite{4144831}. This protocol also achieves forward secrecy, but it presents some privacy issues as shown in~\cite{Ouafi:2008:PRR:1788494.1788513}.

In order to increase efficiency and reduce resource requirements, some protocols use a counter instead of a pseudo-random number generator on the tag's side~\cite{Choi05efficientrfid}. By this technique, tags may dedicate more logic gates to the encryption function and privacy could be guaranteed without updating the key material. However, this type of protocols may be vulnerable to impersonation attacks~\cite{springerlink:10.1007/11807964_27}.

\subsubsection{Lightweight protocols with logarithmic time complexity}

When looking for scalability, the tree-based protocol proposed by Molnar
and Wagner~\cite{Molnar:2004:PSL:1030083.1030112} is considered a secure and highly scalable protocol. It achieves a time complexity in the identification process of $O(b \log^N_b)$ where $N$ is the number of tags and $b$ is the branching factor of the
tree used to store the tag's identifiers. The idea is that each tag in the system is represented by a unique path in the tree, which is simply defined as a sequence of nodes. Then, each tag $T_i$ contains a unique secret key $k_i$ and also contains the keys of the nodes that represent its path in the tree (cf. Figure~\ref{fig:tree1}).

\begin{figure}[!ht]
  \begin{center}
     \includegraphics{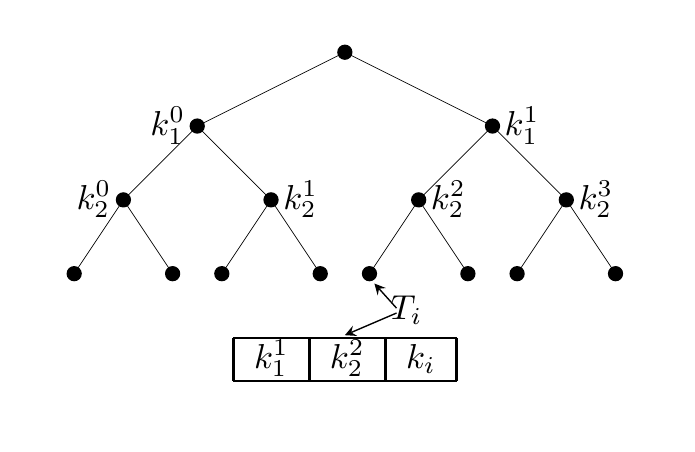}
  \end{center}
  \caption{A tree of depth 3 and branching factor 2. The tag $T_i$ is represented by the path $k_1^1, k_2^2$ and by its unique and secret key $k_i$.}
  \label{fig:tree1}
\end{figure}

Undoubtedly, the unique secret key of each tag is enough to identify it as the improved randomised hash lock~\cite{JuelsW-2007-percom, Juels:2009:DSP:1609956.1609963} does. However, the tree-based protocol uses the path's keys in order to rapidly discard large sets of tags whose keys do not match the received response. It may do so because tags allow readers to check their paths from the root to the leaf. Consequently, assuming that a reader knows a sub-path that matches the tag's response, it may discard all those tags that do not match this sub-path. However, since tags share sub-paths and hence they share keys, compromising some tags is enough to trace other tags in the system~\cite{DBLP:conf/sacrypt/AvoineDO05, Nohl06quantifyinginformation}. Although a trade-off between privacy and efficiency may be achieved considering the privacy measure proposed in~\cite{DBLP:conf/sec/NohlE08}, the tree-based protocols are considered non-private.

There exist two other tree-based protocols aimed at resisting compromise~\cite{LuLHHN-2007-percom, 10.1109/SITIS.2007.51}. However, both use an update phase in order to change the tag's keys in each successful execution of the protocol. As shown below, protocols based on updating the key material do not need complex structures, like trees, in order to be scalable. Furthermore, as shown in~\cite{10.1109/SITIS.2007.51}, synchronising keys that are shared by several tags is challenging.

\subsubsection{Lightweight protocols with sub-linear time complexity}

A similar idea to the tree-based protocols yields the group-based protocols~\cite{4351808}. In these protocols each
tag belongs to a group and has two keys: a group key ($GK$) and an identification
key ($IDK$). Once the tag receives a nonce $r_1$ from the reader, it generates another nonce $r_2$ and sends back $E_{GK}(r_1, r_2, ID)$ and $E_{IDK}(r_1, r_2)$, where $E$ is a symmetric key encryption function and $ID$ is the tag's identifier. The server iterates over all the group's keys until the decryption of $E_{GK}(r_1, r_2, ID)$ succeeds. If so, it recovers $ID$ and then it just needs to check the decryption of $E_{IDK}(r_1, r_2)$.
The time complexity of this protocol is
$O(\frac{N}{k})$ where $k$ is the number of tags in each group.
However, if a tag $T$ is compromised by an attacker, she will be able to recover the identity of
every tag belonging to the group of $T$ just by using the group key.

Another type of group-based protocol was proposed in~\cite{CheonHT-2009-eprint}. However, this protocol is different in the sense that a ``meet-in-the-middle'' strategy is used to efficiently identify tags. By this strategy, they reduce the reader computation to $O(\sqrt{N}\log{N})$, which may be considered sub-linear. However, this protocol is not resistant to compromising attacks either.

\subsubsection{Lightweight protocols with constant time complexity}

Privacy-preserving RFID identification protocols with constant time complexity are based, in general, on synchronisation between tags and readers. The idea is that any response coming from a legitimate tag is somehow expected by the reader. However, in order to preserve privacy, all the responses might be \emph{fresh} at each session. Therefore, at some point, tags should update their key material or renew their next responses. At the same time, the reader must be sure whether and how the tag updated its internal state. Otherwise, the reader should be provided with a mechanism to identify legitimate tags even when their internal states are unknown or unexpected~\cite{YuSMS-2009-esorics}. Note that this \emph{desynchronisation} between tags and readers could easily occur due to communication failures or active attacks to the protocol.

The basic idea to keep a tag and a server synchronised is to design a mutual authentication protocol so that both update the key material once they have mutually authenticated~\cite{1607559, 1276922, Avoine:2005, Chien:2007:MAP:1222669.1222985, Burmester:2008:ARA:1457161.1457162, citeulike:2632811}. In any case, the server should keep the last correct key used by every tag in order to resist desynchronisation attacks. The problem with this approach is that tags are traceable in isolated environments, \emph{i.e.} as long as a tag has not been identified by a legitimate reader, it will send the same response to any reader's interrogation. Note that this type of protocols are passively private only (cf. Definition~\ref{def:universal}).

To tackle this problem, some protocols consider the scalability and privacy issues as a matter of agreement between tags and readers~\cite{5318888, 6014260}. This means that a reader is able to identify a tag in constant time only if
the reader was the last one who interrogated such a tag. Otherwise, the tag's response looks random for the reader and thus the tag cannot be identified in constant time. Therefore, this type of protocols is actively private (cf. Definition~\ref{def:existential}), but not \emph{unconditionally} scalable.

Both passively and actively private approaches may be improved by pre-computing more than one future tag's responses~\cite{LimKwo06}. Typically, in these approaches tag's responses are based on hash chains~\cite{6014260, FernandezCV-2010-setop}. Then, the server may efficiently identify a tag because it had stored \emph{enough} values of the hash chain used by this tag. Therefore, by increasing memory requirements on the server, both privacy and scalability may be improved. However, it should be noticed that the improvement is achieved by demanding more of a resource that is already quite constrained in tags and servers~\cite{Palmer-2011}.

\subsubsection{Identification protocols for high-cost tags}

High-cost tags are similar to smart cards. They are far more expensive than low-cost RFID tags, which could cost as little as $0.05 \$$.
Nevertheless, there exist some applications requiring a high level of security and privacy in which high-cost tags are not only appropriate, but recommended (\emph{e.g.} passports or toll payment).

In such tags, there exists the possibility of implementing some public key cryptosystems, especially those requiring less computational capabilities on tags, namely Elliptic Curve Cryptosystems (ECC) or lattice-based cryptosystems.
A typical ECC implementation could need more than 30K logic gates~\cite{springerlink:10.1007/978-3-540-74442-9_30}, others are able to reduce the computational requirements between 10K and 18K logic gates~\cite{Kumar06arestandards, Hein:2009:ERR:1616707.1616744}, though.

In general, the reduction of computational requirements is achieved by optimising, manipulating, or removing some operations of the original proposal~\cite{1987}. The idea is to reduce as much as possible the computational requirements of public key cryptosystems while guaranteing security and privacy. However, as shown in~\cite{Bringer:2008:CER:1485310.1485324}, the task could be challenging. In a nutshell, they show that previous ECC proposals for RFID systems~\cite{4519370, 4911179} may be vulnerable to tracking and impersonation attacks.

Other approaches perform some pre-computation in order to reduce the computational overhead on the tag's side. In~\cite{Oren:2009:LPI:1514274.1514283}, Oren and Feldhofer improve Shamir's public key scheme~\cite{1994-2756} by replacing a
260-byte long pseudo-random sequence by a reversible
stream cipher of less than 300 bits. Another trade-off between efficiency and memory is shown by Hoffstein {\em et al.} in~\cite{springerlink:10.1007/3-540-36563-X_9}. They propose a lattice based cryptosystem referred to as NTRU, which claims to be faster than ECC during the signature and verification processes.

The well-known randomised Rabin encryption scheme has been also adapted to fit the RFID requirements~\cite{OrenF-2008-rfidsec}. Even though this protocol initially had some shortcomings, it was later improved in~\cite{4911191}. Finally, differently to classical public key cryptosystems, a lightweight identification protocol requiring around 3000 logic gates was proposed in~\cite{Kim:2007:RSP:1345530.1345612}. This protocol uses some ECC elements to strengthen RFID security. However, in order to decrease computational demands it does not provide a trapdoor function as ECC cryptosystems generally do. Hence, this approach presents the same scalability problem of other symmetric key cryptosystems.

\subsubsection{Other approaches}

According to the EPC standard, each tag contains a KILL password. A reader knowing the KILL password of a tag is able to disable this tag permanently. Therefore, after the shipping check-out process tags could be ``killed'' in order to preserve the privacy of their holders. Although this measure prevents privacy disclosure, it may not be practical in the long-term because tags cannot be reused. That is why Spiekermann and Berthold proposed a simple scheme so that users are able to disable/enable RFID tags when needed~\cite{springerlink:10.1007/0-387-23462-4_15}. A more sophisticated, yet complex idea, is proposed in~\cite{Engberg04zero-knowledgedevice}. When a tag enters the post-purchase phase, it supports the ability to change into \emph{privacy mode}. In this mode, tags only accept zero-knowledge proofs from legitimate devices.

There exist other proposals relying on re-encryption in which tags offload most of the computational effort during encryption to the readers or third devices~\cite{Saito04enhancingprivacy, Juels02squealingeuros:, Golle02universalre-encryption}. Those proposals are somehow based on updating the key material in tags. However, they use public key cryptography to re-encrypt the plaintext stored in tags. Therefore, readers knowing the proper private key can obtain the plaintext by decrypting only once, instead of several times as it is usually the case with symmetric key cryptography. Other examples of proposals using a third device are RFID guardian~\cite{Rieback05rfidguardian:}, RFID enhancer proxy~\cite{juels:proxies}, noisy tags~\cite{Castelluccia06noisytags:} and the blocker tag~\cite{Juels03theblocker, JuelsB-2004-wpes}. In addition to those proposals, there exist other proposals based on assumptions not considered by most RFID solutions. In particular, implementation of physically unclonable functions (PUF) in tags have been shown to be useful to cope with the tampering issues of RFID tags~\cite{Tuyls06rfid-tagsfor, Bolotnyy:2007:PUF:1263542.1263714, Bringer:2008:IPT:1432967.1432976}.

\section{Other issues in RFID systems}

Tracking and identifying are the main goals of an RFID system. In consequence, those challenges related to the identification process may seem much more relevant than others. However, RFID systems should face many other challenges depending on their application. For instance, RFID solutions for access control require tags to be in the near proximity of readers. However, the RFID technology is not able to measure the distance from readers to tags as the GPS technology can do. This clearly opens the challenge of designing \emph{distance-bounding protocols} dedicated to RFID tags~\cite{HanckeK-2005-securecomm}. Furthermore, the proper use of RFID data still is an open issue. Thanks to the RFID technology, the trajectories of individuals could be easily collected and released by supermarkets, hospitals, or amusement parks. Therefore, efficient \emph{trajectory anonymisation algorithms} are a need for protecting the privacy of RFID users.

\subsection{Distance checking}

In 1987, Desmedt, Goutier and Bengio~\cite{DBLP:conf/crypto/DesmedtGB87} presented an attack that defeated any authentication protocol. In this attack, called \emph{Mafia Fraud}, the adversary passes through the authentication process by simply relaying the messages between a legitimate reader (the verifier) and a legitimate tag (the prover). In that way, she does not need to modify or decrypt any exchanged data. Initially, this attack was thought to be rather unrealistic because the prover should actively participate in it. However, RFID tags respond to any reader request without any agreement or awareness of their bearer, a feature that clearly opens the door to this type of attack.

Actually, there exist some proofs of concept showing the feasibility of the mafia fraud. In 2005, Hancke showed that two colluders 50 meters apart can perform a mafia fraud attack through a radio channel~\cite{Hancke05apractical}. This is particulary dangerous because that distance is long enough to mount a mafia fraud attack in almost every payment or access control systems. Not surprisingly, this attack has been successfully applied to other technologies~\cite{Francis:2010:PNP:1926325.1926331, Hancke:2006:PAP:1130235.1130384, Kfir:2005:PVP:1128018.1128470, DBLP:conf/iscis/LeviCAKC04} namely, Bluetooth, contactless smart card, and NFC.

Another attack based on cheating the distance between provers and verifiers was introduced in 1993 by Brands and Chaum~\cite{188361}. In this attack, named \emph{Distance Fraud}, a dishonest prover claims to be closer to the verifier than she really is. Figures~\ref{fig:mafia_fraud} and~\ref{fig:distance_fraud} illustrate both mafia and distance fraud respectively. For both figures, the circle represented the maximum distance at which a prover should be authenticated. Formally, we may define both frauds as follows~\cite{AvoineBKLM-2011-jcs}:

\begin{definition}[Mafia fraud]
A mafia fraud is an attack where an adversary passes an authentication protocol by using a man-in-the-middle strategy between the reader and an honest tag located outside the neighbourhood of the verifier.
\end{definition}

\begin{definition}[Distance fraud]
A distance fraud is an attack where a dishonest and lonely prover claims to be in the neighbourhood of the verifier when actually she is not.
\end{definition}

\begin{figure}[!ht]
\begin{tabular}{cc}
\begin{minipage}[c]{0.5\linewidth}
\centering
\includegraphics{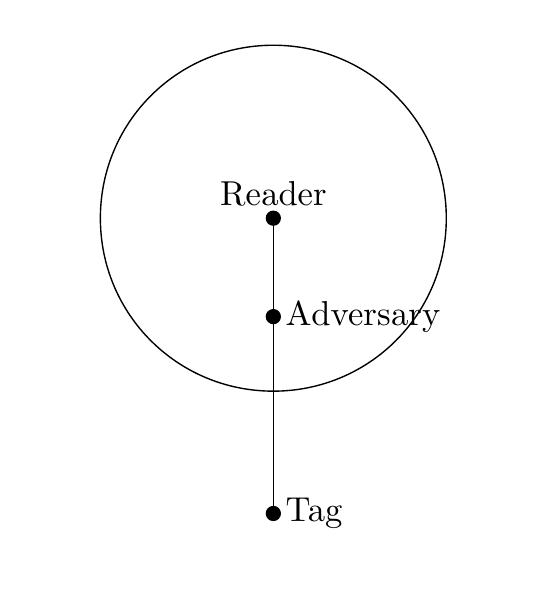}
\caption{Mafia fraud: an adversary trying to be authenticated by applying a man-in-the-middle attack.}
\label{fig:mafia_fraud}
\end{minipage}
&
\begin{minipage}[c]{0.5\linewidth}
\centering
\includegraphics{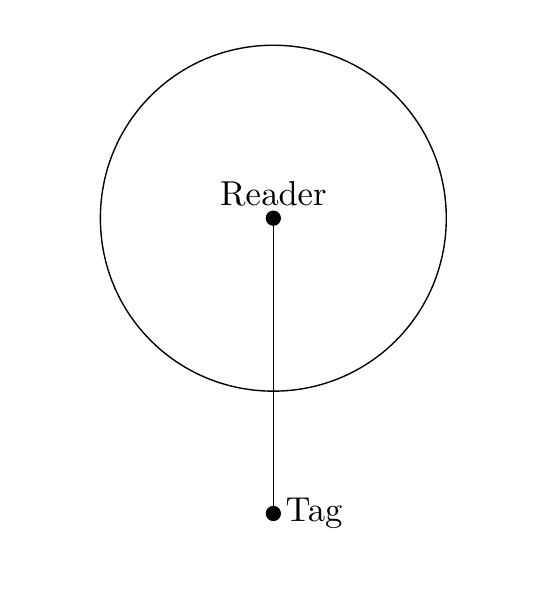}
\caption{Distance fraud: a legitimate prover is farther from the verifier than it is expected.}
\label{fig:distance_fraud}
\end{minipage}
\end{tabular}
\end{figure}

\subsubsection{RFID distance-bounding protocols}

In 1993, Brands and Chaum~\cite{188361} proposed a countermeasure that prevents such attack by computing an upper bound of the distance between the reader and the tag to authenticate: the \textit{distance-bounding protocol}. By doing so, mafia and distance frauds could not be completely prevented, but these protocols may effectively thwart them. However, it was not until 2005 that the first distance-bounding protocol dedicated to RFID came to the light~\cite{HanckeK-2005-securecomm}. The protocol is split in two phases: a \emph{slow phase}, in which reader and tag exchange two nonces, and carry on resource-consuming operations; followed by a \emph{fast phase} divided into $n$ rounds where, in each one, the reader measures the \emph{round trip time} (RTT) of the challenge/response process. Considering that radio waves cannot propagate faster than light, the reader is able to bound the distance between itself and the tag. These communications also provide an identity proof of the tag. Unfortunately, the adversary success probability regarding mafia and distance frauds is $(3/4)^n$ while one may expect
$(1/2)^n$ (the adversary's success probability at each round is expected to be $\frac{1}{2}$). Since then, several RFID distance-bounding protocols have been proposed in order to improve the resistance to both frauds.

Among all the RFID distance-bounding protocols, we differentiate two main families. (i) Those requiring an additional slow phase after the fast phase. This final phase may be used to sign the messages transmitted during the fast phase or to check any other information. (ii) Those that, closer to the Hancke and Kuhn proposal~\cite{HanckeK-2005-securecomm}, end the protocol after the fast phase.

Brands and Chaum~\cite{188361} proposed the first distance-bounding protocol relying on a signature after the fast phase. In the first slow phase, the prover commits to the verifier a sequence of $n$ bits $m_1, \cdots, m_n$. Then, during the fast phase, the verifier sends a challenge $c_i$ to the prover, who should reply with $r_i = c_i \oplus m_i$. Also, the prover concatenates and signs with his private key all the challenges and responses, \emph{i.e.} he sends to the verifier $Sign_{k}(c_1 || r_1 || \cdots || c_n || r_n)$. If some response $r_i$ delays more than a threshold $\Delta_t$, the verifier assumes the prover is out its neighbourhood. Finally, if the prover succeeds in all rounds, the verifier checks the received signature in order to authenticate the prover. This protocol is considered strong in the sense that both mafia and distance fraud attacks cannot succeed with probability higher than $\left ( \frac{1}{2} \right)^n$. 

There exist other distance-bounding protocols based on the Brands and Chaum proposal~\cite{KimAKSP-2008-icisc, springerlink:10.1007/0-387-25660-1_15}. Since both distance and mafia fraud resistance cannot be improved, those protocols aim at improving the resistance to a new type of fraud called \emph{terrorist fraud}~\cite{AvoineBKLM-2011-jcs} (cf. Definition~\ref{def:terrorist}), this type of fraud is out of the scope of this dissertation, though.

\begin{definition}[Terrorist fraud]\label{def:terrorist}
A terrorist fraud is an attack where an adversary defeats a distance bounding protocol using a man-in-the-middle strategy between the reader and a dishonest tag located outside of the neighbourhood, such that the latter actively helps the adversary to maximise her attack success probability, without giving her any advantage for future attacks.
\end{definition}

In practice, the final signature represents an additional delay. Besides, according to~\cite{AvoineT-2009-isc}, as the authentication entirely relies on this phase, if the latter is interrupted or not reached, then the whole process is lost. This means that secure distance-bounding protocols not requiring a final signature are preferred.

Among the protocols without a final signature, Avoine and Tchamkerten's protocol~\cite{AvoineT-2009-isc} is the most resilient to mafia and distance frauds. They introduced the notion of distance-bounding protocols based on trees. The idea is that prover and verifier agree on a decision tree of depth $n$, which contains in its nodes the correct responses for any sequence of challenges $c_1, \cdots, c_i$ ($1 \leq i \leq n$). Since the values of the nodes are randomly chosen at the beginning of the protocol, the probability that two different sequences of challenges $c_1, \cdots, c_i$ and $\tilde{c}_1, \cdots, \tilde{c}_i$ contain the same response is $\frac{1}{2}$. Intuitively, this property dramatically reduces the mafia and distance fraud success probability. However, storing a tree of depth $n$ is prohibitive for most RFID tags.

In comparison with Avoine and Tchamkerten's protocol~\cite{AvoineT-2009-isc}, just Kim and Avoine's~\cite{KimA-2009-cans, KimA-2011-ieeetwc} protocol achieves such a high resistance to mafia fraud. Furthermore, their protocol only requires $4n$ bits of memory on the tag's side where $n$ is the number of rounds. In this protocol, the prover is armed with a mechanism to detect whether it is the target of a mafia fraud attack. Then, once the prover detects the attack, it responds randomly to the subsequent rounds. Therefore, the probability of success of the adversary considerably decreases. However, the more efficient the mechanism, the weaker the protocol against a distance fraud attack. In consequence, the Kim and Avoine's protocol~\cite{KimA-2009-cans, KimA-2011-ieeetwc} might not be appropriate when both mafia and distance frauds need to be thwarted.


\subsection{Trajectory anonymisation}

The location of an individual can be determined by different techniques. Possibly, the most conventional and ancestral of these techniques is the visual identification of that individual in some place at some moment. Nowadays, this task is far easier since there is no need for a person monitoring or harassing another person. Instead, several technologies widely adopted worldwide can perform this task for us automatically (\emph{e.g.} surveillance cameras, credit card transactions, RFID identification, among others). In addition, today's pervasiveness of location-aware devices like mobile phones and GPS receivers helps companies and governments to easily collect huge amount of information about the movements of people.

Analysing and mining this type of information, also known as trajectories or spatio-temporal data, might reveal new trends and previously unknown knowledge to be used in traffic, sustainable mobility management, urban planning, supply chain management, etc. By doing so, resources can be optimised and business and government decisions can be solid and well-founded. As a result, it is considered that both companies and citizens profit directly from the publication and analysis of databases of trajectories. However, there are obvious threats to the privacy of individuals if their
trajectories are published in a way which allows re-identification of the
individual behind a trajectory.

A tentative solution to preserve individuals' privacy is de-identification, that is, to remove all the identifying attributes of individuals. However, this is often insufficient to preserve individuals' privacy. Another set of attributes, known as \emph{quasi-identifiers}, together with external information, can be used to re-identify the individual behind a record. For instance, it has been shown that the tuple $\{$zip-code, gender, and date of birthday$\}$ is unique for the $87\%$ of the population of United States~\cite{sweeney02a}. As an example in the context of spatio-temporal databases, let us consider a GPS application recording the trajectories of some people. Daily routine indicates that a user's trajectory in the morning is likely to begin at home and end at her work place. This information can be easily linked to a single user, whose identity might be obtained from an external source of information like telephone directories or social networks.

Estimating how much external information is available to an adversary is a challenging task~\cite{kaplan10}. Furthermore, the time information and its relation with the spatial information gives a distinctive nature to the spatio-temporal data over the microdata, \emph{i.e.} over records describing users' data without a sequential order. That is why traditional anonymisation and sanitisation methods for microdata~\cite{fung10} are not suitable for spatio-temporal data and viceversa. Therefore, specific anonymisation algorithms devoted to thwarting privacy attacks on published databases of trajectories are increasingly needed.

\subsubsection{$k$-Anonymity and $\ell$-diversity}
\label{micro}


A lot of work has been done in anonymising microdata
and relational/transactional
databases~\cite{samarati98,sweeney02a,truta06,machanavajjhala06,wong06,li07,nergiz07,domingo08,Zhang07aggregatequery, springerlink:10.1007/978-3-642-25237-2_8, Sun20122211, DBLP:journals/corr/abs-1101-2604}; see also the recent survey~\cite{fung10}.
A usual goal in anonymisation is to achieve
$k$-anonymity~\cite{samarati98,sweeney02a}, which is the ``safety in
numbers'' notion.

An anonymised  microdata set is said to satisfy $k$-anonymity if each
combination of quasi-identifier attribute values is shared
by at least $k$ records. Therefore,
this property guarantees that an adversary is unable
to identify the individual to whom an anonymised record
corresponds with probability
higher than $1/k$.

Another useful privacy notion is $\ell$-diversity~\cite{machanavajjhala06}, which improves $k$-anonymity by diversifying the sensitive attributes values of each group of records that can be isolated by an attacker. This privacy notion is motivated by the fact that even when an adversary cannot identify the individual's record among a set of $k$ records with probability greater than $1/k$, she could easily retrieve the individual's sensitive values with high level of confidence, \emph{e.g. } if the $k$ records have the same sensitive attributes values. In~\cite{machanavajjhala06, li07}, different considerations regarding the $\ell$-diversity privacy notion can be found.

\subsubsection{Microaggregation}

$k$-Anonymity cannot be directly achieved with
spatio-temporal data, because any point or time can be regarded as
a quasi-identifier attribute~\cite{abul08}.
Direct $k$-anonymisation would require
a set of original trajectories to be transformed into a set of
anonymised trajectories such that each of the latter
is identical to at least $k-1$ other anonymised trajectories.
This would obviously cause a huge information loss.

Generalisation was the computational approach
originally proposed to achieve $k$-anonymity~\cite{samarati98,sweeney02a}.
Later, Zhang {\em et al.} introduced
the permutation-based approach~\cite{Zhang07aggregatequery},
that has the advantage of not being constrained by domain
generalisation hierarchies. In~\cite{domingo05} it was shown
that $k$-anonymity could also
be achieved through microaggregation of quasi-identifiers.
Microaggregation~\cite{domingo02} works in two stages:
\begin{enumerate}
\item {\em Clustering}. The original records are
partitioned into clusters based
on some similarity measure (some kind of distance)
among the records with the restriction that each
cluster must contain at least $k$ records.
Several microaggregation heuristics are available
in the literature, some yielding fixed-size clusters
all of size $k$, except perhaps one ({\em e.g.} the
MDAV heuristic~\cite{domingo05}), and some yielding
variable-size clusters, of sizes between $k$ and $2k-1$
({\em e.g.} $\mu$-Approx~\cite{muapprox} or V-MDAV~\cite{DBLP:conf/ijcnn/SolanasGM10}).
\item {\em Anonymisation}. Each cluster is anonymised individually.
Anonymisation of a cluster may be based on an aggregation
operator like the average~\cite{domingo02} or the median~\cite{domingo05},
which is used to compute the cluster centroid; each record in the
cluster is then replaced by the cluster centroid.
Anonymisation of a cluster
can also be achieved by replacing the records in the cluster with synthetic
or partially synthetic data; this is called hybrid data
microaggregation~\cite{domingo10a} or condensation~\cite{aggarwal04}.
\end{enumerate}

\subsubsection{Clustering algorithms for trajectories}

Just like in microdata records, suppressing direct identifiers
from trajectories is not enough for privacy~\cite{kaplan08}.
Consequently, several anonymity notions and methods
for trajectories have been
proposed~\cite{hoh05,gruteser05,hoh07,bonchi08,pensa08,bonchi09,abul08,nergiz08,terrovitis08,mohammed09,nergiz09,yarovoy09,monreale10,abul10,hoh10,5374402, DBLP:journals/tdp/MonrealeTPRB11}.
Among those works, we next review the ones that try to achieve some notions of trajectory $k$-anonymity.
Other comparisons of several trajectory anonymisation
methods can be found in~\cite{bonchi09,abul10}.

A naive approach for achieving $k$-anonymity is by suppression of attribute values, which is generally used on categorical nominal data where perturbation methods are not well suited. One of the first suppression-based methods for trajectory anonymisation is due to Terrovitis and Mamoulis~\cite{10.1109/MDM.2008.29}. They consider trajectories to be sequences of addresses taken from an address domain $\mathcal{P}$ and adversaries controlling subsets of addresses of $\mathcal{P}$. Thus, the adversary's knowledge can be represented as a database of projections of original trajectories over the addresses in $\mathcal{P}$ that she controls. Then, they propose a greedy algorithm aimed at guaranteeing that no address unknown by the adversary can be linked with any user with probability higher that some threshold. The main problem with this approach is that dealing with all possible adversary's knowledge causes an anonymisation problem harder than the simpler $k$-anonymity problem in relational databases, which is already known to be NP-Hard~\cite{Meyerson:2004:COK:1055558.1055591}.

Abul, Bonchi and Nanni proposed a notion of trajectory $k$-anonymity assuming uncertainty on the data provided by technologies like GPS~\cite{abul08,abul10}. They also proposed two methods to achieve privacy according to their notion of privacy. In the
original method --Never Walk Alone (NWA)~\cite{abul08}--,
the set of trajectories is partitioned
into disjoint subsets in which trajectories begin and end at roughly
the same time; then trajectories
within each set are clustered using the Euclidean distance. In the follow-up
method --Wait For Me (W4M)~\cite{abul10}--, the original trajectories are
clustered using the edit distance on real sequences (EDR)~\cite{chen05}.
Both approaches proceed by anonymising each cluster separately. Two
trajectories $T_1$ and $T_2$ are said to be co-localised
with respect to $\delta$ in
a certain time interval $[t_1,t_n]$ if for each triple $(t, x_1,y_1)$
in $T_1$ and each triple $(t, x_2,y_2)$ in $T_2$
with $t \in [t_1,t_n]$, it holds that
the spatial Euclidean distance between both triples
is not greater than $\delta$.
Anonymity in this context
means that each trajectory is co-localised with at least $k-1$ other
trajectories ($(k,\delta)$-anonymity).
Anonymisation is achieved by spatial translation of trajectories inside a
cluster of at least $k$ trajectories having the same time span. In the
special case when $\delta=0$, the method produces one centroid/average
trajectory that represents all trajectories in the cluster.
{\em Ad hoc} preprocessing and outlier removal facilitate the process.
Utility is evaluated in terms of trajectory distortion and impact
on the results of range queries.
The problem with the NWA method is that partitioning the set of all
trajectories into subsets sharing the same time span may produce too
many subsets with too few trajectories inside each of them;
clearly, a subset with less than $k$ trajectories cannot be
$k$-anonymised. Also, setting a value for $\delta$ may be
awkward in many applications ({\em e.g.}
trajectories recorded using RFID technology).

Another $k$-anonymity based notion for trajectories consisting of ranges of
points and ranges of times has been proposed in~\cite{nergiz08}
and~\cite{nergiz09}. It uses clustering to minimise the ``log cost
metric'', which measures the spatial and temporal translations
with user-provided weights.
Minimising the log cost therefore maximises utility. The
clusters are anonymised by matching points of the trajectories and
generalising them into minimum bounding boxes. Unmatched points are
suppressed and so are some trajectories. The anonymised data are not
released; instead, synthetic ``atomic'' trajectories (having
unit x-range, y-range and time range) are
generated by sampling the bounding boxes. This approach
does not release standard trajectories but
only trajectories with unit ranges.

In~\cite{monreale10}, $k$-anonymity means that an original
trajectory $T$ is
generalised into a trajectory $g(T)$ (without the
time information) in such a way that $g(T)$ is
a sub-trajectory of the generalisations
of at least $k-1$ other original trajectories.
Ignoring the time information during anonymisation and complex plane
tessellations used to achieve the $k$-anonymity are the main drawbacks of
this method. Utility is measured by comparing clustering results.

\cite{hu10} is another proposal for achieving $k$-anonymity
of trajectories by means of generalisation. The difference
lies in the way generalisation is performed: the authors propose
a technique called \emph{local enlargement}, which guarantees that
user locations are enlarged just enough to reach $k$-anonymity,
which improves utility of the anonymised trajectories.

The adapted $k$-anonymity notion for trajectories in~\cite{yarovoy09}
is stated in terms of a bipartite attack
graph relating original and anonymised trajectories such that
the graph is symmetric and the degree of each vertex representing
an anonymised trajectory is at least $k$.
The quasi-identifiers used to define identities are the times of the positions
in a trajectory, and the anonymity is achieved by generalising points of
trajectories into areas on the grid. An
information loss metric defined for
such areas is used to evaluate the utility of the anonymised data.

Some approaches assume that the data owner anonymising the database knows
exactly what the adversary's knowledge is. If the adversary is assumed to
know different parts of the trajectories, then those are removed from the
published data~\cite{terrovitis08}. However, this work only considers
sequential place visitation without real time-stamps. If the adversary is
assumed to use some prediction of continuation of a trajectory based on
previous path and speed, then uncertainty-aware path
cloaking~\cite{hoh07,hoh10} can suppress these trajectories; this
procedure, however, results in high information loss.

Additional related work about anonymisation of spatio-temporal data can be
found in the literature about location privacy, focused on applications
such as privacy-aware location-based services (LBS) or
privacy-aware monitoring
of continuously moving objects. Location privacy
in the LBS-setting was first proposed in~\cite{gruteser03}.
See~\cite{palanisamy11,hu10tkde} for recent papers on location privacy,
in which mobile objects protect the privacy of
their continuous movement. Location
privacy is enforced on individual sensitive locations or unlinked
locations in an on-line mode; often, data are anonymised on a per-request
basis and in the context of obtaining a location-based service.
In this dissertation, we focus on
off-line publishing of whole spatio-temporal databases rather than
protecting specific individuals from LBS providers
or on-line movement monitoring. In general, a solution to location privacy is
not a solution for publishing anonymised trajectories, and vice versa.

\chapter{Improving Scalability by Means of Distributed Readers}

\label{chap:4}

\emph{This chapter describes our first RFID identification
proposal, which is based on collaborative readers. Its aim is to
improve flexibility and efficiency.}

\minitoc

The idea of making tags and/or readers collaborate has been proposed
and tested. With regard to tags, in \cite{BohnEUSAI2004} and \cite{LimIEA2006}
a distribution of tags is used to guide mobile robots equipped with RFID
readers and perform precise indoor positioning, respectively. Also,
in~\cite{Conti2011} tags cooperate in order to detect when and for how long
a tag has been tampered with. With regard to readers, to improve the
scalability of hash-based solutions without increasing the number
of rounds of the protocol, Solanas {\em et al.} proposed an approach that used
collaborative readers deployed in a grid structure \cite{SolanasDMD-2007-cn}.
Instead of having a centralised database with all the tag IDs, each reader
maintains a local database (\emph{e.g.} in a local cache) in which it stores the IDs of
the tags located in its cover area and the ones in its adjacent neighbours'
area. By doing so, readers no longer need to check all possible IDs to identify
a tag but only a smaller subset of IDs in their local cache. Although the
proposal in \cite{SolanasDMD-2007-cn} is a step forward in terms of scalability,
it replicates too many tag IDs and imposes several constraints to the system
(\emph{e.g.} readers must know the exact distance to the tags and the reader
distribution is very rigid). In \cite{1364147}, Ahamed, Rahman and Hoque
modified the proposal of Solanas {\em et al.} and proposed a more natural neighbourhood
structure using a hexagonal grid (cf. Figure~\ref{fig:grid}). Note that this
solution reduces the number of neighbours from nine (in the squared grid) to six
(in the hexagonal grid). However, this proposal has the same limitations of
\cite{SolanasDMD-2007-cn}.

The idea of distributing tags amongst a number of readers placed in a grid
or in hexagonal cells might resemble the antenna structure of the well-known GSM
system for mobile communications. In fact, readers store information about tags
similarly to what visitor location registers (VLR) do with cell phones in GSM.
However, there are some fundamental differences that make this problem different in the RFID context:
\begin{itemize}
\item In GSM, cell phones are active and they are responsible for the registration
of their ID in the VLR.
\item Visitor location registers (generally) do not exchange information amongst
them. They mainly communicate with a centralised database known as the
home location register (HLR).
\item A centralised database such as the HLR might not exist.
\end{itemize}

\begin{figure}[!ht]
\centering
\scalebox{0.75}{
\includegraphics{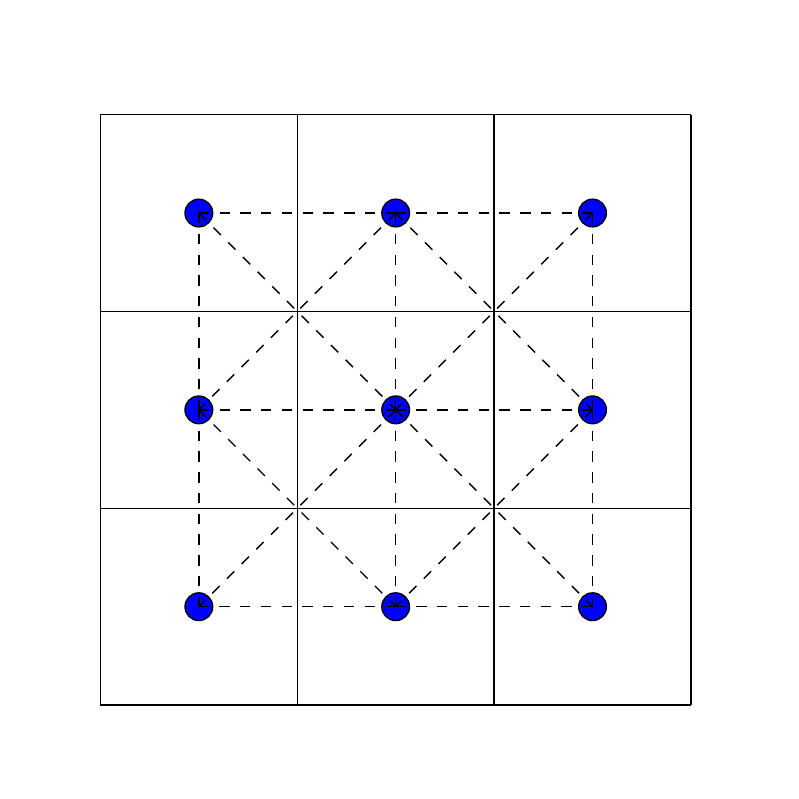}
}
\scalebox{0.6}{
\includegraphics{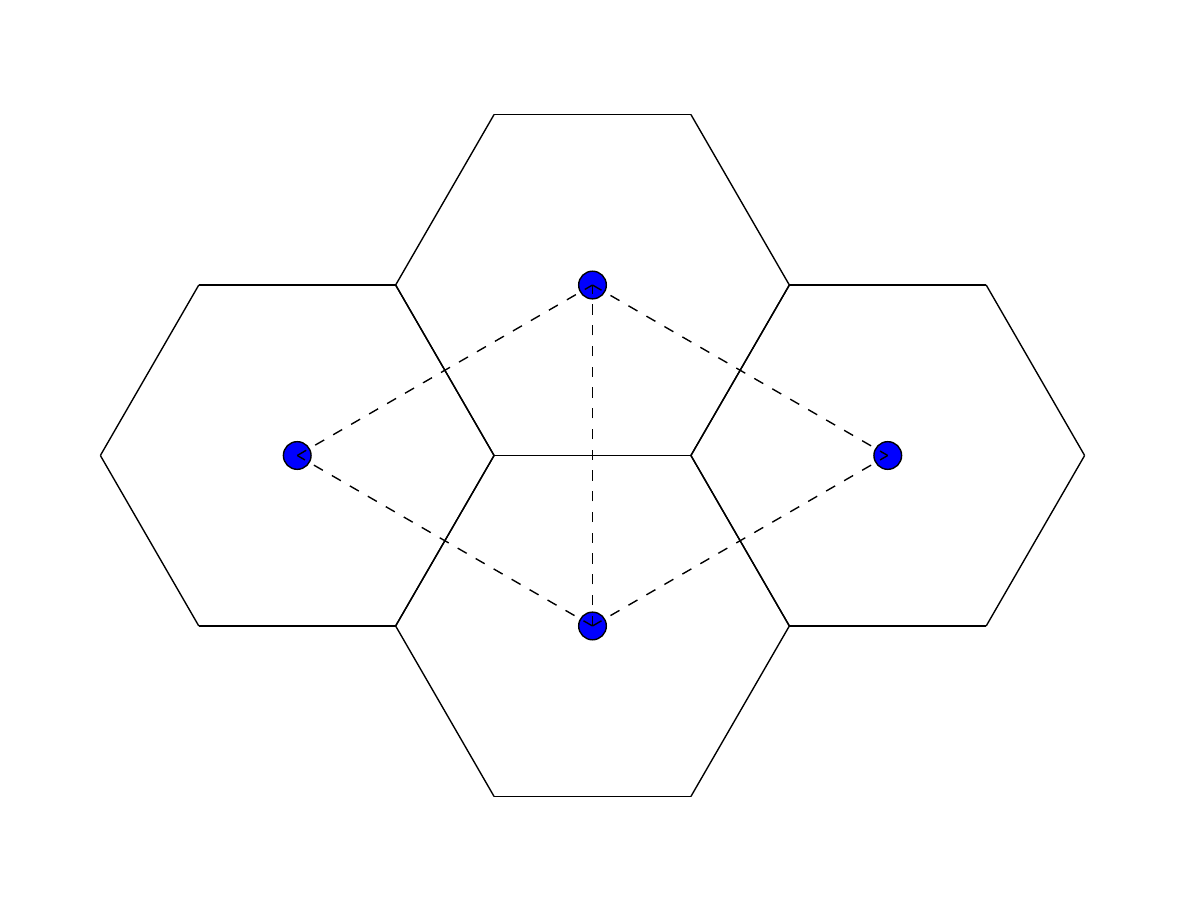}
}
\caption{\textbf{Left:} Scheme of nine collaborative readers using a
squared grid neighbourhood~\cite{SolanasDMD-2007-cn}. \textbf{Right:} Scheme of
four readers using a hexagonal grid neighbourhood~\cite{1364147}. Dashed lines
represent neighbourhood relations amongst readers.}
\label{fig:grid}
\end{figure}

\section{An efficient RFID identification protocol by means of collaborative readers}

Hash-based identification protocols for RFID tags have shown to be private and
secure but they require a significant computational effort on the readers' side
that is generally overcome by using a centralised mainframe, which can lead to
bottlenecks and delays. Specifically, the number of operations performed by
the mainframe to identify a single tag is a function of the number of tags
$(n)$ in the system, \emph{i.e.} $f(n)$.

An alternative to the centralised solution is the collaborative approach that
was first described in~\cite{SolanasDMD-2007-cn}, whose main idea is to
distribute the list of tag IDs amongst all the readers in the system and allow
them to identify tags within their cover range without contacting a central
mainframe. The solution proposed by Solanas {\em et al.} improves the scalability of
the system with regard to the centralised solution. Ideally, if we consider a
number of readers $(m)$ and a number of tags $(n)$, the number of operations
that must be performed by a reader to identify a tag is a function of
$(\frac{n}{m})$, \emph{i.e.} $f(\frac{n}{m})$. Unfortunately, the protocol proposed by
Solanas {\em et al.} requires
the readers to store the IDs of the tags controlled by
neighbour readers and this leads to a significant increase of redundant IDs. If we
assume that the redundancy can be expressed by a factor $k$, the number of
operations that a reader performs to identify a tag using the protocol described
in~\cite{SolanasDMD-2007-cn} is $f(\frac{k\times n}{m})$, where
\[
 f(\frac{n}{m}) < f(\frac{k\times n}{m}) < f(n)
\]

Our protocol leverages the idea of collaboration from~\cite{SolanasDMD-2007-cn},
but implements a new set of messages that permit the reduction of redundant
information. Ideally, we want $k\rightarrow 1$. To do so, thanks to our protocol,
readers can be initialised with a parameter $p \in [0,1]$ that represents the
probability for a reader to store tag IDs from its neighbours. Note that when
$p=0$, the number of redundant IDs is zero and we reach the optimal situation
where the number of operations required to identify a tag is $f(\frac{n}{m})$.

In addition, network designers/engineers can balance the reader's computational
cost and its bandwidth usage by tuning $p$. The smaller $p$ the lower
the number of operations, but the bandwidth requirements are higher.

\subsection{Brief recap of the \textit{Original} protocol}
The protocol described in~\cite{SolanasDMD-2007-cn}, that we call \textit{original}, was
designed to allow multiple readers to collaborate in order to exchange information
about tags so as to improve the scalability of the improved randomised hash-locks (IRHL)
identification procedure.

In the \textit{original} protocol, each reader was responsible for a squared cell and
they were all distributed in a grid structure. Note that, using this distribution,
the areas covered by each reader were disjoint and, by construction, a tag in a given
location could only be queried by a single reader (this is an important difference
with regard to the protocol described in this chapter).

In the \textit{original} protocol three main procedures/subprotocols were described:
\begin{enumerate}
	\item Tag arrival protocol: This protocol is applied when a tag enters the system
	through a System Access Point or SAP. A reader controlling this SAP identifies
	the tag using IRHL and communicates to \textbf{all} its neighbours the ID of
	that tag. Then if that tag moves to any of the cells controlled by these
	neighbours, they will be able to identify it.
	\item Roaming protocol: This protocol is used when a tag changes its location from
	a cell controlled by a reader to another cell. In this case, the reader controlling
	the destination cell informs all its neighbours that he is the new owner of the
	tag and forwards the ID information of the tag to all its neighbours. Also, the
	previous owner sends a message to its neighbours so as to inform that it is no
	longer the owner of the tag.
	\item Departure protocol: This protocol is used when a tag leaves the system. In
	this case a reader controlling a System Exit Point (SEP) simply forwards to its
	neighbours the message of deleting that tag from their caches.
\end{enumerate}

\subsection{Assumptions and definitions}
\label{subsec:assumptions}
In our proposal, instead of using the concept of \textit{unshared cover area}, as
described in~\cite{SolanasDMD-2007-cn}, we use the more general concept of
\emph{shared cover area}.

\begin{definition} [Unshared Cover Area ($A^u$)] \label{def:cover_area}
The unshared cover area of a reader $R$ is the set of locations controlled by
$R$ from which tags can communicate \textbf{only} with $R$.
\end{definition}
\begin{definition} [Shared Cover Area ($A^s$)] \label{def:cover_range}
The shared cover area of a reader $R$ is the set of locations from which tags
in the system can communicate with $R$ and possibly with other readers.
\end{definition}
From these definitions it can be derived that given two shared cover areas
$A^s_i$ and $A^s_j$, $A^s_i \cap A^s_j$ might be different from the
$\emptyset$, whilst given two unshared cover areas $A^u_i$ and $A^u_j$,
$A^u_i \cap A^u_j$ is always $\emptyset$.
Although this property of the unshared cover areas might be theoretically
useful, it is extremely hard to realise it in practise. Thus, from now on,
when we use the term \textit{cover area} we will refer to the more realistic
concept of \textit{shared cover area} described in
Definition~\ref{def:cover_range} and, for the sake of clarity, we avoid using
the superscript $s$.

Let $A_i$ be the cover area of a reader $R_i$ and let $A$ be the area covered by
all the readers in the system. We assume that
$A \subseteq {\bigcup^i} {A_i}, \forall{i}$.

Considering our definition of shared
cover area, we define the neighbourhood relation as follows:
\begin{definition}[Neighbourhood relation]
Two readers $R_i$ and $R_j$ are neighbours if their cover areas
$A_i$ and $A_j$ are not disjoint, \emph{i.e.} $A_i \cap A_j \neq \emptyset$.
\end{definition}

Our notions of cover area and neighbourhood are more flexible and
realistic than those proposed in \cite{SolanasDMD-2007-cn} and \cite{1364147}.
Also, they lead to a simple criterion for connecting readers, \emph{i.e.} only
neighbour readers will share a communication link to exchange protocol messages.
We assume that each reader in the system is connected to its neighbours (\emph{e.g.}
using  WLAN + SSL) and maintains a local database with a list of pairs $(ID_T,
ID_R)$, where $ID_T$ is the identifier of a given tag and $ID_R$ is the
identifier of the reader. We also assume that each tag is controlled by a
single reader, which is its owner.

Note that by using the notion of shared cover areas the tags moving in a region
shared by two readers are controlled by only one of them. On the contrary, if
unshared cover areas are used, a tag moving from one unshared cover area to
another leads to the change of owner from one reader to another. In
Figure~\ref{fig:grid2}, an example of this behaviour is shown. If we use shared
cover areas, the tag $T$ is controlled by $R2$ throughout its way. However, if we
consider the notion of unshared cover area, the tag $T$ is controlled by $R2$ at
locations $(1), (3)$ and $(5)$; and it is controlled by $R1$ at locations $(2)$
and $(4)$. This unnecessary change of ownership requires communication
between readers and increases the bandwidth usage. Consequently, using shared
cover areas may decrease the utilised bandwidth.

\begin{figure}[!ht]
\small
\centering
\scalebox{1}{
\includegraphics{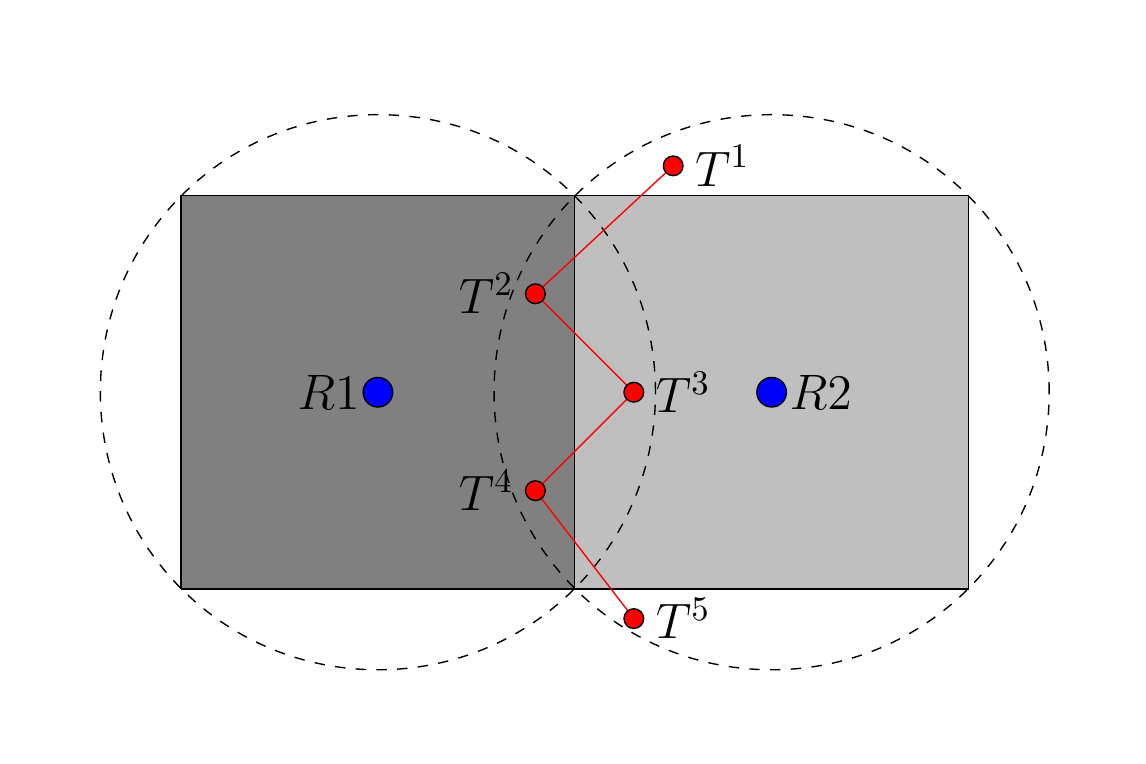}
}
\caption{Graphical example of two readers $R1$, $R2$ and a moving tag $T$. The
tag $T$ is captured in different positions at different instants $T^1$, $T^2$,
$T^3$, $T^4$, and $T^5$ ($T^x$ indicates the position of tag $T$ at time $x$).
The squares represent the unshared areas of $R1$ and
$R2$. The circles represent the shared areas of $R1$ and $R2$.}
\label{fig:grid2}
\end{figure}

\subsection{Messages}
\label{subsec:messages}

In our protocol, readers use a number of messages to exchange information about
the ownership of tags and collaborate to identify them. Each message sent by a
source reader ($R_{ID_S}$) to a destination reader ($R_{ID_D}$) makes the latter
perform an action regarding a tag ($ID_T)$ (cf. Figure~\ref{fig:messageformat}
for a graphical scheme of the message format and its flow).
Depending on the message, the information sent about the tag can be:
\begin{itemize}
	\item \textbf{The tag ID -- ($ID_T$)}: If $R_{ID_S}$ can identify the tag
	because it has the required information in its cache, it can send $ID_T$ to
	$R_{ID_D}$. This might happen for the following messages of the protocol:
	\textit{Delete}, \textit{I am the owner}, \textit{You are the owner}, and
	\textit{Search} messages.
	\item \textbf{The response of the tag $r=(r_2, h(r_1||r_2||ID_T))$ and the
	challenge $r1$}: If $R_{ID_S}$ is not able to identify the tag, it sends to
	$R_{ID_D}$ the challenge $r1$ that it sent to the tag and the answer $r$
	received from the tag. This happens for the \textit{Identify} message.
\end{itemize}

\begin{figure}[!ht]\centering
\begin{tabular}{|c|c|c|c|}
\hline
\multicolumn{4}{|>{\columncolor[rgb]{0.9,0.9,0.9}}c|}{\textbf{Message}}\\
\hline
\textbf{$\ $ Operation}&\textbf{Source}&\textbf{Destination}&\textbf{Tag}\\
\hline
3 bits & 32 bits & 32 bits & 128 bits\\
\hline
	   & $ID_S$ & $ID_D$ & $ID_T$\\
\hline
	   & $ID_S$ & $ID_D$ & $r$,$r_1$\\
\hline
\end{tabular}
\\
$ $
\\
\begin{tabular}[!ht]{|c c|}
\hline
\multicolumn{2}{|>{\columncolor[rgb]{0.9,0.9,0.9}}c|}{\textbf{Flow}}\\
\hline
 $\ \ \ \ R_{ID_S} \xrightarrow[\fbox{Message about the tag $(ID_T$)}]{\ }  R_{ID_D} $ & \\
 $\ \ \ \ R_{ID_S} \xleftarrow[\fbox{\ Information, ACK or NACK}]{\ }  R_{ID_D} $ & \\

\hline
\end{tabular}
\caption{Message format and flow}
\label{fig:messageformat}
\end{figure}
The messages of the protocol are explained in more detail below:

\begin{protocol_message}
\textbf{\rule{1ex}{1ex} \fbox{Delete - ($ID_T$)}} 
When $R_{ID_D}$ receives this message, it removes the identifier $ID_T$
from its local cache.
\end{protocol_message}
\begin{protocol_message}
\textbf{\rule{1ex}{1ex} \fbox{I am the owner - ($ID_T$)}} 
When $R_{ID_D}$ receives this message, it realises that $R_{ID_S}$ claims the
ownership of the tag $ID_T$. If $R_{ID_D}$ was the former owner, it sends a
\textit{Delete} message to its neighbours, excepting $R_{ID_S}$
and its neighbours, to let them know that it is no longer the owner of that tag.
If $R_{ID_D}$ was not the former owner, then it would generate a
random number $x \in [0,1]$, and if $x \ge p$ it would update its
cache with the new ownership information.
\end{protocol_message}
\begin{protocol_message}
\textbf{\rule{1ex}{1ex} \fbox{You are the owner - ($ID_T$)}} 
When $R_{ID_D}$ receives this message, it takes control over the tag $ID_T$. It
stores  the new ownership information in its cache and sends an
\textit{I am the owner} message to all its neighbours, so as to propagate the
new ownership information.
\end{protocol_message}
\begin{protocol_message}
\textbf{\rule{1ex}{1ex} \ \fbox{Identify - ($r, r_1$)}} 
This message is sent by $R_{ID_S}$ when it is not able to determine the ID of
a tag (using the Hash Lock protocol). With this message, $R_{ID_S}$ asks
$R_{ID_D}$ to identify the tag and return the ownership information stored in
its cache. If $R_{ID_D}$ identifies the tag and finds its owner, it sends the ID
of the owner back to $R_{ID_S}$, otherwise it responds with a NACK message.
\end{protocol_message}
\begin{protocol_message}
\textbf{\rule{1ex}{1ex} \fbox{Search - ($ID_T$)}} 
When $R_{ID_D}$ receives this message it checks whether the tag $ID_T$ is in its
cover area. If it finds the tag, it sends an ACK message back to $R_{ID_S}$,
otherwise it responds with a NACK.
\end{protocol_message}

\subsection{Protocol execution}
\label{subsec:protocolexecution}
\begin{figure}[!ht]\centering
\scalebox{0.85}{
\includegraphics{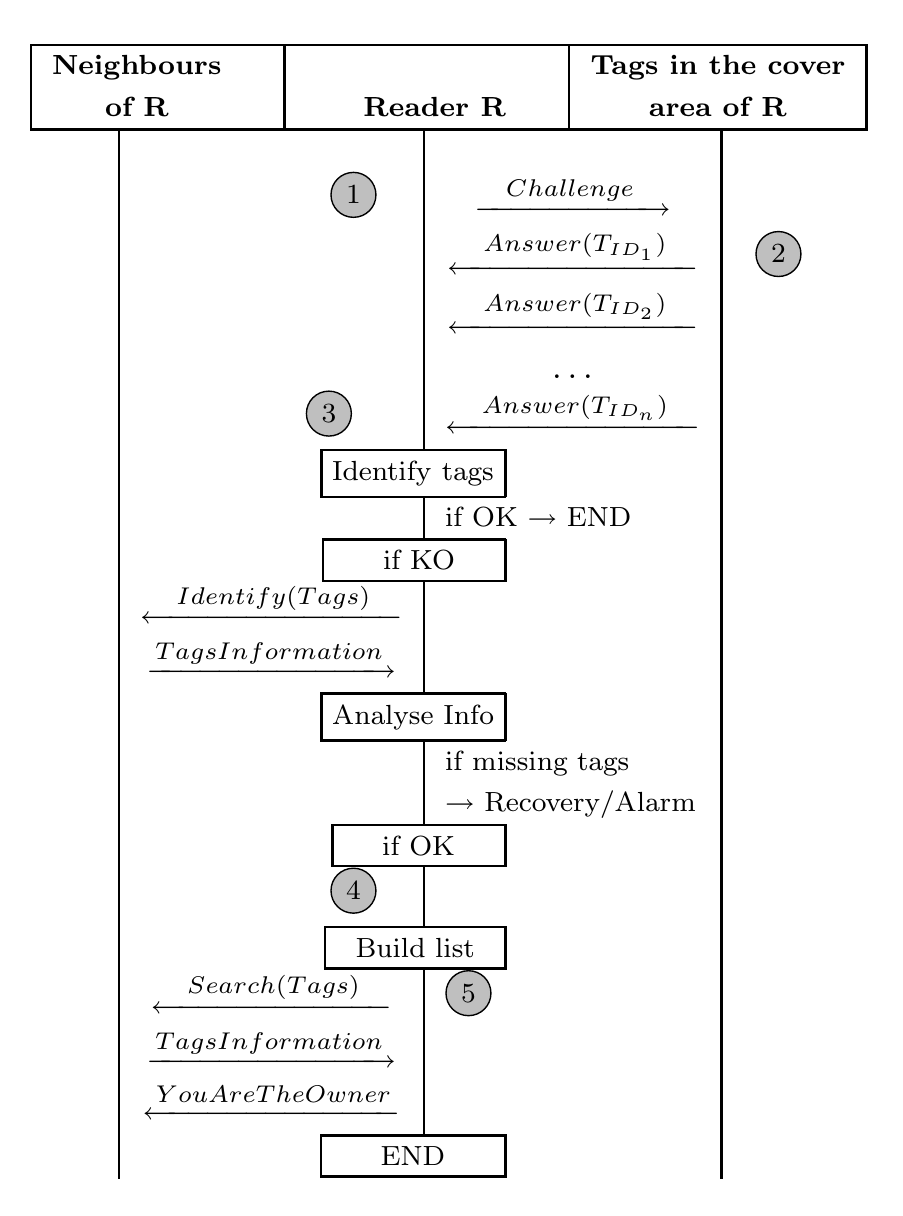}
}
\caption{Scheme of the flow of the identification protocol}\label{fig:protocol}
\end{figure}

Thanks to the probabilistic nature of our protocol, the number of IDs stored in
the local caches of the readers can be reduced with respect to the protocols
presented in \cite{SolanasDMD-2007-cn} and~\cite{1364147}; however, the flow of messages is a bit
more complex.
The identification protocol proposed in this chapter considers three main
actors: (i) the tags in the system, (ii) a reader, and (iii) the neighbours of
that reader. The protocol
depicted in Figure~\ref{fig:protocol} works as follows:

\begin{enumerate}
  \item A reader ($R$) sends a challenge ($r_1$) to the tags in its cover range.
  \item All tags in the cover range of $R$ answer the challenge.
  \item For each tag ($T$) responding to the challenge, $R$ tries to identify it
  using the hash-locks scheme~\cite{JuelsW-2007-percom} applied to its local cache.
  \begin{enumerate}
    \item If it identifies the responding tag, the process finishes.
    \item Otherwise, $R$ sends an \textit{Identify} message to its neighbours
    and stores their answers in its cache. If any of its neighbours
    identifies the tag, $R$ executes a recovery procedure described in the
    next section.
  \end{enumerate}
 \item Then, $R$ builds a list ($L$) containing all the tags that it owns
  (\emph{i.e.} which are under its control) and that have not responded to the
  challenge (\emph{e.g.} those tags that have left its cover range).
  \item For each tag  $T \in L$, $R$ sends a \textit{Search} message to its
  neighbours. After receiving the answers from its neighbours, $R$ sends a
  \textit{You are the owner} message to the first neighbour that responded
  positively (\emph{i.e.} ACK) to the search message.
\end{enumerate}

All the readers in the system periodically use this protocol. By doing so, all
tags can be controlled without the intervention of a centralised database. In
addition, due to the fact that readers only store information about
the tags of their neighbours with a given probability $p$, the number of
redundant IDs is reduced with respect to \cite{SolanasDMD-2007-cn, 1364147}
and, therefore, the computational effort performed by the readers is also reduced.

\subsection{Alarm/recovery protocol}
When a reader is not able to identify a tag and its neighbours do not have information
about this tag either, one may be in two possible situations:
\begin{itemize}
	\item An unauthorised tag has entered the system.
	\item A tag has been covered (so as to hide it from the readers) and uncovered
	in a different location controlled by another reader whose neighbours have no
	information about the tag.
\end{itemize}
When this situation arises, we propose two possible solutions:
\begin{itemize}
	\item A centralised solution: This solution is based on maintaining a
	backup of all tag's IDs in a centralised server. Doing so, when neither a
	reader nor its neighbours could identify a tag, that reader could request
	the identification of this tag to the centralised server. Note that, this
	solution has a high computational cost but does not create bottlenecks
	because the centralised server is supposed to be used in exceptional
	cases only.
	\item A fully decentralised solution: In this case readers can iteratively
	query their neighbours so as to find the previous owner of the tag in the
	system. First the reader queries its adjacent neighbours (located at one
	hop), then it queries the neighbours located at two hops, etc. This
	procedure finishes when the tag ID is found or when all readers have been
	queried. In the first case, our protocol keeps working normally, in the
	second case, an alarm is raised. This procedure is depicted in Figure~\ref{fig:alarm}.
	Note that in the worst case, in which all readers in the
	system were to be queried, the computational cost would be linear in
	the number of tags $n$. Although the computational cost is high and the
	communication overhead might be significant, this situation should happen
	rarely; hence, it should not affect the overall efficiency of the proposed
	protocol.
\end{itemize}
\begin{figure}[!ht]
\centering
\scalebox{0.6}{
\includegraphics{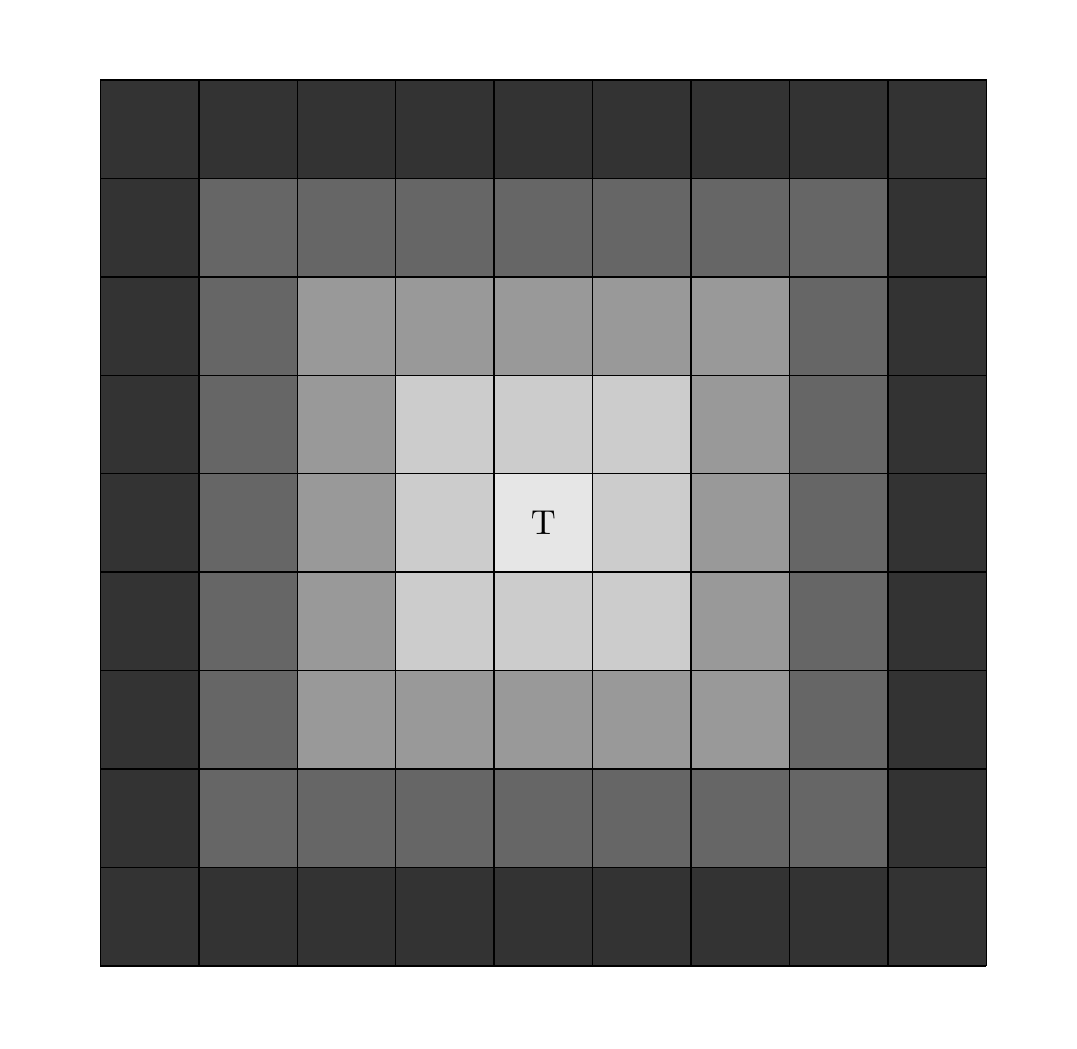}}
\caption{A representation of cells covering the monitored area. An unidentified
tag is located in the central cell. The reader in that cell will iteratively query
other readers to identify $T$. Readers in lighter coloured cells are queried first.}
\label{fig:alarm}
\end{figure}

\subsection{The role of $p$}
\label{subsec:role}
The number of operations performed in a reader to identify a tag is linear with
the number of tag IDs stored in its cache. A reader stores the IDs of the tags
in its cover area (for which it is responsible) -- we say that that reader is the
\textbf{owner} of those tags. In addition, a reader may store
the IDs of tags located in the cover area of its neighbours. This way, if a tag
moves from the cover area of one of its neighbours, it can identify that tag
without querying its neighbours.

Each reader is initialised with a parameter $p$. This parameter defines the
probability for a given reader to store neighbour tag's IDs in its cache. If
$p=1$ the reader stores all the IDs of its neighbour tags, on the contrary if
$p=0$ the reader stores no information about its neighbours' tags. If $p$ takes
a value in $(0,1)$ the reader stores a number of IDs \textit{proportional} to
that value. The main goal of $p$ is to reduce the number of redundant IDs stored
in the cache of neighbour readers.

The number of IDs stored by a reader $i$ $(\#ID^i)$ can be computed as
$$
\#ID^i = n_i + p_i\sum_{j=1}^{b_i} n_j^i
$$
where $n_i$ is the number of tags owned by $i$, $b_i$ is the number of
neighbours of reader $i$, $n_j^i$ is the number of tags owned by the $j$-th
neighbour of reader $i$, and $p_i$ is the probability for the reader $i$
to store IDs of tags owned by its neighbours. The total number of IDs stored
in the system ($\#ID$) can be computed as $\sum_{i=1}^m \#ID_i$, where $m$
is the total number of readers.

In the example shown in Figure~\ref{fig:p}, it is apparent that, by reducing the
value of $p$, the number of IDs stored in the caches of the readers is also
reduced. Consequently, the number of operations required to identify a tag
is also reduced and the whole process of identifying tags scales better.

Note that the protocols described in \cite{SolanasDMD-2007-cn}
and~\cite{1364147} do no support the addition of this probabilistic property.
Thus, the main goal of the proposed protocol, explained in the following
sections, is to allow the use of the parameter $p$ and, as a result, to improve
the scalability of the identification process on the readers' side.

\begin{figure}[tb]
\centering
\begin{tabular}{|c|c|c|}
\hline
2 & 2 & 2 \\
\hline
2 & 2 & 2 \\
\hline
2 & 2 & 2 \\
\hline
\multicolumn{3}{c}{$ $}\\
\multicolumn{3}{c}{$p=0$}\\
\end{tabular}
\begin{tabular}{|c|c|c|}
\hline
5 & 7 & 5 \\
\hline
7 & 10 & 7 \\
\hline
5 & 7 & 5 \\
\hline
\multicolumn{3}{c}{$ $}\\
\multicolumn{3}{c}{$p=0.5$}\\
\end{tabular}
\begin{tabular}{|c|c|c|}
\hline
8 & 12 & 8 \\
\hline
12 & 18 & 12 \\
\hline
8 & 12 & 8 \\
\hline
\multicolumn{3}{c}{$ $}\\
\multicolumn{3}{c}{$p=1$}\\
\end{tabular}
\caption{Number of IDs stored in the readers for different values of $p$
considering that each reader is the owner of 2 tags. The neighbourhood
relations are the ones described in Figure~\ref{fig:grid}-left.}
\label{fig:p}
\end{figure}

\subsection{Our protocol in a centralised back-end}
Although our protocol has been designed to work in a distributed way, it could
be ``simulated'' by a centralised database (\emph{i.e.} a back-end) connected
to a properly distributed set
of readers. By doing so, the back-end would be able to
identify tags and ``logically'' cluster them in regions (\emph{e.g.} virtually covered
by the readers). Thus, intelligent search of a tag into these regions might be
scalable in terms of computational cost. In addition, this approach averts the
communication overhead associated with the exchange of messages between readers
because all the communication might be ``simulated'' within the back-end.

The main problems of using this approach are: (i) using a single centralised
database leads to a single point of failure and, (ii) the communication of
a single back-end with a (possibly) large number of readers might create
bottle-necks and undesired delays.

It might be said that, depending on the special characteristics of the environment
in which the RFID system is to be deployed, engineers may decide whether to use
our protocol ``simulated'' within a back-end, or use it as a fully distributed
non-centralised protocol.


\section{Experimental results and evaluation}
\label{sec:simulation}

We have developed a simulator to analyse the number of operations performed by
the collaborative readers during the execution of our probabilistic protocol,
and their bandwidth usage. The simulator allows the deployment of readers
without constraints. The number of readers, their cover range, their location,
the number of moving tags, and the scenario in which they move can be defined
at the beginning of the simulation.

We have concentrated on simulations to analyse the theoretical
properties of our protocol and we have left for the future the implementation
and testing of a practical prototype. Although there are some limitations in
the off-the-shelf RFID tags, there exist some EPC UHF Gen 2 tags that can
compute hash functions and random numbers (using
ARMADILLO~\cite{armadillo}) and can be read at distances of up to 1 meter.
Currently, newer versions with larger reading distances (\emph{i.e.} 3 m) are under
development (cf. \url{www.oridao.com}).

With the aim to evaluate our probabilistic protocol, we compare it with the \textit{original} protocol
presented in \cite{SolanasDMD-2007-cn}.  Although our protocol has no limitations related to the
deployment and range of the readers, the \textit{original} protocol does
have some.
Consequently, we simulate the regular distribution of 24 readers ($4\times 6$)
depicted in Figure~\ref{fig:screenshot} that the \textit{original} protocol can
handle.

\begin{figure}[!ht]
\small
\centering
\scalebox{0.9}{
\includegraphics{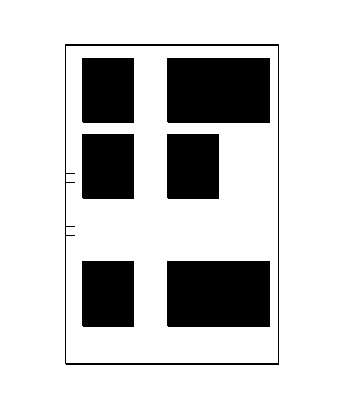}
}
\scalebox{0.9}{
\includegraphics{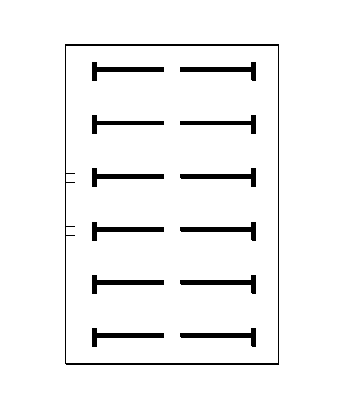}
}
\scalebox{0.9}{
\includegraphics{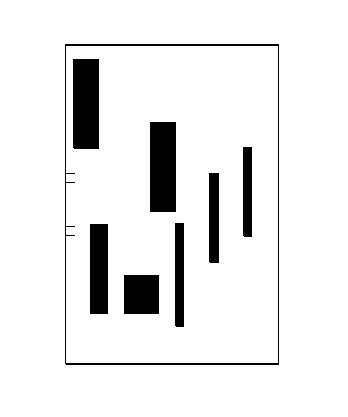}
}
\scalebox{0.9}{
\includegraphics{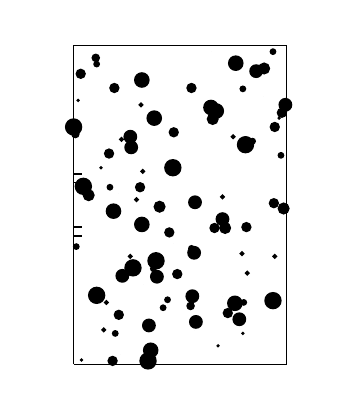}
}
\caption{Graphical scheme of the simulated scenarios. (From left to right) Scenario
with narrow corridors, scenario with wide corridors, scenario with random large
obstacles, scenario with random small obstacles.}
\label{fig:scenarios}
\end{figure}

We have considered five different scenarios\footnote{Some of these  scenarios
were already used in~\cite{SolanasDMD-2007-cn}}: (i) An empty scenario in which
tags can freely move, (ii) a scenario with narrow corridors, (iii) a scenario
with wide corridors, (iv) a scenario with randomly placed large obstacles and, (v) a
scenario with randomly placed small obstacles (cf. Figure~\ref{fig:screenshot} for
a screenshot of the simulator
and Figure~\ref{fig:scenarios} for a graphical scheme of the four non-empty
scenarios). For each scenario we have simulated the movement of $10^3$ and
$10^4$ tags. We have considered two different tag behaviours: (i) a random
movement and, (ii) a semi-directed movement: tags move half of the times
randomly and half of the times toward a far, randomly-selected point within the scenario.
Each simulation has been repeated 30 times for each value of $p$ in $(0,1)$ with
$0.1$ increments.  Globally a total of 7200 simulations have
been conducted: 2 types of
movement $\times$ 5 different scenarios $\times$ 12 protocols (11 different $p$ +
the \textit{original} one) $\times$ 30 repetitions $\times$ 2 different tag populations
($10^3$, $10^4$).

\begin{figure}[!ht]
\centering
	\includegraphics[scale=0.45, angle=90]{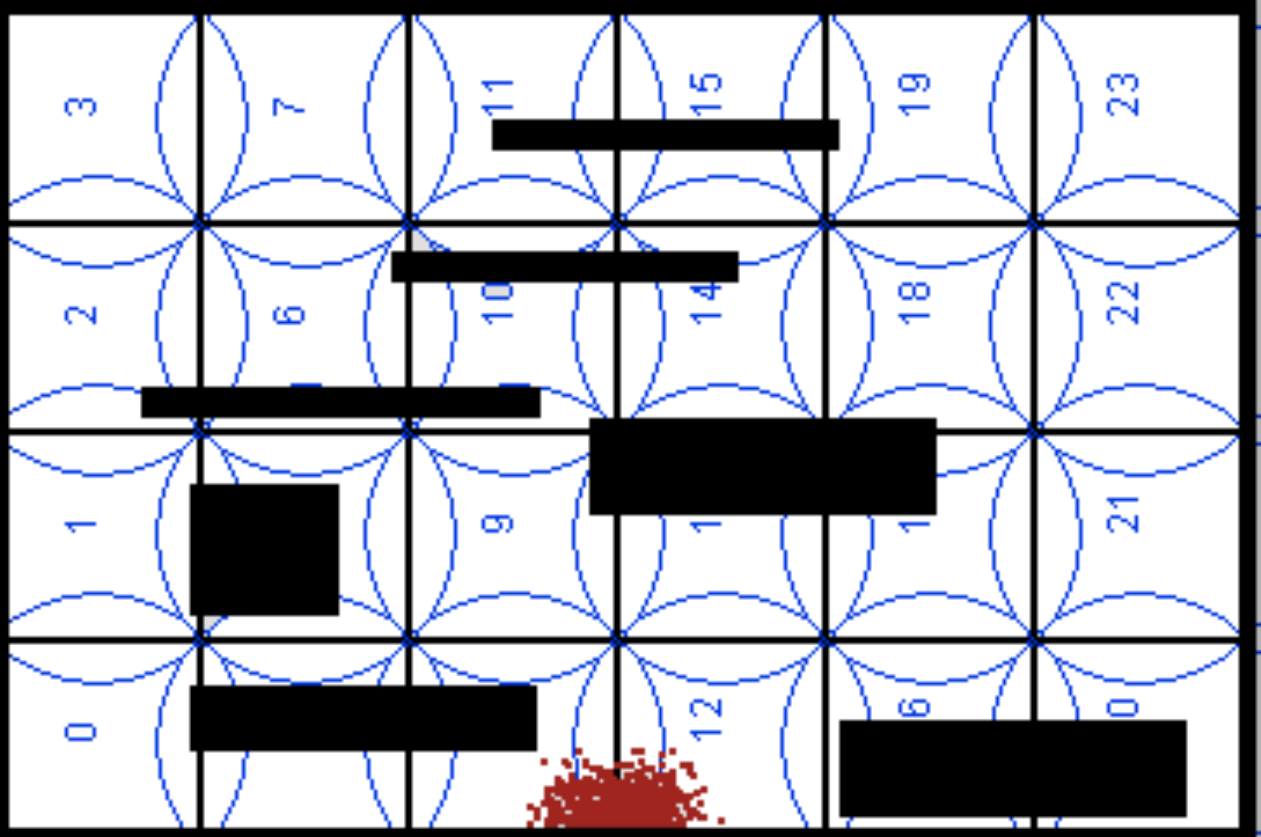}
\caption{Screenshot of the simulator. The cloud of red dots represents
the tags entering the system. Blue circles represent the \textbf{shared} cover
area of the readers, which are identified by a number. Thick black lines
represent obstacles. Finally, thin black lines represent the \textbf{unshared}
cover areas that the protocol in~\cite{SolanasDMD-2007-cn} would use.}
\label{fig:screenshot}
\end{figure}

\begin{figure}[p]
\centering
  \subfigure[Random movement - $10^3$ tags]{
  	\includegraphics[width=3.0in, angle=270]{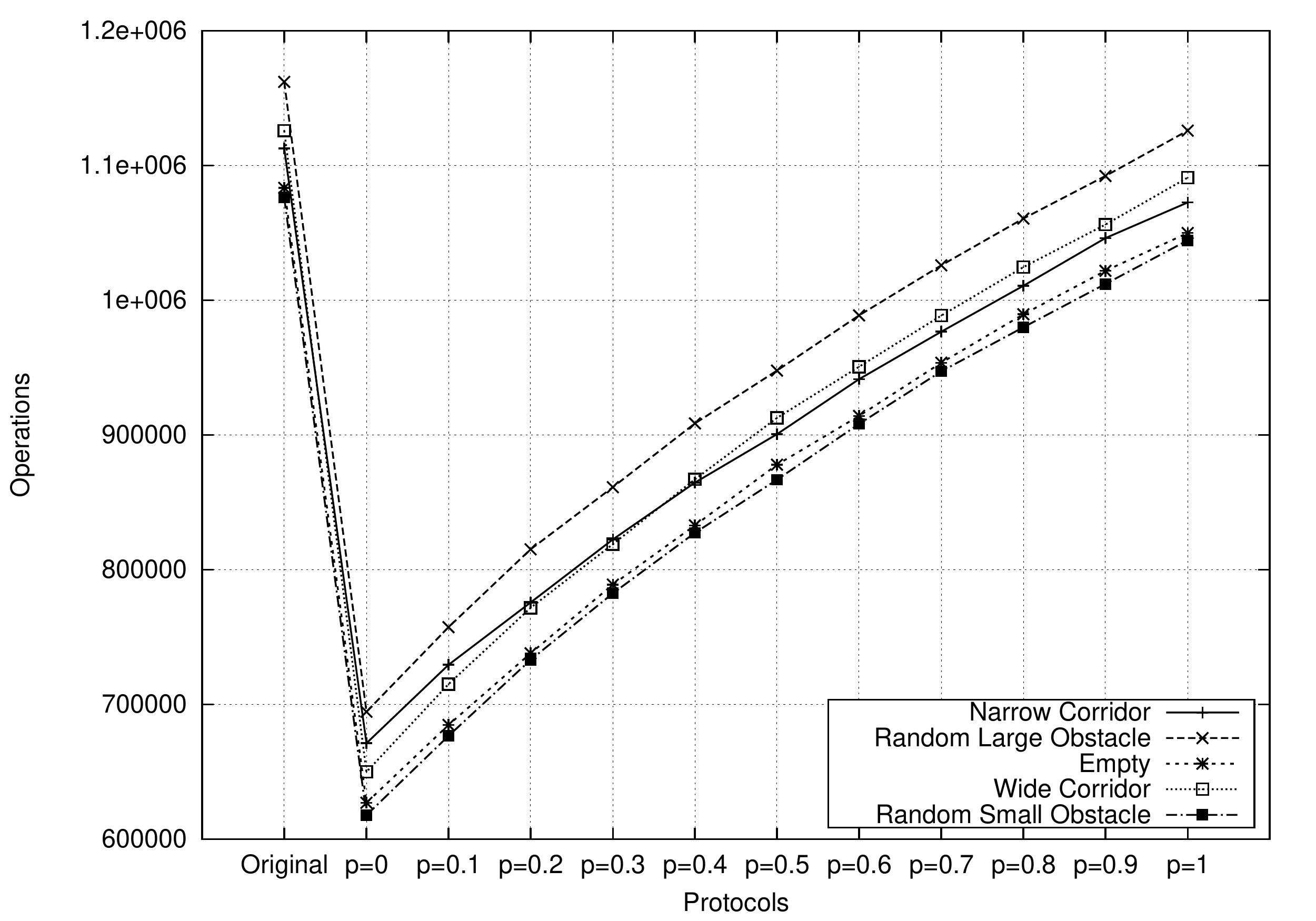}
  }
  \subfigure[Random movement - $10^4$ tags]{
  	\includegraphics[width=3.0in, angle=270]{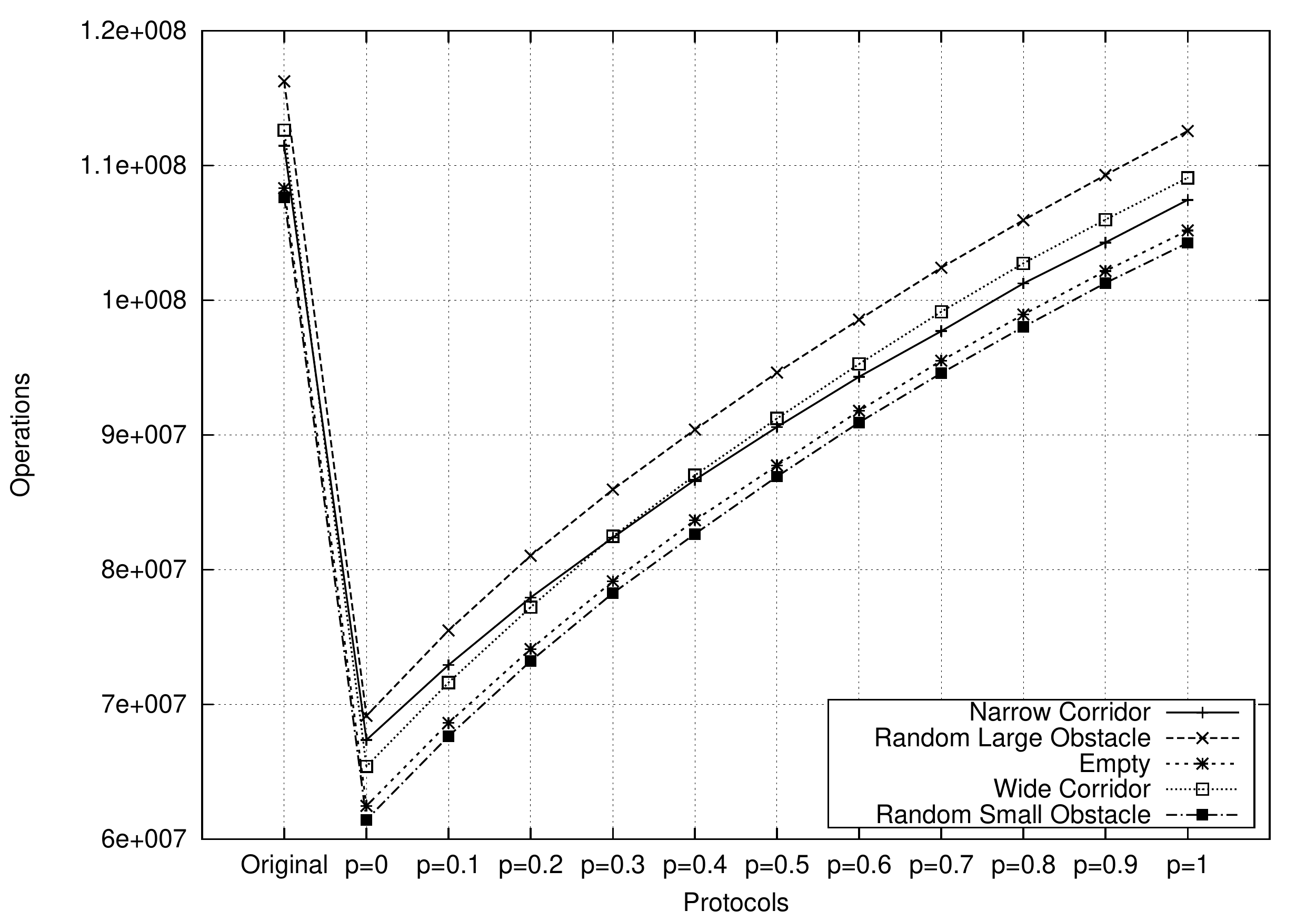}
  }
  \caption{Operations performed by the readers controlling
  $10^3$ and $10^4$ tags for different values of $p$ in all scenarios and considering 
  a random movement pattern. \textbf{The lower the better.}
  }
  \label{fig:cost1}
\end{figure}

\begin{figure}[p]
\centering
  \subfigure[Semi-directed movement - $10^3$ tags]{
  	\includegraphics[width=3.0in, angle=270]{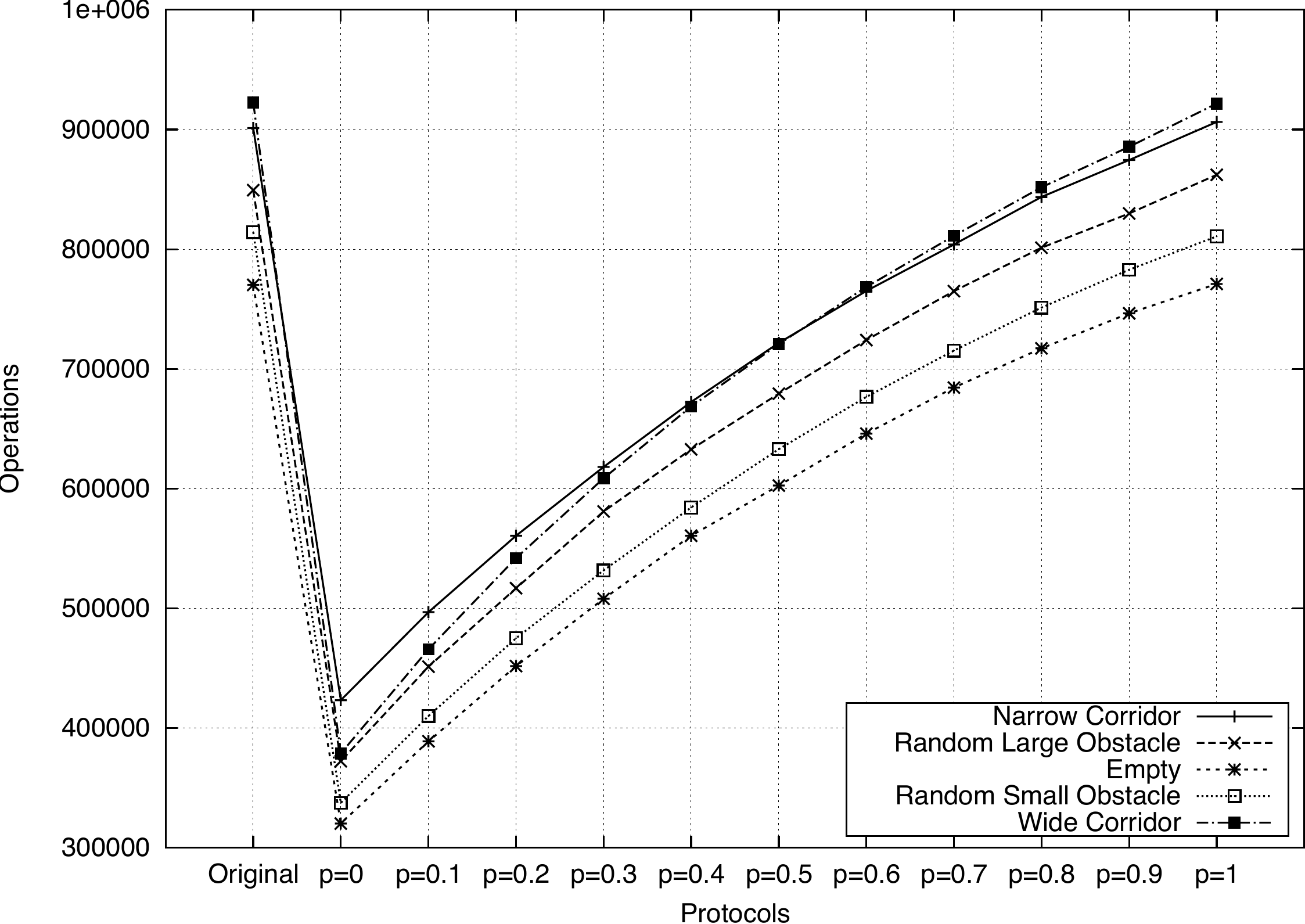}
  }
  \subfigure[Semi-directed movement - $10^4$ tags]{
  	\includegraphics[width=3.0in, angle=270]{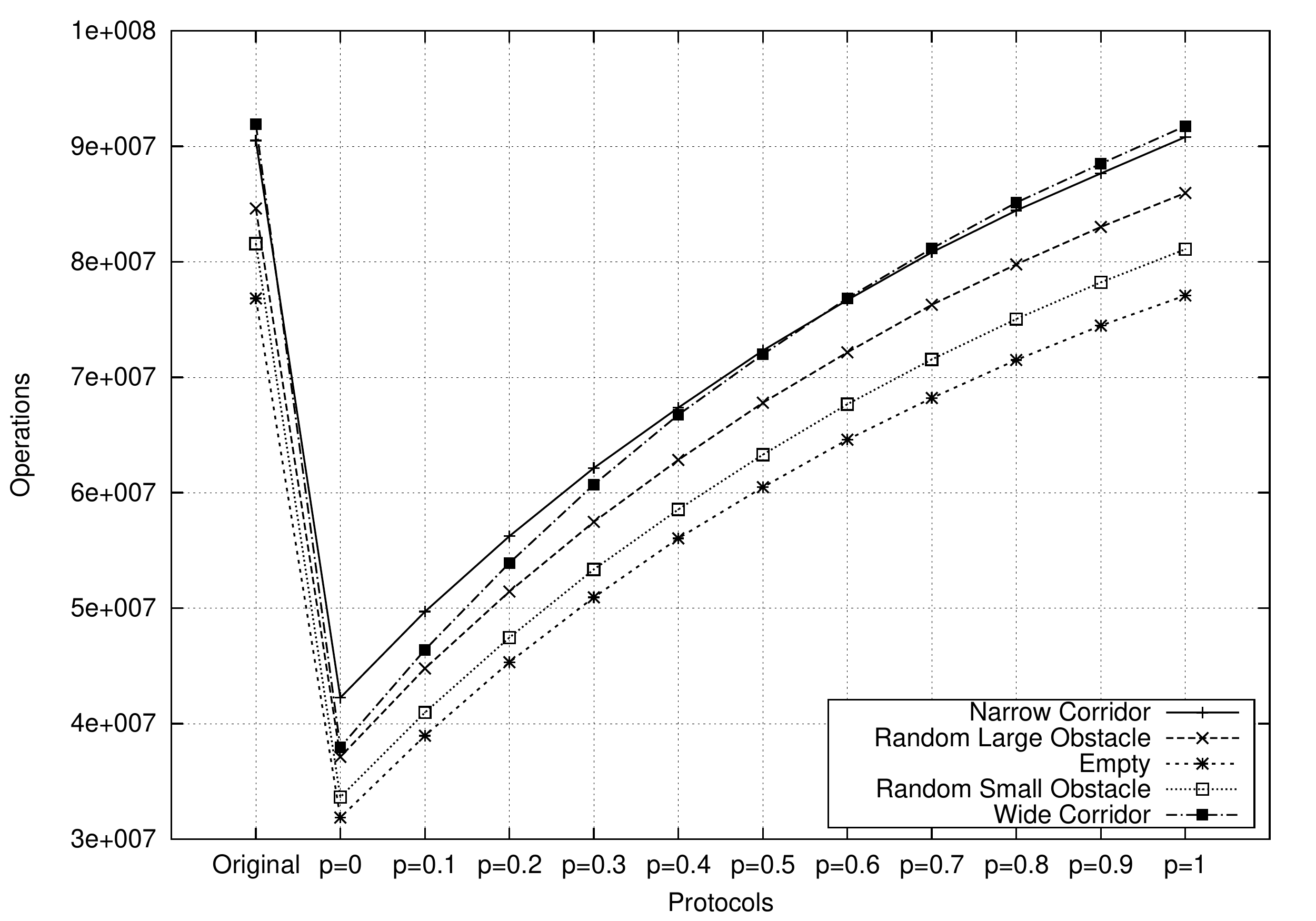}
  }
  \caption{Operations performed by the readers controlling
  $10^3$ and $10^4$ tags for different values of $p$ in all scenarios and considering 
  a semi-directed movement pattern. \textbf{The lower the better.}
  }
  \label{fig:cost2}
\end{figure}

\begin{figure}[p]
\centering
  \subfigure[Random movement - $10^3$ tags]{
  	\includegraphics[width=3.0in, angle=270]{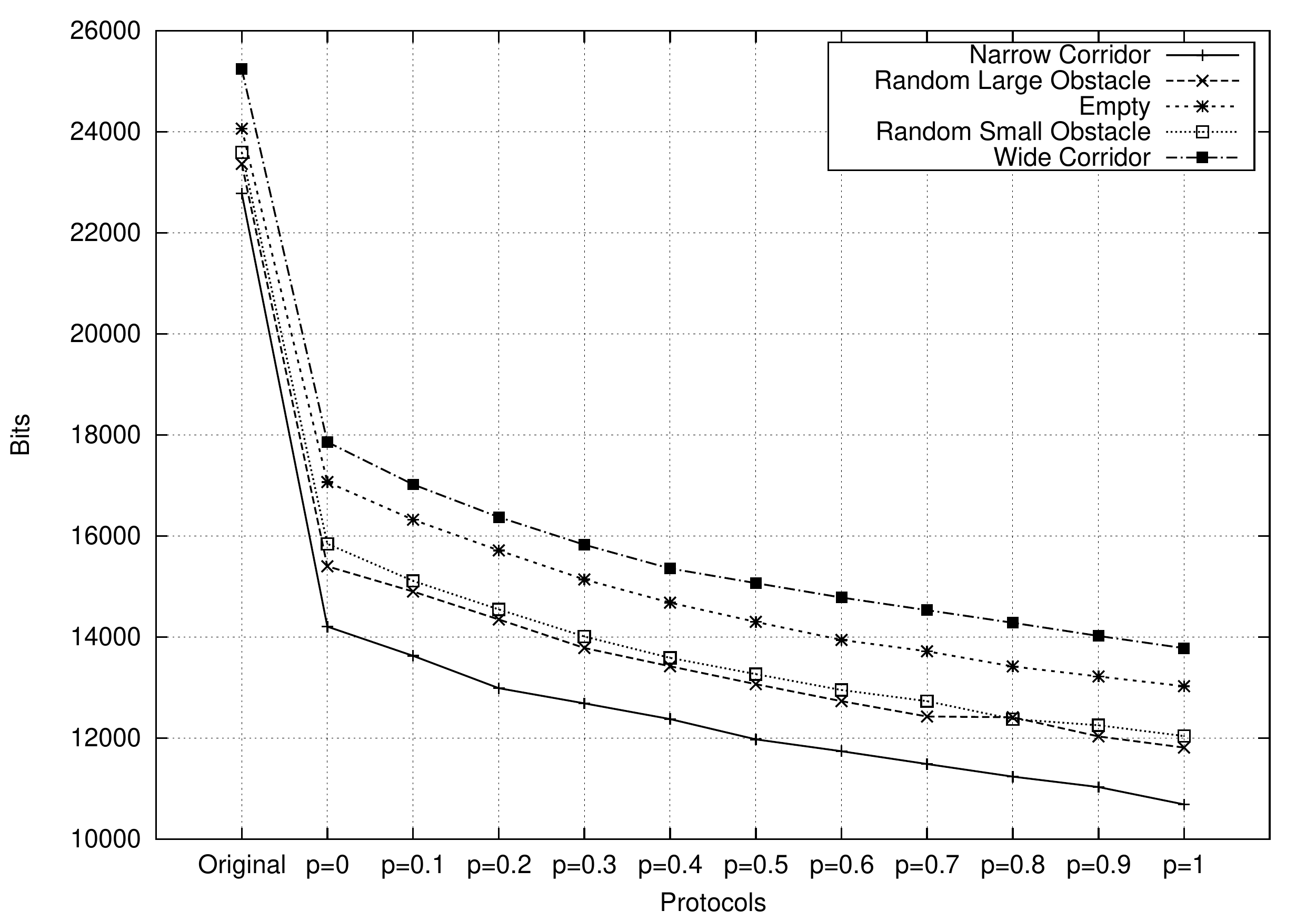}
  }
  \subfigure[Random movement - $10^4$ tags]{
  	\includegraphics[width=3.0in, angle=270]{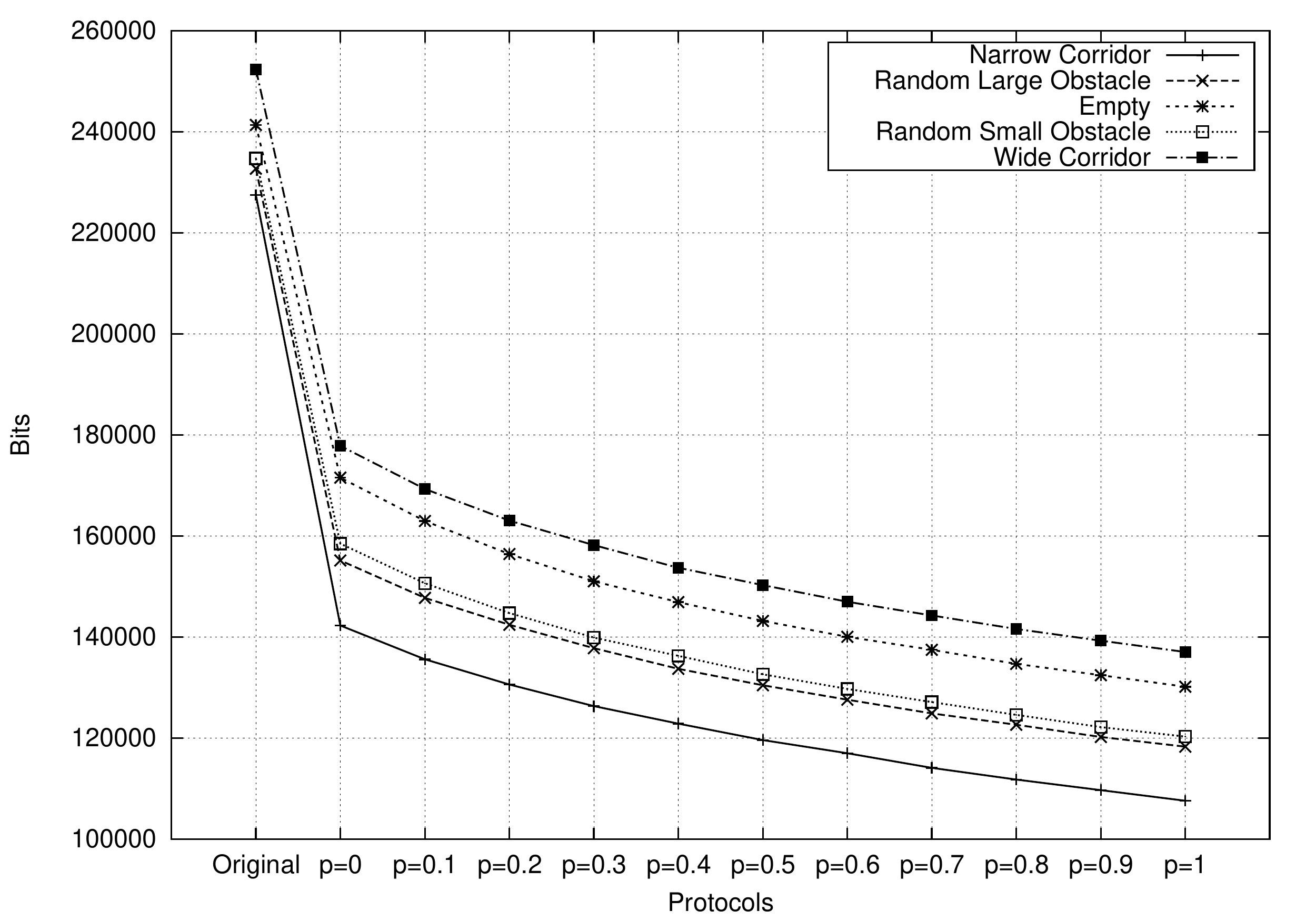}
  }
  \caption{Total number of bits transmitted by the readers controlling
  $10^3$ and $10^4$ tags for different values of $p$ in all scenarios and
  considering a random movement pattern. \textbf{The lower the better.}
  }
  \label{fig:bits1}
\end{figure}

\begin{figure}[p]
\centering
  \subfigure[Semi-directed movement - $10^3$ tags]{
  	\includegraphics[width=3.0in, angle=270]{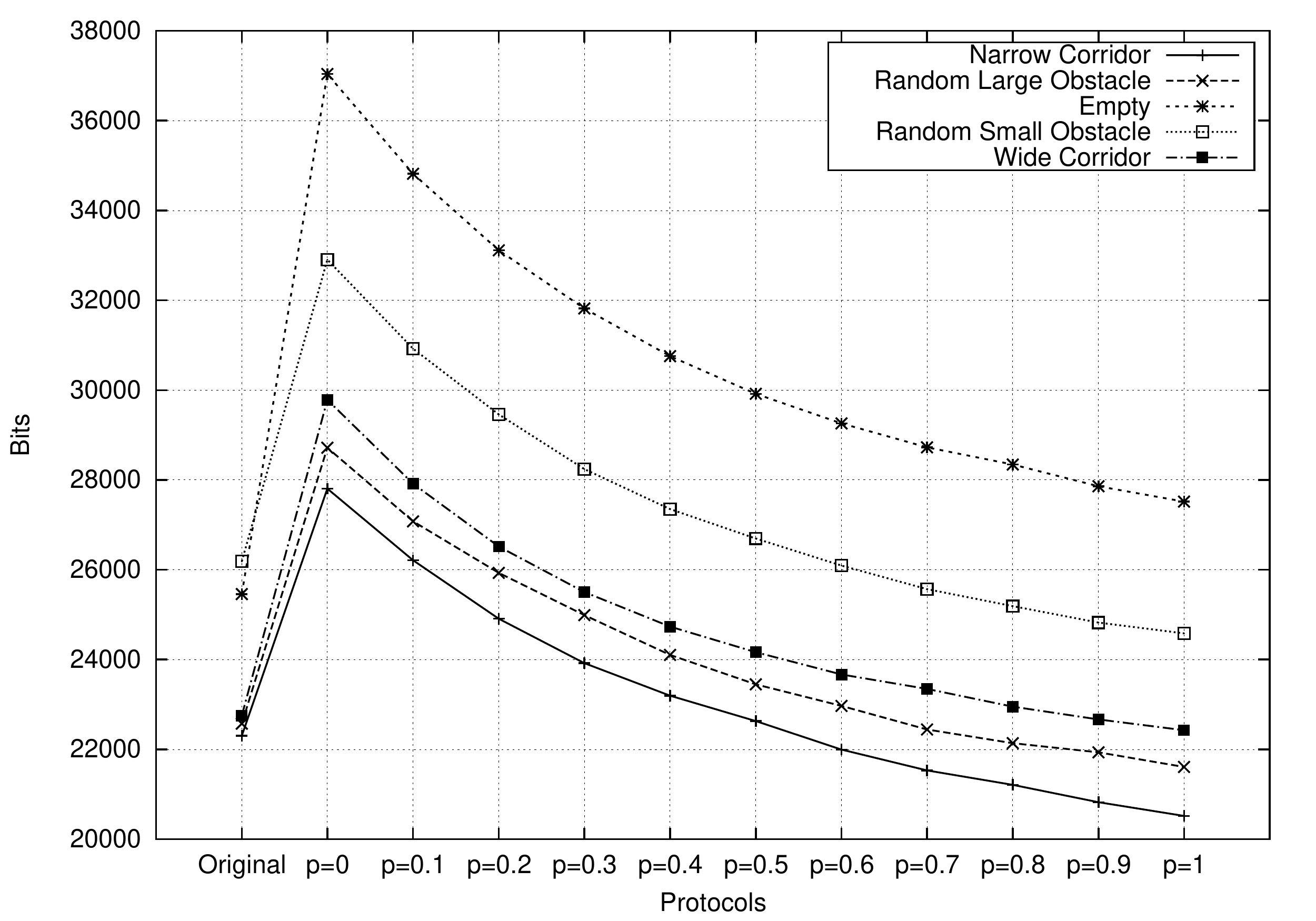}
  }
  \subfigure[Semi-directed movement - $10^4$ tags]{
  	\includegraphics[width=3.0in, angle=270]{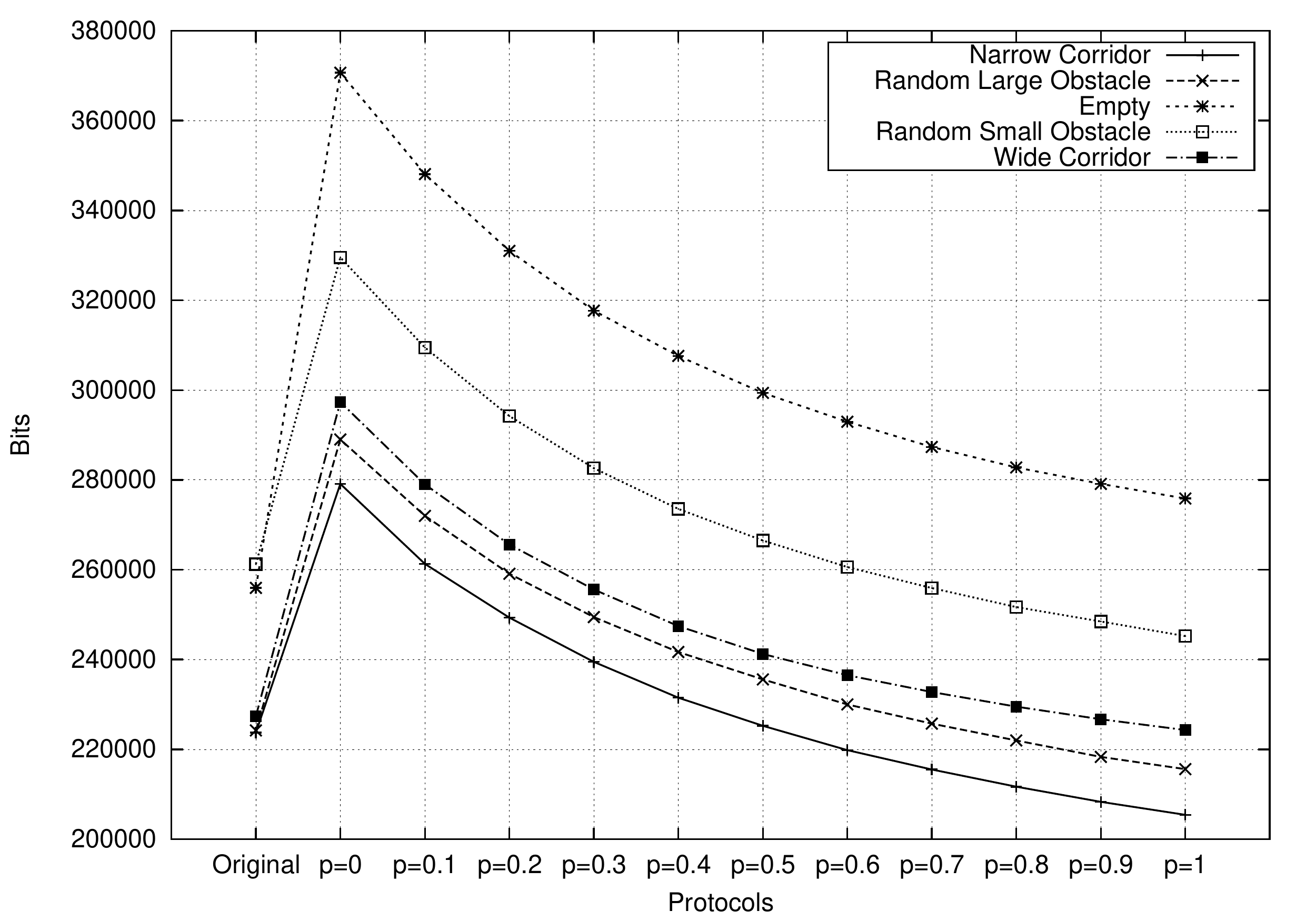}
  }
  \caption{Total number of bits transmitted by the readers controlling
  $10^3$ and $10^4$ tags for different values of $p$ in all scenarios and
  considering a semi-directed movement pattern. \textbf{The lower the better.}
  }
  \label{fig:bits2}
\end{figure}

For each scenario, we have concentrated on analysing the computational cost (in
terms of number of operations performed by readers) and the bandwidth usage (in
terms of total number of bits sent). Figures~\ref{fig:cost1} and~\ref{fig:cost2} show the results
for the computational cost and Figures~\ref{fig:bits1} and~\ref{fig:bits2} show the results for the
bandwidth usage. It can be observed that our protocol has a significantly lower
computation cost than the \textit{original} protocol. This is especially
apparent when the probability $p$ is low\footnote{Note that when the probability
$p$ tends to 1, our protocol tends to resemble the \textit{original} protocol in terms
of computational cost. However, it is still better in most cases.}.

Regarding the bandwidth usage, two different behaviours can be observed:
\begin{itemize}
	\item With random movements: Tags change from a cell
	to another with low probability (in our protocol). Thus, the number of
	required messages to update the state of the readers' caches is smaller.
	In this situation our protocol is clearly more efficient than the \textit{original} one.
	\item With semi-directed movements: Tags follow a clear path and change from one
	cell to another with a higher probability. In this case, our protocol requires
	more messages (especially in the case of using a low $p$). Thus, in this situation
	the \textit{original} protocol is more efficient for smaller $p$.
\end{itemize}

In general, the computational cost is the main concern in RFID identification
protocols and, as we have shown above, our proposal clearly outperforms the \textit{original}
protocol in this regard for all scenarios. Indeed, if bandwidth usage is not a
concern at all, our proposal with $p=0$ is the optimal solution. However, our protocol
requires more bandwidth to improve the computational cost.

Capturing the trade-off between computational cost and bandwidth is not trivial.
Note that the computational cost and the bandwidth usage are measured
in different units.
However, it is possible to define a measure in order to compare our proposal with the \textit{original} protocol in terms of both computational cost and bandwidth usage.

\begin{definition}[Trade-off measure]
Let $\alpha$ be a real value in the range $[0..1]$. Let $c$ and $b$ be the
computational cost and the bandwidth usage, respectively,
of the \textit{original} protocol
for a given configuration\footnote{A configuration will be defined by the number of
tags in the system, the number of readers and their distribution, the scenario,
etc.}. Let $c_p$ and $b_p$ be the computational cost and the bandwidth usage of our
protocol using the same configuration and $p$ the probability value. Then, the
trade-off measure that we propose is computed as follows:
$$
d(\alpha, p) = \left(\left(\frac{c_p}{c} - 1 \right)\times 100 \right)\times \alpha + \left(\left(\frac{b_p}{b} - 1 \right)\times 100\right)\times (1 - \alpha)
$$
\end{definition}

Intuitively, the proposed trade-off measure $d(\alpha, p)$ represents the
performance of the \textit{original} protocol with regard to our protocol using $p$ as the
probability value and considering $\alpha$ the weight given to the computational
cost and $1-\alpha$ the weight given to the bandwidth usage. Note that when
$\alpha = 0$ the bandwidth usage is the only concern, whilst when $\alpha = 1$ only
the computational cost is considered.

Figures~\ref{fig:3d_corridors1}, ~\ref{fig:3d_corridors2}, \ref{fig:3d_empty}, \ref{fig:3d_obstacles1} and \ref{fig:3d_obstacles2} depict the
performance of the \textit{original} protocol with regard to our protocol using the
trade-off measure described above. At the bottom of each figure there is a
three-dimensional chart showing the values of $d(\alpha, p)$ for each
$\alpha \in \{0, 0.1, \cdots ,0.9, 1\}$ and each $p \in \{0, 0.1, \cdots ,0.9, 1\}$.
Also, at the top left side and at the top right side of the figure there are the
projections of the three-dimensional charts for the x-axis and y-axis,
respectively.
In the x-axis projection, for each value of $\alpha$ the values of
$d(\alpha, p)$, $\forall p \in [0,1]$, are shown, whilst in the y-axis projection
the plot of the linear functions $d(\alpha, p)$ with $\alpha$ fixed is shown.

\begin{figure}[p]
\centering
  \subfigure[Random movement - ``Narrow corridor'' scenario]{
  	\includegraphics[width=3.5in, angle=270]{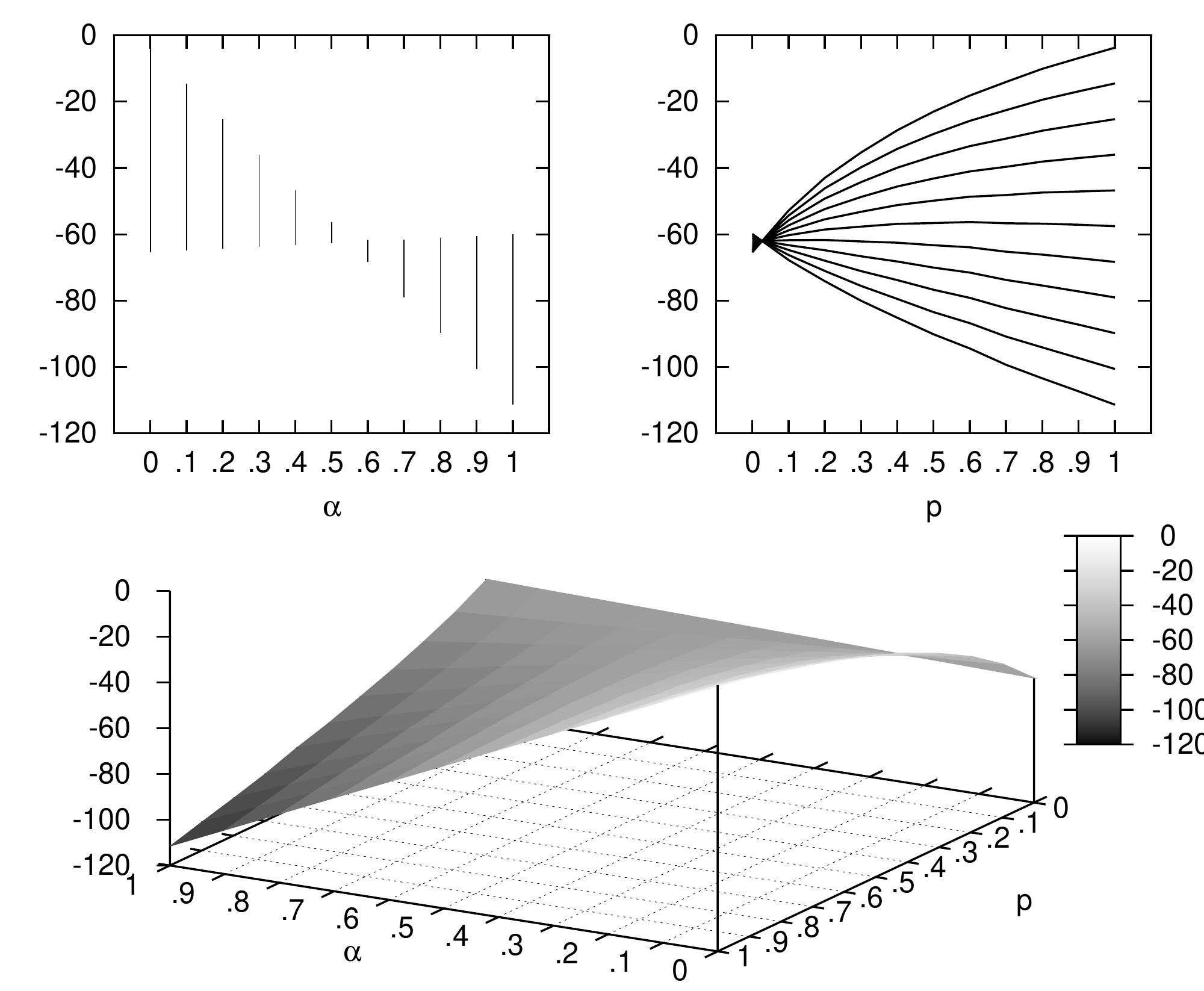}
  }
  \subfigure[Random movement - ``Wide corridor'' scenario]{
  	\includegraphics[width=3.5in, angle=270]{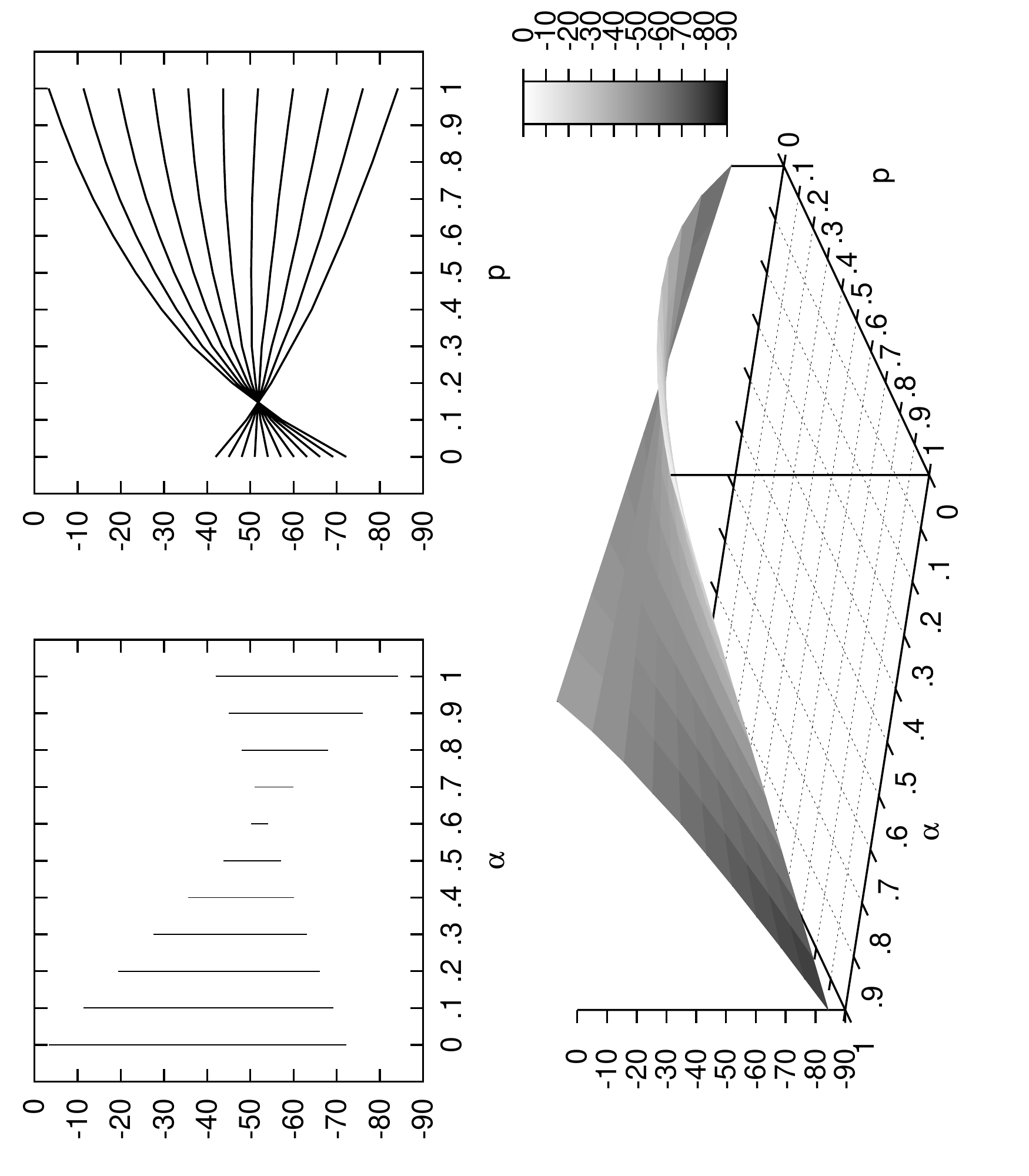}
  }
  \caption{$d(\alpha, p)$ results for $10^4$ tags and different
  values of $p$ and $\alpha$ in the scenarios with corridors and with a random movement pattern. \textbf{Values below
  zero indicate that our protocol is better with respect to the \textit{original} protocol}.}
  \label{fig:3d_corridors1}
\end{figure}

\begin{figure}[p]
\centering
  \subfigure[Semi-directed movement -``Narrow corridor'' scenario]{
  	\includegraphics[width=3.5in, angle=270]{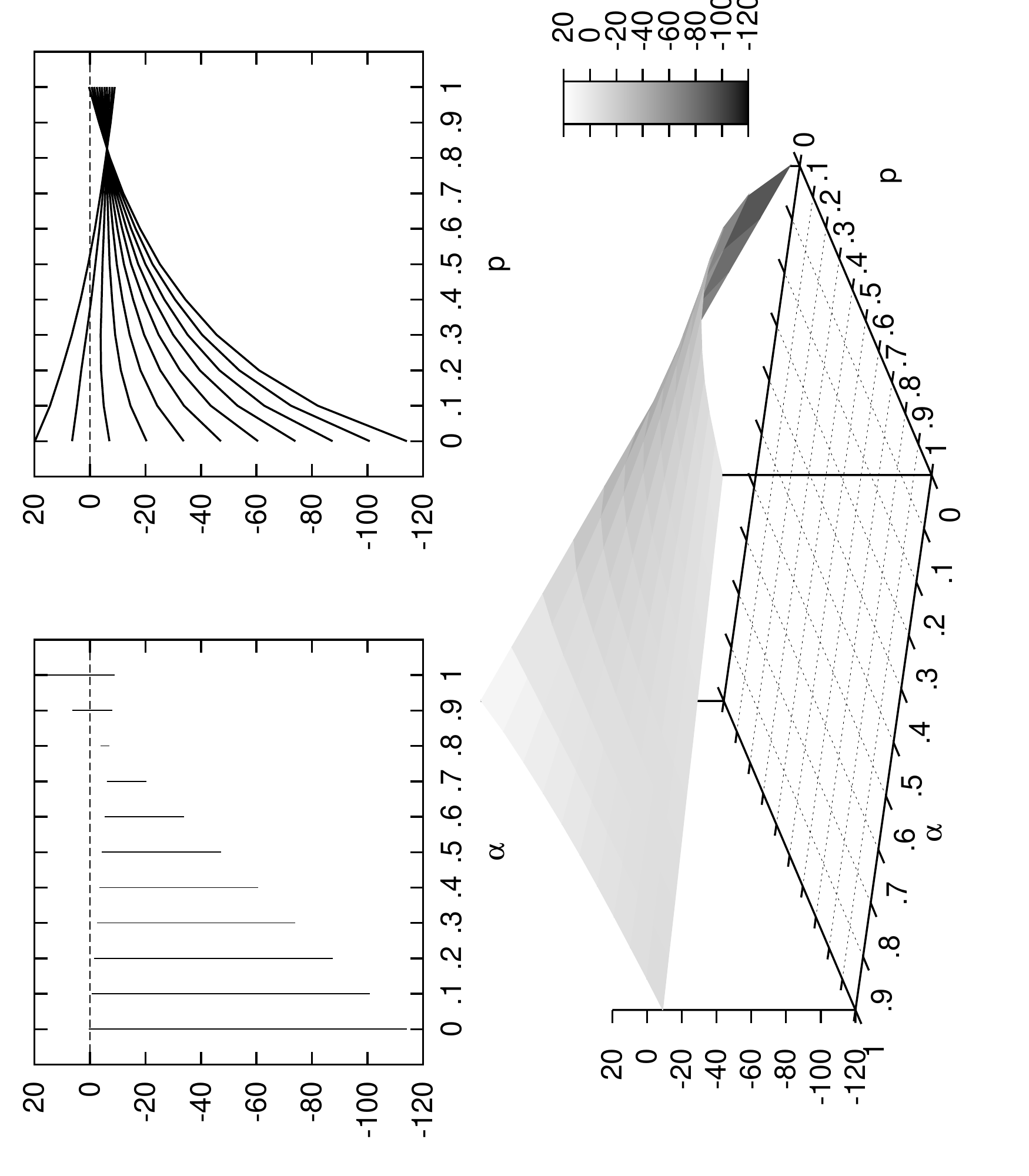}
  }
  \subfigure[Semi-directed movement - ``Wide corridor'' scenario]{
  	\includegraphics[width=3.5in, angle=270]{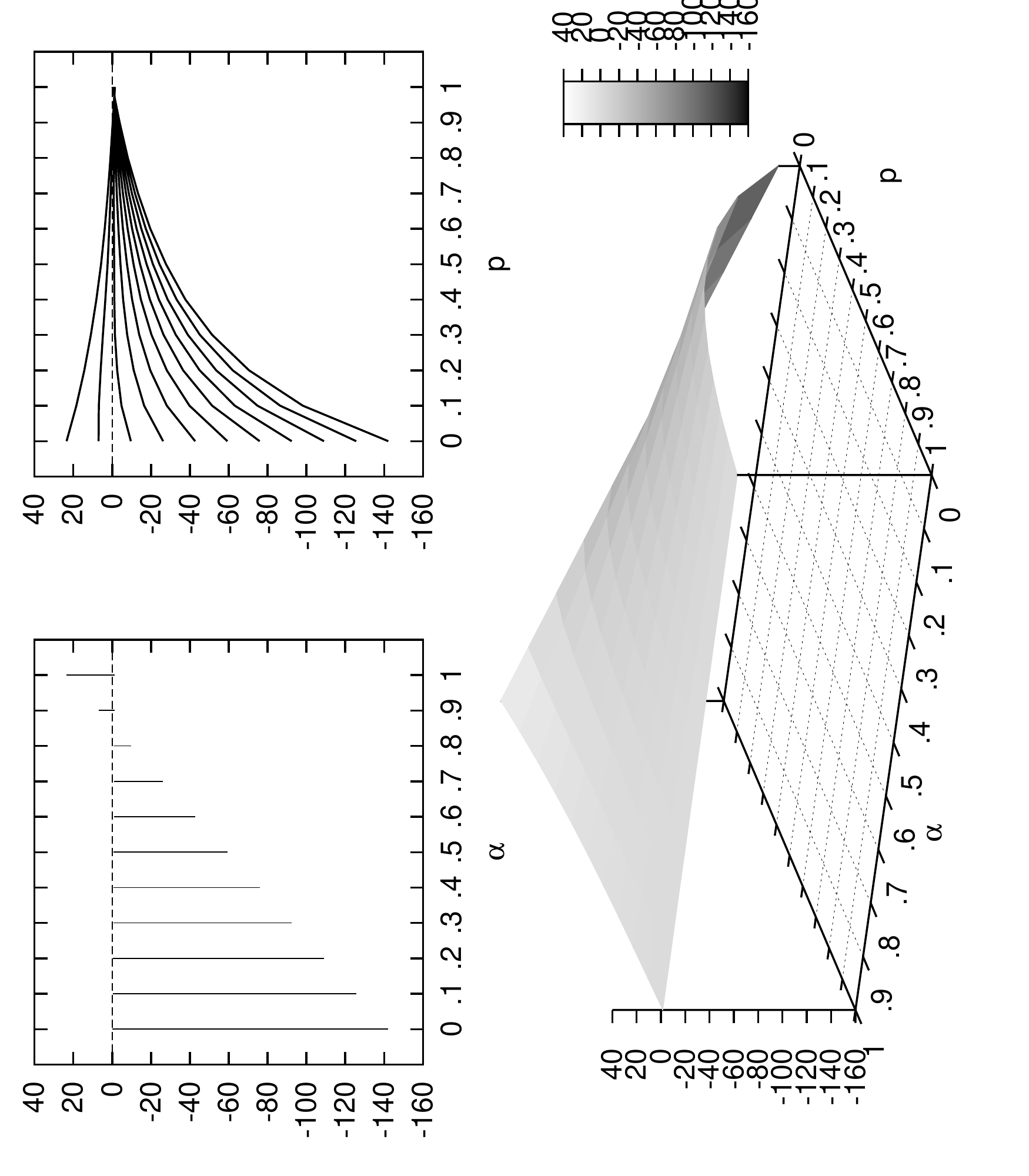}
  }
  \caption{$d(\alpha, p)$ results for $10^4$ tags and different
  values of $p$ and $\alpha$ in the scenarios with corridors and with a semi-directed movement pattern. \textbf{Values below
  zero indicate that our protocol is better with respect to the \textit{original} protocol}.}
  \label{fig:3d_corridors2}
\end{figure}

\begin{figure}[p]
\centering
  \subfigure[Random movement]{
  	\includegraphics[width=3.5in, angle=270]{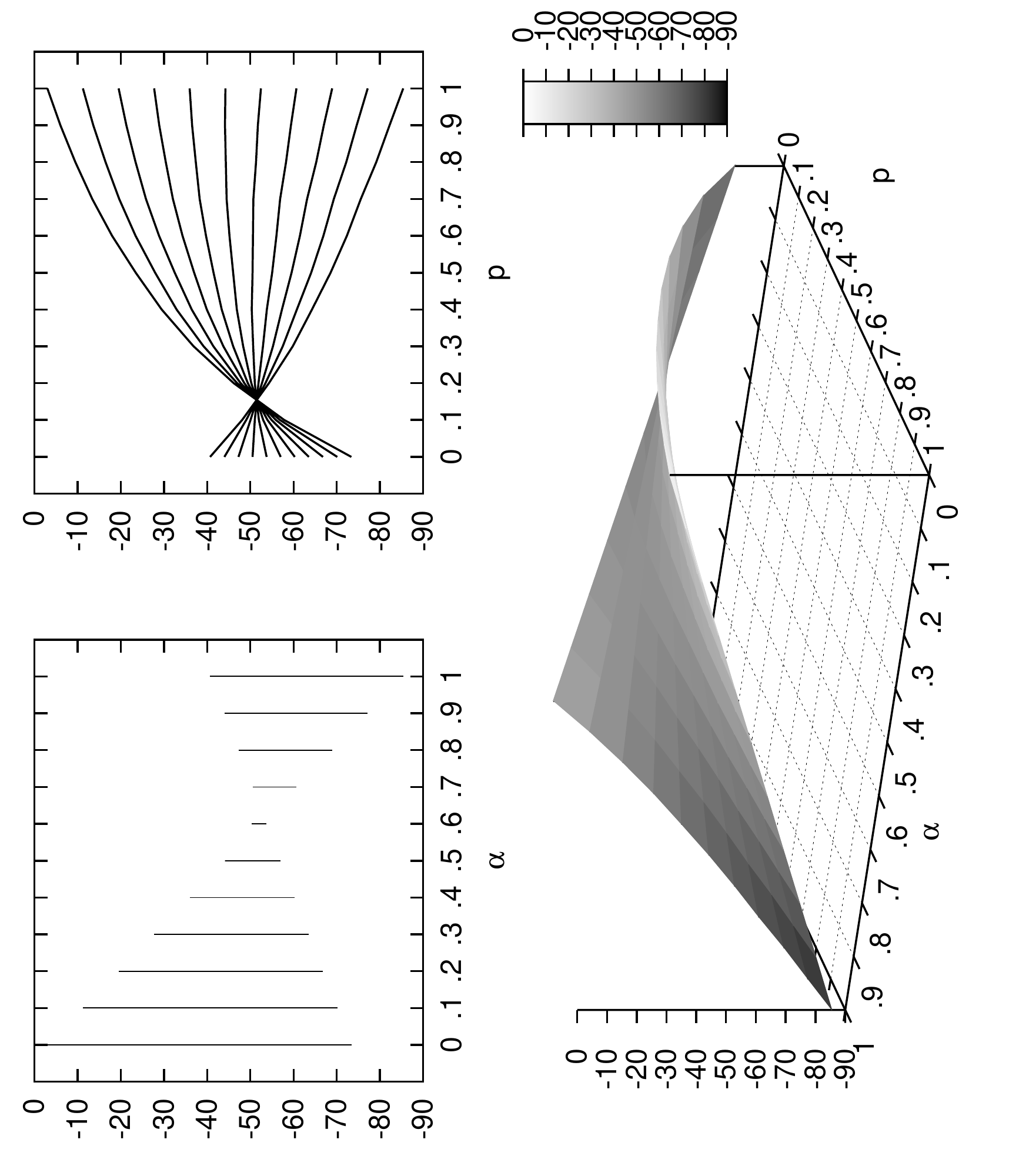}
  }
  \subfigure[Semi-directed movement]{
  	\includegraphics[width=3.5in, angle=270]{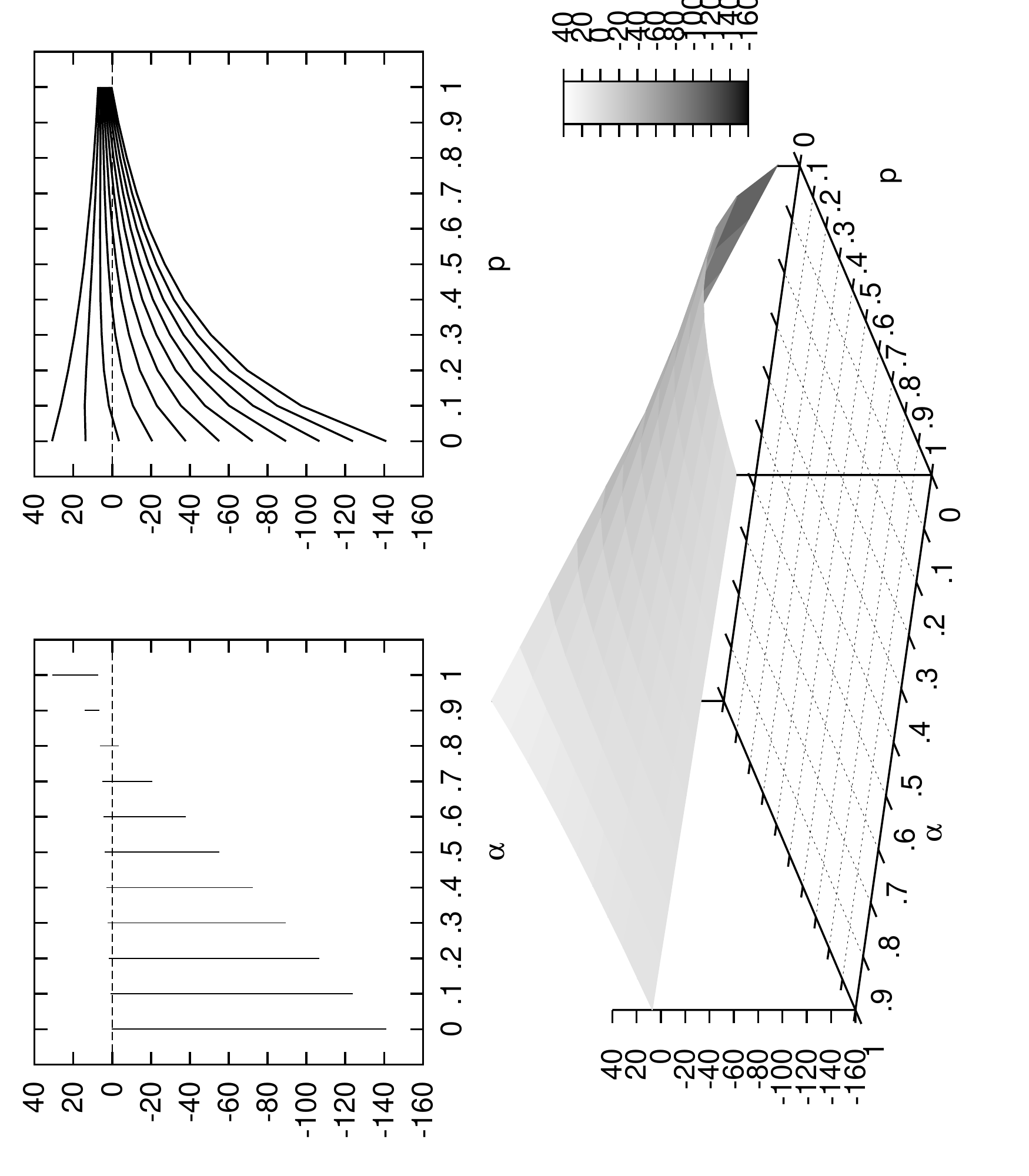}
  }
  \caption{$d(\alpha, p)$ results for $10^4$ tags and different
  values of $p$ and $\alpha$ in the empty scenario. \textbf{Values below
  zero indicate that our protocol is better with respect to the \textit{original} protocol}.}
  \label{fig:3d_empty}
\end{figure}

\begin{figure}[p]
\centering
  \subfigure[Random movement - ``Random Large Obstacles'' scenario]{
  	\includegraphics[width=3.5in, angle=270]{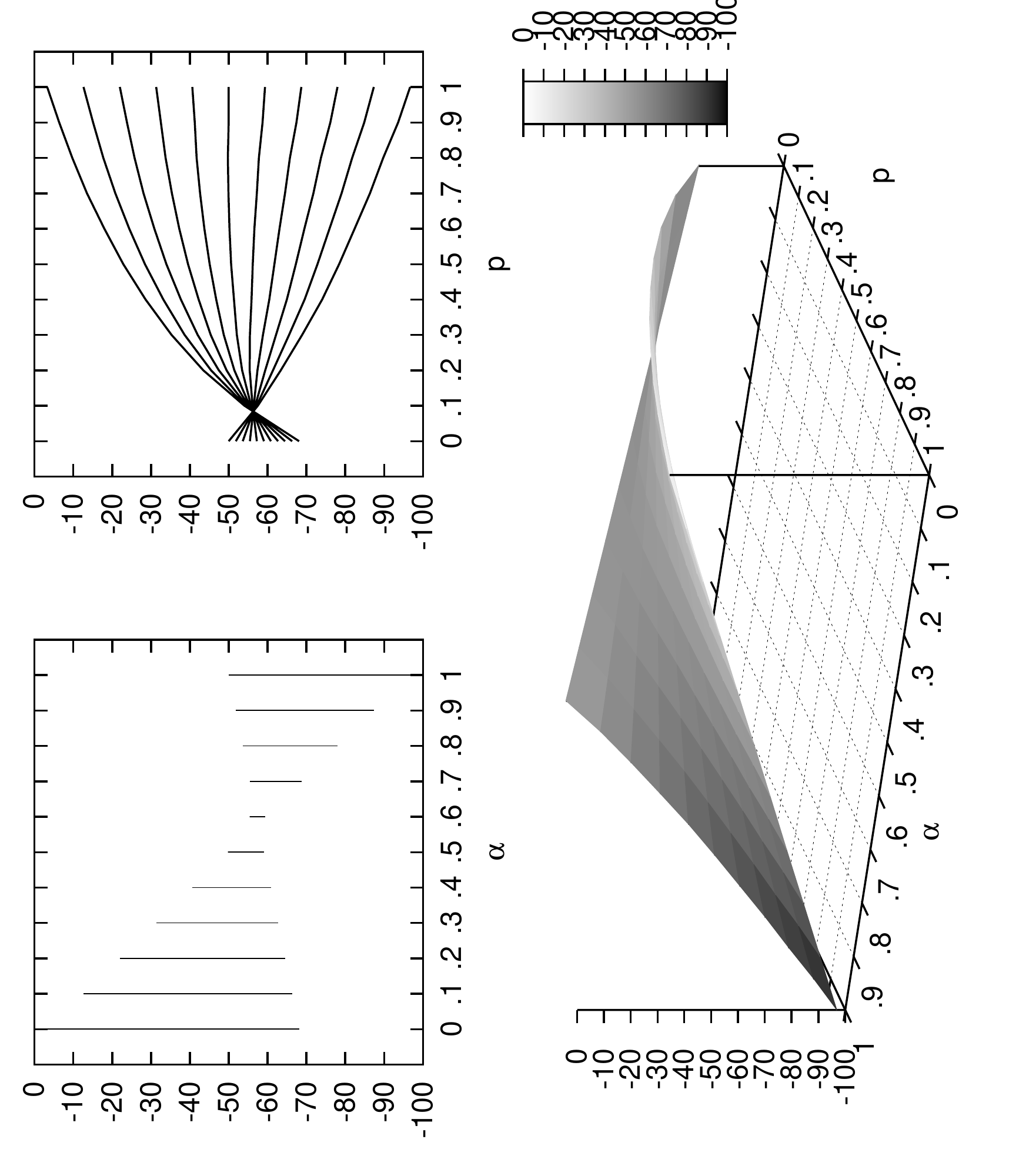}
  }
  \subfigure[Random movement - ``Random Small Obstacles'' scenario]{
  	\includegraphics[width=3.5in, angle=270]{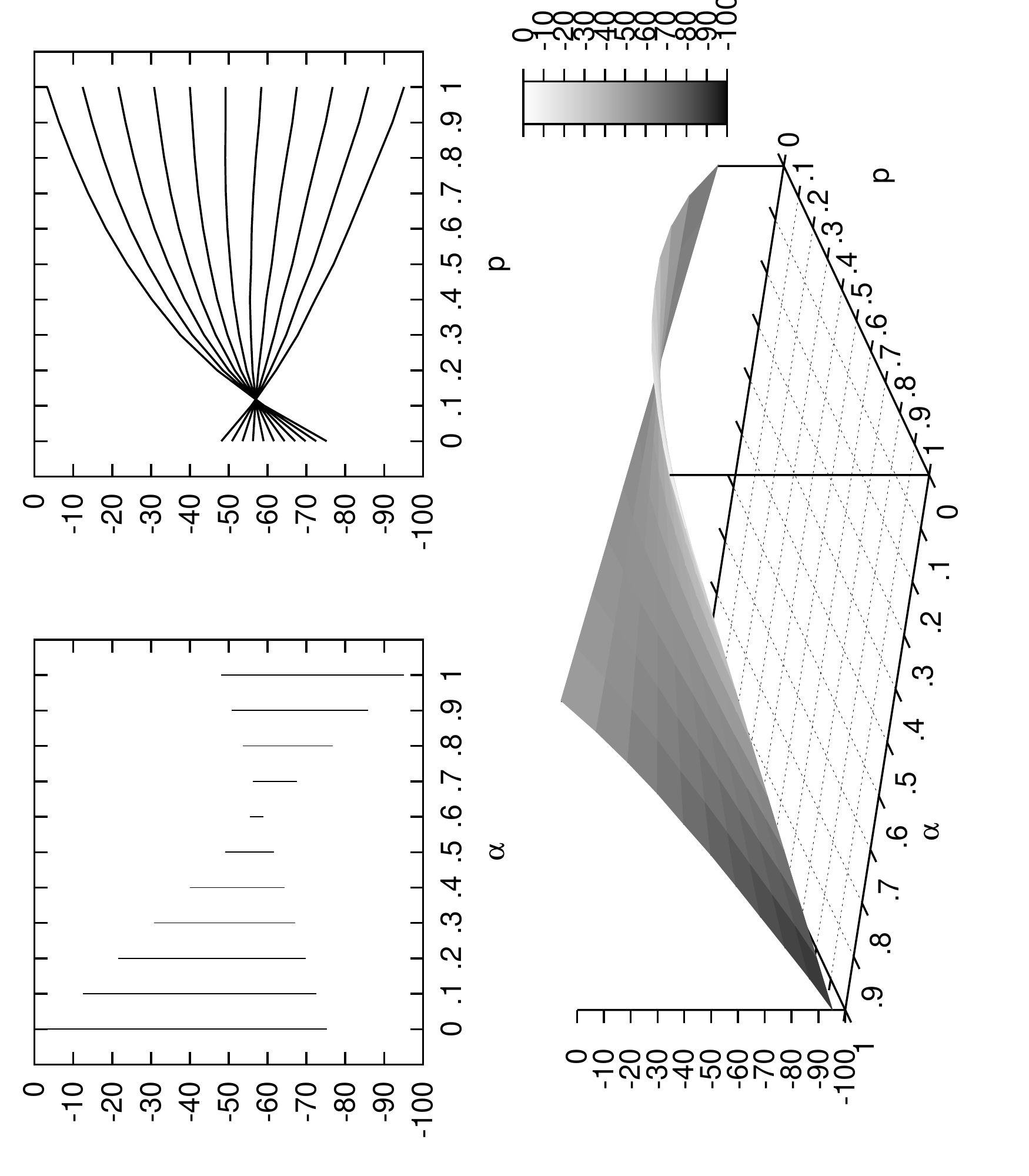}
  }
  \caption{$d(\alpha, p)$ results for $10^4$ tags and different
  values of $p$ and $\alpha$ in the scenarios with random obstacles and with a random movement pattern. \textbf{Values below
  zero indicate that our protocol is better with respect to the \textit{original} protocol}.}
  \label{fig:3d_obstacles1}
\end{figure}

\begin{figure}[p]
\centering
  \subfigure[Semi-directed movement - ``Random Large Obstacles'' scenario]{
  	\includegraphics[width=3.5in, angle=270]{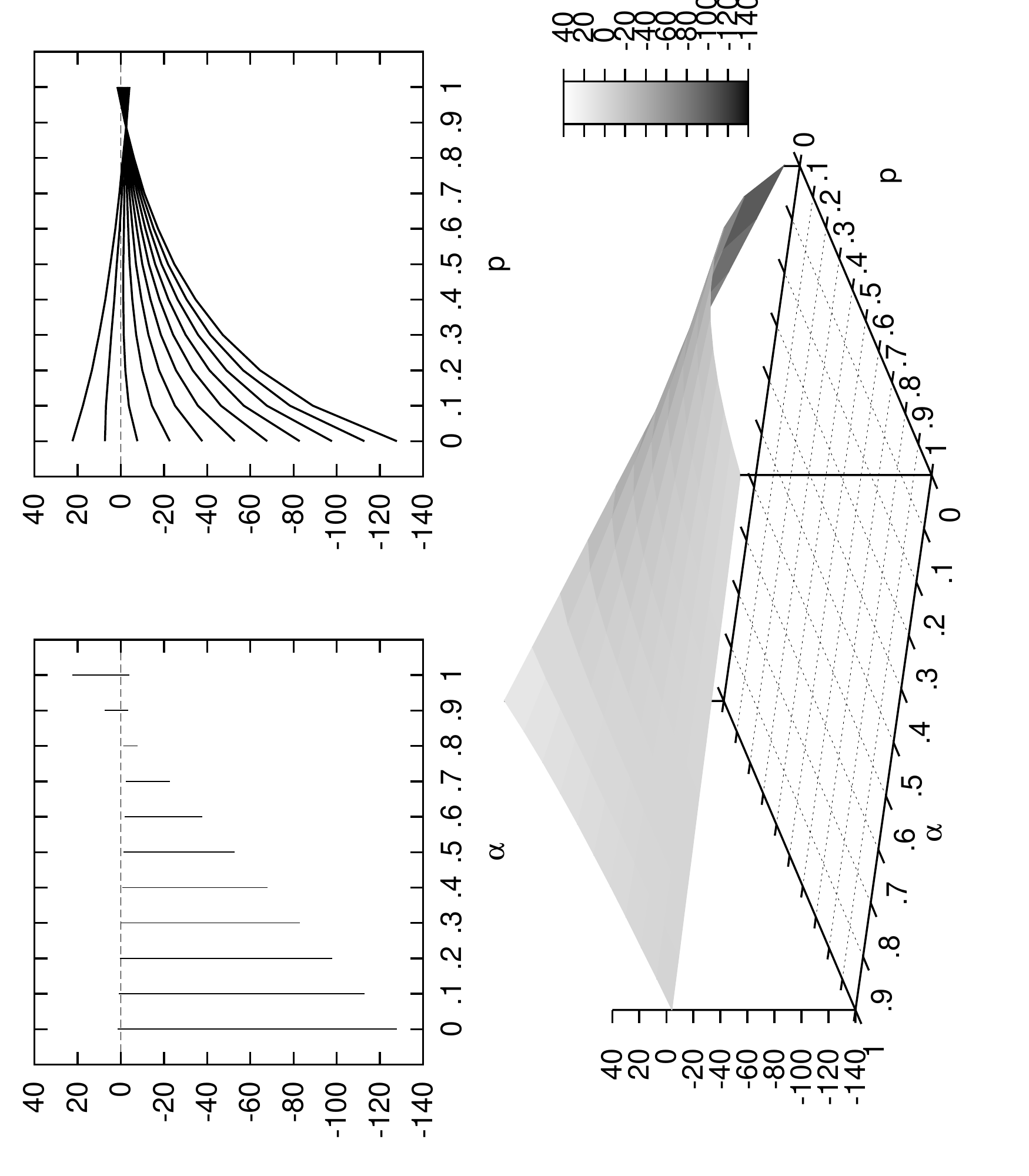}
  }
  \subfigure[Semi-directed movement - ``Random Small Obstacles'' scenario]{
  	\includegraphics[width=3.5in, angle=270]{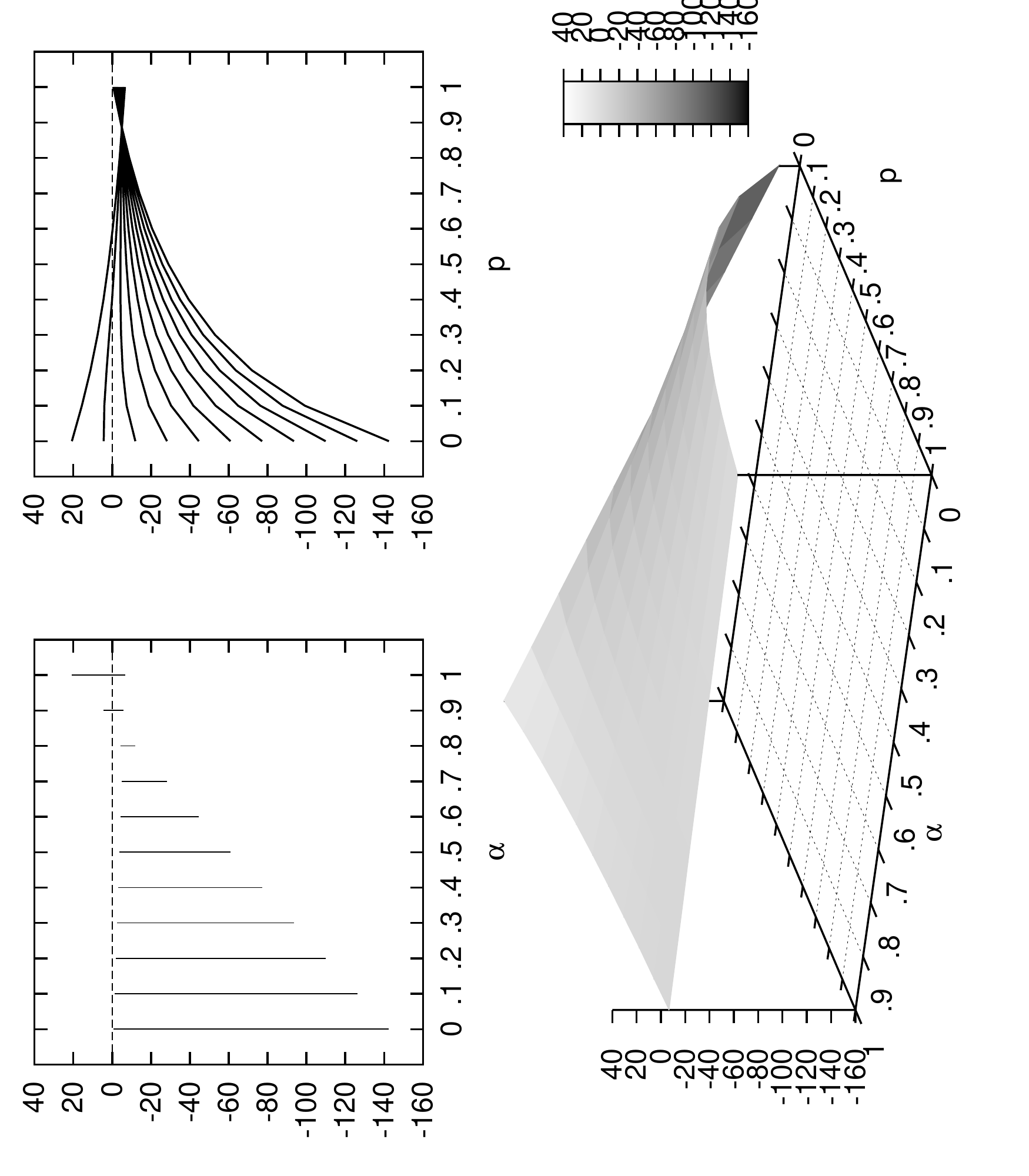}
  }
  \caption{$d(\alpha, p)$ results for $10^4$ tags and different
  values of $p$ and $\alpha$ in the scenarios with random obstacles and with a semi-directed movement pattern. \textbf{Values below
  zero indicate that our protocol is better with respect to the \textit{original} protocol}.}
  \label{fig:3d_obstacles2}
\end{figure}

It can be observed that our protocol outperforms the previous proposal in most
cases. When the movement of the tags is random, our protocol is always better for
all possible configurations. When the movement of the tags is semi-directed our
proposal is better in 81\% of the cases. That leads to a global improvement
in more than 90\% of all configurations.

\section{Conclusions}

In this chapter, we have presented an efficient communications protocol for collaborative RFID readers
to privately identify RFID tags. With the presented protocol, the centralised management
of tags can be avoided along with bottlenecks and undesired delays.

Our protocol is not a simple modification of previous proposals but a completely
different approach that clearly improves the efficiency and flexibility of the
whole system. In addition, due to the probabilistic nature of our protocol, the system
becomes very flexible, \emph{i.e.} the relation between computational cost and communications
overhead can be easily tuned by means of $p$. The simulation results confirm that our
protocol outperforms previous approaches like~\cite{SolanasDMD-2007-cn}.

Although the presented protocol is an improvement, there are some open issues that
should be considered in the future, namely (i) study the effect of the number of
neighbours, (ii)  propose methods to dynamically vary $p$ so as to adapt it to
the movements of tags, (iii) propose hybrid methods that mix hash-based solutions and tree-based
solutions with collaborative readers, etc.

\chapter{Predictive Protocol for Scalable Identification of RFID Tags
through Collaborative Readers}
\label{chap:4_5}

\emph{This chapter presents a natural improvement of previous RFID identification protocols based on collaborative readers. The described protocol improves the identification process be predicting the locations of the moving tags.}

\minitoc

Let us consider an RFID system intended for identification and tracking (\emph{e.g.} tracking of goods in a supply chain or luggage control in an airport). In such applications, several RFID readers are distributed over the system in order to identify tags passing through the RFID reader positions~\cite{SolanasDMD-2007-cn, Fouladgar:2008:SPP:1461464.1461467, Tavares-JPR, Bornhovd:2004:IAD:1316689.1316790, 10.1109/IDEAS.2006.47, Cao_architecturalconsiderations}. By doing so, it is possible to obtain the trajectory of a tag by concatenating the reader's positions where the tag has been identified. Even in applications without tracking purposes, it makes sense to distribute a set of readers covering strategic points or the whole monitored area~\cite{SolanasDMD-2007-cn} in order to identify the tags moving in it. Supermarkets with several
entry/exit doors or department stores are genuine examples of such applications.

Although there are several applications where many tags should be identified using some readers, to the best of our knowledge, only two protocols~\cite{SolanasDMD-2007-cn, Fouladgar:2008:SPP:1461464.1461467} exploiting this particular property have been proposed so far. The first of them~\cite{SolanasDMD-2007-cn}
introduced the idea of using multiple collaborative readers to make the
identification process scalable whilst maintaining the high level of privacy of
the IRHL scheme~\cite{JuelsW-2007-percom}. Their proposal is aimed at efficiently identifying tags in
applications where each tag must be continuously monitored while it remains
in the system. This implies that readers must cover the whole system. Under
this assumption, tags are constrained to move along neighbour
readers\footnote{Two readers are said to be neighbours if their cover areas are
not disjoint.} and therefore, neighbour readers collaborate in order to
guarantee efficiency during the identification process. Efficiency is achieved
by means of the so-called \emph{reader's cache}, which is defined as a storage
device where a reader saves tag identification data\footnote{This cache can
be either an external database securely connected to the reader or a database
internally managed by the reader itself.}. The protocol reduces the size of the
readers' cache by considering that only the closest reader to some tag and its
neighbours must store the identification information of this tag. By reducing
the size of the cache the identification procedure becomes more efficient.
Despite the benefits in terms of computational cost provided by this protocol,
assuming that readers are able to compute their accurate distance to tags is a
bit unrealistic.

On the other hand, in the context of using multiple readers (connected to a centralised back-end),
Fouladgar and Afifi~\cite{Fouladgar:2008:SPP:1461464.1461467} point out that, in many
applications, tags are usually queried by the same set of readers. Therefore,
they propose to cluster tags according to the readers that identify them more
often. This idea improves the group-based proposals in the sense that tags are
not randomly assigned to groups, but intelligently clustered according to the
spatial location of the readers that identify them. By doing so, when a reader
receives a tag's response, it first performs a search on the group of tags that
it usually identifies. If it does not succeed, an exhaustive search is performed
over the whole set of tag identifiers. The problem of this proposal is that
tags may have a long life-cycle and move through a wide variety of readers. In
this scenario, the protocol could scale as poorly as previous protocols based on
symmetric key cryptography~\cite{JuelsW-2007-percom}.

We show that the scalability problems of some private protocols can be alleviated not only distributing readers throughout the system, but also by exploiting the spatial location of tags. Indeed, a tagged item usually follows a pre-established life-cycle and then it could be intelligently identified according to its expected spatial location. In this chapter, we propose an adaptive and distributed architecture aimed at efficiently identifying RFID tags based on their expected spatial location. Unlike previous proposals~\cite{SolanasDMD-2007-cn}, our architecture is suitable for all possible scenarios and adapts itself to the type of tag movement. We show empirical results based on synthetic data confirming the superiority of our architecture with respect to previous proposals \cite{SolanasDMD-2007-cn} and \cite{Fouladgar:2008:SPP:1461464.1461467}.

\section{Trajectory-based RFID identification protocol}

In the Solanas {\em et al.} proposal~\cite{SolanasDMD-2007-cn}, the
readers' cache contains identification data of tags but it
lacks information about
the expected time at which the tags might next be identified by a reader or
where they were identified in the past. Assuming that it is possible to
approximately know the instant at which a tag will be identified by a given
reader, it is greatly beneficial to use this spatio-temporal information to
speed up the searching process in the readers' cache. Therefore, we propose to
structure the readers' cache as an ordered list where the expected time of
arrival (ETA) is the ordering criterion.
\begin{quotation}
\begin{definition}[Cache]
Given the set of tags $\mathcal{T}$ and readers $\mathcal{R}$ in the
system, the cache of a reader $R\in \mathcal{R}$ consists of a sequence of
ordered tuples
$$C(R) =\ <t_1, ID_1, R_{prev}^{ID_1}, R_{next}^{ID_1},Y|N>, \cdots, $$
$$,\cdots, <t_N, ID_N,R_{prev}^{ID_N}, R_{next}^{ID_N},Y|N>$$
where the order is given by the timestamps $t_1 \leq \cdots \leq t_N$. The tag
identifiers $ID_i \in \mathcal{T}, \ \forall \ 1 \leq i \leq N$, and
$R_{prev}^{ID_i} \in \mathcal{R}$ and $R_{next}^{ID_i}\in \mathcal{R}$ are the
reader that sent the $ID_i$ to $R$ and the reader that will receive the $ID_i$
from $R$, respectively. $Y|N$ is used as a flag to show whether the tag has been
already identified by this reader.\\
\end{definition}
\end{quotation}

From the above definition it can be observed that our protocol will use the
spatial information about the trajectory of the tags to predict
which reader will be the next reader to receive a given tag. Our protocol will also use the temporal
information of such trajectories to predict when a given tag will be read in the
future by the next reader. Table~\ref{tab:cache} is an example of the cache
of a reader. In this example, the reader $R_{512}$ expects to receive the tag
$T_{90876534}$ from reader $R_{1012}$ at time \textit{2011-07-28 11:31:38}, and
will forward the identification information to the next reader $R_{201}$. Also,
it can be seen that the tag has not been identified by the reader yet.

\begin{table}[tb]
\centering
\begin{tabular}{|c|c|c|c|c|}
\hline
\multicolumn{5}{|c|}{\textbf{Cache of Reader 512}}\\
\hline
    &        & Previous & Next   & First\\
ETA & Tag ID & Reader   & Reader & Time\\
\hline
2011-07-28 11:31:38 & 90876534 & 1012 & 201 & Yes\\
2011-07-28 11:41:33 & 10311299 & 1011 & 1201& No\\
$\cdots$ & $\cdots$ & $\cdots$ & $\cdots$   & $\cdots$\\
2011-07-30 22:01:08 & 21134211 & 1012 & 201 & No\\
\hline
\end{tabular}
\caption{Example of the cache of a reader.}
\label{tab:cache}
\end{table}

By using this ordered cache, when a tag response arrives at a given timestamp
$t$, a reader is able to optimise the searching process in its cache by first
considering the tags that it expects to identify at a timestamp $t'$ close to
$t$. Note that if the ETA is accurate, the identification of tags might be very
fast. The better the prediction, the faster the identification process. In the
worst case the computational cost is $O(n)$, where $n$ is the number of
identifiers in the cache of the reader.

\subsection{Trajectory prediction algorithms}
\label{sec:trajectory}

We propose to use trajectory predictors extensively so as to be able to
inform readers about which tags they will receive and when, before they actually
receive them. However, when this prediction fails, we propose to use other predictors to find the reader that
might have the information about a tag. These latter predictors consider the movement of all tags globally, \emph{i.e.} they look for global trends instead of predicting the moves of a single tag as does by the former predictors.

In general, a trajectory is understood as a timely ordered set of consecutive
points $(\mathcal{P})$ defined in an $n$-dimensional space $(\mathcal{S})$.
However, due to the fact that we can only control the location of the tags when
they are detected by a reader, we define our concept of trajectory as follows:
\begin{quotation}
\begin{definition}[Trajectory]
Given a set of readers $\mathcal{R}$ and tags $\mathcal{T}$. The trajectory of
a tag $T_i \in \mathcal{T}$ is defined as a sequence
$$S_i = <t_1, R_{1}>, <t_2, R_{2}>, \cdots, <t_{s(i)}, R_{s(i)}>$$
where $s(i)$ is the size of the sequence, $t_1 < t_2 < \cdots < t_{s(i)}$ are
timestamps and, $R_{j} \in \mathcal{R} \ \forall 1 \leq j \leq s(i)$ are the
readers that identified the tag $T_i$ at the timestamp $t_j$.
\\
\end{definition}
\end{quotation}

\begin{figure}[tb]
\centering
\includegraphics[scale=0.25,angle=0]{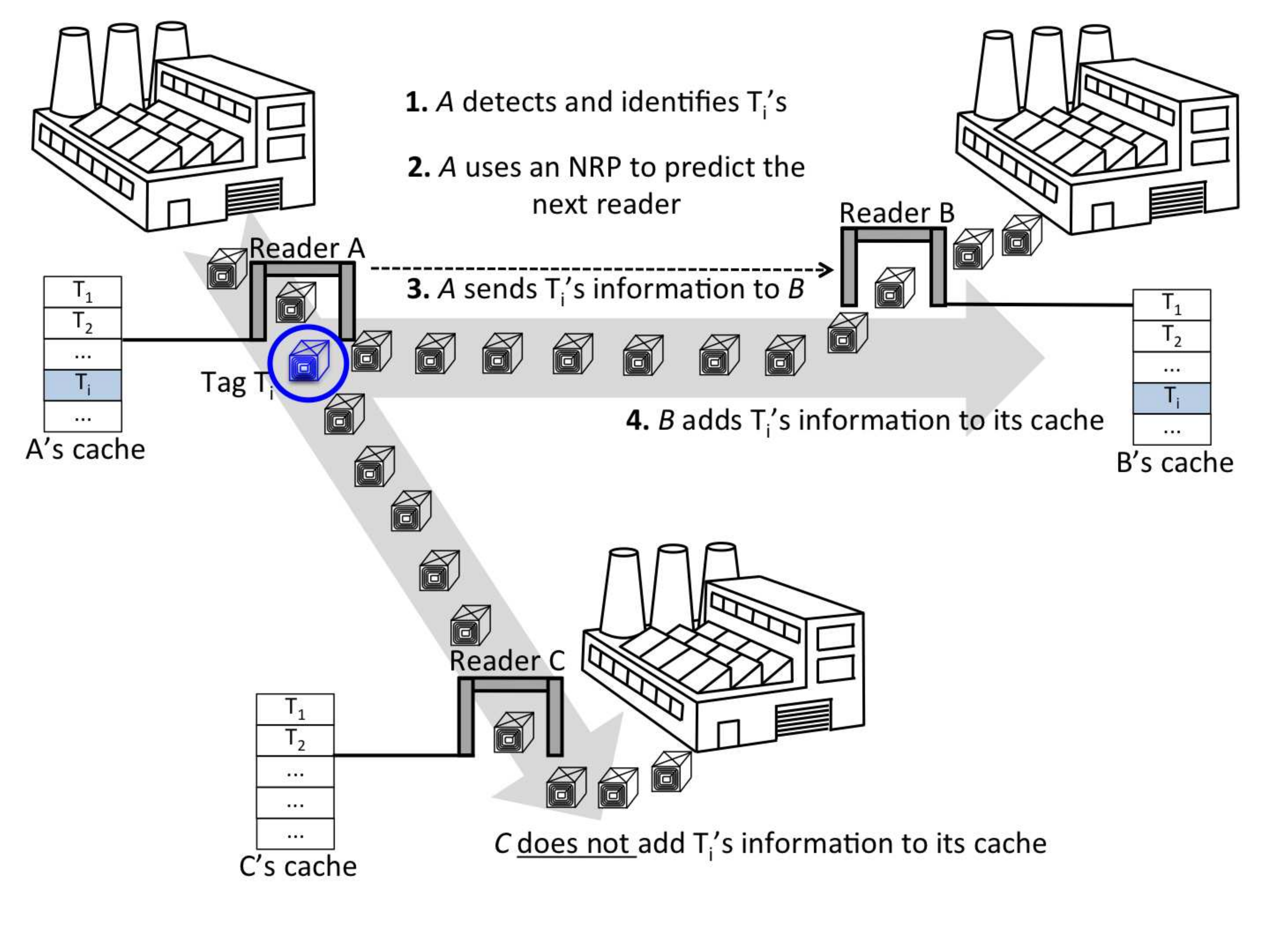}
\caption{Illustration of the identification success of the tag $T_i$ by reader
$A$. After successfully identifying $T_i$, the reader $A$ applies a Next Reader
Predictor (NRP) to predict the next reader (in this case, reader $A$ decides
that reader $B$ is the best candidate) and sends the $T_i$ identification
information to reader $B$. The reader $B$ stores the information about $T_i$
in its cache so as to be able to identify it (if necessary).}
\label{fig:NRP}
\end{figure}

When a tag arrives at the cover area of a reader, the reader tries to identify
it by applying the already explained IRHL protocol. During the identification
process two situations could arise:
\begin{enumerate}
	\item \textbf{Identification success}: The reader finds the identification
	information of the tag in its cache and can identify it. Then it has to
	decide to which reader should this information be forwarded (see the example
	of Figure~\ref{fig:NRP}).
	\item \textbf{Identification failure}: The cache of the reader does not
	contain the identification information of the tag and the reader cannot
	identify it. The reader has to decide which other reader to ask for help
	(see the example of Figure~\ref{fig:PRP}).
\end{enumerate}
In the first case (\textit{identification success}), after properly identifying
a tag, the reader will proceed by using a \emph{Next
Reader Predictor} (NRP) algorithm to determine which reader will be the next one to which
the tag will move. Once this next reader is determined, the current reader
sends the identification information of the tag to that reader. An NRP can be
defined as follows:

\begin{quotation}
\begin{definition}[Next Reader Predictor (NRP)]
Let $T_i$ be a tag of the system and let $S_i = <t_1, R_{1}>, <t_2, R_{2}>,
\cdots, <t_{j}, R_{j}>$ be its trajectory.
An \emph{NRP} is a polynomial-time algorithm (let us call it
$\mathcal{A}_{next}$) that, on input $T_i$ and $S_i$, outputs the pair
$<t_{j+1}, R_{j+1}>$.
$$ \mathcal{A}_{next}(T_i, S_i) \longrightarrow  <t_{j+1}, R_{j+1}>$$
This output pair means that it is expected that the tag $T_i$
will be identified at time $t_{j+1} > t_{j}$ by the reader $R_{j+1}$.
\end{definition}
\end{quotation}

Note that the result of the NRP is correct only with a probability that highly
depends on the utilised algorithm and the degree of regularity of the movement
of tags\footnote{It is apparent that in a chaotic system where no regularities
exist, the prediction of the next move of a tag would be extremely
inefficient.}.
Thus, if the prediction is wrong, the reader which is currently identifying the
tag $T_i$ will forward the identification information to a wrong reader. As a
consequence, when that tag reaches the next reader, the latter will not be able
to identify the tag (because the identification information will not be in its
cache) and will need the help of other readers to do so (this is the second
case enumerated above).

In the second case (\textit{Identification failure}), when a reader cannot
identify a tag, it proceeds by using a \emph{Previous
Reader Predictor} (PRP) algorithm to identify the reader that might have identified the
tag previously and might have the identification information of the tag. A PRP
can be defined as follows:

\begin{figure}[tb]
\centering
\includegraphics[scale=0.25,angle=0]{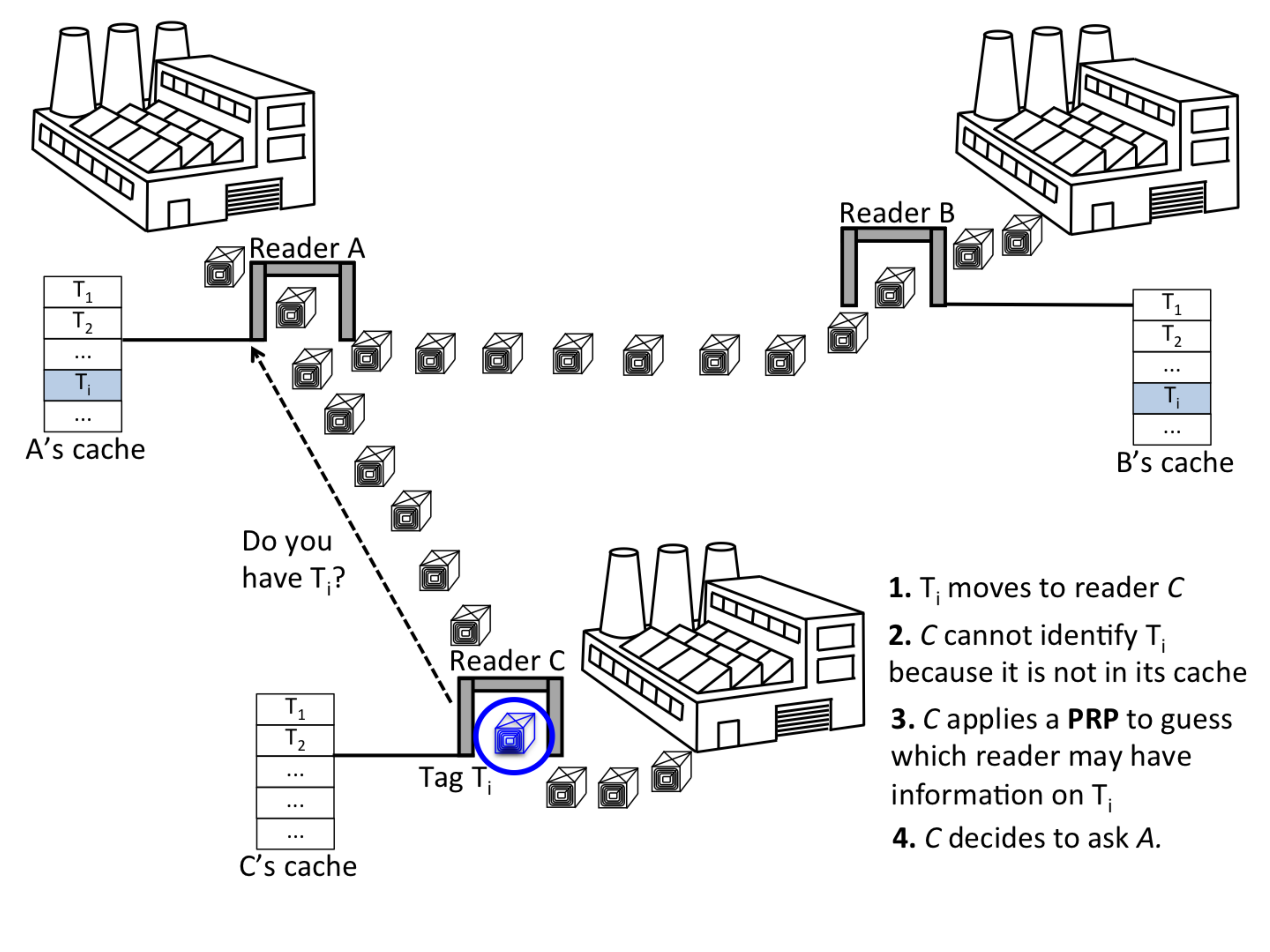}
\caption{Illustration of the identification failure. Reader $C$ tries to
identify $T_i$ but it fails because it has not got the information in its cache.
It uses a PRP to decide which reader to ask for help (in this case $C$ asks
$A$).}
\label{fig:PRP}
\end{figure}

\begin{quotation}
\begin{definition}[Previous Reader Predictor (PRP)]
Let $\mathcal{T}=\{T_1, \cdots, T_{N}\}$ be the set of tags in the system and
let $\mathcal{S} \ = \ \{S_1, \cdots, S_{N}\}$ be the set of trajectories of the
tags in $\mathcal{T}$ until a given time $t$. Let $\mathcal{T^R}\subset
\mathcal{T}$ be the subset of tags known by reader $\mathcal{R}$ and let
$\mathcal{S^R}\subset\mathcal{S}$ be the trajectories of the subset of tags
known by reader $\mathcal{R}$.  A \emph{PRP} is a polynomial-time algorithm
(let us call it $\mathcal{A}_{prev}$) that, on input a reader $\mathcal{R}$
and a set of trajectories of tags $\mathcal{S^R}$, outputs the sequence of $k$
readers $R_{1}, R_{2}, \cdots, R_{k}$ that are candidates to be the previous
reader that identified a tag:
$$ \mathcal{A}_{prev}(\mathcal{R}, \mathcal{S^R}) \longrightarrow
<R_{1}, R_{2}, \cdots, R_{k}>$$
\end{definition}
\end{quotation}

\begin{quotation}
\begin{remark}
The order of the sequence of candidate readers depends on the specific
implementation of the predictor. However, the following condition must hold:
$$P(success|R_1)\geq P(success|R_2)\geq \cdots \geq P(success|R_k)$$
This means that $R_1$ has more chances of being the actual previous reader than
$R_2$, etc.
\end{remark}
\end{quotation}
Note that we have defined the theoretical concept of
NRP and PRP algorithms. However, the specific
implementation of these algorithms would highly affect the performance of the
whole system. Further below, we will give details on the implementations
that we have used for our experimental analysis.

\subsection{Our protocol}
\label{sec:protocol}

We define our protocol as a distributed algorithm in the context of a set of
collaborative readers $\mathcal{R}$ that share identification information on
a number of tags $\mathcal{T}$. For the sake of completeness, in addition to
$\mathcal{R}$, we consider a special reader $\mathcal{O}_R$ that acts as an
oracle, \emph{i.e.} it has the same role of classical back-ends that have
the information of all tags in the system. $\mathcal{O}_R$ can
identify any tag in $\mathcal{T}$, hence no false negative identifications
occur. However,
in our collaborative context, the oracle
should be understood as the \textit{``last
resort''} to identify a tag if all the other mechanisms fail\footnote{Note that this situation might happen rarely and
probably it would be caused by a communication failure amongst the
collaborative readers or by an active attack.},
because
the computational cost associated to the identification of tags by the oracle
is very high.

\begin{figure}[tb]
\centering
\includegraphics[scale=0.20,angle=0]{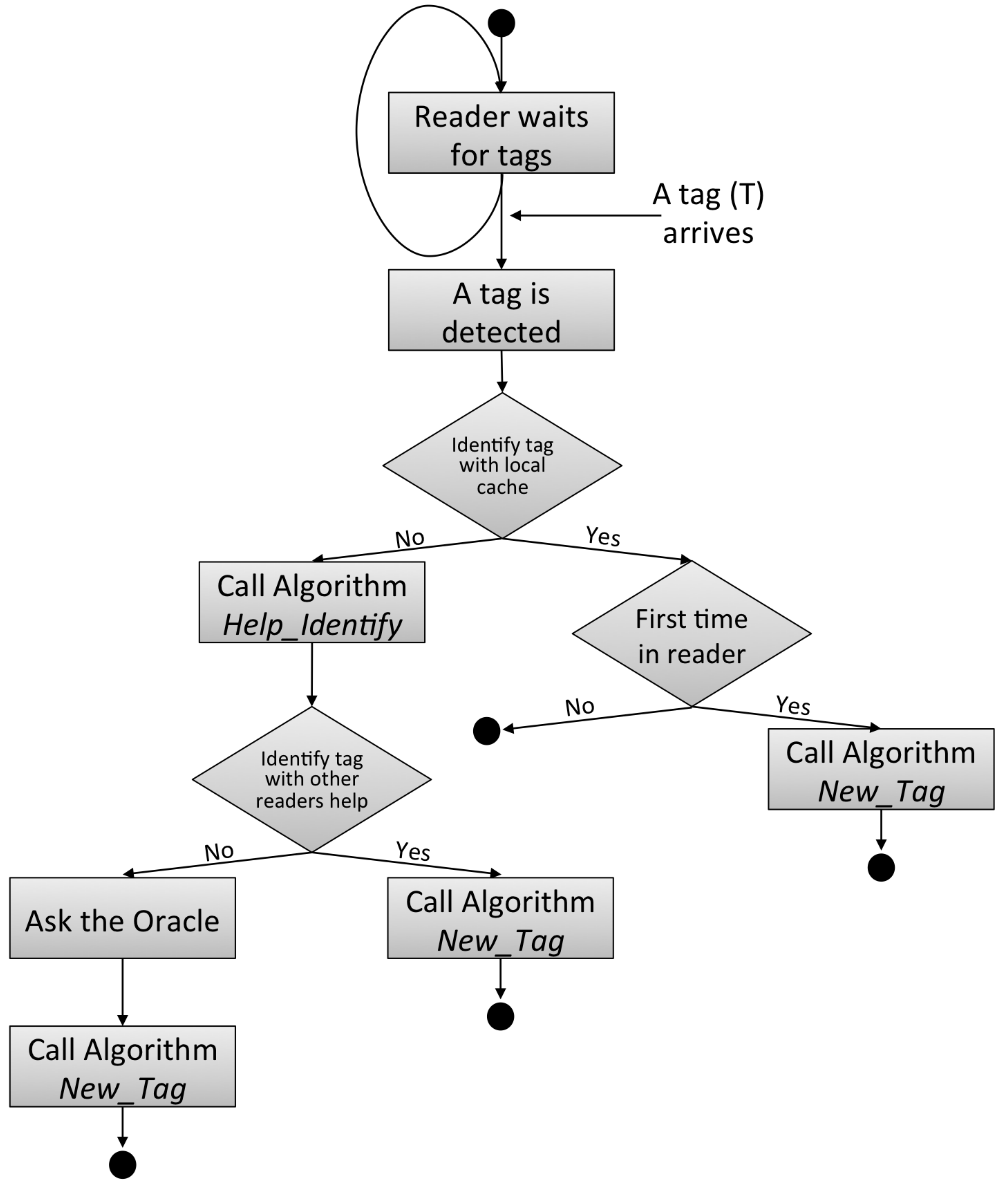}
\caption{Conceptual logical flow of the protocol.}
\label{fig:Protocol}
\end{figure}

\begin{algorithm}[!ht]
\caption{Main protocol} \label{alg:main_protocol}
\begin{algorithmic}[1]
\STATE \textbf{Require:} (i) A tag $T$ to be identified by a reader $R$ at time $t$; (ii) The set of tags' trajectories $\mathcal{S^R}$ known by $R$; (iii) The cache of reader $R$, $C(R)$.\\
$ $\\
\textit{- - Try to identify T using the local cache}
\FORALL {$t' \in \{t, t+1, t-1, t+2, t-2, \cdots\}$}
    \FORALL{$T_i \in C(R)$ with $ETA=t'$}
        \IF {$T$ is identified as $T_i$}
            \IF {$T_i$ is a ``first time'' tag}
                \STATE Call \emph{New\_Tag}
                ($R$, $T_i$, $t$, $R_{prev}^{T_i}$, $S_i$)
            \ENDIF
            \STATE \textbf{Return:} (\textit{Tag identified correctly});
        \ENDIF
    \ENDFOR
\ENDFOR \\
\textit{- - Try to identify the tag with the help of other readers}
\FORALL {$R'\in <R_1, R_2, \cdots, R_k> \longleftarrow \mathcal{A}_{prev}(R, \mathcal{S^R})$}
    \STATE Call \emph{Help\_Identify} ($T$,$R'$)
    \IF {$R'$ identifies $T$}
        \STATE $R$ receives $<t', T_i, R_{prev}^{T_i}, R_{next}^{T_i}>$ from $R'$.
        \STATE Call \emph{New\_Tag}
        ($R$, $T_i$, $t$, $R'$ and, $S_i$).
        \STATE \textbf{Return:} (\textit{Tag identified correctly});
    \ENDIF
\ENDFOR
\\
\textit{- - (\textit{last resort}) Ask the Oracle}
\IF {$\mathcal{O}_R$ identifies $T$}
    \STATE Call \emph{NewTag} ($R$, $T$, $t$, $\mathcal{O}_R$, $\emptyset$).
    \STATE \textbf{Return:} (\textit{Tag identified correctly});
\ENDIF
\STATE \textbf{Return:} \textbf{Invalid tag $T$ found}
\end{algorithmic}
\end{algorithm}

Algorithm~\ref{alg:main_protocol} shows a pseudocode description of our protocol
and Figure~\ref{fig:Protocol} depicts the logical flow of the proposal.
Our protocol works as follows. The reader $R$, that receives an identification
message from an unidentified tag $T$ at time $t$, tries to identify it by
following the Improved Randomised Hash Lock Scheme (IRHL)~\cite{JuelsW-2007-percom, Juels:2009:DSP:1609956.1609963} but using the identification information stored in its own cache only \textit{(lines 1 to
10 in Algorithm~\ref{alg:main_protocol})}).
In order to perform this identification efficiently, the
reader uses the cache structure described above. First, it tries to identify $T$
as one of the tags that were expected to arrive at time $t$. If the tag is not
identified amongst these candidate tags, the reader tries with tags that were
expected to arrive a bit later at time $t+1$ and a bit earlier at time
$t-1$, and so on. Searching in this way, if the ETA of $T$ was properly
predicted and forwarded, $T$ is identified almost instantly. However, if
the prediction was wrong, the reader $R$ might need to search over all its
cache. If $T$ is identified, $R$ checks whether it is the first time that this
tag enters its interrogation zone, \emph{i.e.} it is a ``first time'' tag. If it is, it calls the procedure
\textit{New\_Tag}($R, T_i, t, R^{T_i}_{prev},S_i$) and the identification
finishes. If the tag is not a ``first time'' tag, the identification procedure
simply finishes.

If the identification information of $T$ was not properly forwarded to
$R$\footnote{Note that this might happen due to a wrong prediction of the
next reader by the previous reader.}, it will search over all its cache and will
not be able to identify $T$. In this situation, it has to ask for help to the
other collaborative readers that might have the information it needs
(\textit{lines 11 to 18 in Algorithm~\ref{alg:main_protocol}}). To do so, $R$
calls the PRP algorithm $\mathcal{A}_{prev}(R, S)$ so as to obtain a list of
readers that may have information about $T$. For each reader $R'$ in the list
returned by $\mathcal{A}_{prev}$, the procedure \textit{Help\_Identify}
($T$,$R'$) is called. If this procedure succeeds in identifying $T$, the
collaborative reader that succeeds sends the tuple of its cache that contains
the information about $T$, \emph{i.e.}
the tuple $<t', T_i, R_{prev}^{T_i}, R_{next}^{Ti}>$) is sent to $R$.
By using the information in this tuple, the identification process correctly
finishes after calling the procedure
\textit{New\_Tag}($R$, $T_i$, $t$, $R'$, $S_i$).

Finally, if no reader $R'$ can identify $T$, $R$ asks the oracle
$\mathcal{O}_R$ (\textit{lines 19 - 23 in Algorithm~\ref{alg:main_protocol}}).
If $\mathcal{O}_R$ cannot identify $T$, the latter can be considered
an illegitimate tag\footnote{In this case, the proper actions are to be
taken, namely raise an alarm, locate and eliminate the tag, etc.}. Otherwise,
$R$ finishes successfully the identification process by calling procedure
\textit{New\_Tag} ($R$, $T$, $t$, $\mathcal{O}_R$, $\emptyset$).

\begin{algorithm}[!ht]
\caption{New\_Tag} \label{alg:new_tag}
\begin{algorithmic}[1]
\STATE \textbf{Require:} (i) A reader $R$ that has identified a tag $T_i$ at time $t$; (ii) The reader $R^i_{prev}$; (iii) The trajectory $S_i$ of $T_i$.\\
$ $\\
\STATE $R$ asks $R_{prev}^{T_i}$ to remove $T_i$ from its cache;
\STATE $R$ predicts the next reader and ETA
$<t_i, R_{next}^{T_i}> = \mathcal{A}_{next} (T_i, S_i)$;
\STATE $R$ asks $R_{next}^{T_i}$ to insert the record
$<t_i, T_i, R, null, Y>$ in its cache;
\STATE $R$ removes the record about $T_i$ from $C(R)$ \textit{(if it exists)};
\STATE $R$ inserts the record  $<t, T_i, R_{prev}^{T_i}, R_{next}^{T_i}, N>$
into $C(R)$;
\STATE $R$ adds $<t, R>$ to $T_i$'s trajectory ($S_i$);
\end{algorithmic}
\end{algorithm}

The main protocol described in Algorithm~\ref{alg:main_protocol} uses two
procedures (\textit{New\_Tag} and \textit{Help\_Identify}) to
update the state of the caches of other collaborative readers and to identify
tags for which the identifying reader has no information.

The \textit{New\_Tag} procedure, described in Algorithm~\ref{alg:new_tag}, is
called when a reader $R$ determines that a newly identified tag, $T_i$, has
entered its interrogation zone for the first time, \textit{i.e.} it is a
``First time'' tag, and thus, $T_i$'s trajectory must be updated. In this case,
$R$ sends a message to the previous reader $R_{prev}^{T_i}$ that identified
$T_i$ so as to let it remove the information it has about $T_i$\footnote{This
information is no longer necessary and removing it from the cache speeds up
the identification procedure.} (note that when $R_{pre}^{T_i} = \mathcal{O}_R$
this message is not sent). Then, $R$ uses an NRP to determine the next reader
that will be visited by $T_i$ and sends a message to it to let it insert the
tuple $<t_i, T_i, R, null, Y>$ in its cache (this way, when the tag reaches this
reader, it will be able to identify it efficiently). Finally, the record
corresponding to $T_i$ in $C(R)$ is updated with proper information about the
next reader $<t, T_i, R_{prev}^{T_i}, R_{next}^{T_i}, N>$.

\begin{algorithm}[!ht]
\caption{Help\_Identify} \label{alg:help}
\begin{algorithmic}[1]
\STATE \textbf{Require:} $T$ a tag to be identified by a reader $R$;
\STATE Determine $t_{old}$ the oldest timestamp in $C(R)$;
\FORALL {$t' \in \{t_{old}, t_{old}+1, t_{old}+2 \cdots\}$}
    \FORALL {$T_i\in C(R)$ with $ETA=t'$}
        \IF {$T$ is identified as $T_i$}
            \STATE $R$ asks $R_{prev}^{i}$ to remove the information of $T_i$
            from its cache;
            \STATE \textbf{Return:} $<t', T_i, R_{prev}^{T_i}, R_{next}^{T_i}>$
        \ENDIF
    \ENDFOR
\ENDFOR
\STATE \textbf{Return:} \textit{$<null>$ (Tag not identified)}
\end{algorithmic}
\end{algorithm}

The \textit{Help\_Identify} procedure, described in Algorithm~\ref{alg:help}, is
called when a reader $R$ cannot identify a tag with the information stored in
its cache. This procedure is executed by the readers that collaborate with $R$.
Due to the fact that these collaborative readers might have seen the unknown tag
quite in the past, they start searching tuples in their caches whose timestamps
are old. If a collaborative reader $R'$ identifies $T$ as $T_i$ it sends a
message to $R_{prev}^{T_i}$ in order to let it remove the information on $T_i$
from its cache. Finally, $R'$ returns the tuple about $T_i$ stored in its cache.

\subsection{Practical implementation of the predictors}
\label{sec:predictors}
Previously, we have theoretically defined the
concepts of Next Reader Predictor (NRP) and Previous Reader Predictor (PRP). Below, we propose practical implementations for each of
these predictors.

\subsubsection{Next reader predictor}
We propose to use a location prediction algorithm based on a Markov
model~\cite{SongEtAl-2004-Infocom}. A Markov-based predictor of order $k$,
$O(k)$, is defined over the sequence of the last $k$ locations of a given
moving entity. Let $L = \ell_1, \cdots, \ell_n$ be the location history of
a given entity and let $L(i, j) = \ell_i, \cdots, \ell_j$ be a subsequence of
$L$. Let $X_i$ be the random variable that represents a location at time $t$.
Then, the Markov assumption is that:
\begin{equation*}
\Pr(X_{n+1} = x | X_1 = \ell_1, \cdots, X_n = \ell_n) = \\
\end{equation*}
\begin{equation}\label{eq_4_5:1}
\Pr(X_{n+1} = x | X_{n-k+1} = \ell_{n-k+1}, \cdots, X_n = \ell_n)
\end{equation}

And that for every $i \in \{1, 2, \cdots, n-k\}$:

\begin{equation*}
\Pr(X_{n+1} = x | X_{n-k+1} = \ell_{n-k+1}, \cdots, X_n = \ell_n) =
\end{equation*}
\begin{equation}\label{eq_4_5:2}
\Pr(X_{i+k} = x | X_{i-k} = \ell_{n-k+1}, \cdots, X_{i+k-1} = \ell_{n})
\end{equation}

Simply stated, Equation~\ref{eq_4_5:1} says that the probability of being in
a given location depends on the previous $k$ locations only, whilst
Equation~\ref{eq_4_5:2} says that this probability is time independent.
Therefore, this probability can be represented by a transition matrix $M$
labelled with all possible sequences of locations of size $k$:

\begin{equation*}
\Pr(X_{n+1} = x | X_1 = \ell_1, \cdots, X_n = \ell_n) =
\end{equation*}
\begin{equation}
M(L(n-k+1, n), L(n-k+1, n)||x)
\end{equation}

And the value of $M(a, b)$ may be estimated by

\begin{equation}
M(a, b) = \frac{N(a, L)}{N(b, L)}
\end{equation}

Where $N(s_1,\ s_2)$ is the number of times the subsequence $s_1$ occurs in the
sequence $s_2$.

In our protocol, locations are represented by the readers $\mathcal{R}$ and
a next reader predictor (NRP) is only used by readers $\mathcal{R}$ once they
realise that a tag $T$ is in their interrogation zone. Thus, the last location
of $T$ is the current reader $R$, \emph{i.e.} $\ell_{n} = R$. Therefore, we
believe that a reader could be able to implement a Markov-based predictor of
order 1 or 2 using a reasonably small amount of memory. In addition, counting
the number of times that a tag is identified by a reader after having been
identified by another reader can be easily done when calling the
\textit{New\_Tag} procedure described in Algorithm~\ref{alg:new_tag}.
Our Markov-based predictor is computationally efficient. It has
a logarithmic computational cost with respect to the number of readers
$\mathcal{R}$.

Regarding the time prediction, we use a very simple approach. Let
$t_m$ be the average time in which a tag $T$ is identified by two consecutive
readers. Let $t$ be the current time in which $T$ is identified by a reader. We
estimate that the next reader will identify $T$ at time $t+t_m$. Note that the
readers store, share and update $t_m$. To update the value of $t_m$, the reader
applies the following equation:
$$
	t_m = \frac{t_m\times(c-1)+t-t_{last}}{c},
$$
where $c$ is the number of times that the tag has been identified and
$t_{last}$ is the last time in which that tag was identified.

\subsubsection{Previous reader predictor}

We propose the use of an heuristic to implement the previous reader predictor ($\mathcal{A}_{prev}$). In a nutshell, the proposed predictor works as follows: When a reader $R_i$
identifies a tag, it increments a counter $G(R_i,R_j)$, where $R_j$ is the last
reader that identified that tag. By doing so, when
$\mathcal{A}_{prev}(\mathcal{R}_i,\mathcal{S}^{\mathcal{R}_i})$ is called, it
outputs the sequence,
$$R_{1}, R_{2}, \cdots, R_{k}$$
such that
$$G(R_i, R_{1}) \geq
G(R_i, R_{2}) \geq \cdots \geq G(R_i, R_{k}).$$

The computational cost of $\mathcal{A}_{prev}$ is logarithmic
with respect to the number of readers $\mathcal{R}$. Note that there is no
need for sorting the output list every time the algorithm is called, \textit{i.e.} this
might lead to a computational complexity $O(|\mathcal{A}|\log |\mathcal{A}|)$.
On the contrary, the list could be stored already sorted and simply updated
after increasing the value of $G(R_i, R_j)$ for any pair of readers
$R_i$ and $R_j$.

Note that the PRP is essentially a ``global'' predictor in the sense that it
is based on the information of the trajectories of multiple tags. Consequently,
it can be seen as a trend analyser (\emph{e.g.} if most of the tags that are
identified by a reader $R_y$ move to a reader $R_x$, when the reader $R_x$ uses
the PRP, the first result will be $R_y$). On the contrary, the NRP previously described
is essentially ``local'' in the sense that it only
depends on the information of a single tag.


\section{Experimental results and evaluation}

We split this experimental section in two subsections.
The first is devoted to comparing
our proposal with both the Solanas {\em et al.}~\cite{SolanasDMD-2007-cn} and the Fouladgar and Afifi~\cite{Fouladgar:2008:SPP:1461464.1461467} proposals. The second subsection does not consider the Solanas {\em et al.}~\cite{SolanasDMD-2007-cn} proposal anymore because it was designed over assumptions quite different from ours. In turn, we show in this part of the experimental section how a good implementation of the next reader predictor algorithm ($\mathcal{A}_{next}$) improves the efficiency of the identification process.

\subsection{Experiments considering the Solanas {\em et al.} proposal}

As we explain above, the Solanas {\em et al.}~\cite{SolanasDMD-2007-cn} proposal considers a scenario where tags are continuously monitored by readers. To do so, readers must have a large interrogation field so as to cover the whole scenario. Consequently, a tag is likely to be identified several consecutive times by the same reader. Under this assumption, a next reader predictor algorithm ($\mathcal{A}_{next}$) does not make sense. Note that $\mathcal{A}_{next}(T_i, S_i)$ will output, with high probability, the last reader of the trajectory $S_i$ because that reader is likely to identify again $T_i$ in the next slot of time. Then, in this first part of the experiments, we remove $\mathcal{A}_{next}$ from our protocol and we recall this variant as \emph{Partial-predictive}. In turn, our proposal using both predictors is simply called \emph{predictive} and it will be evaluated in the second half of the experimental section.

\subsubsection{Scenarios}

With the aim to overcome the limitation of obtaining real datasets
of tag movement in a fine-grained fashion as required in~\cite{SolanasDMD-2007-cn}, we define two types of tag movements and three different scenarios with which we evaluate and compare our partial-predictive proposals with ~\cite{SolanasDMD-2007-cn} and ~\cite{Fouladgar:2008:SPP:1461464.1461467}.

The first scenario is an open area (cf. Figure \ref{fig:open_area}) where tags can freely move. The area is completely covered by 96 readers uniformly distributed over the whole area. By doing so, we meet the constraints of the Solanas {\em et al.} protocol w.r.t. the distribution of readers~\cite{SolanasDMD-2007-cn}.

\begin{figure}[tb]
\small
\centering
   \includegraphics[scale=1]{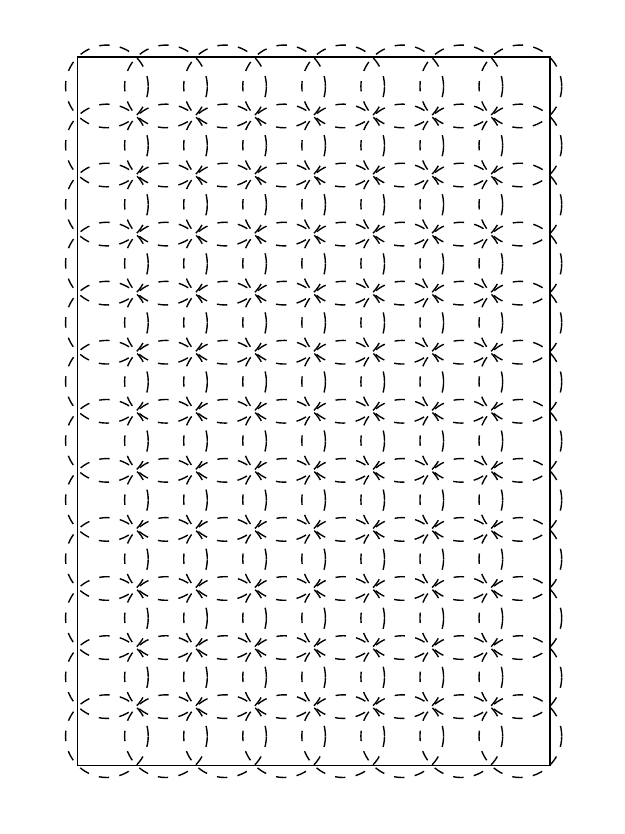}
\caption{Open area completely covered by 96 readers \label{fig:open_area}}
\end{figure}

The second and third scenarios are representations of the seven bridges of K\"{o}nigsberg\footnote{The seven bridges of K\"{o}nigsberg is a notable historical problem in mathematics. In 1735, Leonhard Euler proved that no Eulerian path existed for the K\"{o}nigsberg city. This result laid the foundations of graph theory.}. In these scenarios, people's movements are constrained by the river and thus, they can only use bridges in order to move to different sides of the city. Like people, we assume that tags should not be on the river and we design the second and third scenarios using two different distributions of readers.
The second scenario (cf. Figure~\ref{fig:covered_area}) is a representation of the K\"{o}nigsberg city where 14 readers cover the entire city. Note that this scenario also meets the constraints of the Solanas {\em et al.} protocol w.r.t. the reader distribution~\cite{SolanasDMD-2007-cn}. Since covering a city by 14 RFID readers can be not practical, we design a third scenario (cf. Figure~\ref{fig:not_covered_area}) that only differs from the previous one in the reading ranges and positions of the readers. Notice that in the second scenario a tag can be monitored in every part of the city, while in the third scenario a tag can only be read when passing through a bridge. However, due to the movement constraints in the city and the strategic position of the readers, it is easy to know in which side of the city each tag is located. This is a good example of how, by cleverly placing readers, it is possible to obtain accurate trajectories of
tags.

\begin{figure}[tb]
\small
\centering
   \includegraphics[scale=1]{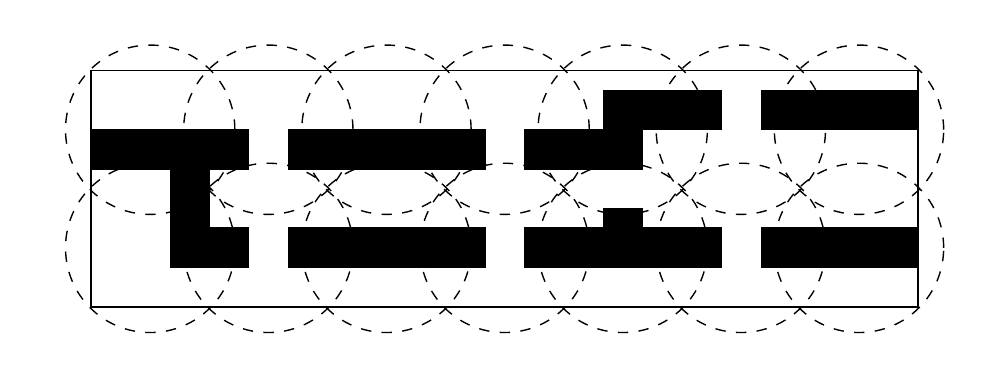}
\caption{K\"{o}nigsberg city representation where 14 readers cover the entire city. Black blocks represent the river water and the seven bridges are represented by the square spaces between black blocks. \label{fig:covered_area}}
\end{figure}

\begin{figure}[tb]
\small
\centering
   \includegraphics[scale=1]{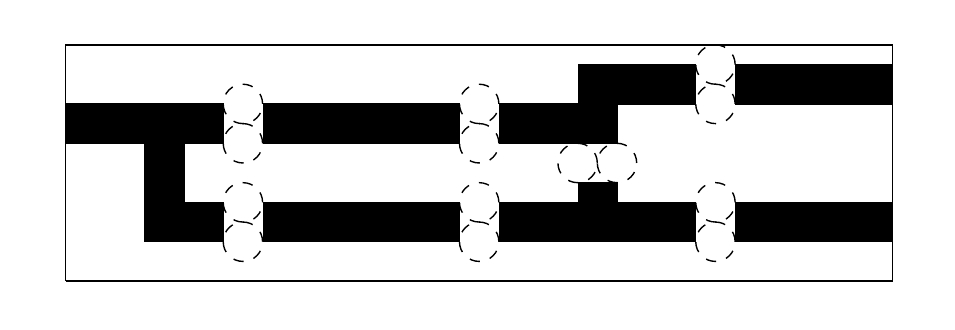}
\caption{K\"{o}nigsberg city representation where 14 readers are monitoring the
two ends of each bridge. Black blocks represent the river water and the seven bridges are represented by the square spaces between black blocks. \label{fig:not_covered_area}}
\end{figure}

\subsubsection{Movement of tags}

In this stage of the experiments, we should generate fine-grained moving data as required by~\cite{SolanasDMD-2007-cn}. To do so, we consider that tags move according to two types of movement:

\begin{itemize}
    \item \textbf{Random movement}. At each step, a tag chooses a random direction and moves in this direction.
    \item \textbf{Semi-directed movement}. In this movement, a tag always has a target point. Once the tag reaches its target, it changes the target point to a new random and valid point in the scenario. Then, at each step, with probability $0.5$ the tag chooses whether to move randomly or move in the target's direction.
\end{itemize}

Between both movements, semi-directed movement can be considered closer to real movement patterns of people. However, unpredictable movement patterns can be only evaluated using a random movement.

\subsubsection{Simulations}

In order to compare our partial-predictive proposal against the two previous proposals~\cite{SolanasDMD-2007-cn, Fouladgar:2008:SPP:1461464.1461467}, we perform simulations on the three scenarios defined above. For each scenario, different settings defined by the number of tags in the system and the type of movement are used. A simulation process consists of $10^4$ tags moving according to some pattern (random or semi-directed) in one of the three scenarios. For each simulation process, tags are identified using four different methods:

\begin{enumerate}
    \item The Fouladgar {\em et al.} method~\cite{Fouladgar:2008:SPP:1461464.1461467} assuming that
    each tag is in the cache of only one reader. We refer to this method as
    \textbf{Fouladgar 1-1}.
    \item The Fouladgar {\em et al.} method~\cite{Fouladgar:2008:SPP:1461464.1461467} assuming that
    each tag may be in the cache of several readers. The authors propose to
    store the tag data in the cache of those readers that may read it most
    often. As this is not possible for the two data sets considered in this
    work, we make the assumption that a tag will be in the cache of the readers
    that have identified it previously. We refer to this method as
    \textbf{Fouladgar 1-M}.
    \item The Solanas {\em et al.} method~\cite{SolanasDMD-2007-cn}. We refer to this method as
    \textbf{Solanas}.
    \item The previously mentioned \textbf{partial-predictive} proposal.
\end{enumerate}

In order to give statistically sound results, each simulation process is executed 30 times and the average number of cryptographic operations performed by each method is computed. Figure~\ref{fig:random_10000} and Figure~\ref{fig:semi_10000} show the experimental results obtained for $10^4$ tags moving according to the random movement and the semi-directed movement, respectively. In both figures, it can be observed that the Partial-predictive proposal improves on the previous ones by more than $50\%$. This means that, for any scenario and any type of movement, our partial-predictive proposal needs, in the worst case, half the number of cryptographic operations executed by previous proposals~\cite{SolanasDMD-2007-cn, Fouladgar:2008:SPP:1461464.1461467}.

\begin{figure}[!t]
\centering
\includegraphics[width=2.5in, angle = 270]{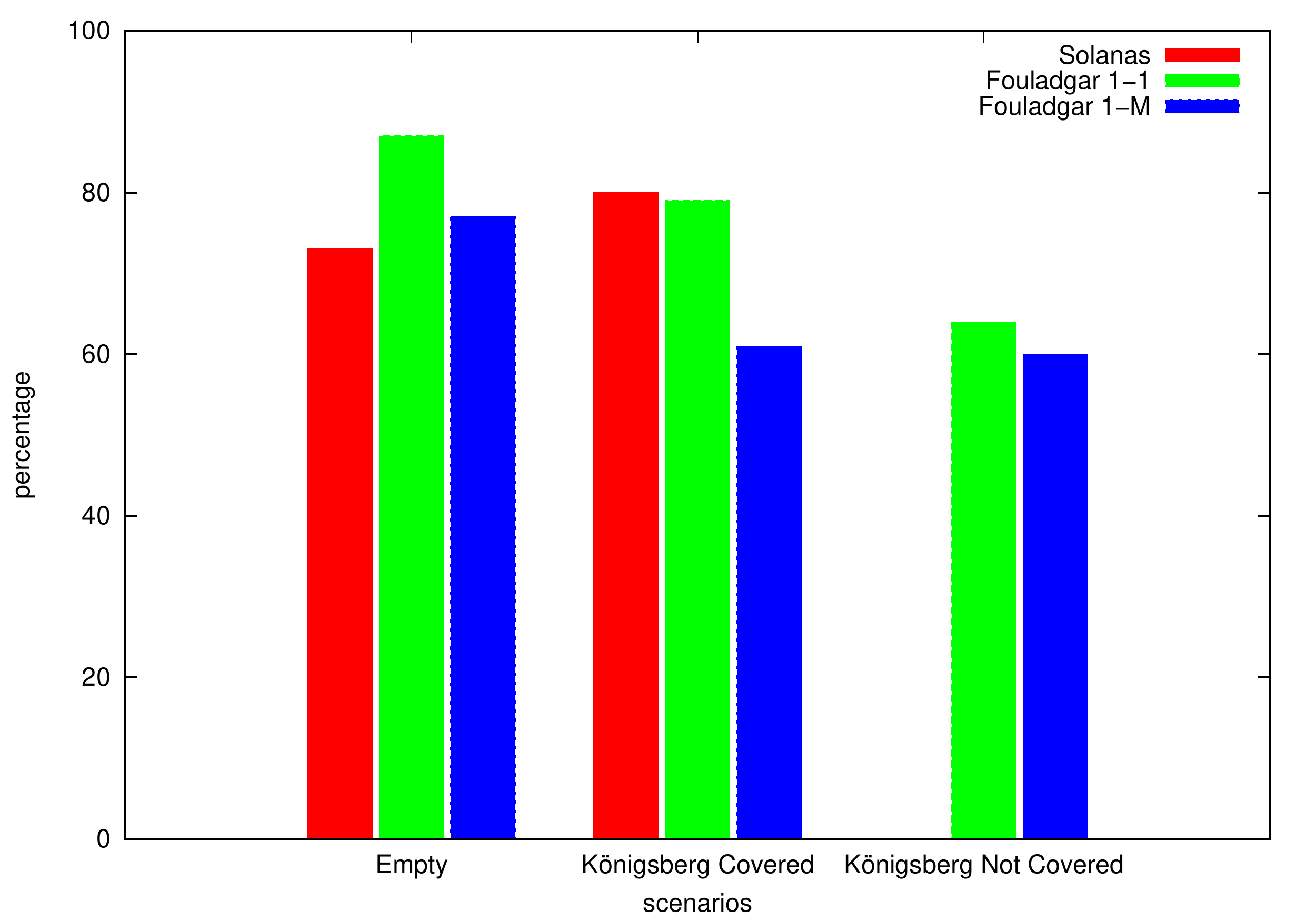}
\caption{Percentage of improvement of the partial-predictive proposal w.r.t. previous ones considering random movement and $10^4$ tags}
\label{fig:random_10000}
\end{figure}

\begin{figure}[!t]
\centering
\includegraphics[width=2.5in, angle = 270]{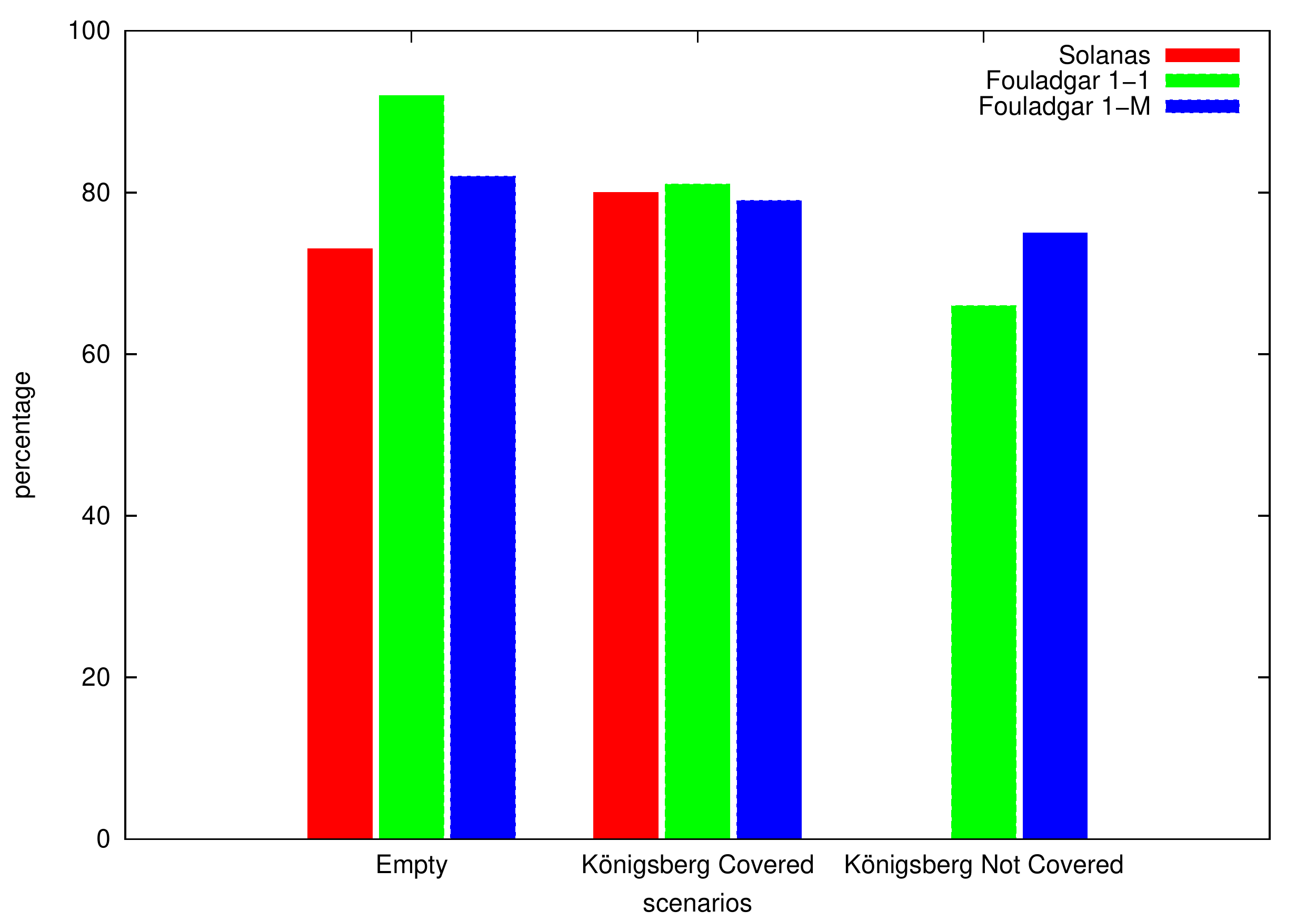}
\caption{Percentage of improvement of the partial-predictive proposal w.r.t. previous ones considering semi-directed movement and $10^4$ tags}
\label{fig:semi_10000}
\end{figure}

The partial-predictive proposal performs better than previous ones mainly due to three reasons: (i) after the identification of a tag, the reader saves in its own cache the tag's data in order to identify it faster in the future, (ii) the size of the caches of readers is minimised in such a way that two readers never share tag information, (iii) and when a reader can not identify a tag using its own cache, it is able to heuristically ($\mathcal{A}_{prev}$) find another reader that could identify this tag.

\subsection{Experiments considering the $\mathcal{A}_{next}$ algorithms}

In the second half of the experimental section, we consider coarse-grained datasets of tag movement (data sets of trajectories), in the style of
tracking data recorded by RFID systems. A data set of trajectories contains a historical log with all the identification
events produced by readers during the identification of tagged objects in a
given scenario. Thus, with these data sets, it is possible to determine the
precise moment in which a tagged object was identified by a given reader in
an exact location. By using these data sets of trajectories, whether real or
synthetic, we are able to measure the performance of our proposal in terms of
computational cost and bandwidth usage, and compare it to others without the
need for an expensive and very time consuming implementation of real prototypes.

However, obtaining real data sets of trajectories of RFID tagged objects
moving through, for example, supply chains is very difficult, \emph{i.e.}
these data are generally kept by private companies that are quite reluctant
to share them. Hence, the use of synthetic data obtained by means of
simulation is a common practice~\cite{GonzalezEtAl-2010-TKDE}
\cite{GonzalezEtAl-2006-ICDE}~\cite{GonzalezEtAl-2007-ICDE}. However, a
synthetic data set might fall short of capturing the real complexity of the
motion of objects. With the aim to lessen this problem and in order to perform
a comprehensive comparison of our proposal with previous ones, we use two
different data sets of trajectories:
\begin{enumerate}
\item A synthetic data set generated by simulating the movement of tagged
objects in supply chains. This data set has been generated by using techniques
proposed in previous articles~\cite{GonzalezEtAl-2010-TKDE}
\cite{GonzalezEtAl-2006-ICDE}~\cite{GonzalezEtAl-2007-ICDE}, which deal with
moving objects in supply chains.
\item A real data set consisting of a historical log of the movement of
wireless cards through several access points at Dartmouth
College~\cite{dartmouth-campus-2007-02-08}. This real data set of trajectories
captures the movement of students in the Dartmouth College when they connect to
the wireless access points of the campus.
\end{enumerate}

\subsubsection{Generating the synthetic data set}

As stated above, we generate a synthetic data set of moving objects in supply
chains. Similarly to~\cite{GonzalezEtAl-2006-ICDE}, we consider several
distribution centres or factories that may exchange tagged products/items in
both directions by means of input/output gates (controlled by RFID readers).
Once a distribution centre has $M$ items in any of its output gates, it sends
these items to another randomly selected distribution centre. Upon reception of
a set of items by a distribution centre, these items are processed according to
the distribution centre policy. Like in previous models
\cite{GonzalezEtAl-2010-TKDE}
\cite{GonzalezEtAl-2006-ICDE}~\cite{GonzalezEtAl-2007-ICDE},
the distribution centre policy is defined by a
graph. Locations where items arrive and depart are the nodes of the graph,
whilst the edges represent the possibility of moving between locations. In
particular, we define a random graph for each distribution centre and random
Poisson distributions to model the departure of items in each location. By doing
so, we simulate that items move in small groups or individually inside each
distribution centre whilst they move in large groups between distribution
centres. Note that this kind of movement is similar to the one given
in~\cite{LeeChung-2008-SIGMOD} where two types of data are considered: (i) groups
of items (GData) and (ii) single items (IData).

Similarly to~\cite{GonzalezEtAl-2010-TKDE}, we define five distribution centres
and twenty locations in each of them. For each distribution centre we define a
random graph using an Erd\H{o}s-R\'{e}nyi model $G(n, p)$ where $n = 20$ and
$p = 0.5$. Also, we assign to each location a Poisson distribution $P(\lambda)$
where $\lambda = 10$. Finally, the minimum number of items that are sent as a
group between distribution centres is defined as $M = 100$.

In order to define the movement pattern of items we consider that they have
different probabilities to departure towards different locations. For each
out-edge of the graph, each item has a probability of taking this edge to leave.
In our experiments, we have defined that for every node having $n$ out-edges,
the sequence of probability values assigned to these out-edges is a permutation
of the sequence $\{\frac{1}{2}, \frac{1}{2^2}, \cdots, \frac{1}{2^{n-1}},
\frac{1}{2^{n-1}}\}$ (note that any other probability distribution could be
defined.) Finally, considering all these settings, we generate a synthetic
data set with $10^5$ trajectories having an average length of $200$ points.

\subsubsection{Generating the real data set}
Dartmouth College has 566 Cisco 802.11b access points installed to cover
most of its campus. The college has about 190 buildings with 115 subnets so that
clients roaming between buildings can change their IP addresses. This roaming
information is recorded in different files for different clients by using
syslog events~\cite{dartmouth-campus-2007-02-08}. In total, more than $14,000$
trajectories collected over almost $2$ years can be found in this data set.

For our experiments, we have selected the shortest $10,000$ trajectories of this
data set. This subset of trajectories is created by parsing all the files having
less than $46$ Kb. We have selected the shortest trajectories because longer
trajectories have useless, larger gaps in the data, generally caused by power
failures, access points failures, or long periods of time in which clients
were not in the campus. Note that those big gaps should not appear in data
sets of items moving through supply chains because, in this scenario, items
cannot be considered lost for a long time. The trajectories of the resulting
data set have an average length of $400$ points.

\subsubsection{Implementing predictors}
In Section~\ref{sec:predictors}, we have defined an effective algorithm to
predict the next location of a moving object based on a Markov model. Also, we
have shown that it is possible to give an estimation of the time when an
object should visit the next location.

In order to provide a better evaluation of our proposal, we have run experiments
using two different predictors:

\begin{enumerate}
    \item A Markov-based predictor: The predictor described in
    Section~\ref{sec:predictors} used to estimate both the next location and
    the time when the object should visit that location.
    \item An Oracle predictor: A predictor that \underline{always correctly}
    guesses the next location and the time when the object should visit that
    location.
\end{enumerate}

It should be emphasised that the Oracle predictor is only possible because we
know in advance the trajectories of the data sets, otherwise it is not possible
to create it. The Oracle predictor can be understood as the optimal predictor,
\emph{i.e.} an upper bound in prediction accuracy.

The Markov-based predictor that we have implemented for our experiments guesses
correctly the next location and ETA of tags $41\%$ of the times with
synthetic data, and $67\%$ of the times with real data. The Oracle predictor has
$100\%$ of success for both data sets.

As it has previously been stated, the performance of our protocol in terms of
computational cost and bandwidth usage strongly depends on the accuracy of the
predictors. Although the obtained results outperform all previous proposals,
there is still room for improvement (\emph{e.g.} by developing better
predictors).

\subsubsection{Performance of protocols}

We will compare the performance of the following proposals:

\begin{enumerate}
    \item The previously mentioned \textbf{Fouladgar 1-1}, \textbf{Fouladgar 1-M}, and \textbf{Partial-predictive} proposals.
    \item Our proposal using a Markov-based predictor. We refer to this
    proposal as \textbf{Predictive (Markov)}.
    \item Our proposal using an Oracle predictor. We refer to this proposal
    as \textbf{Predictive (Oracle)}.
\end{enumerate}


From a scalability point of view, the number of cryptographic operations
performed on the server side is the main concern. Consequently, most of
the hash-based protocols are not considered scalable. However,
RFID protocols based on collaboration between readers have less
computational cost than hash-based protocols but may require more
bandwidth. Therefore, for all the studied protocols we compute
the number of cryptographic operations and, also, the number of messages sent
between readers.

With the aim to study both the computational cost and the bandwidth usage
simultaneously, we have defined a measure that for every protocol
outputs the percentage of \textit{closeness} of the protocol to the optimal
case; the higher \textit{(closer)} the better.

\begin{definition}[Trade-off measure] \label{def:trade-off}
Let $\mathcal{P}$ be the set of protocols under evaluation. Let $\alpha$ be
a real value in the range $[0..1]$. Let $P_c$ and $P_b$ be the number of
cryptographic operations and the number of sent messages of a given protocol
$P \in \mathcal{P}$. Let $min_c = \min(P^i_c), \ \forall \ P^i \in
\mathcal{P}$, $min_b = \min(P^i_b), \ \forall \ P^i \in \mathcal{P}$,
$max_c = \max(P^i_c), \ \forall \ P^i \in \mathcal{P}$, and
$max_b = \max(P^i_c), \ \forall \ P^i \in \mathcal{P}$. Then, the
trade-off measure that we propose is computed as follows:
$$
d(\alpha, P, \mathcal{P}) = \left(\frac{max_c - P_c}{max_c-min_c}\times
100 \right)\times \alpha + \left(\frac{max_b - P_b}{max_b - min_b}\times
100\right)\times (1 - \alpha)
$$
\end{definition}

Using this measure, it is possible to globally analyse the performance of all
protocols at the same time. In addition, thanks to the use of $\alpha$, it
is simple to weight the importance of either the computational cost or the
bandwidth usage. Thus, it can be easily observed which of the analysed protocols
performs best in given conditions.


\subsubsection{Experiments with the synthetic data set}

Figure~\ref{fig:random_hash} depicts the number of cryptographic operations
performed by each protocol over the synthetic data set.
In the beginning of the simulation (in the start-up phase) the ``Predictive
(Oracle)'' and the ``Predictive (Markov)'' have a performance similar to the
``Partial predictive'' protocol. However, after \emph{learning} the movement
pattern of items, they immediately outperform the ``Partial predictive'' protocol. It can also be observed that the predictive protocols and the
``Partial predictive'' protocol are clearly superior to the ``Fouladgar 1-1''
and the ``Fouladgar 1-M'', thus confirming the results
presented in the first part of the experimental section. Figure~\ref{fig:total_hash_weighted} shows the average number of cryptographic
operations per identification. From this figure, it is clear that our new
proposals outperform the previous ones in terms of computational cost and,
by extension, they improve scalability also.

\begin{figure}[p]
\centering
 \subfigure{
   \includegraphics[width=2in, angle = 270]
   {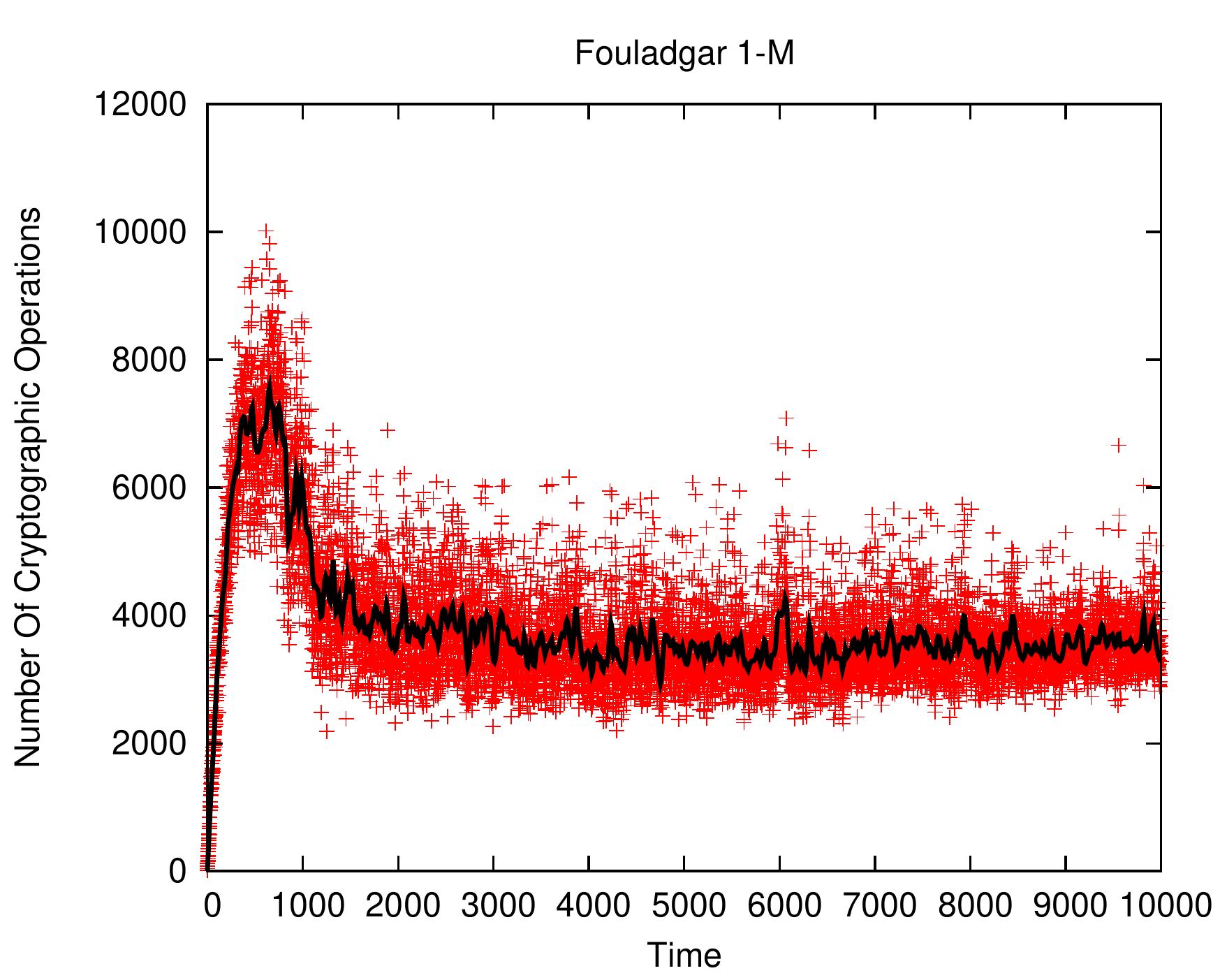}
 }
 \subfigure{
   \includegraphics[width=2in, angle = 270]
   {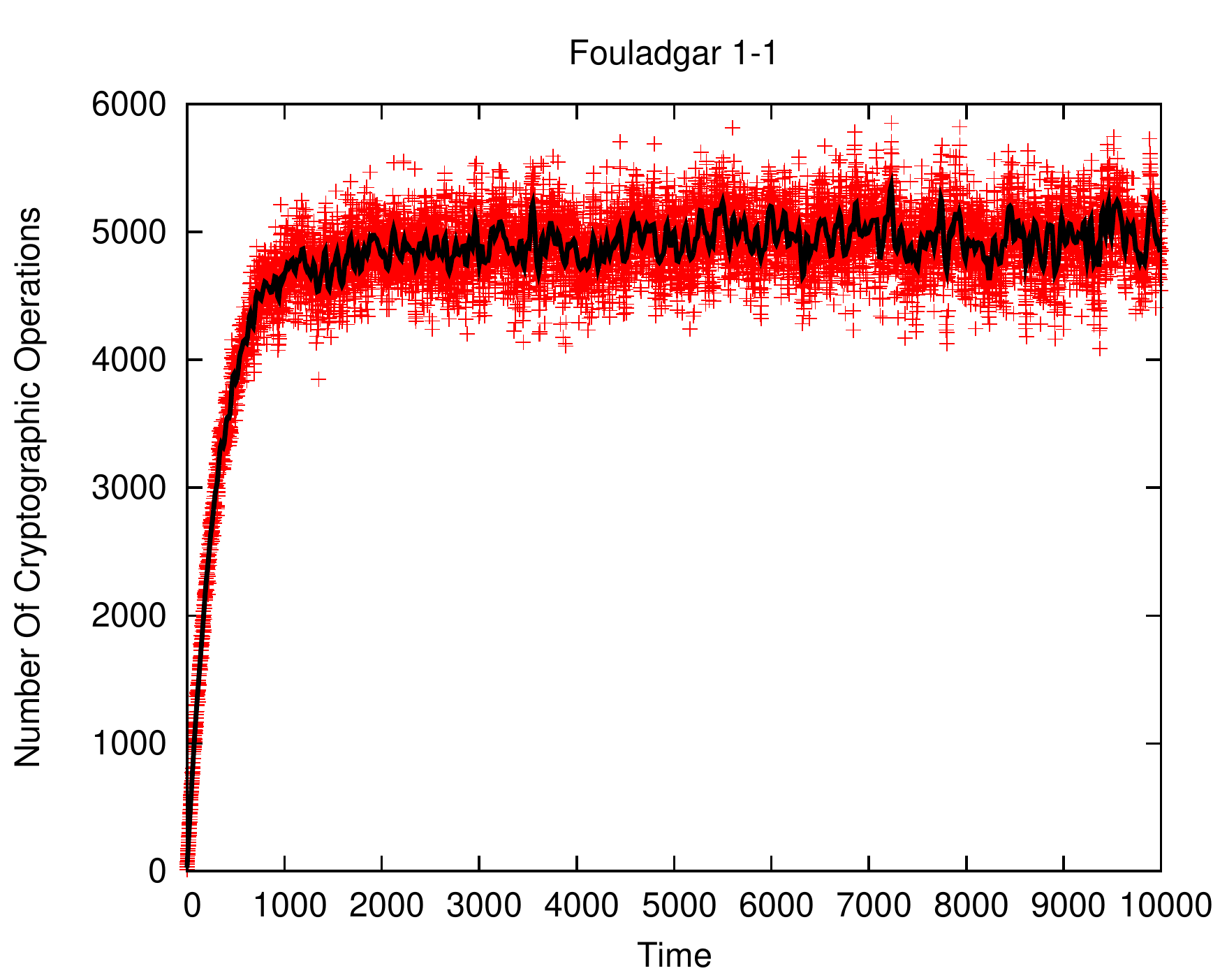}
 }
 \subfigure{
   \includegraphics[width=2in, angle = 270]
   {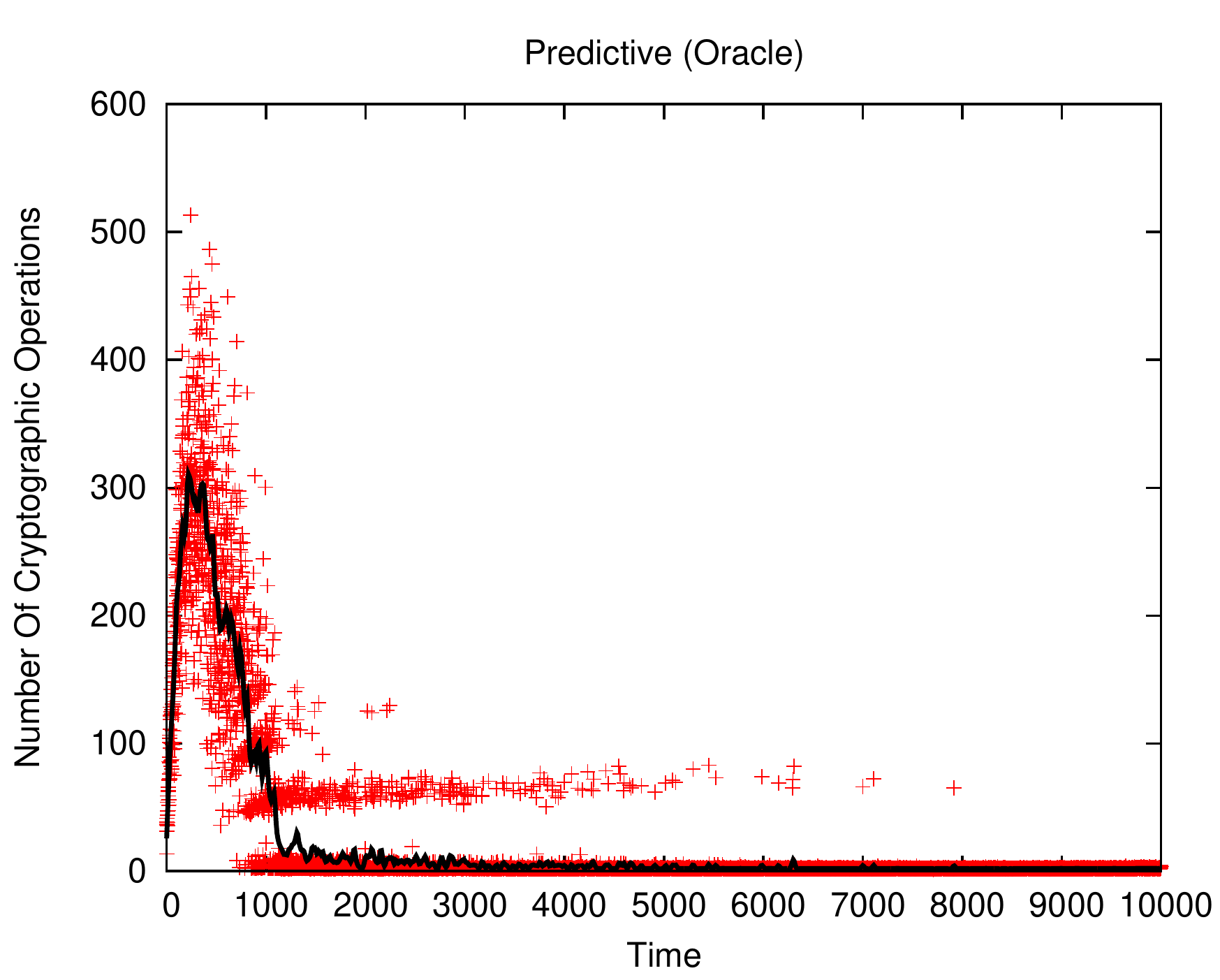}
 }
 \subfigure{
   \includegraphics[width=2in, angle = 270]
   {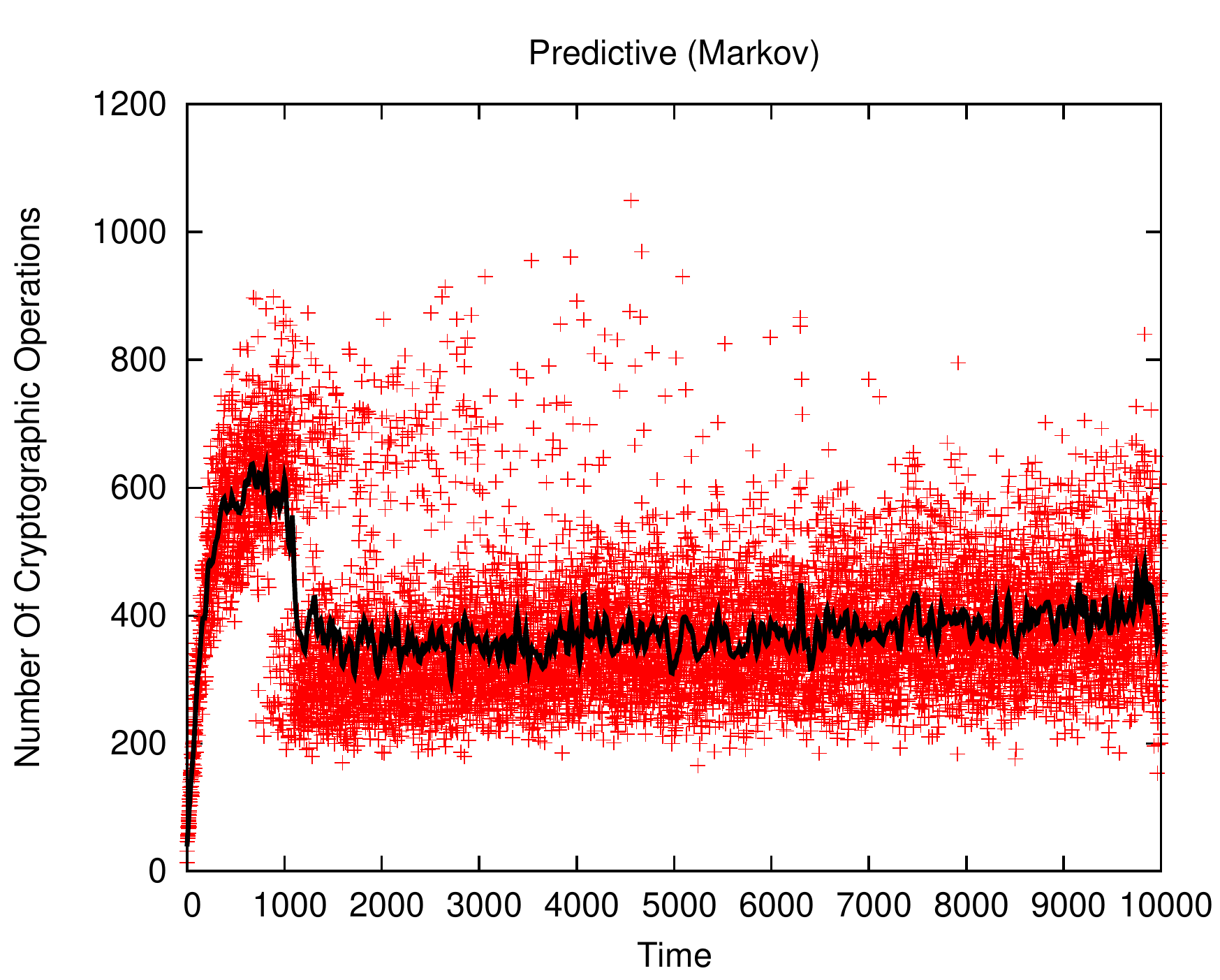}
 }
 \subfigure{
   \includegraphics[width=2in, angle = 270]
   {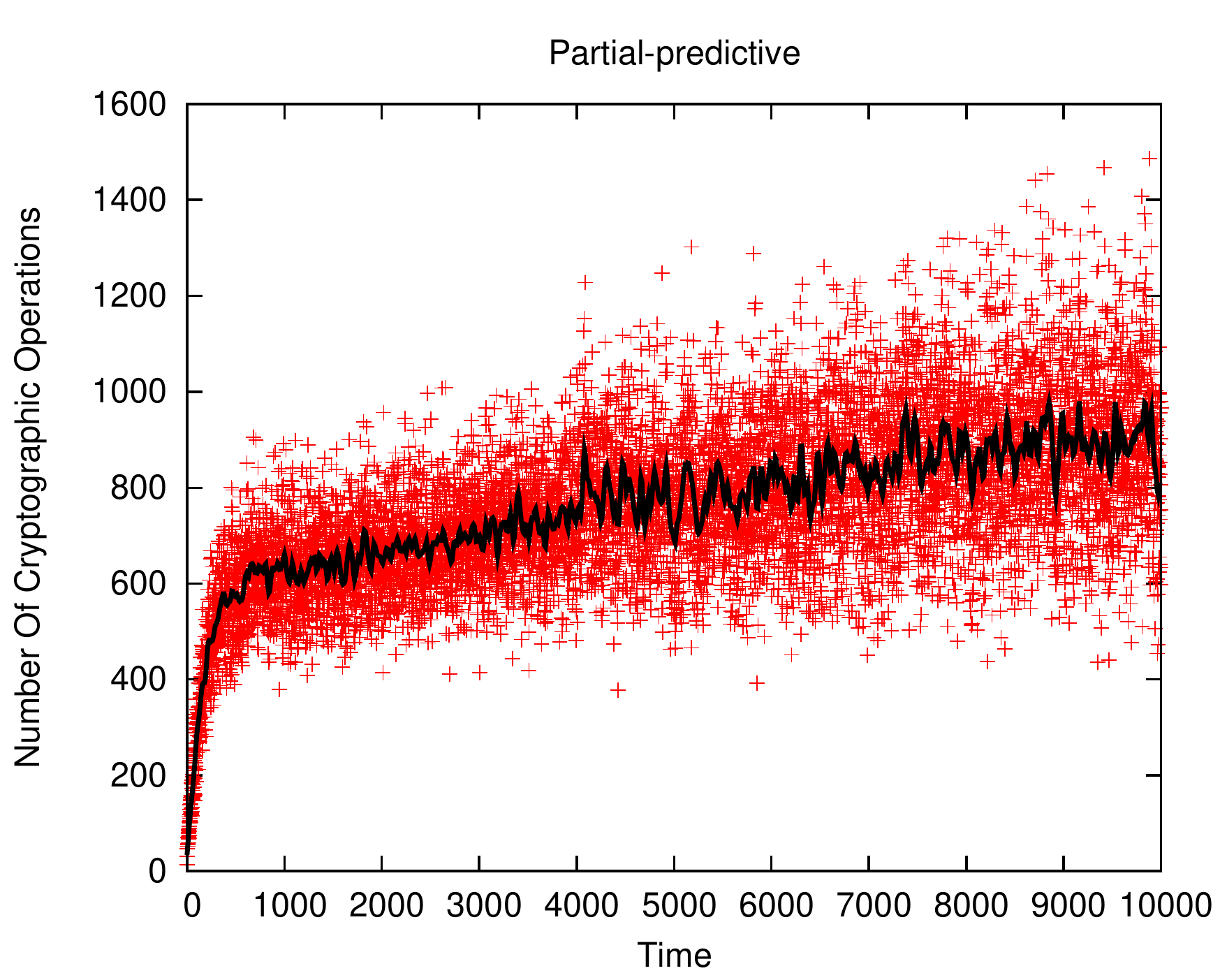}
 }
\caption{Average number of cryptographic operations performed by each
protocol during the complete simulation with the synthetic data set (in red).
The black line represents the moving average of those values in subsets of 100
elements. The time axis represents simulation steps.}
\label{fig:random_hash}
\end{figure}

\begin{figure}[!h]
\centering
    \includegraphics[width=2.5in, angle = 270]{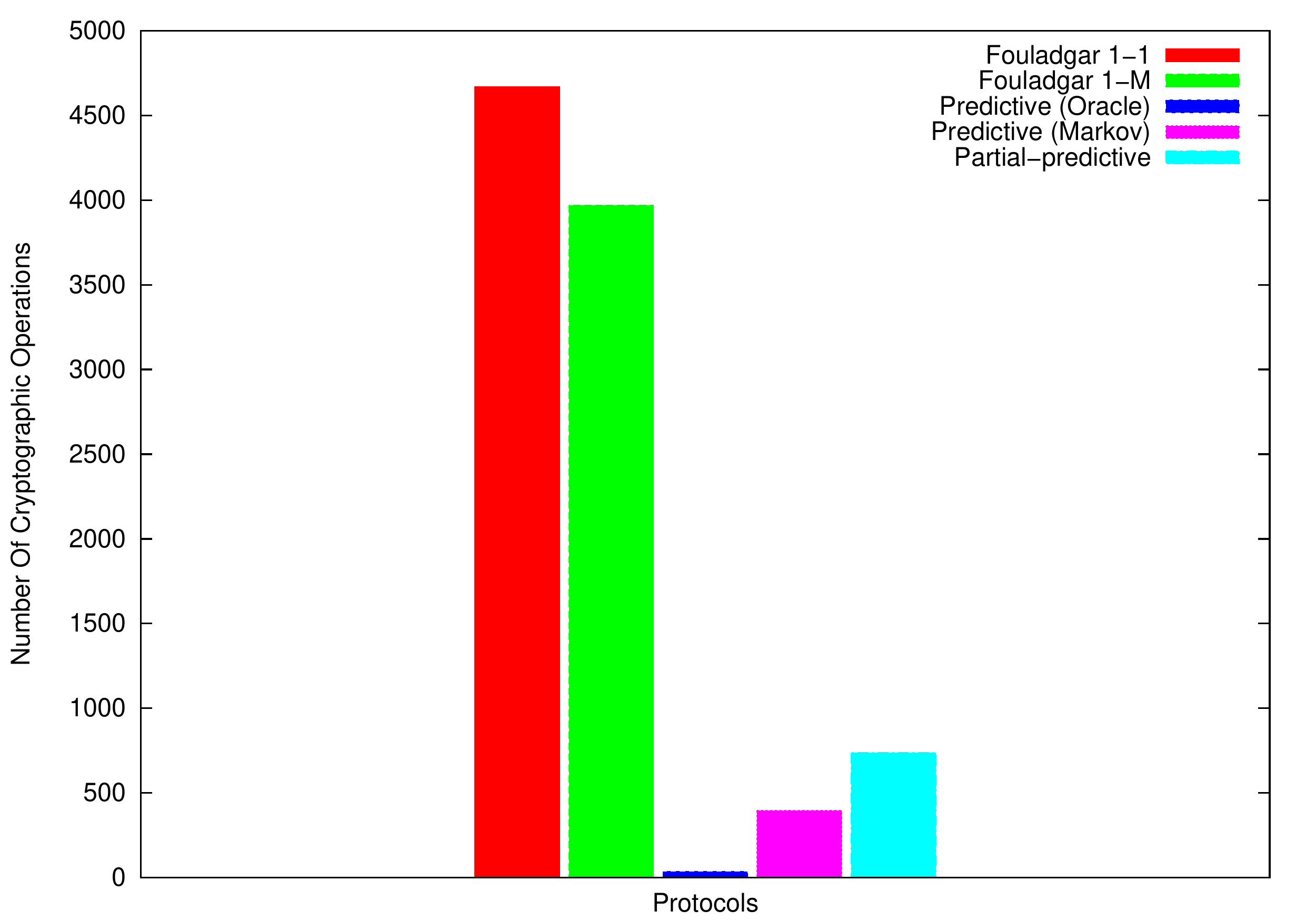}
    \caption{Average number of cryptographic operations per identification
    for the simulation with the synthetic data set. (The lower the better)}
\label{fig:total_hash_weighted}
\end{figure}

Figure~\ref{fig:weighted_msg} shows the number of messages sent by readers in
the studied protocols over the same data set, and
Figure~\ref{fig:total_msg_weighted} depicts the number of those messages on
average. It can be observed in both figures that the ``Fouladgar 1-M'' method
sends fewer messages because it replicates the identification information of
tags in several readers, at the cost of a poor scalability. It is also clear
that our new proposals send a very similar number of messages to the
``Fouladgar 1-M'' proposal but they perform significantly better in terms of
scalability.

\begin{figure}[p]
\centering
 \subfigure{
   \includegraphics[width=1.75in, angle = 270]
   {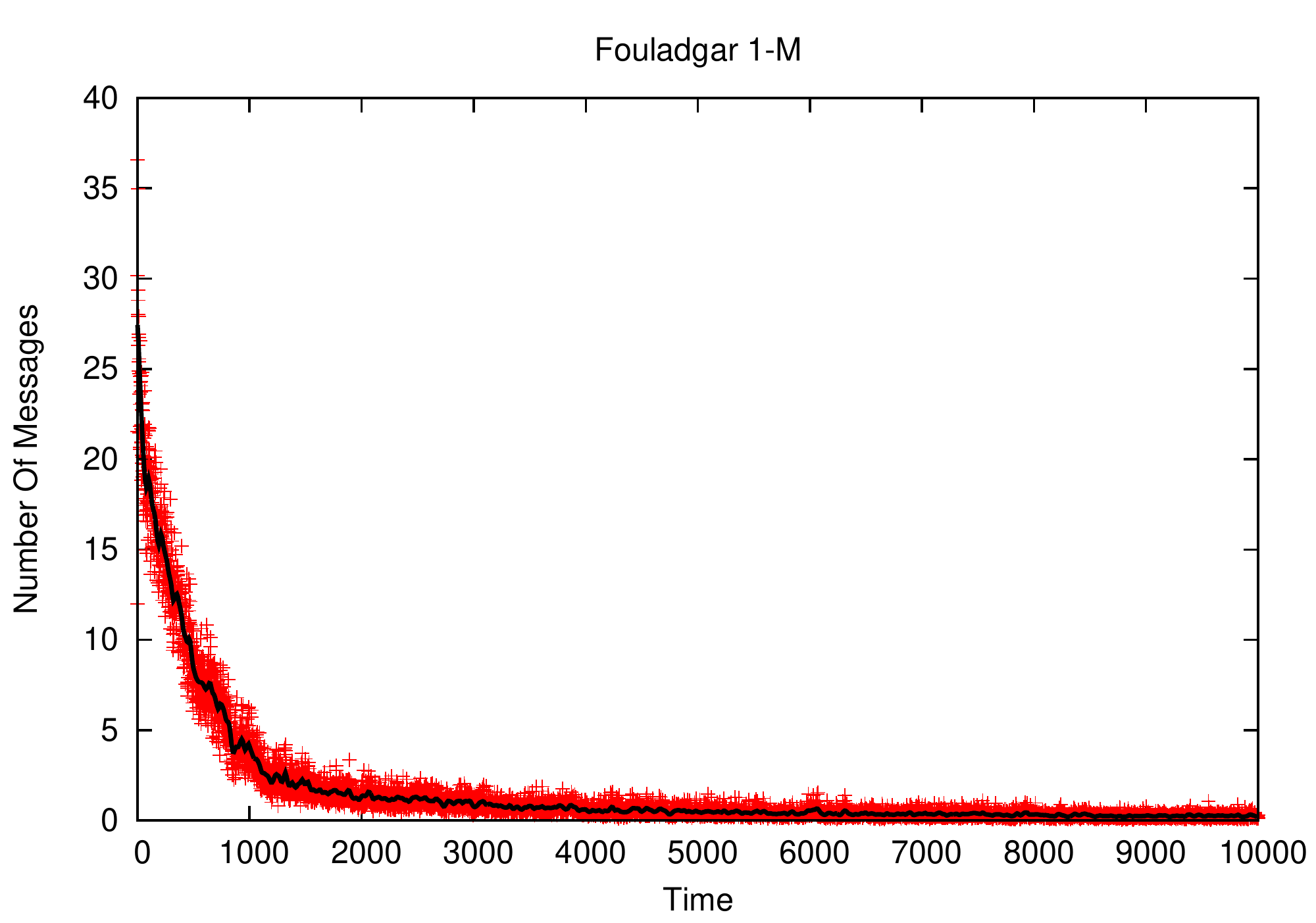}
 }
 \subfigure{
   \includegraphics[width=1.75in, angle = 270]
   {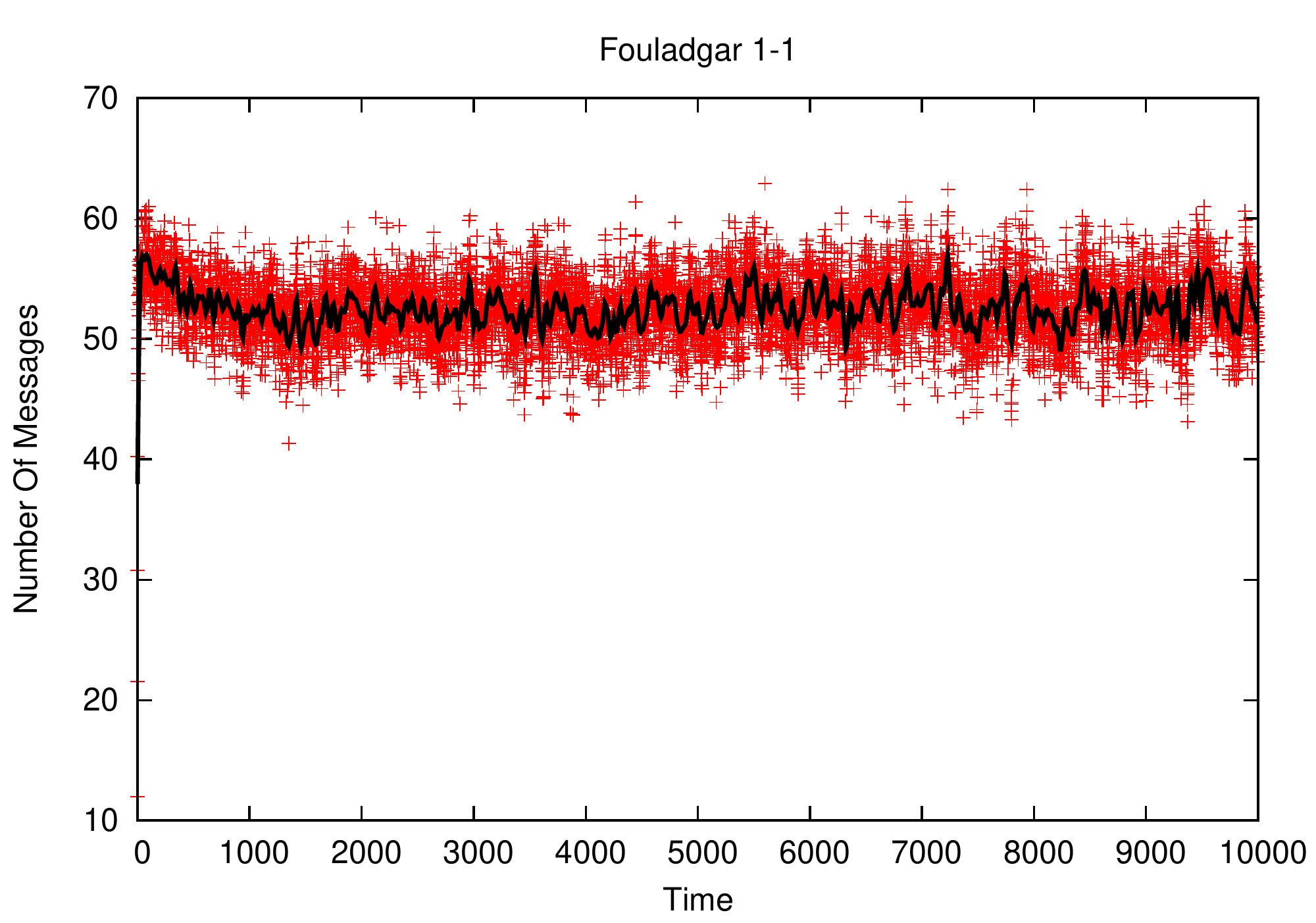}
 }
 \subfigure{
   \includegraphics[width=1.75in, angle = 270]
   {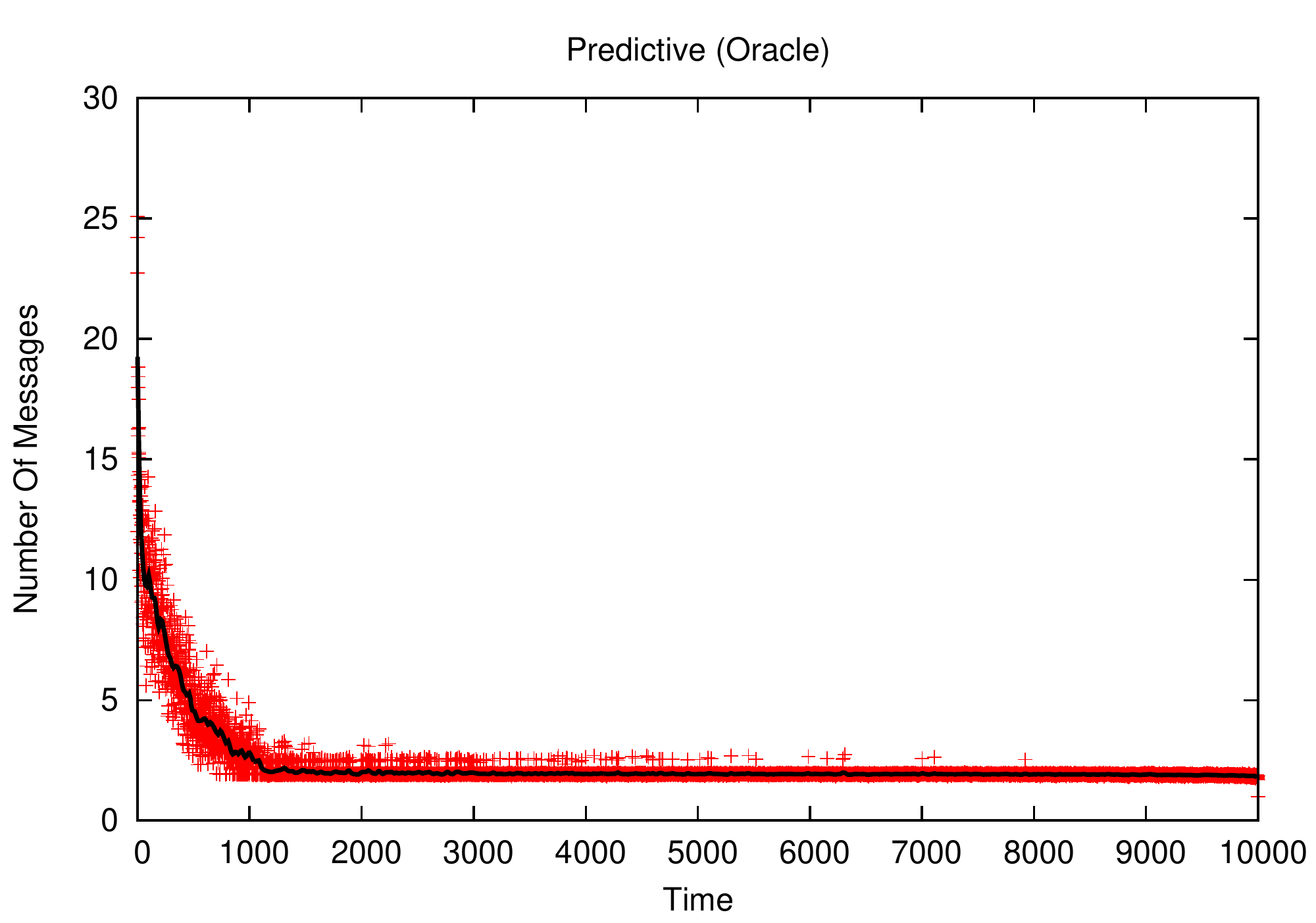}
 }
 \subfigure{
   \includegraphics[width=1.75in, angle = 270]
   {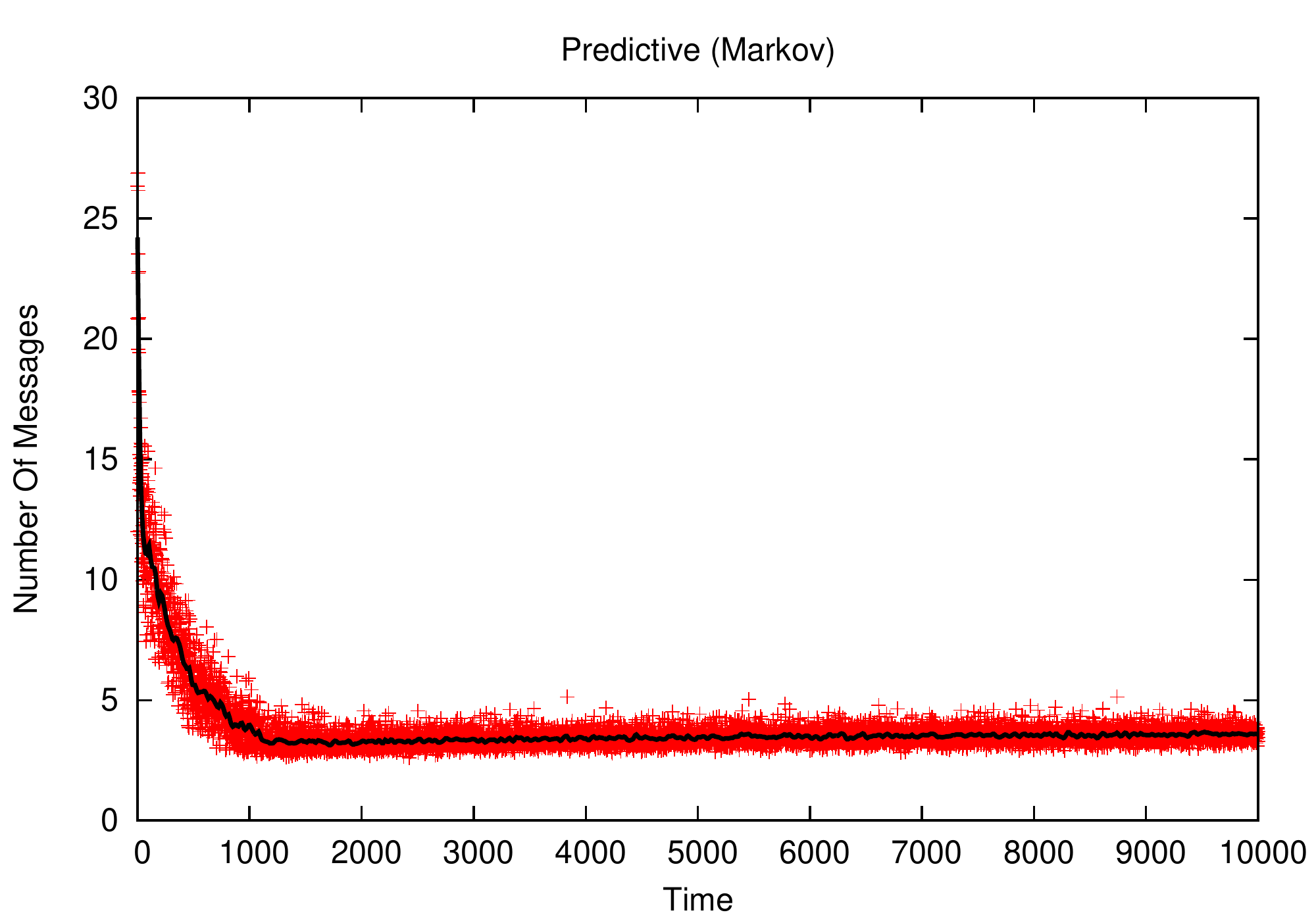}
 }
 \subfigure{
   \includegraphics[width=1.75in, angle = 270]
   {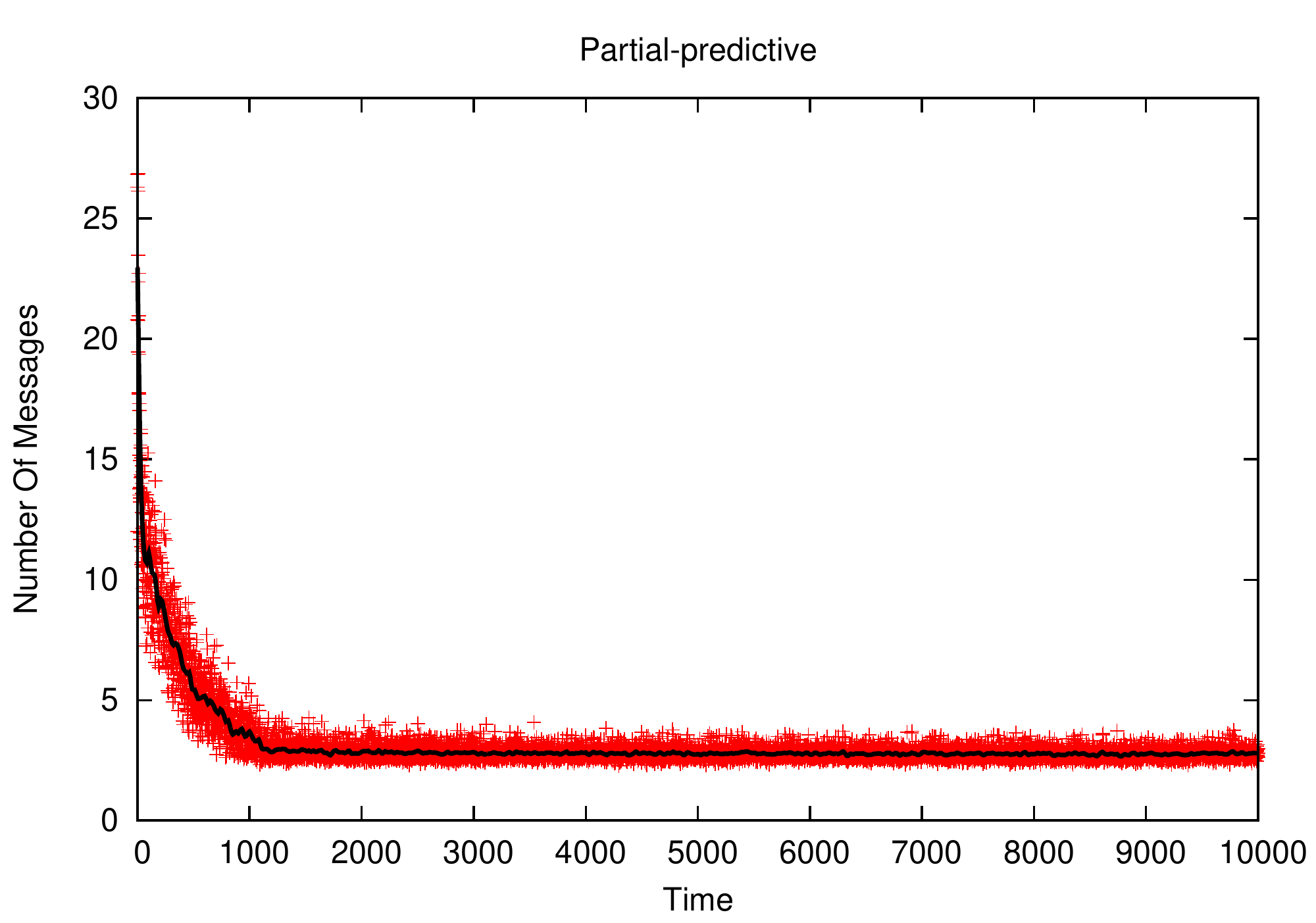}
  }
\caption{Average number of messages sent by each protocol during the
complete simulation with the synthetic data set (in red). The black line
represents the moving average of those values in subsets of 100 elements. The
time axis represents simulation steps.}
\label{fig:weighted_msg}
\end{figure}

\begin{figure}[p]
\centering
 \includegraphics[width=2.5in, angle = 270]{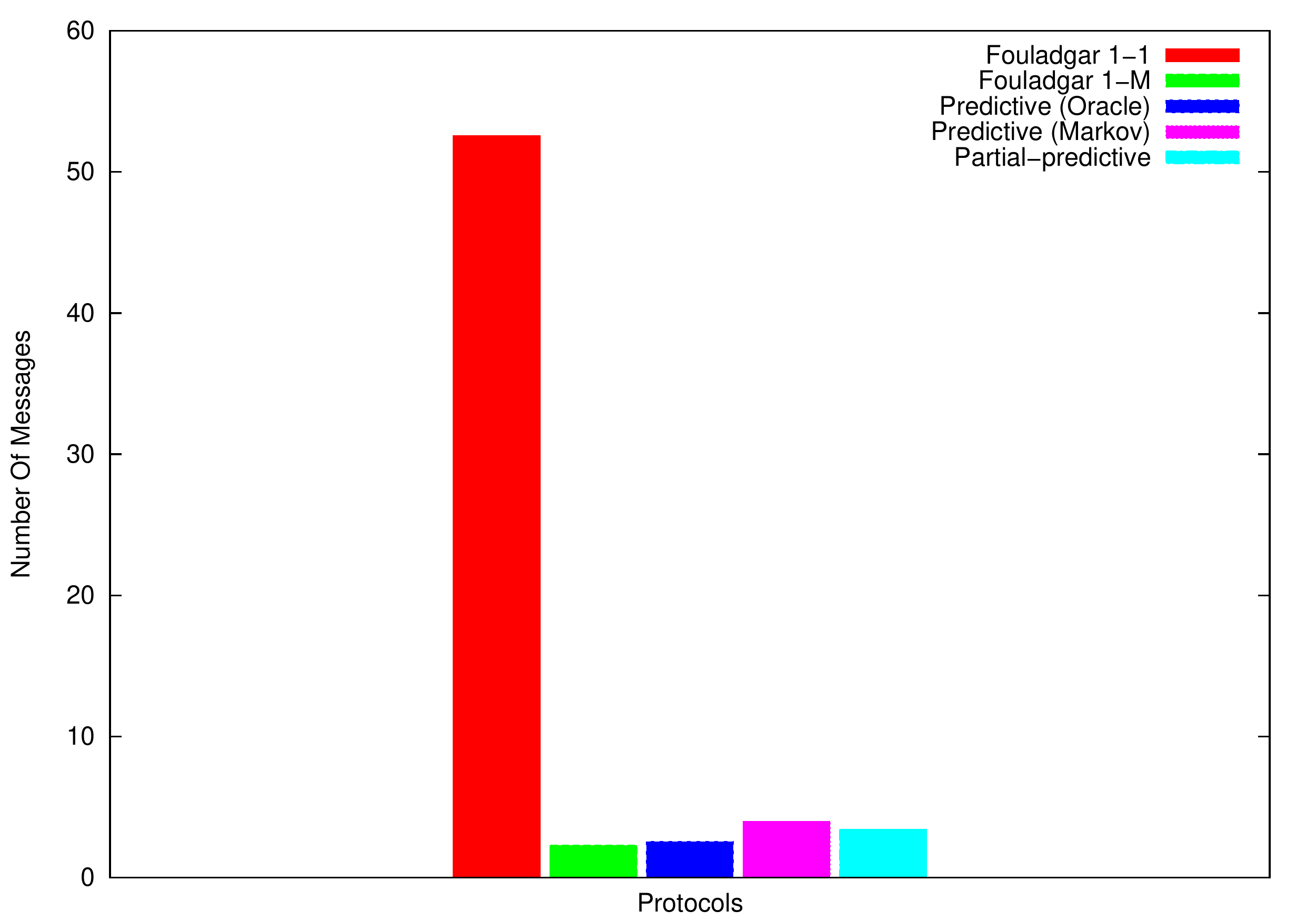}
\caption{Average number of messages sent per identification with the
synthetic data set}
\label{fig:total_msg_weighted}
\end{figure}


Using the trade-off measure described in Definition~\ref{def:trade-off} we
have compared all the protocols considering different values of $\alpha$ (cf.
Figure~\ref{fig:comparison_weighted}). It is apparent that the protocol
presented in this chapter (in its two variants) is the best for almost all
values of $\alpha$. Only in the region of $\alpha$ values very close to 0
(meaning that only the number of messages counts) our proposal is not the best.
Hence, we can conclude that our proposal is better than previous proposals for
the analysed synthetic data set.



\begin{figure}[p]
\centering
 \includegraphics[width=2.5in, angle = 270]
 {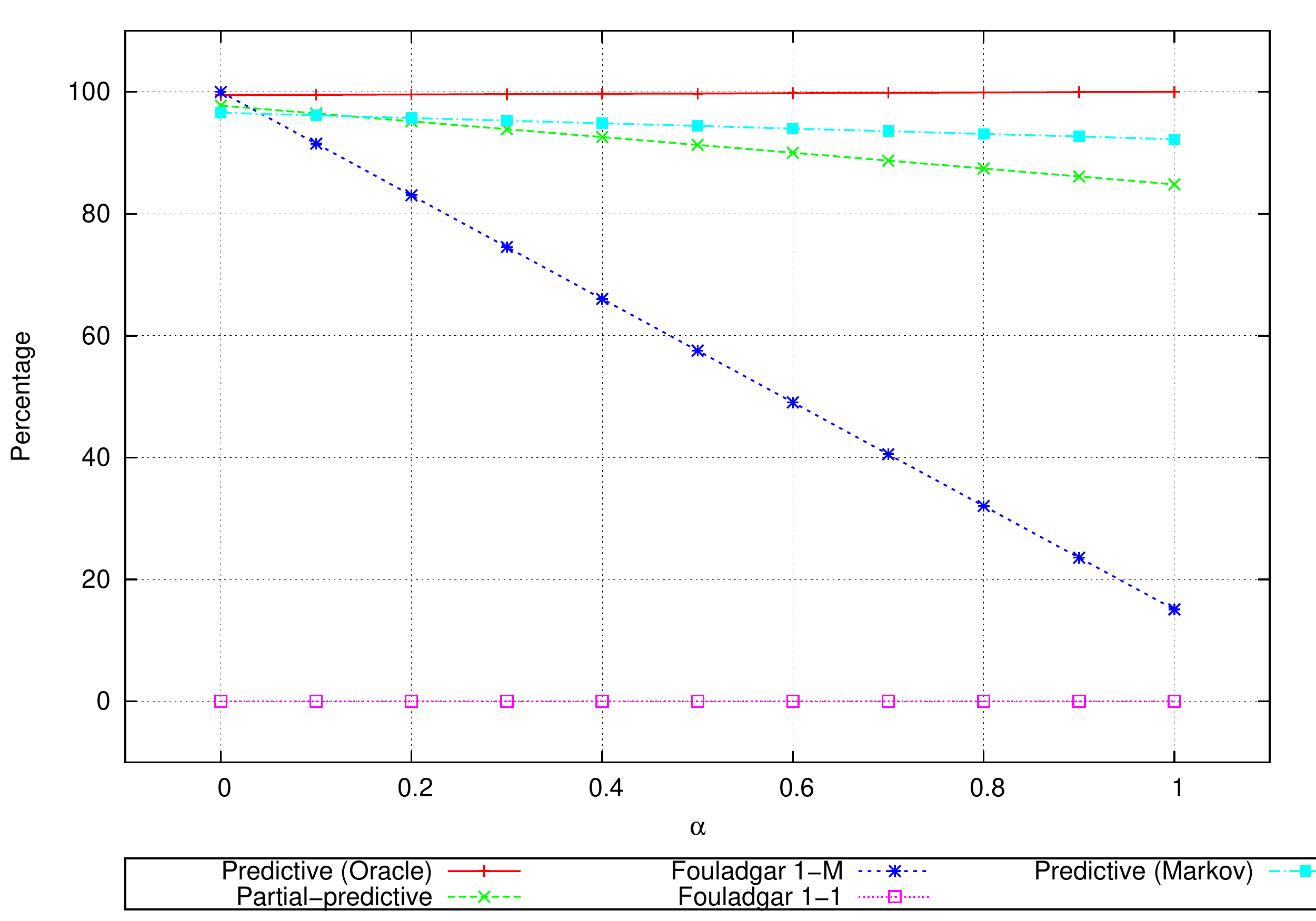}
\caption{The trade-off measure for each protocol and different values of
$\alpha = \{0, 0.1, 0.2, \cdots, 1.0\}$ for the
synthetic data set. The higher the better. (Note that the
colours assigned to the protocols do not coincide with the colours of
Figure~\ref{fig:total_msg_weighted}).}
\label{fig:comparison_weighted}
\end{figure}

\subsubsection{Experiments with the real data set}

In the case of the real data set, we consider the same measures described
above, \textit{i.e.} the number of cryptographic operations, the number of sent
messages, and the trade-off measure. Figure~\ref{fig:real_hash} and
Figure~\ref{fig:real_msg} show the number of cryptographic operations and the
number of sent messages for each protocol. Figure~\ref{fig:total_hash_real} and
Figure~\ref{fig:total_msg_real} show those values on average. Finally,
Figure~\ref{fig:comparison_real} depicts the \emph{closeness} of all protocols
to the optimal case by using the trade-off measure described in
Definition~\ref{def:trade-off}.

The results are very similar to the ones obtained with synthetic data. Again,
our proposal outperforms all previous proposals. Note that the different shape
of Figures 8 and 10 with respect to Figures 13 and 15 is due to the very nature
of the analysed data (\textit{i.e.} synthetic vs. real).
\begin{figure}[p]
\centering
 \subfigure{
   \includegraphics[width=2in, angle = 270]
   {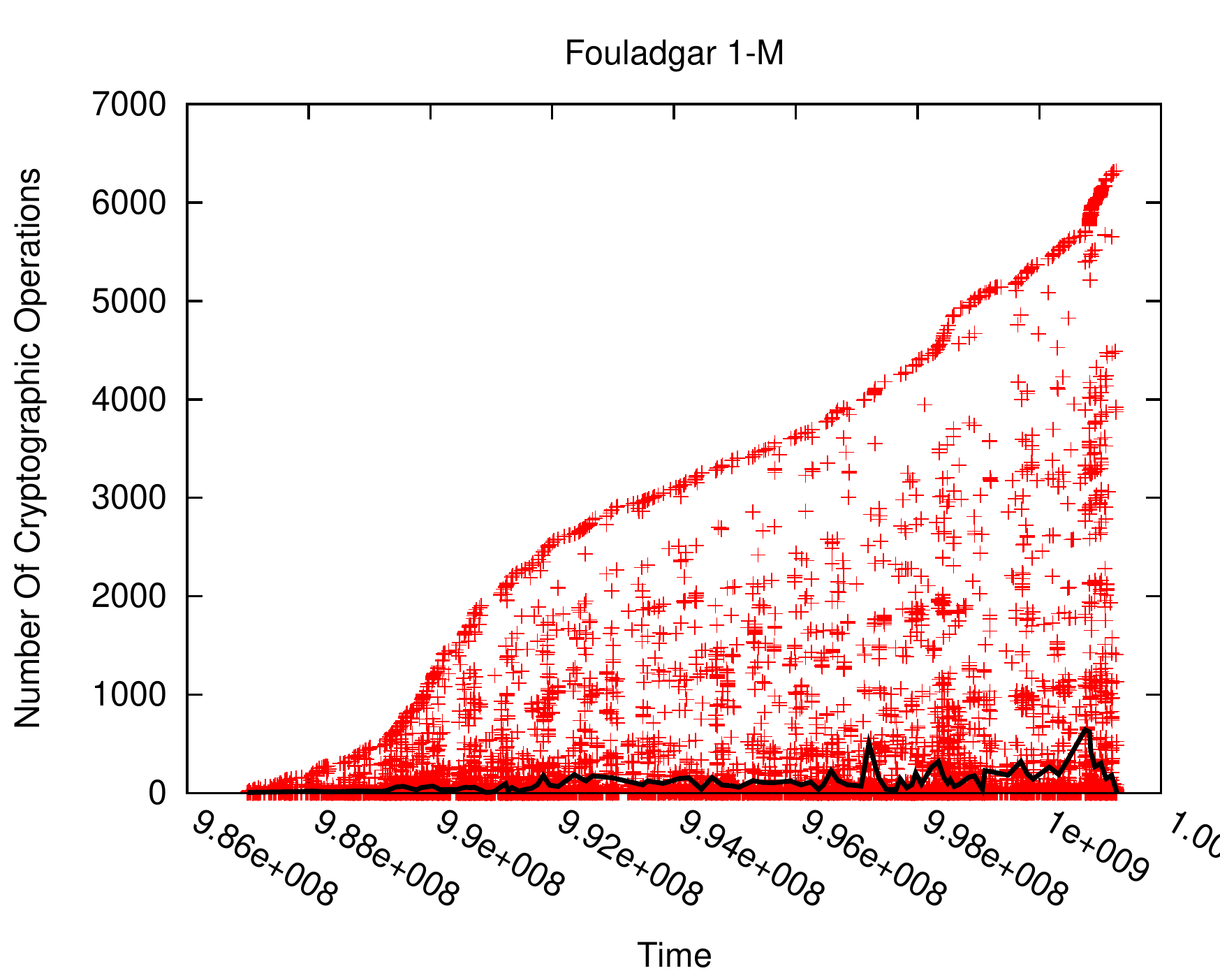}
 }
 \subfigure{
   \includegraphics[width=2in, angle = 270]
   {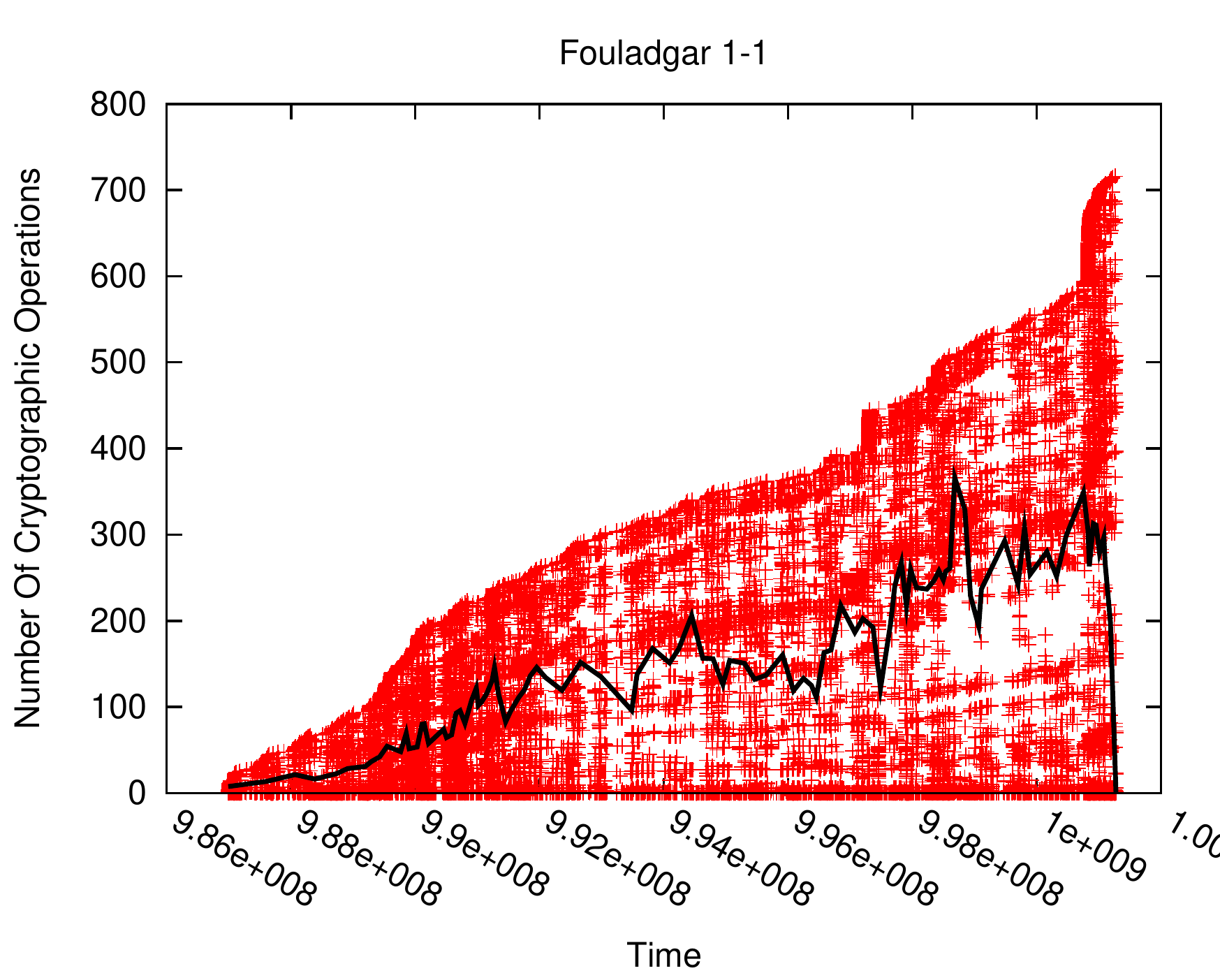}
 }
 \subfigure{
   \includegraphics[width=2in, angle = 270]
   {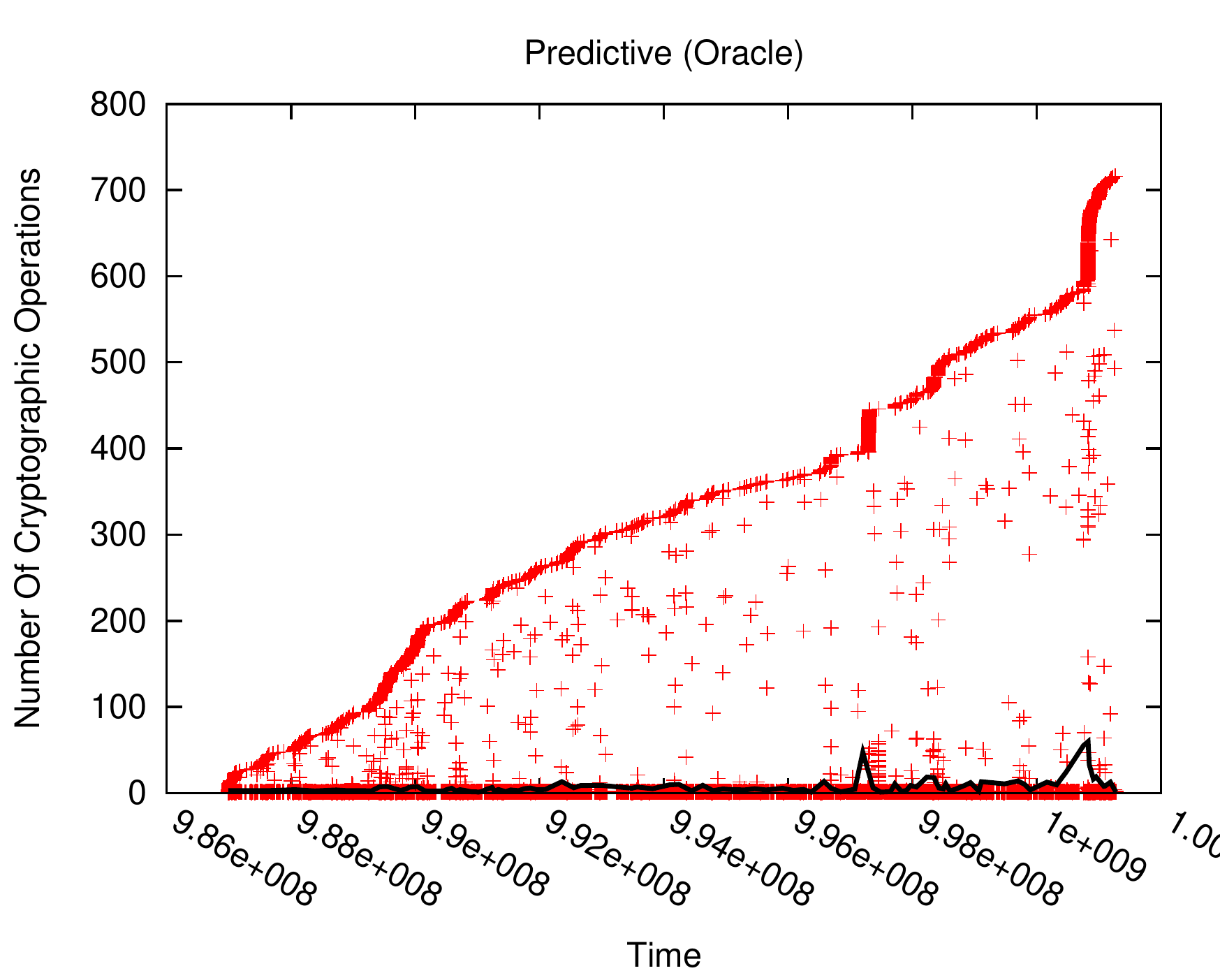}
 }
 \subfigure{
   \includegraphics[width=2in, angle = 270]
   {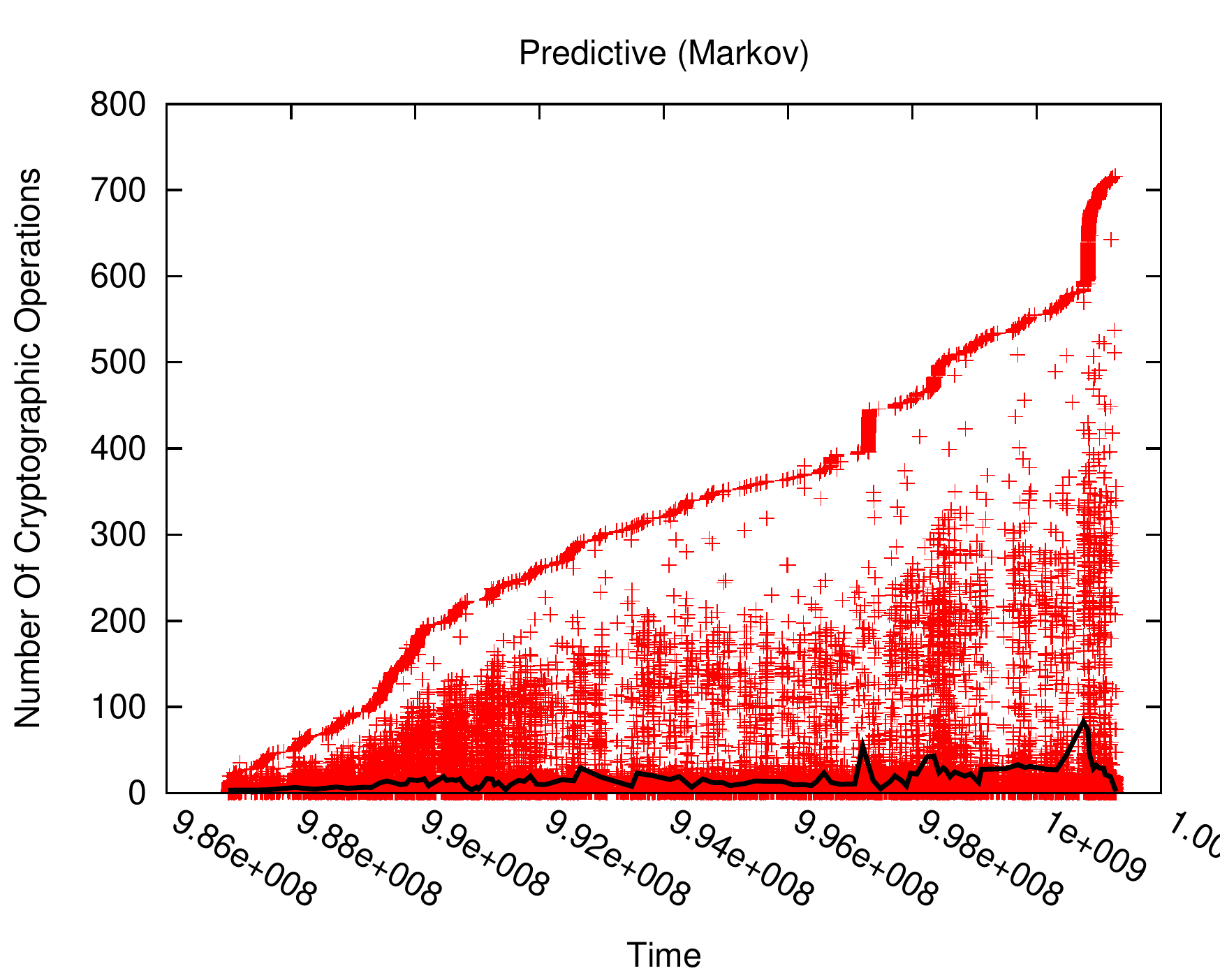}
 }
 \subfigure{
   \includegraphics[width=2in, angle = 270]
   {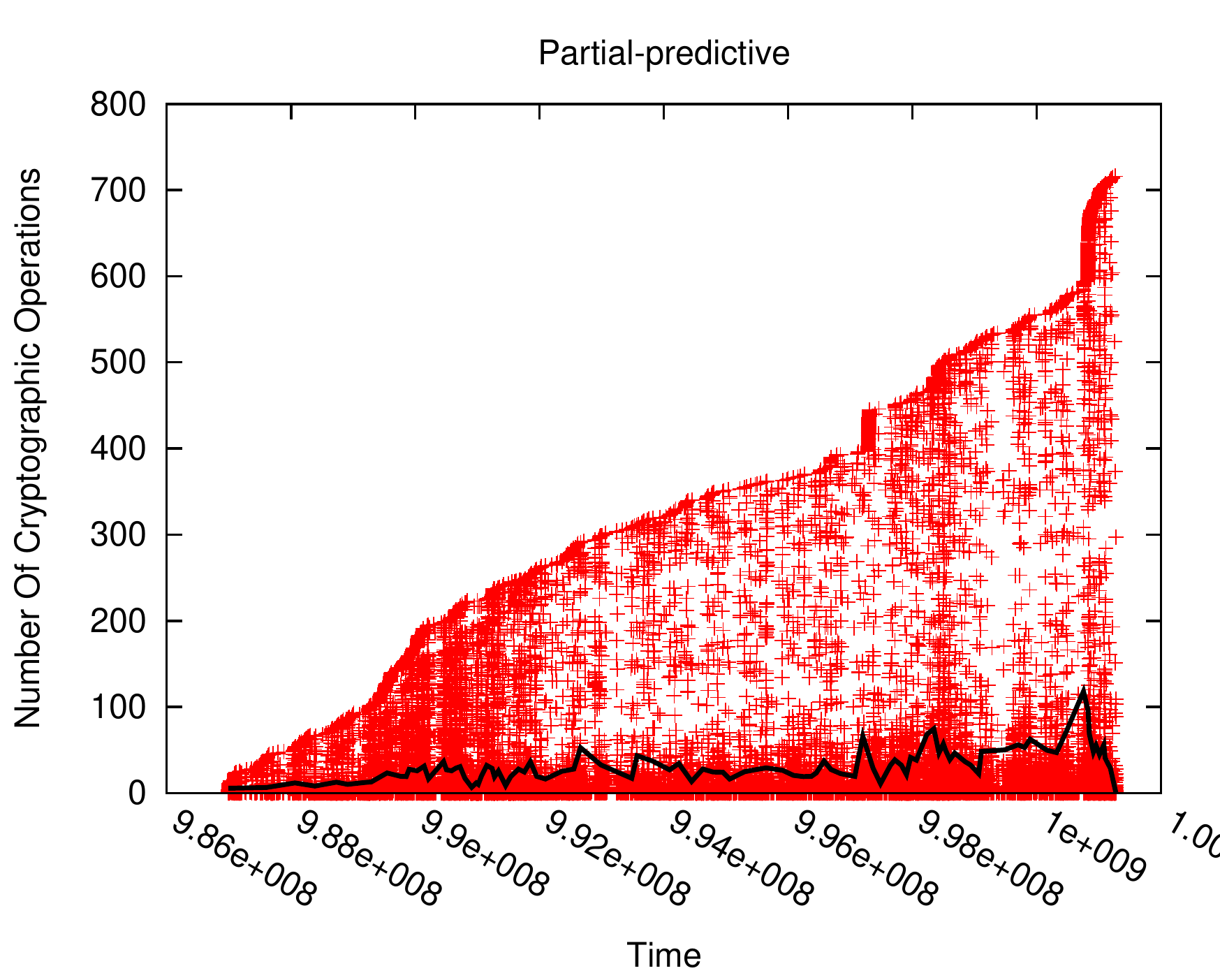}
 }
\caption{Average number of cryptographic operations performed by each
protocol during the complete simulation with the real data set (in red).
The black line represents the moving average of those values in subsets of 1000
elements. The time scale is in milliseconds.}
\label{fig:real_hash}
\end{figure}

\begin{figure}[p]
\centering
 \includegraphics[width=2.5in, angle = 270]
 {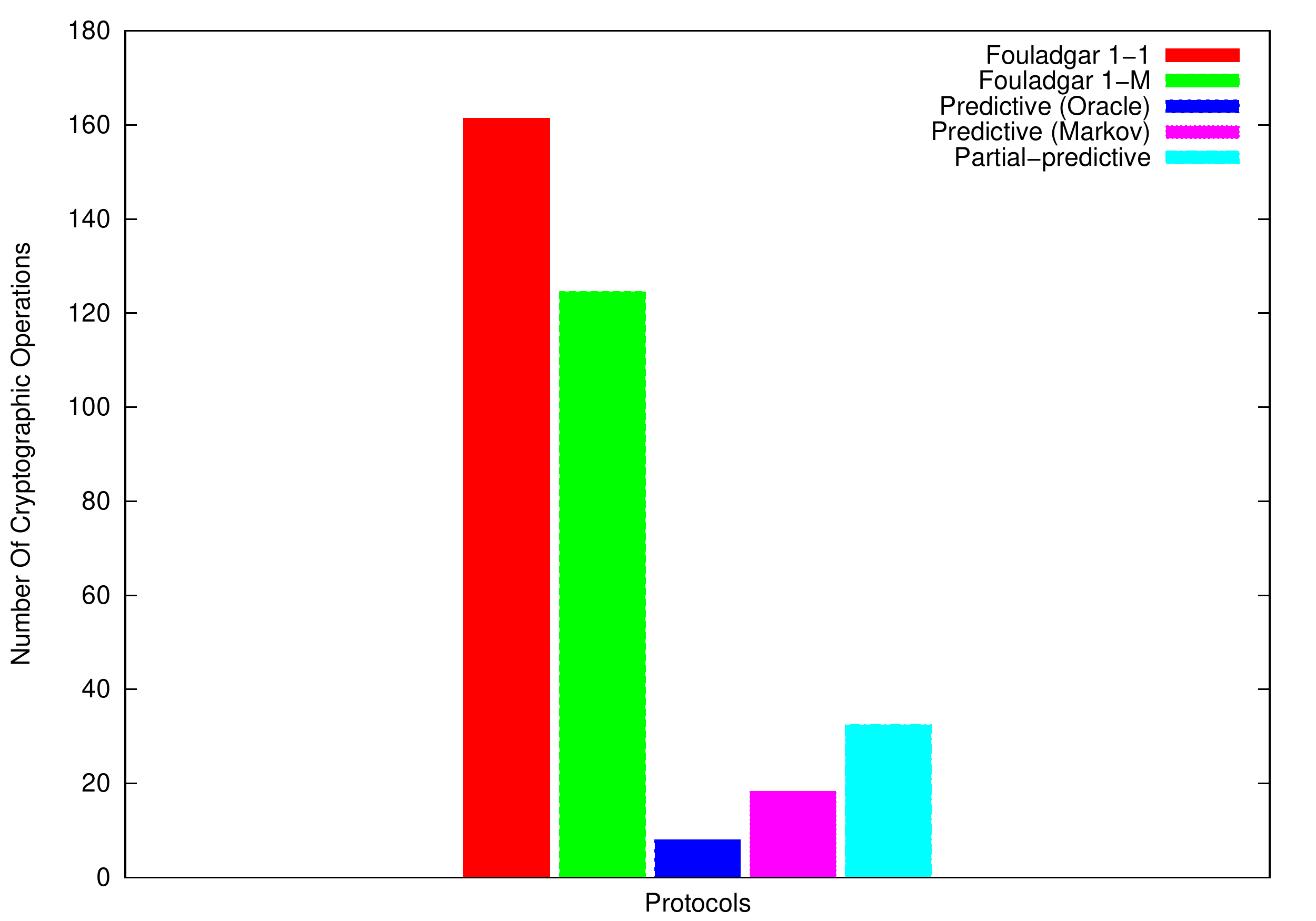}
\caption{Average number of cryptographic operations per identification with
the real data set}
\label{fig:total_hash_real}
\end{figure}

\begin{figure}[p]
\centering
 \subfigure{
   \includegraphics[width=2in, angle = 270]
   {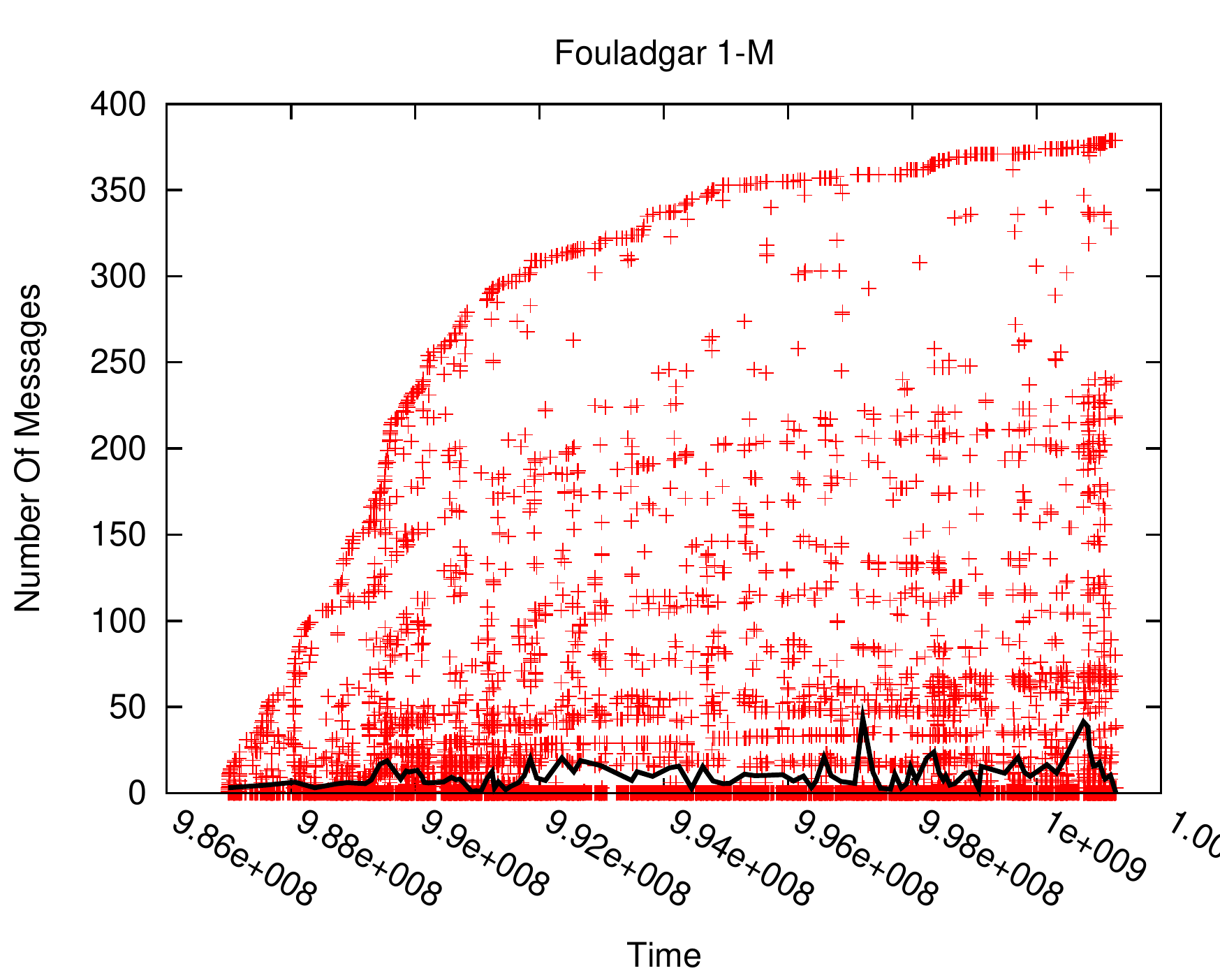}
 }
 \subfigure{
   \includegraphics[width=2in, angle = 270]
   {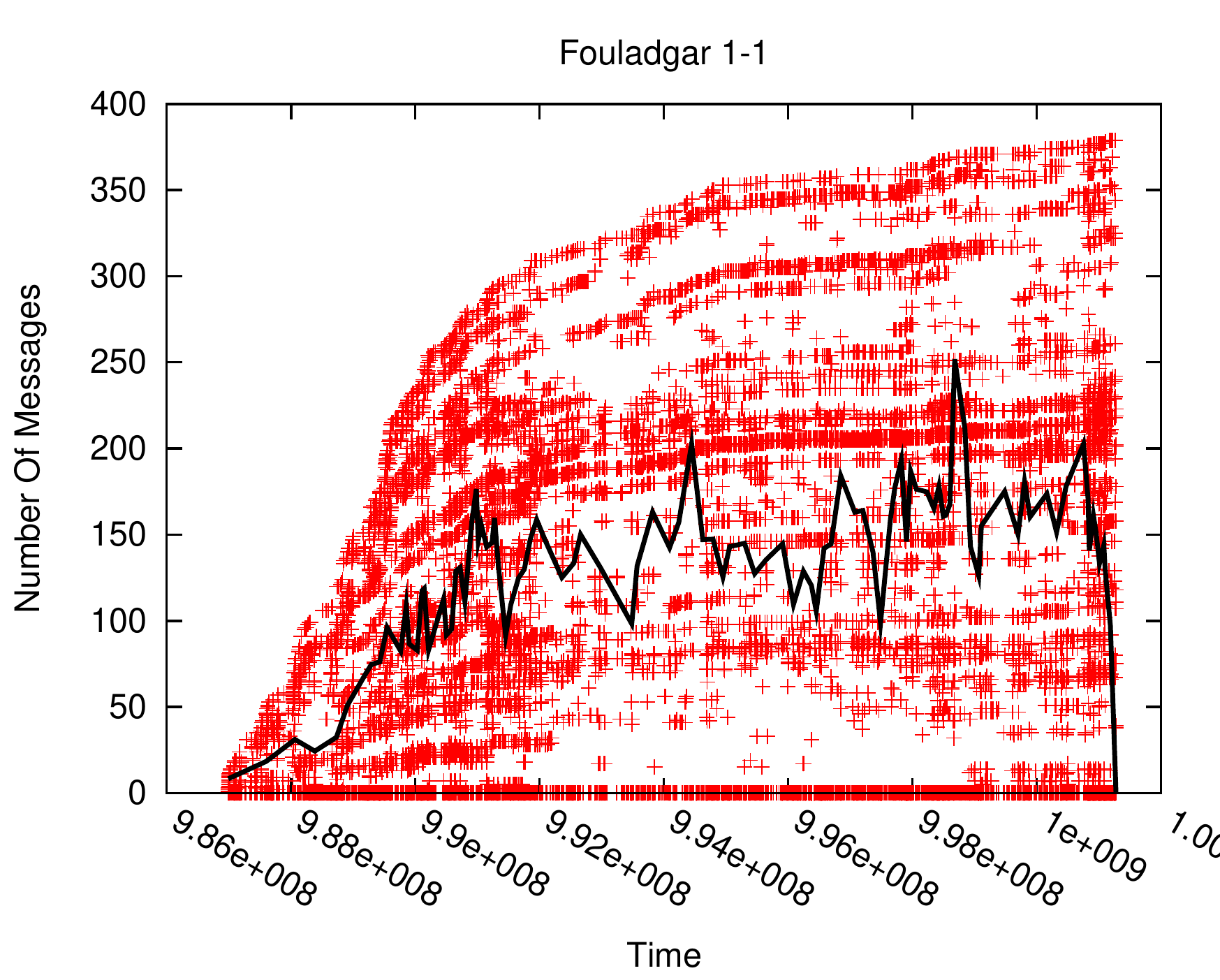}
 }
 \subfigure{
   \includegraphics[width=2in, angle = 270]
   {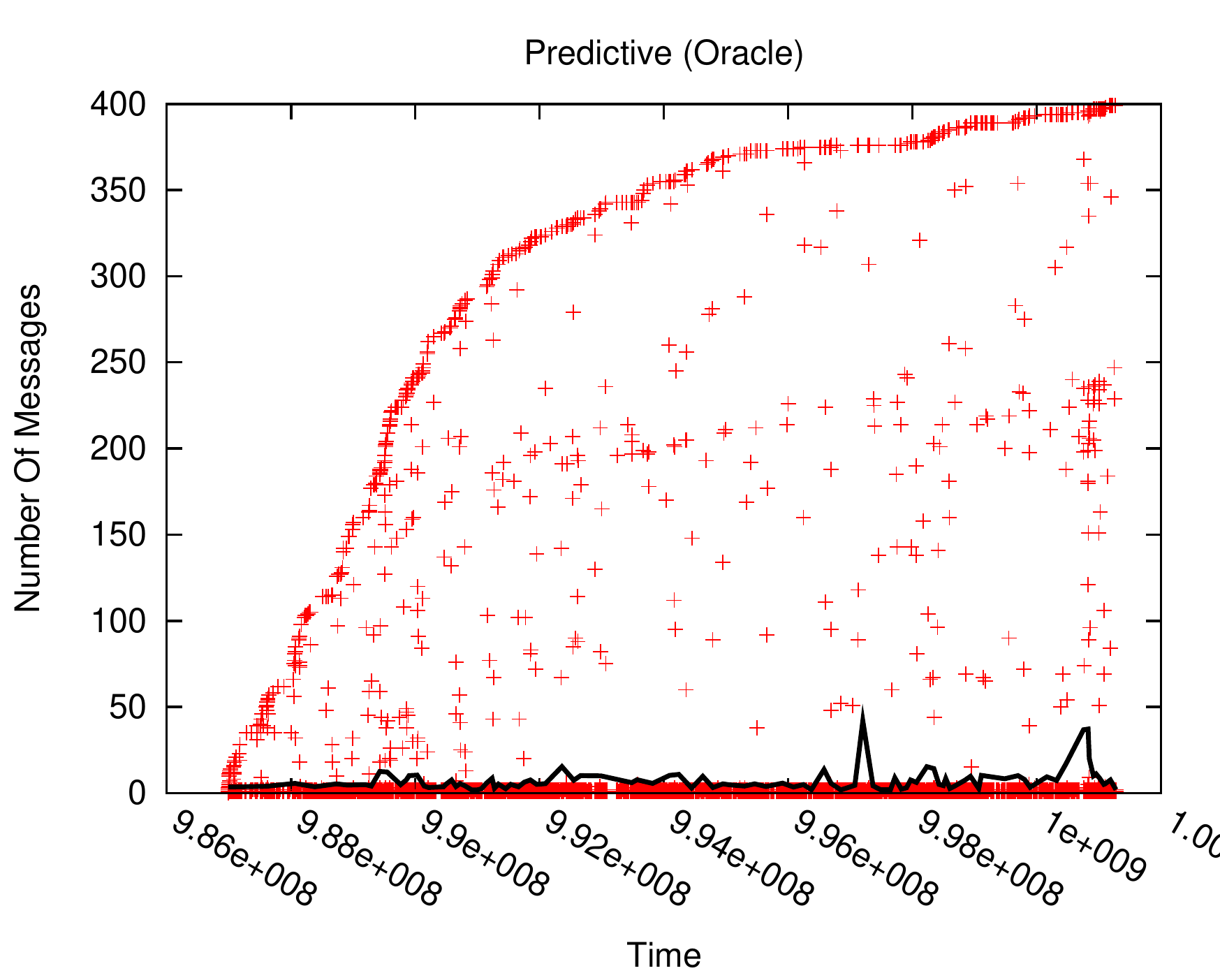}
 }
 \subfigure{
   \includegraphics[width=2in, angle = 270]
   {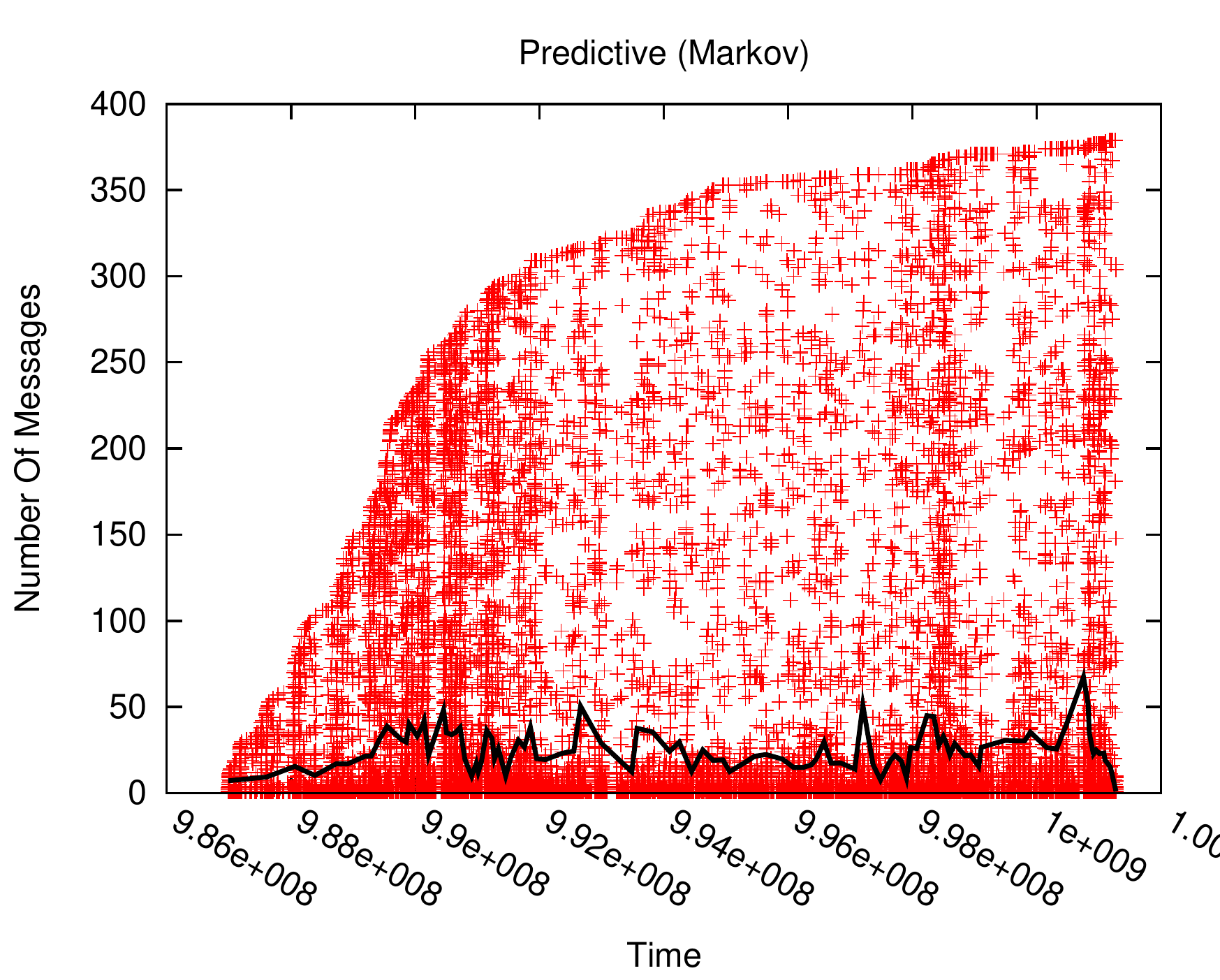}
 }
 \subfigure{
   \includegraphics[width=2in, angle = 270]
   {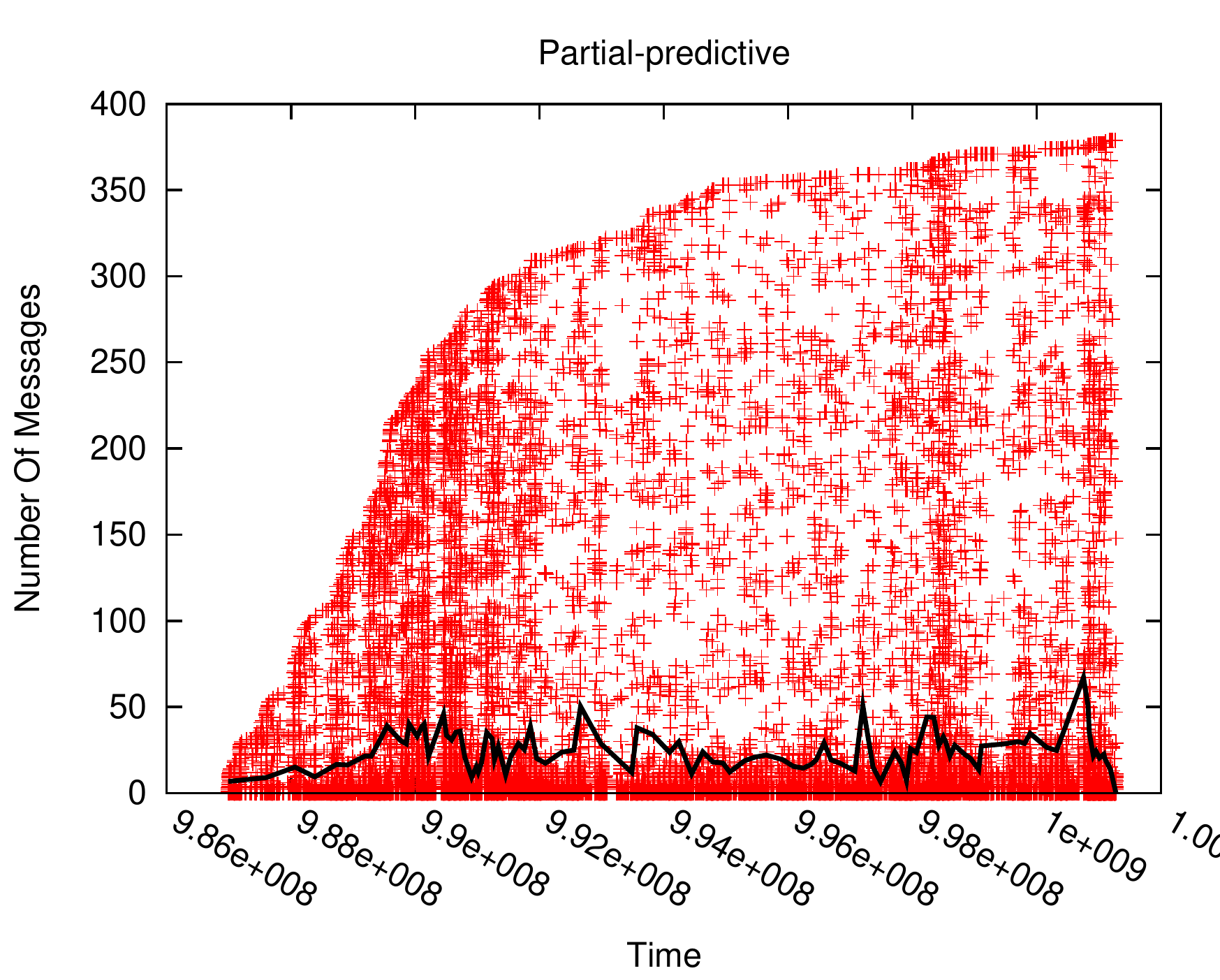}
 }
\caption{Average number of messages sent by each protocol during the
complete simulation with the real data set (in red). The black line represents
the moving average of those values in subsets of 1000 elements. The time scale
is in milliseconds.}
\label{fig:real_msg}
\end{figure}

\begin{figure}[p]
\centering
 \includegraphics[width=2.5in, angle = 270]
 {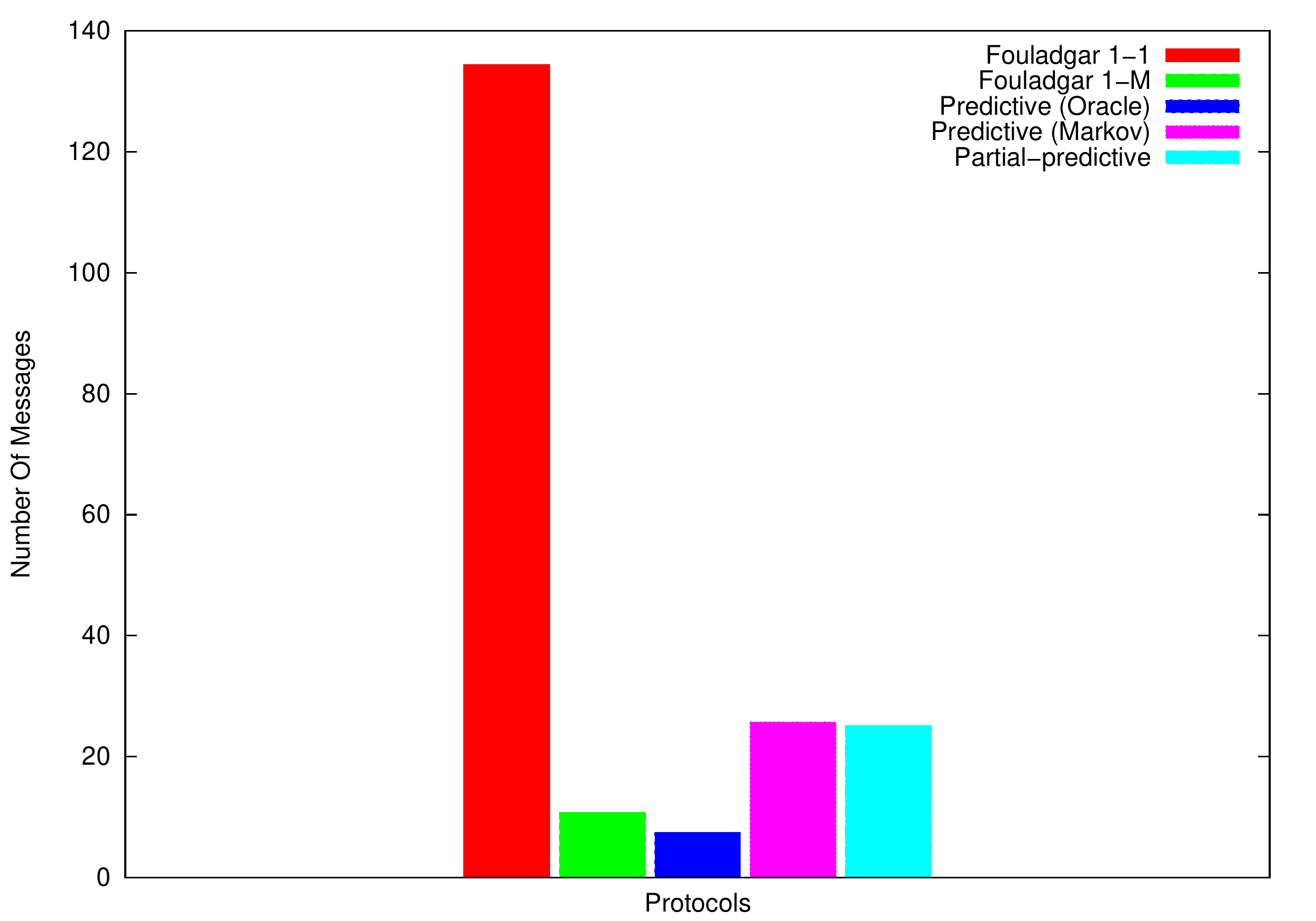}
\caption{Average number of messages sent per identification with the real
data set}
\label{fig:total_msg_real}
\end{figure}

\begin{figure}[p]
\centering
 \includegraphics[width=2.5in, angle = 270]
 {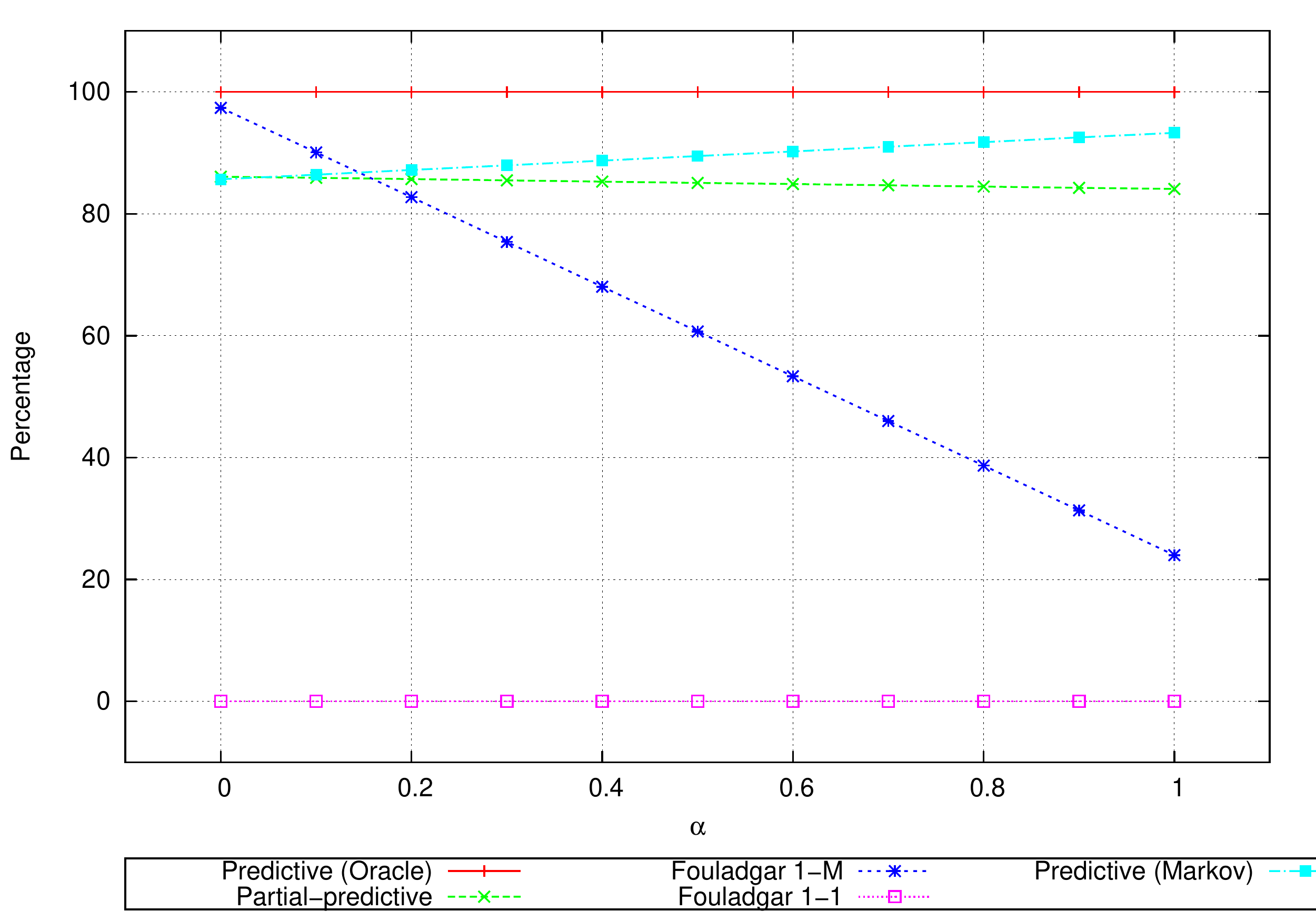}
\caption{The trade-off measure for each protocol and different values of
$\alpha = \{0, 0.1, 0.2, \cdots, 1.0\}$ for the
real data set. The higher the better. (Note that the
colours assigned to the protocols do not coincide with the colours of the
previous Figure).}
\label{fig:comparison_real}
\end{figure}


\section{Conclusions}
\label{sec:conclu}
In this chapter, we have presented a novel protocol that allows efficient identification of
RFID tags by means of a set of collaborative readers. Our proposal
uses location and time of arrival predictors to improve the efficiency of the
widely accepted IRHL scheme. We have shown that our protocol outperforms
previous proposals in terms of scalability whilst guaranteeing the same level
of privacy and security.

From the experimental results obtained, we can conclude that our proposal could be
comparable to highly scalable protocols like the tree-based protocols. However,
we do not sacrifice any privacy to achieve this goal.

Usually, algorithms aimed at location prediction work well in some scenarios,
but their performance decreases in others. Although we have provided some
practical implementations for the predictors, the definition of our protocol is
flexible enough to accept the use of any location predictor. Due to the fact
that the efficiency of our proposal highly depends on the accuracy of the
predictors, in the future we plan to study and compare a
variety of predictors in different
scenarios.

\chapter{The Poulidor Distance-bounding Protocol}
\label{chap:5}

\emph{This chapter describes a novel distance-bounding protocol resistant to both mafia and distance fraud. The experimental results show that this new 
proposal strikes a good balance of memory usage, mafia fraud resistance, and distance fraud resistance.}

\minitoc

The most widespread and low-cost tags are \textit{passive}, meaning that they do not have their own power source, and are supplied by the electromagnetic field of a reader. Although the capacities of such tags are quite limited, some of them benefit from cryptographic building blocks and secure authentication protocols. Nevertheless, Desmedt, Goutier and Bengio~\cite{DBLP:conf/crypto/DesmedtGB87} presented in 1987, an attack that defeated any authentication protocol. In this attack, called \emph{mafia fraud}, the adversary passes through the authentication process by simply relaying the messages between a legitimate reader (the verifier) and a legitimate tag (the prover). Thus she does not need to modify or decrypt any exchanged data. Later in 1993, Brands and Chaum~\cite{188361} proposed a countermeasure that prevents such attack by estimating the distance between the reader and the tag to be authenticated: the \textit{distance-bounding protocol}. They also introduced in~\cite{188361} a new kind of attack, named \emph{distance fraud}, where a dishonest prover claims to be closer to the verifier than she really is.

Since then, many distance-bounding protocols have been proposed to thwart these attacks. In 2005, Hancke and Kuhn~\cite{HanckeK-2005-securecomm} proposed the first distance-bounding protocol dedicated for RFID. The protocol is considered simple in the sense that it only requires an initial \emph{slow phase} followed by a \emph{fast phase} in order to perform both authentication and distance checking. Unfortunately, the adversary success probability regarding mafia and distance frauds is $(3/4)^n$ while one may expect $(1/2)^n$. As a result, many other protocols~\cite{AvoineT-2009-isc,KimA-2009-cans,KimAKSP-2008-icisc,MunillaP-2008-ntms,ReidGTS-2007-asiaccs,TuP-2007-rfidtechnology} have been proposed attempting to improve the Hancke and Kuhn proposal.

Amongst them, to the best of our knowledge, the Kim and Avoine protocol~\cite{KimA-2009-cans} and the Avoine and Tchamkerten protocol~\cite{AvoineT-2009-isc} have the best resistance considering only mafia fraud. However, the Kim and 
Avoine protocol~\cite{KimA-2009-cans} severely sacrifices the distance fraud security, whereas the Avoine and Tchamkerten proposal~\cite{AvoineT-2009-isc} requires an exponential amount of memory ($2^{n+1}-2$ in its standard configuration) to achieve such a high mafia fraud resistance. Neither the Hancke and Kuhn protocol nor the two latter protocols achieve a good balance between memory, mafia fraud resistance and distance fraud resistance.

In this chapter, we perform a detailed analysis of the mafia and distance fraud resistance of the protocols \cite{AvoineT-2009-isc} and \cite{KimA-2009-cans}. Then, we  introduce the concept of distance-bounding protocols based on graphs, and we propose a new distance-bounding protocol based on a particular type of graph. Our goal is not to provide the best protocol in terms of mafia fraud or distance fraud, but to design a protocol that ensures a good trade-off between these concerns, while still using a linear amount of memory with respect to the number of rounds. This means that our protocol is never the best one when considering only one property, but is a good option when considering the three properties altogether. This is why we name our protocol \emph{Poulidor} after a famous French bicycle racer known as \emph{The Eternal Second}: never the best in any race, but definitively the best on average.

\section{Previous proposals}

In terms of efficiency and resource consumption, our proposal is comparable to the Hancke and Kuhn~\cite{HanckeK-2005-securecomm} and Kim and Avoine~\cite{KimA-2009-cans} protocols. Therefore, we explain below those two proposals. We also detail the Avoine and Tchamkerten protocol~\cite{AvoineT-2009-isc} because our aim is to be as resilient as this protocol to mafia and distance frauds.

\subsection{Hancke and Kuhn's protocol}

Hancke and Kuhn's protocol (HKP)~\cite{HanckeK-2005-securecomm}, depicted in Figure~\ref{figure:HK}, is a key-reference protocol in terms of distance bounding devoted to RFID systems. HKP is a simple and fast protocol, but it suffers from a high adversary success probability.

\subsubsection{Initialisation}
The prover ($P$) and the verifier ($V$) share a secret $x$ and agree on: (i) a security parameter $n$, (ii) a public hash function $H$, whose output size is $2n$, and (iii) a certain timing bound $\Delta t_{\text{max}}$.

\subsubsection{Protocol} HKP consists of two phases: a slow one followed by a fast one. During the slow phase $V$ generates a random nonce $N_V$ and sends it to $P$. Reciprocally, $P$ generates $N_P$ and sends it to $V$. Both $V$ and $P$ compute $H^{2n}:=H(x,N_P,N_V)$. In what follows, $H_i$ ($1\leq i\leq 2n$) denotes the $i$-th bit of $H^{2n}$, and $H_i\dots H_j$ ($1\leq i<j\leq 2n$) denotes the concatenation of the bits from $H_i$ to $H_j$. Then $V$ and $P$ split $H^{2n}$ into two registers of length $n$: $R^0:=H_1\dots H_n$ and $R^1:=H_{n+1}\dots H_{2n}$. The fast phase then consists of $n$ rounds. In each of them, $V$ picks a random bit $c_i$ (the challenge) and sends it to $P$. The latter immediately answers $r_i := R^{c_i}_i$, the $i$-th bit of the register $R^{c_i}$.

\subsubsection{Verification} At the end of the fast phase, the verifier checks that the answers received from the prover are correct and that $\Delta t_i\leq\Delta t_{\text{max}}$  ($1\leq i\leq n$) .

\begin{figure}[!h]
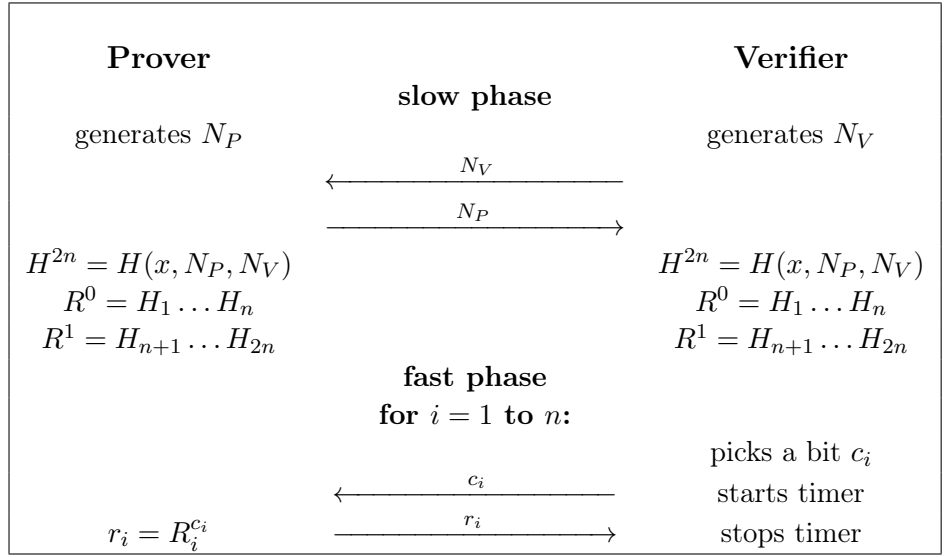
\centering
\begin{tabular}{|ccc|}
\hline
&& \\
\large\textbf{Prover}&&\large\textbf{Verifier}\\
&\textbf{slow phase}&\\
generates $N_{P}$&&generates $N_{V}$\\
&$\xleftarrow{\ \ \ \ \ \ \ \ \ \ \ \ \ N_{V}\ \ \ \ \ \ \ \ \ \ \ \ \ }$&\\
&$\xrightarrow{\ \ \ \ \ \ \ \ \ \ \ \ \ N_{P}\ \ \ \ \ \ \ \ \ \ \ \ \ }$&\\
$H^{2n}=H(x,N_{P},N_{V})$&&$H^{2n}=H(x,N_{P},N_{V})$\\
$R^0=H_{1}\dots H_{n}$&&$R^0=H_{1}\dots H_{n}$\\
$R^1=H_{n+1}\dots H_{2n}$&&$R^1=H_{n+1}\dots H_{2n}$\\
&\textbf{fast phase}&\\
&\textbf{for $i=1$ to $n$:}&\\
&&picks a bit $c_{i}$\\
&$\xleftarrow{\ \ \ \ \ \ \ \ \ \ \ \ \ c_i\ \ \ \ \ \ \ \ \ \ \ \ \ }$&starts timer\\
$r_{i}=R^{c_{i}}_{i}$&$\xrightarrow{\ \ \ \ \ \ \ \ \ \ \ \ \ r_i\ \ \ \ \ \ \ \ \ \ \ \ \ }$&stops timer\\
\hline
\end{tabular}
\caption{\label{figure:HK}Hancke and Kuhn's protocol}
\end{figure}

\subsection{Kim and Avoine's protocol}

Kim and Avoine's protocol (KAP)~\cite{KimA-2009-cans}, represented in Figure \ref{figure:AK}, basically relies on \emph{predefined} challenges. Predefined challenges allow the prover to detect that an attack occurs as follows: the prover and the verifier agree on some predefined 1-bit challenges; if the adversary sends in advance a challenge to the prover that is different from the expected predefined challenge, then the prover detects the attack and since then, it sends random responses to the adversary. The complete description of the KAP protocol is provided below.

\subsubsection{Initialisation} The prover (P) and the verifier (V) share a secret $x$ and agree on: (i) a security parameter $n$, (ii) a public hash function $H$, whose output size is $4n$, and (iii) a certain timing bound $\Delta t_{\text{max}}$.

\subsubsection{Protocol}
As previously, $V$ and $P$ exchange nonces $N_V$ and $N_P$. From these values they compute $H^{4n}=H(x,N_P,N_V)$, and split it in four registers. $R^0:=H_1\dots H_n$ and $R^1:=H_{n+1}\dots H_{2n}$ are the potential responses. The register $D:=H_{3n+1}\dots H_{4n}$ contains the potential predefined challenges. Finally, the register $T:=H_{2n+1}\dots H_{3n}$ allows the verifier to decide whether a predefined challenge should be sent: in round $i$, if $T_i=1$ then a random challenge is sent; if $T_i=0$ then the predefined challenge $D_i$ is sent instead.

\subsubsection{Verification} At the end of the fast phase, the verifier checks whether the answers received from the prover are correct and that $\Delta t_i\leq\Delta t_{\text{max}}$  ($1 \leq i \leq n$).

\begin{figure}[!h]
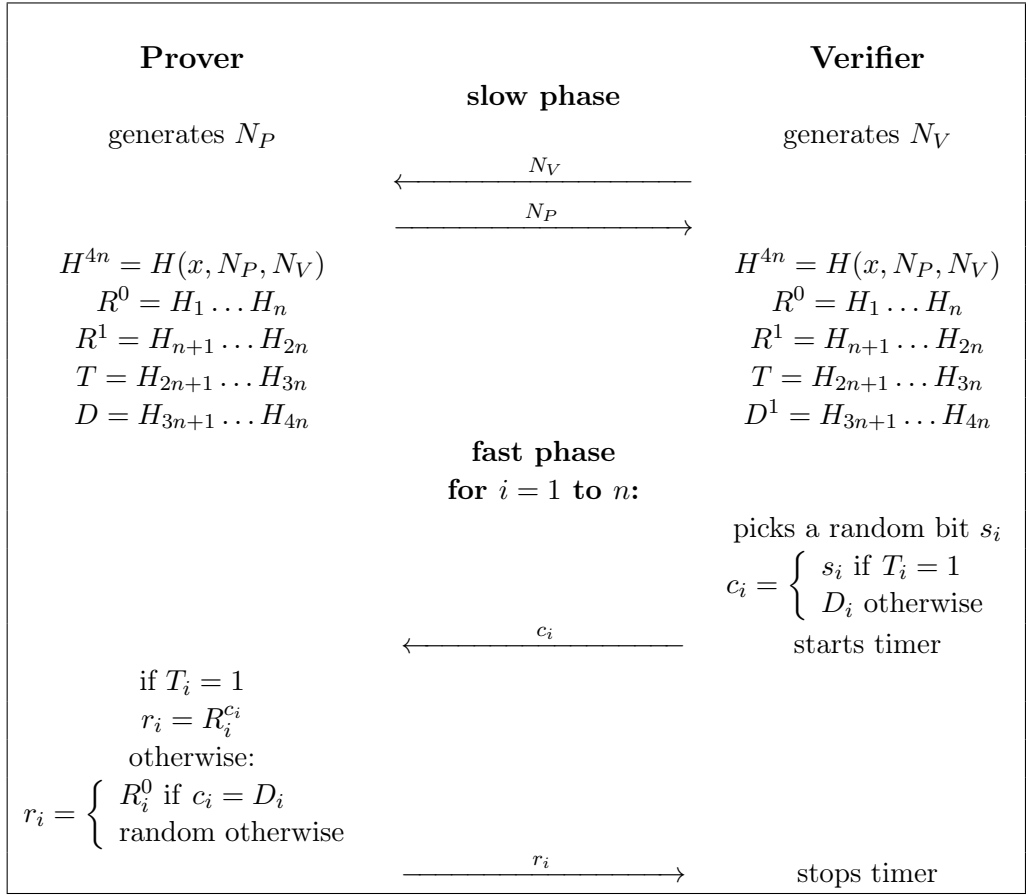
\centering
\begin{tabular}{|ccc|}
\hline
&& \\
\large\textbf{Prover}&&\large\textbf{Verifier}\\
&\textbf{slow phase}&\\
generates $N_{P}$&&generates $N_{V}$\\
&$\xleftarrow{\ \ \ \ \ \ \ \ \ \ \ \ \ N_{V}\ \ \ \ \ \ \ \ \ \ \ \ \ }$&\\
&$\xrightarrow{\ \ \ \ \ \ \ \ \ \ \ \ \ N_{P}\ \ \ \ \ \ \ \ \ \ \ \ \ }$&\\
$H^{4n}=H(x,N_{P},N_{V})$&&$H^{4n}=H(x,N_{P},N_{V})$\\
$R^0=H_1\dots H_n$&&$R^0=H_1\dots H_n$\\
$R^1=H_{n+1}\dots H_{2n}$&&$R^1=H_{n+1}\dots H_{2n}$\\
$T=H_{2n+1}\dots H_{3n}$&&$T=H_{2n+1}\dots H_{3n}$\\
$D=H_{3n+1}\dots H_{4n}$&&$D^1=H_{3n+1}\dots H_{4n}$\\
&\textbf{fast phase}&\\
&\textbf{for $i=1$ to $n$:}&\\
&&picks a random bit $s_{i}$\\
&&$c_i=\left\{\begin{array}{l}s_i\mathrm{\ if\  }T_i=1\\D_i\mathrm{\ otherwise\ }\end{array}\right.$\\
&$\xleftarrow{\ \ \ \ \ \ \ \ \ \ \ \ \ c_i\ \ \ \ \ \ \ \ \ \ \ \ \ }$&starts timer\\
if $T_i=1$&&\\
$r_i=R^{c_i}_i$&&\\
otherwise:&&\\
$r_i=\left\{\begin{array}{l}R^0_i\mathrm{\ if\  }c_i=D_i\\\mathrm{random\ otherwise}\end{array}\right.$&&\\
&$\xrightarrow{\ \ \ \ \ \ \ \ \ \ \ \ \ r_i\ \ \ \ \ \ \ \ \ \ \ \ \ }$&stops timer\\
\hline
\end{tabular}
\caption{\label{figure:AK}Kim and Avoine's protocol}
\end{figure}

\subsection{Avoine and Tchamkerten's protocol} \label{sec:avp}

The Avoine and Tchamkerten's protocol (ATP)~\cite{AvoineT-2009-isc} is slightly different from the other existing distance bounding protocols. This protocol is also based  on single bit challenge/response exchanges. However, the authors propose the use of a decision tree to set up the fast phase. Figure~\ref{figure:AT} depicts the protocol detailed below.

\subsubsection{Initialisation}
The prover and the verifier share a secret $x$, and they agree on: (i) two  security parameters $n=\alpha k$ and $m$, (ii) a pseudo-random function $PRF$ whose output size is at least $m+\alpha(2^{k+1}-2)$ bits, and (iii) a timing bound $\Delta t_{\text{max}}$.

\subsubsection{Protocol} The prover $P$  and the verifier $V$ generate two nonces $N_{P}$ and $N_{V}$ respectively. The verifier sends his nonce to $P$. Upon reception, the latter computes $PRF(x,N_{P},N_{V})$ and sends $[PRF(x,N_{P},N_{V})]_{1}^m$, the first  $m$ bits of $PRF(x,N_{P},N_{V})$, and $P$ also sends $N_P$. These bits are used for the authentication.

$P$ and $V$ use the remaining $\alpha(2^{k+1}-2)$ bits to label the nodes of $\alpha$ binary decision trees of depth $k$. Each node of the trees\footnote{Except the roots.} is labeled by one bit from $[PRF(x,N_{P},N_{V})]_{m+1}^{m+\alpha(2^{k+1}-2)}$ (the remaining bits) in a one-to-one way. These labels represent the prover's responses during the fast phase. The challenges are represented by the edges of the trees; left and right edges are labeled with 0 and 1 respectively.

Afterwards, the fast phase begins, for $1\leq i \leq\alpha$, and $1\leq j\leq k$, $V$ picks a bit $c_{j}^i$ at random, starts a timer and sends $c_j^i$ to $P$. The latter immediately answers  a bit $r_{j}^i=\mathrm{node}(c_{1}^i,\dots c_{j}^i)$, \emph{i.e.} the value of the node located in the $i$-th tree and reached from the root by taking the sequence of decision bits $c_{1}^i,\dots,c_{j}^i$. Once $V$ receives $P$'s response, he stops his timer and computes $\Delta t_j^i$.

\subsubsection{Verification}
The verifier authenticates the prover if the $m$ bits, sent during the slow phase, are correct. The prover succeeds in the distance-bounding stage, if all his responses are correct and if for all $1\leq i\leq\alpha$ and $1\leq j\leq k$, $\Delta t_{j}^i\leq\Delta t_{\text{max}}$.

\begin{figure}[!h]
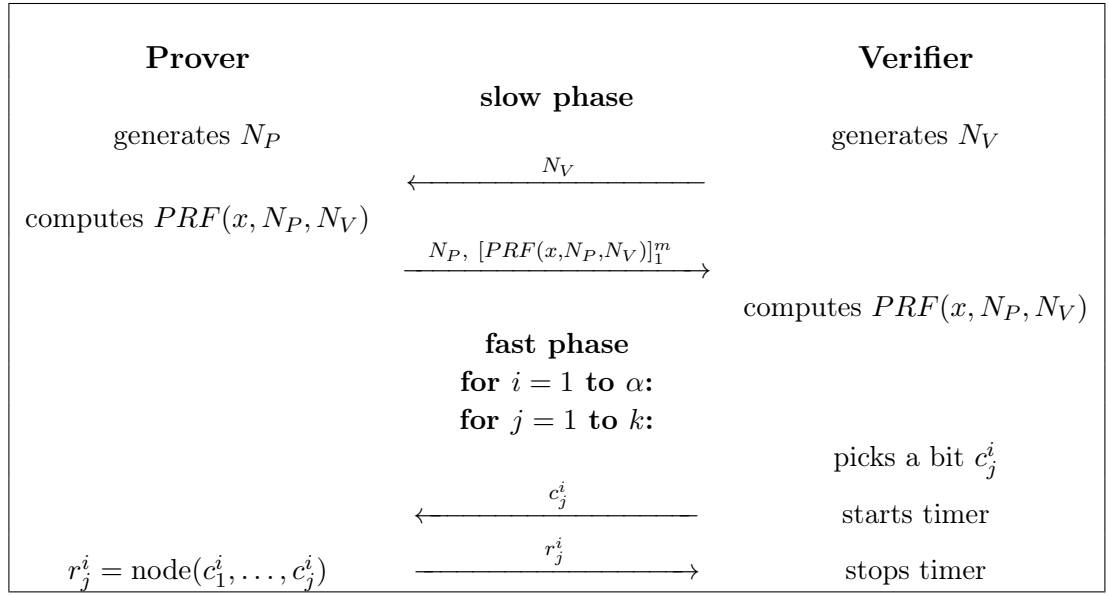
\centering
\begin{tabular}{|ccc|}
\hline
&& \\
\large\textbf{Prover}&&\large\textbf{Verifier}\\
&\textbf{slow phase}&\\
generates $N_{P}$&&generates $N_{V}$\\
&$\xleftarrow{\ \ \ \ \ \ \ \ \ \ \ \ \ N_{V}\ \ \ \ \ \ \ \ \ \ \ \ \ }$&\\
computes $PRF(x,N_{P},N_{V})$&&\\
&$\xrightarrow{\ \ N_{P},\ [PRF(x,N_{P},N_{V})]_{1}^m\ \ \ }$&\\
&&computes $PRF(x,N_{P},N_{V})$\\
&\textbf{fast phase}&\\
&\textbf{for $i=1$ to $\alpha$:}&\\
&\textbf{for $j=1$ to $k$:}&\\
&&picks a bit $c_{j}^i$\\
&$\xleftarrow{\ \ \ \ \ \ \ \ \ \ \ \ \ c_j^i\ \ \ \ \ \ \ \ \ \ \ \ \ }$&starts timer\\
$r^{i}_j=\mathrm{node}(c_{1}^i,\dots,c_{j}^i)$&$\xrightarrow{\ \ \ \ \ \ \ \ \ \ \ \ \ r_j^i\ \ \ \ \ \ \ \ \ \ \ \ \ }$&stops timer\\
\hline
\end{tabular}
\caption{\label{figure:AT}Avoine and Tchamkerten protocol}
\end{figure}


\section{Graph-based distance-bounding protocol} \label{sec:proposal}
\label{section:proposal}

The ATP protocol~\cite{AvoineT-2009-isc} in its standard configuration ($\alpha = 1$) relies on a binary tree. The amount of memory needed to build this binary tree is exponential regarding the number of rounds. Although the authors in~\cite{AvoineT-2009-isc} proposed to split the binary tree in order to reduce the memory requirements, they pointed out that this procedure leads to a significant decrease in the security level of the protocol. We go a step forward and propose protocols based on graphs rather than trees. The graph-based protocols, as presented below, provide a greater design flexibility, a high security level and a low memory consumption.

\subsection{Initialisation}

\subsubsection{Parameters}

The prover $P$ and the verifier $V$ agree on four public parameters: (i) a security parameter $n$ that represents the number of rounds in the protocol, (ii) a timing bound $\Delta t_{\text{max}}$, (iii) a pseudo random function $PRF$ whose output size is $4n$ bits, and (iv) a directed graph $G$ whose characteristics are discussed below. They also agree on a shared secret $x$.

\subsubsection{Graph}

 To achieve $n$ rounds, the proposed graph requires $2n$ nodes $\{q_0,q_1, \dots,q_{2n-1}\}$, and $4n$ edges $\{s_0, s_1, \cdots, s_{2n-1}, \ell_0, \ell_1, \cdots, \ell_{2n-1} \}$ such that $s_i$ $(0 \leq i \leq 2n-1)$ is an edge from $q_i$ to $q_{(i+1) \mod 2n}$, and $\ell_i$ $(0 \leq i \leq 2n-1)$ is an edge from $q_i$ to $q_{(i+2) \mod 2n}$. Figure~\ref{fig:graph} depicts the graph when $n = 4$.

\begin{figure}[!h]
\centering
\scalebox{0.5}{
\includegraphics{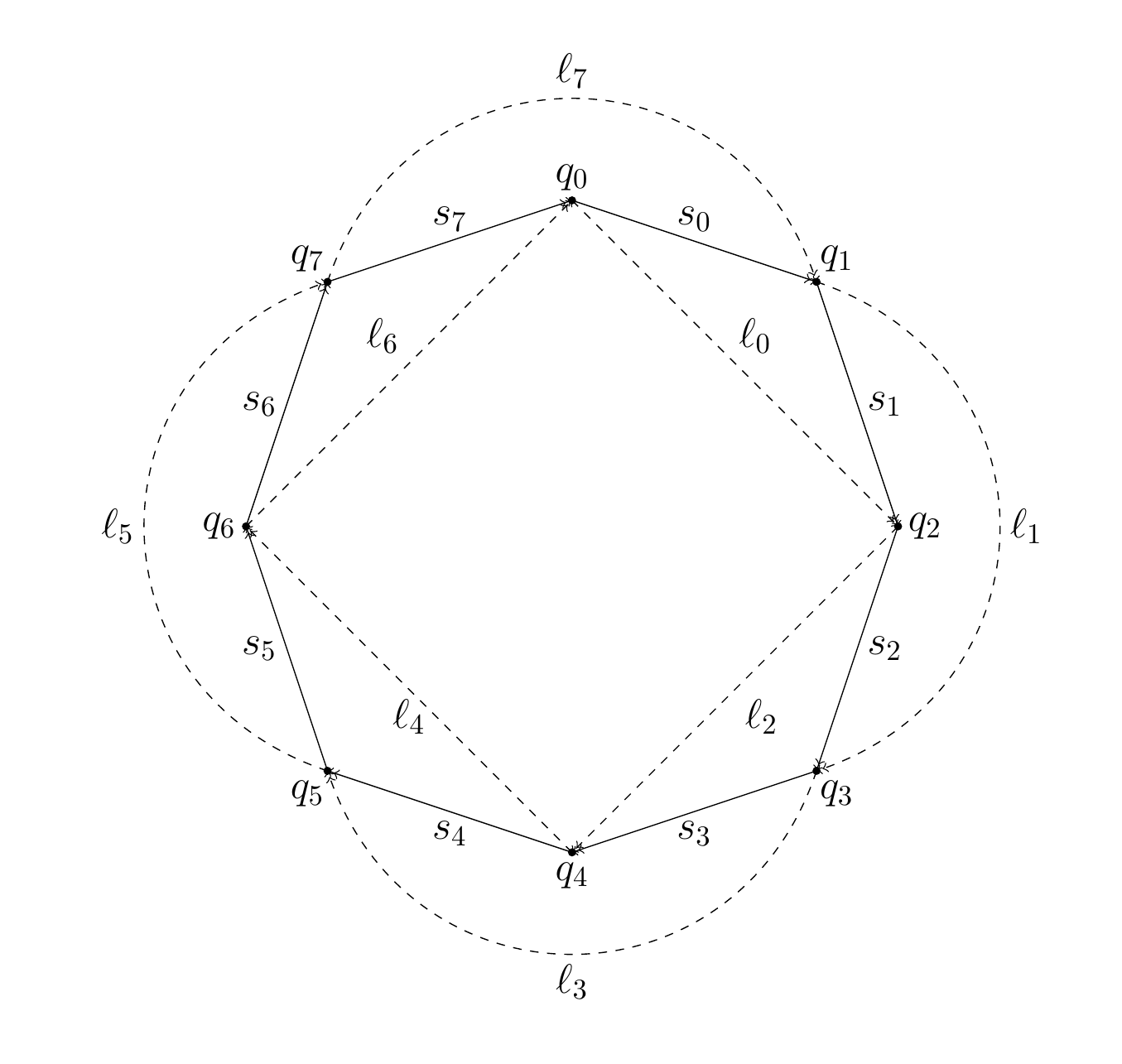}}
\caption{Graph when $n = 4$\label{fig:graph}}
\end{figure}

\subsection{Exchanges}

As described below, the protocol is divided in two phases, a slow one and a fast one. Figure~\ref{fig:proposal} summarises the protocol.
\paragraph{Slow phase --}
 $P$ and $V$  generate nonces $N_P$ and $N_V$, respectively, and exchange them. From these values and the secret $x$, they compute $H_{1}||\dots||H_{4n}=PRF(x,N_{P},N_{V})$ where $H_i$ denotes the $i$-th bit of the output of $PRF(x,N_{P},N_{V})$. The bits $H_{1},\dots,H_{4n}$ set up the graph $G$ as follows: the first $2n$ bits are used to value the nodes while the remaining bits are used to value the edges $s_i$ ($0 \leq i \leq 2n-1$); finally, $\ell_i = s_i \oplus 1$ ($0 \leq i \leq 2n-1$).

\paragraph{Fast phase --}

This phase consists of $n$ stateful rounds numbered from $0$ to $n-1$. In the $i$-th round $P$'s state and $V$'s state are represented by the nodes $q_{p_i}$ and $q_{v_i}$ respectively: initially $q_{p_0} = q_{v_0} = q_0$. Upon reception of the $i$-th challenge $c_i$, $P$ moves to the node $q_{p_i}$ to $q_{p_{i+1}}$ in the following way: $q_{p_{i+1}} = q_{(p_i+1) \mod 2n}$ if $s_i$ is labeled with $c_i$, otherwise $q_{p_{i+1}} = q_{(p_i+2) \mod 2n}$. Finally, the prover sends as response $r_i$ the bit-value of the node $q_{p_{i+1}}$. Upon reception of the prover's answer $r_{i}$, the verifier stops his timer, and computes $\Delta t_i$, \emph{i.e.} the round trip time spent for this exchange. Besides, $V$ moves to the node $q_{v_{i+1}}$ using the challenge $c_i$ (as the prover did but from the node $q_{v_i}$) and checks if $q_{v_{i+1}} = r_i$.

\begin{figure}[!h]
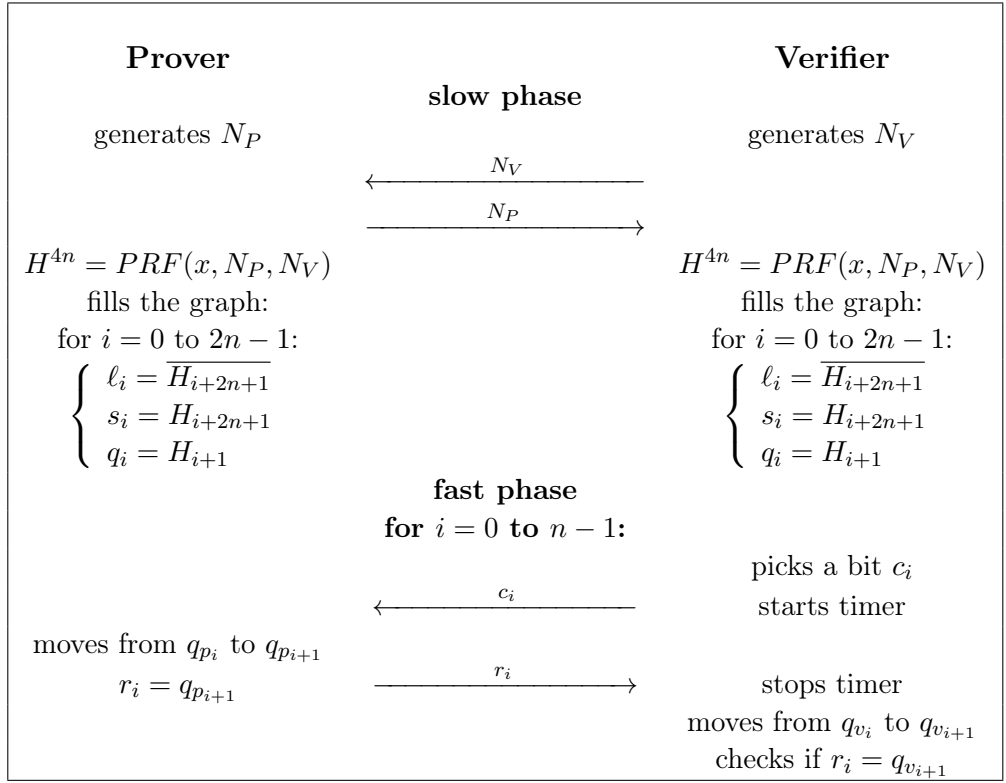

\centering
\begin{tabular}{|ccc|}
\hline
&&\\
\textbf{\large{Prover}}&&\textbf{\large{Verifier}}\\
&\textbf{slow phase}&\\
generates $N_{P}$&&generates $N_{V}$\\
&$\xleftarrow{\ \ \ \ \ \ \ \ \ \ \ \ N_{V}\ \ \ \ \ \ \ \ \ \ \ \ }$&\\
&$\xrightarrow{\ \ \ \ \ \ \ \ \ \ \ \ N_{P}\ \ \ \ \ \ \ \ \ \ \ \ }$&\\
$H^{4n}=PRF(x,N_{P},N_{V})$&&$H^{4n}=PRF(x,N_{P},N_{V})$\\
fills the graph:&&fills the graph:\\
for $i=0$ to $2n-1$:&&for $i=0$ to $2n-1$:\\
 $\left\{\begin{array}{l}

\ell_i = \overline{H_{i+2n+1}} \\
s_i=H_{i+2n+1} \\

q_i=H_{i+1}

\end{array}\right.$&& $\left\{\begin{array}{l}

\ell_i = \overline{H_{i+2n+1}} \\
s_i=H_{i+2n+1} \\

q_i=H_{i+1}
\end{array}\right.$\\
&\textbf{fast phase}&\\
&\textbf{for $i=0$ to $n-1$:}&\\
&&picks a bit $c_{i}$\\
&$\xleftarrow{\ \ \ \ \ \ \ \ \ \ \ \ c_i\ \ \ \ \ \ \ \ \ \ \ \ }$&starts timer\\
moves from $q_{p_i}$ to $q_{p_{i+1}}$&&\\
$r_i=q_{p_{i+1}}$&$\xrightarrow{\ \ \ \ \ \ \ \ \ \ \ \ r_i\ \ \ \ \ \ \ \ \ \ \ \ }$&stops timer\\
&&moves from $q_{v_i}$ to $q_{v_{i+1}}$\\
&&checks if $r_i = q_{v_{i+1}}$\\
\hline
\end{tabular}
\caption{The new graph-based proposal\label{fig:proposal}}
\end{figure}

\subsection{Verification}

The authentication succeeds if all the responses are correct, and each round is completed within the time bound $\Delta t_{\text{max}}$.


\section{Security analysis of the graph-based protocol} \label{sec:security}

As stated in the introduction, mafia fraud and distance fraud are the two main security concerns when considering distance bounding protocols. We analyse in this section the graph-based protocol with respect to these frauds.

\subsection{Mafia fraud} \label{sec:mafia}

To analyse the mafia fraud we consider the adversary abilities complying with the models provided in \cite{AvoineT-2009-isc}, \cite{HanckeK-2005-securecomm} and \cite{KimA-2009-cans}. Below, we define the \emph{head node} and rephrase the well-known pre-ask strategy (see for example \cite{MunillaOP-2006-rfidsec}) with our terminology.

\begin{definition}[Head node] \label{def:head}
Given a sequence of challenges $\{c_1, c_2, \cdots, c_i\}$ ($1 \leq i \leq n$), the head node is the node that should be used by the prover to send the response to the verifier according to this sequence of challenges. The head node is denoted as $\Omega(c_1, c_2, \cdots, c_i)$.
\end{definition}

\begin{definition}[Pre-ask strategy]  \label{def:adv_aby}
 The pre-ask strategy begins at the end of the slow phase and before the beginning of the fast phase. First, the adversary sends a sequence of challenges $\{ \tilde{c_1}, \tilde{c_2}, \cdots, \tilde{c_n}\}$ to the prover and receives a sequence of responses $\{\Omega(\tilde{c_1}), \Omega(\tilde{c_1}, \tilde{c_2}), \cdots , \Omega(\tilde{c_1}, \tilde{c_2}, \cdots, \tilde{c_n})\}$. \\Later, during the fast phase, the adversary tries to use the information obtained from the prover in the best way. Let us consider $\{c_1, c_2, \cdots c_i\}$ the challenges sent by the verifier until the $i$-th round during the fast phase. If $\forall j$ s.t. $1 \leq j \leq i$, we have $c_j = \tilde{c_j}$ then the adversary sends as response $\Omega(\tilde{c_1}, \tilde{c_2}, \cdots, \tilde{c_i})$. Otherwise she sends as response the value $\Omega(\tilde{c_1}, \tilde{c_2}, \cdots, \tilde{c_j})$ where $j$ is selected according to some rule that will be defined later.
\end{definition}

\begin{remark} \label{rem:adv}
Sending a combination of two or more values as response is completely useless for the adversary because the nodes' values in the graph are independent from each other. Furthermore, in the graph-based protocol one node is never used twice to send a response. Therefore, the adversary can neither obtain nor infer more information than the one obtained from the prover. Finally, note that in the security analysis of previous protocols \cite{AvoineT-2009-isc}, \cite{HanckeK-2005-securecomm} and \cite{KimA-2009-cans}, the best adversary strategy is to pick $j = i$ for every round, \emph{i.e.} the adversary sends exactly what she received from the prover in the $i$-th round. However, as we explain below, in the graph-based protocol it makes sense to send a value received in a different round.
\end{remark}

While the challenges sent by the adversary match the challenges sent by the verifier, the adversary is able to send the correct response. However, after the first \emph{incorrect} adversary challenge, she can no longer be convinced about the correctness of her response. Consequently, we analyse below the adversary success probability when the adversary sends at least an \emph{incorrect} challenge to the prover during the pre-ask strategy.

\begin{theorem} \label{theo:first}
Let $(c_1, c_2, \cdots, c_i)$ be the sequence of verifier challenges until the $i$-th round, and let $(\tilde{c_1}, \tilde{c_2}, \cdots, \tilde{c_n})$  be the sequence of adversary challenges in the pre-ask strategy.  Let $F$ be the random variable  representing the first round in which $c_t \neq \tilde{c_t}$ $(1 \leq t \leq n)$. Given $\Omega(\tilde{c_1}, \tilde{c_2}, \cdots, \tilde{c_j})$, the adversary response in the $i$-th round for some $(1 \leq j \leq n)$, we have:

$$
\Pr(\Omega(\tilde{c_1}, \tilde{c_2}, \cdots, \tilde{c_j}) = \Omega(c_1, c_2, \cdots, c_i)|F = t) = \left\{\begin{array}{l l}
1 &  \quad \text{if} \quad i < t \quad \text{and} \quad i = j,\\
\frac{1}{2} & \quad \text{if} \quad i < t \quad \text{and} \quad i \neq j,\\
\frac{1}{2} & \quad \text{if} \quad i \geq t \quad \text{and} \quad j < t,\\
p(t) & \quad \text{if} \quad i \geq t \quad \text{and} \quad j \geq t,
\end{array}\right.
$$
where $p(t) = \frac{1}{2}+\frac{1}{2^{i+j-2t+2}}\sum_{k = 0}^{k = 2n-1} \left (A^{i-t}[1, k]A^{j-t}[2, k] + A^{i-t}[2, k]A^{j-t}[1, k] \right)$, and $A$ is the adjacency matrix of the graph which represents the graph-based protocol.
\end{theorem}

\begin{proof}
We analyse the problem by cases:

\begin{case}[$i < t$ \text{and} $i = j$]
As $i<t$ then $\forall 1 \leq k \leq i$, $\tilde{c_k} = c_k$, therefore $\Omega(\tilde{c_1}, \tilde{c_2}, \cdots, \tilde{c_j}) = \Omega(c_1, c_2, \cdots, c_i)$.
\end{case}

\begin{case}[$i < t$ \text{and} $i \neq j$] \label{case:2}
As $i<t$ then $\Omega(\tilde{c_1}, \tilde{c_2}, \cdots, \tilde{c_i}) = q_{v_i} =\Omega(c_1, c_2, \cdots, c_i)$. On the other hand, as $i \neq j$ then $q_{v_i}$ and $\Omega(\tilde{c_1}, \tilde{c_2}, \cdots, \tilde{c_j})$ are not the same node in the graph. As the node values in the graph are independent, we conclude that, $\Pr(\Omega(\tilde{c_1}, \tilde{c_2}, \cdots, \tilde{c_j}) = \Omega(c_1, c_2, \cdots, c_i)) = \frac{1}{2}$.
\end{case}

\begin{case}[$i \geq t$ and $j < t$]
This case is analog to Case \ref{case:2}.
\end{case}

\begin{case}[$i \geq t$ and $j \geq t$]
Let be $q_{v_i} = \Omega(c_1, c_2, \cdots, c_i)$ and \\$q_{a_j} = \Omega(\tilde{c_1}, \tilde{c_2}, \cdots, \tilde{c_j})$, so:

\begin{equation} \label{eq_5:1}
\Pr(\Omega(\tilde{c_1}, \tilde{c_2}, \cdots, \tilde{c_j}) = \Omega(c_1, c_2, \cdots, c_i)) = \Pr(q_{v_i} = q_{a_j}) \; .
\end{equation}

Now, $\Pr(q_{v_i} = q_{a_j}) = \Pr(q_{v_i} = q_{a_j}| v_i = a_j)\Pr(v_i = a_j) + \Pr(q_{v_i} = q_{a_j}| v_i \neq a_j)\Pr(v_i \neq a_j)$ where $\Pr(q_{v_i} = q_{a_j}| v_i = a_j) = 1$ by definition of the graph-based protocol. On the other hand, $\Pr(q_{v_i} = q_{a_j}| v_i \neq a_j) = \frac{1}{2}$ because the node values are selected at random in the protocol. Then

\begin{equation} \label{eq_5:2}
\Pr(q_{v_i} = q_{a_j}) = \frac{1}{2}+\frac{\Pr(v_i = a_j)}{2} \; .
\end{equation}

As $0 \leq v_i, a_j \leq 2n-1$ then

\begin{equation} \label{eq:tp}
\Pr(v_i = a_j) = \sum_{k = 0}^{k = 2n-1}\Pr(v_i = k)\Pr(a_j = k) \; .
\end{equation}

As $c_t \neq \tilde{c_t}$ for the first time, then two equally probable cases occur: 1) $\Omega(c_1, \cdots, c_t) = q_x$ and $\Omega(\tilde{c_1}, \cdots, \tilde{c_t}) = q_{x+1}$, 2) $\Omega(c_1, \cdots, c_t) = q_{x+1}$ and $\Omega(\tilde{c_1}, \cdots, \tilde{c_t}) = q_x$, where ($0 \leq x \leq 2n-1$) and $\forall x$, $x+1 = (x+1) \mod 2n$. Using these two events in Equation \ref{eq:tp} we obtain:

\small{
    \begin{equation*}
    \Pr(v_i = a_j) =  \frac{1}{2}\left(\sum_{k = 0}^{k = 2n-1}\Pr(v_i = k|\Omega(c_1, \cdots, c_t) = q_x)\Pr(a_j = k|\Omega(c_1, \cdots, c_t) = q_{x+1})\right.\end{equation*}
    \begin{equation}\left.+\sum_{k = 0}^{k = 2n-1}\Pr(v_i = k|\Omega(c_1, \cdots, c_t) = q_{x+1})\Pr(a_j = k|\Omega(c_1, \cdots, c_t) = q_{x})\right) \label{eq:equal}.
    \end{equation}
}

\normalsize

As $A^y[x, k]$ represents the number of walks of size $y$ between nodes $x$ and $k$, then $\Pr(v_i = k | \Omega(c_1, \cdots, c_t) = q_x) = \frac{A^{i-t}[x, k]}{2^{i-t}}$ and $\Pr(v_i = k | \Omega(c_1, \cdots, c_t) = q_{x+1}) = \frac{A^{i-t}[x+1, k]}{2^{i-t}}$; in the same way $\Pr(a_j = k| \Omega(c_1, \cdots, c_t) = q_{x}) = \frac{A^{j-t}[x, k]}{2^{j-t}}$ and $\Pr(a_j = k| \Omega(c_1, \cdots, c_t) = q_{x+1}) = \frac{A^{j-t}[x+1, k]}{2^{j-t}}$. Then using Equation \ref{eq:equal}:

\begin{equation} \label{eq:equal2}
\Pr(v_i = a_j) = \frac{1}{2^{i+j-2t+2}}\sum_{k = 0}^{k = 2n-1} \left (A^{i-t}[x, k]A^{j-t}[x+1, k] + A^{i-t}[x+1, k]A^{j-t}[x, k] \right) \; .
\end{equation}

Given the graph characteristics, we have  $A^{y}[x,k] = A^y[(x-z) \mod 2n, (k-z) \mod 2n]$ for any $z \in \mathbb{N}$. Therefore, $A^{i-t}[x, k] = A^{i-t}[1, (k-x+1) \mod 2n]$ and $A^{i-t}[x+1, k] = A^{i-t}[2, (k-x+1) \mod 2n]$, in the same way, $A^{j-t}[x, k] = A^{j-t}[1, (k-x+1) \mod 2n]$ and $A^{j-t}[x+1, k] = A^{j-t}[2, (k-x+1) \mod 2n]$. So:

\begin{equation*}
\sum_{k = 0}^{2n-1} \left (A^{i-t}[x, k]A^{j-t}[x+1, k] + A^{i-t}[x+1, k]A^{j-1}[x, k] \right) =
\end{equation*}
\begin{equation}
\sum_{k = 0}^{2n-1} \left (A^{i-t}[1, k]A^{j-t}[2, k] + A^{i-t}[2, k]A^{j-t}[1, k] \right) \label{eq:equal3} .
\end{equation}

Equations \ref{eq_5:1}, \ref{eq_5:2}, \ref{eq:equal2}, and \ref{eq:equal3} yield the expected result.

\end{case}

\end{proof}

\begin{remark}
Using Theorem \ref{theo:first}, assuming $c_1 \neq \tilde{c_1}$, for  $i = 1$ we obtain that $\Pr(\Omega(\tilde{c_1}, \tilde{c_2}) = \Omega(c_1)) = \frac{5}{8} > \Pr(\Omega(\tilde{c_1}, \tilde{c_2}, \cdots, \tilde{c_j}) = \Omega(c_1))$ for every $j \neq 2$. This means that in this case it is better for the adversary to send the second response of the prover ($\Omega(\tilde{c_1}, \tilde{c_2}$)). These results only reinforce the ideas shown in Remark \ref{rem:adv}, that the best adversary strategy is not always to pick $j = i$ in the graph-based protocol.
\end{remark}

\begin{corollary} \label{cor:str}
Given $r_i  = \Omega(\tilde{c_1}, \tilde{c_2}, \cdots, \tilde{c_i})$ and $c_i' = \Omega(c_1, c_2, \cdots, c_i)$ for every $1 \leq i \leq n$, the best adversary success probability in the mafia fraud is:

$$
\sum_{t = 1}^{t = n} \frac{1}{2^t}\left( \prod_{i = t}^{i = n} \max(\Pr(r_1 = c_i' | F = t), \cdots, \Pr(r_{n} = c_i'| F = t)) \right) + \frac{1}{2^n}
$$
where $\Pr(r_j = c_i'|F = t)$ is defined in Theorem \ref{theo:first}.

\end{corollary}

\begin{proof}

The adversary success probability in the mafia fraud is:

\begin{equation} \label{eq:total}
\sum_{t = 1}^{t = n} \left(\Pr(\text{success} | F = t)\Pr(F = t)\right) + \Pr(c_1 = \tilde{c_1}, c_2 = \tilde{c_2}, \cdots, c_n = \tilde{c_n}) \; .
\end{equation}

As the challenges are selected at random, then:

\begin{equation} \label{eq:cha}
\begin{array}{l}
\Pr(F = t) = \frac{1}{2^t}  \; .\\
\Pr(c_1 = \tilde{c_1}, c_2 = \tilde{c_2}, \cdots, c_n = \tilde{c_n}) = \frac{1}{2^n}  \; .\\
\end{array}
\end{equation}

Considering the pre-ask attack strategy in Definition \ref{def:adv_aby}:

\begin{equation} \label{eq:tj}
\Pr(\text{success} | F = t) = \prod_{i = t}^{i = n} \max(\Pr(r_1 = c_i'| F = t), \cdots, \Pr(r_{n} = c_i'| F = t)) \; .
\end{equation}

Equations \ref{eq:total}, \ref{eq:cha}, and \ref{eq:tj} yield the expected result.

\end{proof}

\subsection{Distance fraud}

The  distance fraud analysis for most of the distance-bounding protocols is not a hard task. However, for the ATP~\cite{AvoineT-2009-isc} protocol, to the best of our knowledge, nobody has computed the distance fraud success probability. Unfortunately, in the graph-based protocol which has some similarities with the ATP protocol, distance fraud analysis is also not trivial. Then, in this chapter we provide an upper bound on the distance fraud for a sub-family of the distance-bounding protocols, which will be useful for the ATP protocol and for the graph-based protocol.

\begin{definition}[Distance-bounding protocol sub-family] \label{def:sub}
Let us consider $\mathcal{P}$ a distance bounding protocol. $\mathcal{P}$ belongs to the distance-bounding protocol sub-family if it fulfills the following requirements: \begin{itemize}
    \item During the fast phase, in each round the verifier sends a bit as challenge and the prover answers with a bit alike.
    \item There is no final phase.
	\item After the slow phase, it should be possible to build a function $f: \{0,1\}^n \rightarrow \{0,1\}^n$ such that, given any sequence of challenges $\{c_1, c_2, \cdots, c_n\}$, then \\$f(c_1, c_2, \cdots, c_n)$ is the correct response sequence for the verifier. From now on, we are going to call this function ``prover function''.
\end{itemize}

\end{definition}

\begin{definition}[Prover function pre-image] \label{def:pfi}
For a sequence $y \in \{0,1\}^n$ and a prover function $f$, the prover function pre-image is the set $I_{y} = \{ x \in \{0,1\}^n |  f(x) = y \}$.
\end{definition}


\begin{definition}[Adversary capability in the distance fraud attack]
The adversary capability in the distance fraud is twofold:
\begin{enumerate}
	\item The adversary has access to the prover function.
	\item The adversary can send in advance a sequence $y \in \{0,1\}^n$ to the verifier, trying to maximise $\Pr(f(c_1, c_2, \cdots, c_n) = y)$ where $\{c_1, c_2, \cdots, c_n\}$ is a random sequence of challenges.
\end{enumerate}
\end{definition}

\begin{proposition} \label{cor:alg}
Let $y$  be the sequence sent by the adversary in advance,  then the success probability in the distance fraud is $\frac{|I_y|}{2^n}$.
\end{proposition}

Undoubtedly, the best adversary strategy is to find and send a sequence $y \in \{0,1\}^n$ such that for any sequence $x \in \{0,1\}^n$ it holds that $|I_y| \geq |I_x|$.

\begin{theorem} \label{theo:ineq}
Given $x, y \in \{0,1\}^n$ two random sequences, and a prover function $f$, then, for any sequence $z \in \{0,1\}^n$ such that $I_z \neq \emptyset$ we have:

$$\Pr(x \in I_z) \leq \frac{\frac{1}{2^n}+\sqrt{\frac{1}{2^{2n}}-\frac{4}{2^n}+4\Pr(f(x) = f(y))}}{2}$$
\end{theorem}

\begin{proof}
Given that $I_z \neq \emptyset$, we have:

\begin{eqnarray} \label{eq:des_eq1}
\Pr(f(x) = f(y)) &=& \Pr(f(x) = f(y)|y \in I_z)\Pr(y \in I_z) \nonumber\\
&&{}+  \Pr(f(x) = f(y)|y \notin I_z)\Pr(y \notin I_z)
\end{eqnarray}

But, $\Pr(f(x) = f(y)|y \in I_z) = \Pr(x \in I_z) = \Pr(y \in I_z)$ because $x$ and $y$ are random sequences. On the other hand, $\Pr(f(x) = f(y)|y \notin I_z) \geq \frac{1}{2^n}$ because of the ``prover function'' definition. Therefore, using these results in Equation \ref{eq:des_eq1}, we obtain:

\begin{equation} \label{eq:des_eq2}
\Pr(f(x) = f(y)) \geq \Pr(x \in I_z)^2 +  \frac{1}{2^n}(1 - \Pr(x \in I_z))  \; .
\end{equation}

By calculating the discriminant of this quadratic inequality, and obtaining its solutions, we conclude the proof. Note that, this quadratic inequality has real solutions because $\Pr(f(x) = f(y)) \geq \frac{1}{2^n}$, and in this case, the discriminant value is always positive.
\end{proof}

\begin{corollary} \label{cor:ineq}
For every distance-bounding protocol that complies with Definition \ref{def:sub}, the adversary success probability in the distance fraud is upper-bounded by:

$$\frac{\frac{1}{2^n}+\sqrt{\frac{1}{2^{2n}}-\frac{4}{2^n}+4\Pr(f(x) = f(y))}}{2}  \; .$$
\end{corollary}

With this result, we are giving a way to compute an upper bound of a sub-family of the distance-bounding protocols. We show below how to apply this result to the graph-based protocol, and later we apply the same result to the ATP protocol.

\begin{theorem} \label{theo:df}
The distance fraud success probability for the graph-based protocol is upper bounded by:

$$\frac{\frac{1}{2^n}+\sqrt{\frac{1}{2^{2n}}-\frac{4}{2^n}+4p\ }}{2}$$
where $$p = \prod_{i = 1}^{i = n} \left( \frac{1}{2}+\frac{1}{2^{2i+1}}\sum_{k = 0}^{k = 2n-1}(A^i[0, k])^2 \right)  \; .$$

\end{theorem}

\begin{proof}
Let us consider two random sequences $x = \{x_1, x_2, \cdots, x_n\}$ and $y = \{y_1, y_2, \cdots, y_n\}$, then by the definition of the graph-based protocol and the definition of ``Prover Function'':

\begin{equation} \label{eq:distance1}
\Pr(f(x) = f(y)) = \prod_{i = 1}^{i = n}\Pr(\Omega(x_1, \cdots, x_i) = \Omega(y_1, \cdots, y_i)) \; .
\end{equation}

Let be $q_{x_i} = \Omega(x_1, \cdots, x_i)$ and $q_{y_i} = \Omega(y_1, \cdots, y_i)$, then, like in Theorem\ref{theo:first}, we can obtain that

\begin{equation} \label{eq:distance2}
\Pr(q_{x_i} = q_{y_i}) = \frac{1}{2}+\frac{\Pr(x_i = y_i)}{2} 
\end{equation}

and

\begin{equation} \label{eq:distance3}
\Pr(x_i = y_i) = \sum_{k = 0}^{k = 2n-1} \Pr(x_i = k)\Pr(y_i = k) \; .
\end{equation}

Once again, as $A^i[j,k]$ represents the number of walks of size $i$ between the nodes $j$ and $k$, where $A$ is the adjacency matrix of the graph, then $\Pr(x_i = k) = \frac{A^i[0, k]}{2^i} = \Pr(y_i = k)$. Therefore, using Equation \ref{eq:distance3}:

\begin{equation} \label{eq:distance4}
\Pr(x_i = y_i) = \sum_{k = 0}^{k = 2n-1} \left( \frac{A^i[0, k]}{2^i} \right)^2 \; .
\end{equation}

Equations \ref{eq:distance1}, \ref{eq:distance2} and \ref{eq:distance4}
yield

\begin{equation} \label{eq:distance5}
\Pr(f(x) = f(y)) = \prod_{i = 1}^{i = n} \left( \frac{1}{2}+\frac{1}{2^{2i+1}}\sum_{k = 0}^{k = 2n-1}(A^i[0, k])^2 \right) \; .
\end{equation}

By applying Equation \ref{eq:distance5} to Corollary \ref{cor:ineq}, considering that $p = \Pr(f(x) = f(y))$, we conclude the proof of this theorem.
\end{proof}

\section{Experimental results and evaluation} \label{sec:comparison}

We analyse mafia fraud resistance, distance fraud 
resistance and memory consumption. Therefore, we need 
to measure the above features 
for each of the previous protocols. We have detected that the mafia fraud
success probability for the KAP protocol provided in~\cite{KimA-2009-cans} is not correct. Also, as we previously said, the distance fraud success
probability of ATP was not presented in~\cite{AvoineT-2009-isc}. Therefore, we first provide both a correct calculation of the mafia fraud success probability of the KAP protocol and an upper bound 
for the distance fraud success probability of the ATP protocol.

\subsection{Mafia fraud success probability for KAP}\label{ap:kim}
In the Kim and Avoine protocol, the adversary success probability in the mafia fraud depends on the predefined challenges probability ($p_d$). Define the following events:

\begin{itemize}
	\item $L_i$ is the event ``the adversary wins the $i$-th round'';
	\item $D_i$ is the event ``the adversary is detected in 
	the $i$-th round by the tag for the first time'';
	\item $N_i$ is the event ``the adversary is detected by the tag in 
	the $i$-th round'';
	\item $N$ is the event ``the adversary is never detected''.
\end{itemize}
The notation $\bar{A}$ denotes the complement of event $A$.

By the law of total probability:

\begin{equation} \label{eq:mafiaKimMain}
P(\text{success}) = \sum_{i = 1}^{i=n}\Pr(\text{success}|D_i)\Pr(D_i) + \Pr(\text{success}|N)\Pr(N) \; .
\end{equation}

Since $\Pr(N_i) = \frac{p_d}{2}$,

 \begin{equation} \label{eq:yoquese}
 \Pr(N) = (1-\frac{p_d}{2})^n \; .
 \end{equation}

The probability of being detected in the $i$-th round for the first time is:

\begin{equation} \label{eq:di}
 \Pr(D_i) = \prod_{j = 1}^{j = i-1}\Pr(\bar{N_j})\Pr(N_i) = \left(\frac{2-p_d}{2} \right)^{i-1}\left (\frac{p_d}{2} \right) \; .
 \end{equation}

On the other hand

\begin{equation} \label{eq:sucess}
\Pr(\text{success}|D_i) = \prod_{j = 1}^{j = i-1}\Pr(L_j|\bar{N_j})\prod_{j = i}^{j = n}\Pr(L_j|N_j)
\end{equation}

where $\Pr(L_j|N_j) = \frac{1}{2}$, and

\begin{equation} \label{eq:conj}
\Pr(L_j|\bar{N_j}) = \frac{\Pr(L_j \cap \bar{N_j})}{\Pr(\bar{N_j})} 
\end{equation}

where $\Pr(L_j \cap \bar{N_j}) = \Pr(L_j \cap \bar{N_j}|p_d)p_d + \Pr(L_j \cap \bar{N_j}|p_r)p_r$. But $\Pr(L_j \cap \bar{N_j}|p_d) = \frac{1}{2}$, because the adversary must send the correct challenges $c_j$ in this round. And $\Pr(L_j \cap \bar{N_j}|p_r) = \frac{3}{4}$, because this is the same case as in the Hancke and Kuhn protocol. Therefore, $\Pr(L_j \cap \bar{N_j}) = \frac{1}{2}p_d+\frac{3}{4}p_r = \frac{3-p_d}{4}$. Using this result in Equation \ref{eq:conj} we obtain:

\begin{equation} \label{eq:wndi1}
\Pr(L_j|\bar{N_j}) = \frac{3-p_d}{4-2p_d} \; .
\end{equation}

Using Equations \ref{eq:sucess} and \ref{eq:wndi1} we obtain:

\begin{equation} \label{eq:wndi}
\Pr(\text{success}|D_i) = \left(\frac{3-p_d}{4-2p_d} \right)^{i-1}\left(\frac{1}{2}\right)^{n-i+1} \; ,
\end{equation}
and
\begin{equation} \label{eq:wnd2}
\Pr(\text{success}|N) = \left(\frac{3-p_d}{4-2p_d} \right)^{n} \; .
\end{equation}

Using Equations \ref{eq:mafiaKimMain}, \ref{eq:yoquese}, \ref{eq:di}, \ref{eq:wndi} and \ref{eq:wnd2}, we obtain the adversary success probability for the mafia fraud in the Kim and Avoine protocol:

\begin{equation} \label{eq:kim}
P(\text{success}) = \frac{p_d}{2}\sum_{i = 1}^{i=n}\left(\frac{3-p_d}{4}\right)^{i-1}\left(\frac{1}{2}\right)^{n-i+1} + \left(\frac{3-p_d}{4}\right)^{n}.
\end{equation}

\subsection{Distance fraud success probability for ATP} \label{ap:avp}

To find an upper bound for the adversary success probability in the distance fraud for the ATP protocol, we use the result of Theorem \ref{theo:df}. Indeed, this protocol behaves like the graph-based protocol. The only difference between them is that the ATP protocol creates a full tree as a graph. Therefore, in the ATP protocol the distance fraud success probability  is upper bounded by:

$$\frac{\frac{1}{2^n}+\sqrt{\frac{1}{2^{2n}}-\frac{4}{2^n}+4p\ }}{2},$$
where $$p = \prod_{i = 1}^{i = n} \left( \frac{1}{2}+\frac{1}{2^{2i+1}}\sum_{k = 0}^{k = 2n-1}(A^i[0, k])^2 \right).$$

To give a complete equation, we define $A^i[0,k]$ for a tree. For this purpose, we  consider that the nodes in the tree are labeled between $0$ and $2^n-1$ using a breadth-first algorithm. Then:

$$
A^i[0,k] = \left\{\begin{array}{l l}
1 &  \quad \text{if} \quad \mbox{$2^i-1 \leq k < 2^{i+1}-1$},\\
\\
0 & \quad \text{otherwise.}
\end{array}\right.
$$

Finally we obtain:

$$p = \prod_{i = 1}^{i = n} \left( \frac{1}{2}+\frac{1}{2^{i+1}}\right).$$

\subsection{Comparison}

Since memory is a scarce resource in RFID tags and thus it is one of the main concerns in distance-bounding protocols, we relax the ATP protocol to operate with linear memory. As noted in~\cite{AvoineT-2009-isc}, reducing memory in the ATP protocol increases the adversary success probability for both types of fraud. Hence, we pick $\alpha=\frac{n}{3}$, in which case the memory consumption equals $\frac{14n}{3} \approx 5n$, while a sufficient security is still ensured. 
Note that this memory consumption is in the range of the other studied protocol. This instance of the ATP protocol is named ``ATP3''.

\begin{table}
\centering
\begin{tabular}{|c | c | c | c |}
\hline
& Memory & Mafia Fraud & Distance Fraud \\
\hline
HKP & $2n$ \cite{HanckeK-2005-securecomm}& $\left(\frac{3}{4} \right)^n$ \cite{HanckeK-2005-securecomm}& $\left(\frac{3}{4} \right)^n$ \footnotemark[2]\\
\hline
KAP & $4n$ \cite{KimA-2009-cans}& Section \ref{ap:kim} & $\left(\frac{3}{4}+\frac{p_d}{4} \right)^n$ \cite{KimA-2009-cans}\\
\hline
ATP & $2^{n+1}-2$ \cite{AvoineT-2009-isc}& $\left(\frac{1}{2} \right)^n(\frac{n}{2}+1)$ \cite{AvoineT-2009-isc}& Section \ref{ap:avp} \\
\hline
ATP3 & $\frac{14n}{3}$ \cite{AvoineT-2009-isc}& $\left(\frac{1}{2} \right)^n\left(\frac{5}{2} \right)^{\frac{n}{3}}$ \cite{AvoineT-2009-isc} & $(0.3999)^{\frac{n}{3}}$ \footnotemark[3]\\
\hline
GRAPH & $4n$ & Corollary \ref{cor:str} & Theorem \ref{theo:df} \\
\hline
\end{tabular}
\caption{Memory consumption, mafia fraud success probability  and distance fraud success probability for the HKP protocol, the KAP protocol, the ATP protocols (ATP and ATP3), and the graph-based protocol (GRAPH).} \label{tb:values}
\end{table}

Table \ref{tb:values} depicts the values of the three parameters for each protocol that we are considering. In terms of memory, the Hancke and Kuhn protocol is, undoubtedly, the best protocol. As can be seen in Figure~\ref{fig:mafia}, when  considering only mafia fraud resistance, the KAP and the ATP protocols are the best ones.  Only in terms of distance fraud, the lowest adversary success probability is reached by the ATP protocol (cf. Figure~\ref{fig:dist}).

\footnotetext[2]{The distance fraud probability for the HKP protocol is computed using the distance fraud probability for the KAP protocol. Note that the KAP protocol with $p_d = 0$ and the HKP protocol are the same.}

\footnotetext[3]{The distance fraud probability for the ATP3 protocol is an accurate value, not an upper bound like in ATP or GRAPH. It was computed by brute force, \emph{i.e.} for a given instance, we computed the adversary success probability.  Then, considering all possible instances we deduce the probability in the average case.}

\begin{figure}[!ht]
  \begin{center}
       \includegraphics[width=0.7\textwidth, angle=270]{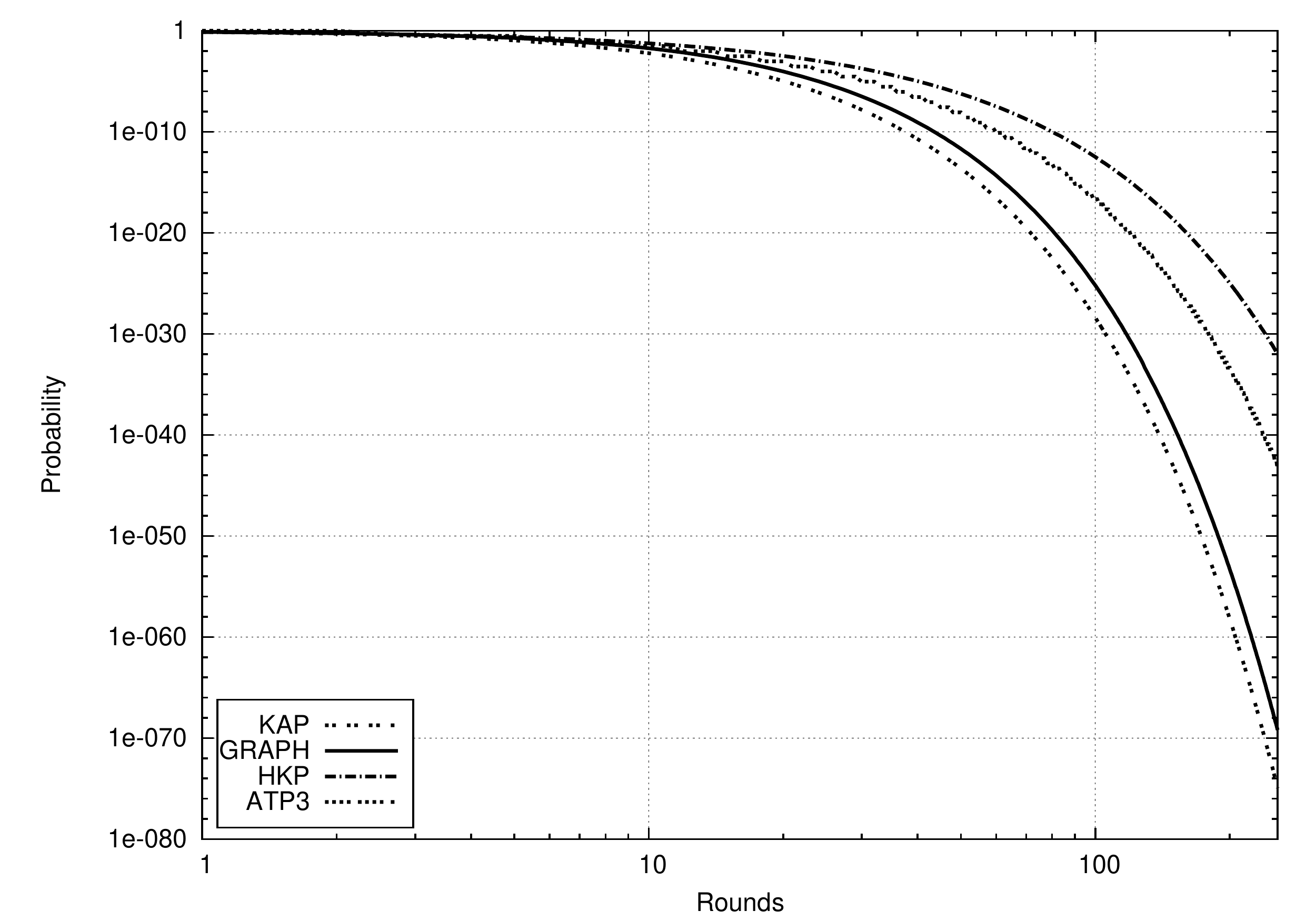}
  \end{center}
  \caption{Adversary success probability in the mafia fraud against the GRAPH protocol, the HKP protocol and the ATP3 protocol. The ATP protocol in its standard configuration is not represented in this chart because it has the same mafia fraud probability as the KAP protocol.
  \label{fig:mafia}}
\end{figure}

\begin{figure}[!ht]
 \begin{center}
       \includegraphics[width=0.7\textwidth, angle=270]{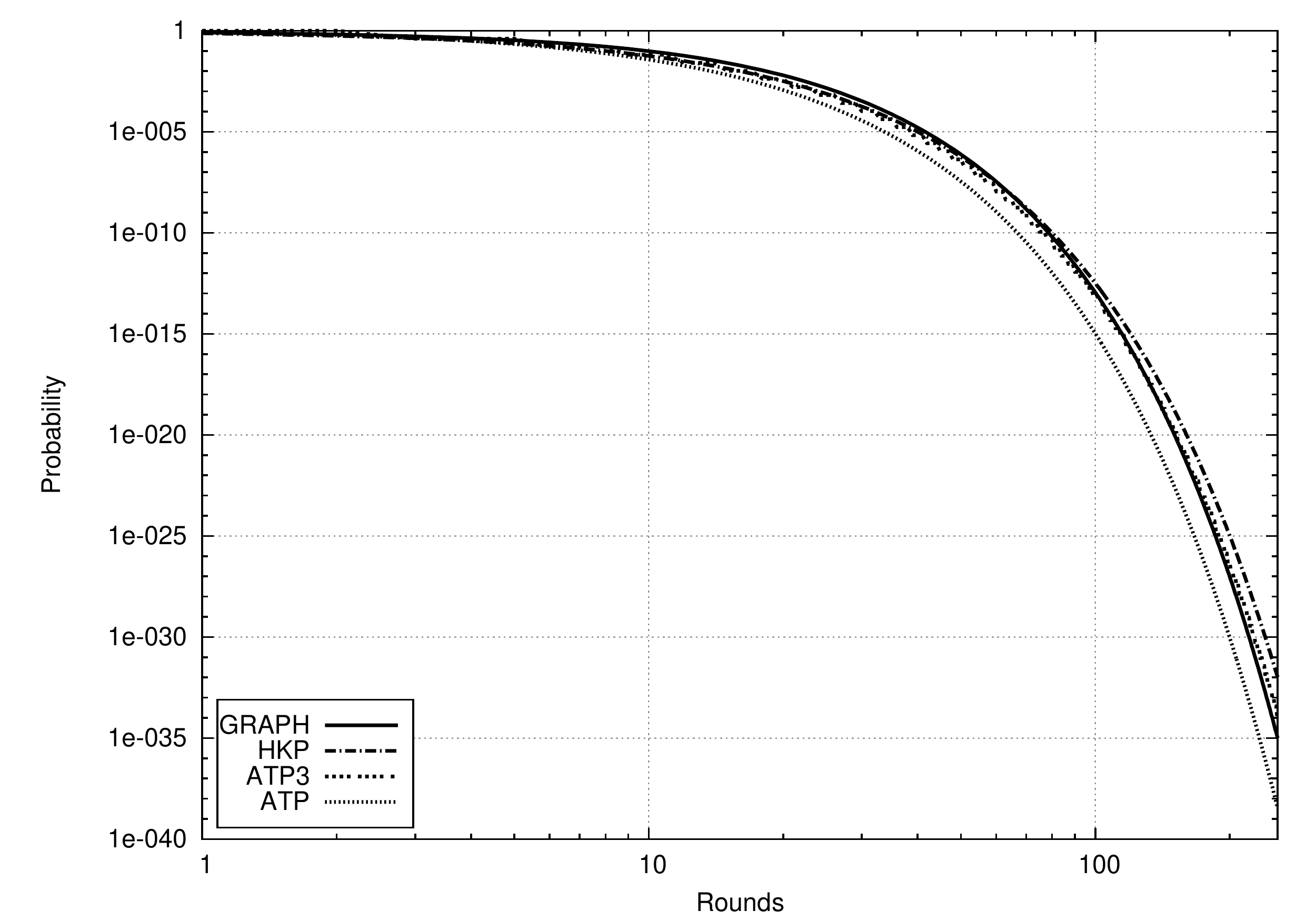}
  \end{center}
  \caption{Adversary success probability in the distance fraud against the GRAPH protocol, the HKP protocol, and the ATP protocols (ATP and ATP3). The KAP protocol is not represented in this chart because it has the same distance fraud probability as the HKP protocol in the best case.
  \label{fig:dist}}
\end{figure}

\begin{figure}[!ht]
  \begin{center}
       \includegraphics[width=0.7\textwidth, angle=270]{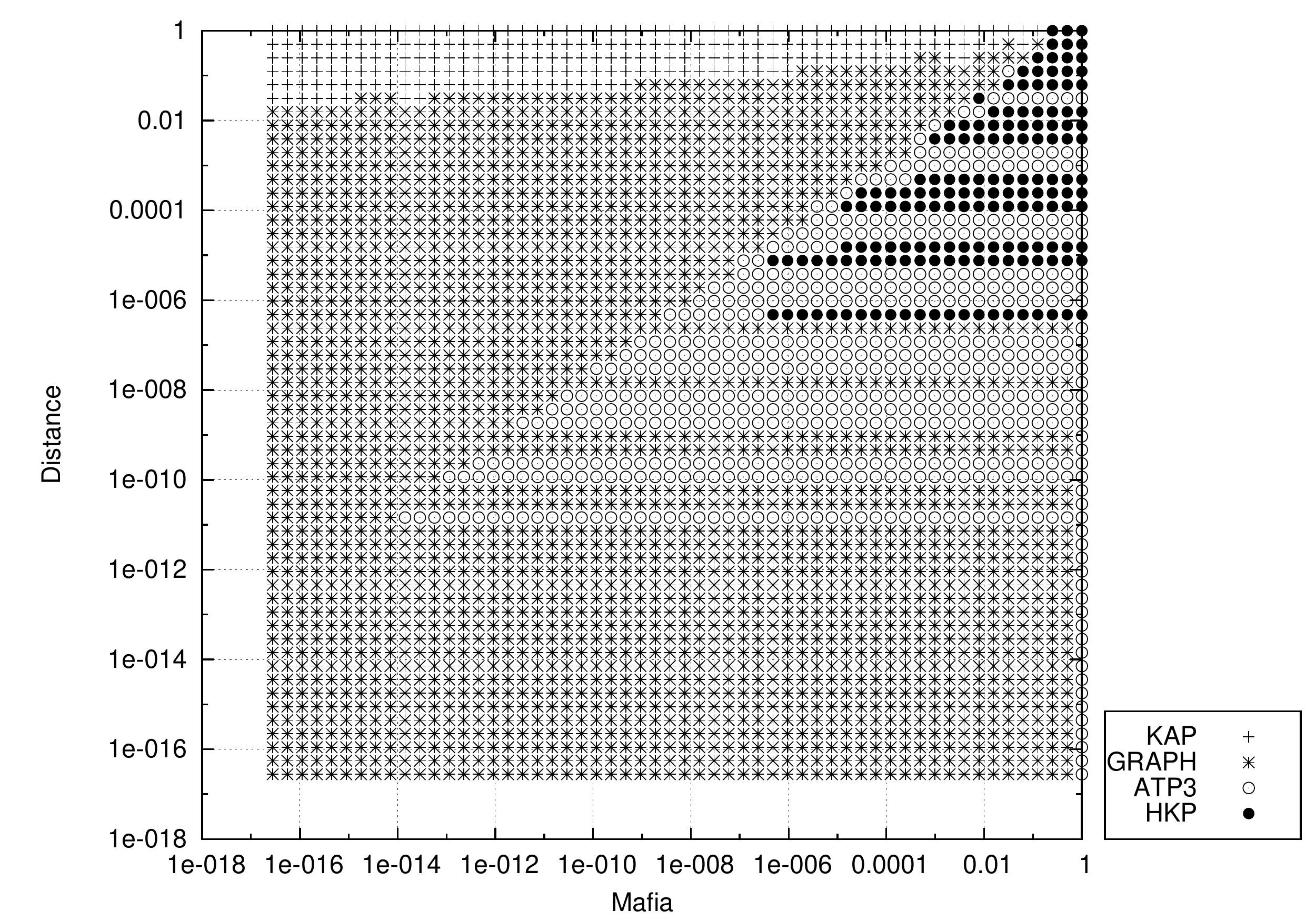}
  \end{center}
  \caption{Best protocols in terms of number of rounds given different values of mafia fraud probability and distance fraud probability. The considered protocols are: the graph-based protocol (GRAPH), the Hancke and Kuhn protocol (HKP), the Kim and Avoine protocol (KAP), and the Avoine and Tchamkerten protocol (ATP3). The ATP protocol in its standard configuration is not considered in this chart because we are comparing only protocols with linear memory consumption.}
  \label{fig:trade_off}
\end{figure}

However, our aim is to find the best protocol given a security level in terms of mafia fraud and distance fraud. To that end, Figure \ref{fig:trade_off} represents, for each pair of mafia and distance fraud success probabilities, the protocol needing a lowest number of rounds to reach these probabilities. As it can be seen in Figure \ref{fig:trade_off}, the graph-based protocol is, in general, the best option when considering memory consumption, distance, and mafia fraud at the same time. In particular, if one requires low success probabilities for both mafia and distance fraud, we stress the particularly good behaviour of the graph-based protocol. It should be remarked that in some cases more than one protocol is optimal in terms of number of rounds; in this case, the best one in terms of memory is chosen.

\section{Conclusions} \label{sec:conclusion}

In this chapter, we contribute to balancing 
mafia fraud resistance, distance fraud resistance and memory consumption
for distance-bounding protocols. In particular, we provide a way to compute an upper bound on the distance-fraud probability, which is useful for analysing previous protocols and designing future ones. In addition, we propose a new distance-bounding protocol, and we show that the achieved security level is better than all previously published distance-bounding protocols when considering 
mafia fraud, distance fraud and memory at the same time.

We do not only provide a simple, fast, and flexible protocol, but we also introduce the graph-based protocol concept and several new open questions. An interesting question is to know if there are graph-based protocols that behave still better than the one presented here. In particular, if the number of rounds is not a critical parameter, prover and verifier may be allowed to increase the number of rounds while keeping a $2n$-node graph. This means that some nodes may be used twice. In such a case, the security analysis provided in this chapter must be refined. On the other hand, although a bound on the distance fraud success probability is provided, calculating the exact probability of success is still cumbersome.

\chapter{Microaggregation- and Permutation-Based Anonymisation of Movement Data}
\label{chap:7}

\emph{This chapter describes a novel distance measure between trajectories not
necessarily defined over the same time span.
By using it, two permutation-based trajectory anonymisation algorithms are proposed. Both algorithms preserve the true original locations of trajectories and provide better utility properties than previous algorithms.}

\minitoc

Various technologies such as GPS, RFID, GSM, etc., can sense and track
the whereabouts of objects (cars, parcels, people, etc.).
In addition, the current storage capacities allow
collecting such object movement data in huge spatio-temporal
databases. Analysing this kind of databases containing
the trajectories of objects
can lead to useful and previously unknown knowledge. Therefore, it is
beneficial to share and publish such databases and let the analysts derive
useful knowledge from them ---knowledge that can be applied, for example, to
intelligent transportation, traffic monitoring, urban and road planning,
supply chain management, sightseeing improvement, etc.

However, the privacy of individuals
may be affected by the publication or the outsourcing
of databases of trajectories. Several kinds of privacy threats
exist. Simple de-identification
realised by removing identifying attributes is insufficient to protect the
privacy of individuals. The biggest threat with trajectories is the
``sensitive location disclosure''. In this scenario, knowing the times at
which an individual visited a few locations can help an adversary to
identify the individual's trajectory in the published database, and
therefore learn the individual's other locations at other times.
Privacy preservation in this context means that no sensitive location
ought to be linkable to an individual.

The risk of sensitive location disclosure is also affected by how much the
adversary knows. The adversary may have access to auxiliary
information~\cite{kaplan10}, sometimes called side knowledge,
background knowledge or external knowledge. The adversary can link such
background knowledge obtained from other sources
to information in the published database.
Estimating the amount and extent of auxiliary information available to the
adversary is a challenging task.

There are quite a few differences between spatio-temporal data and
microdata, {\em i.e.} records describing individuals in a standard
database with no movement data.
One real difference becomes apparent when considering
privacy. Unfortunately, the traditional anonymisation and sanitisation
methods for microdata~\cite{fung10} cannot be directly applied to
spatio-temporal data without considerable expense in computation time and
information loss. Hence, there is a need for specific anonymisation methods
to thwart privacy attacks and therefore reduce privacy risks associated with
publishing trajectories.

Trajectories can be modeled and represented in many ways~\cite{forlizzi00}.
Without loss of generality, we
consider a trajectory to be a timestamped path in a plane. By assuming
movements on the surface of the Earth, the altitude of each location
visited by a trajectory stays implicit; it could be explicitly restored
if the need arose. More formally,
let \emph{timestamped location} be a triple $(t,x,y)$ with $t$ being a
timestamp and
$(x,y)$ a \emph{location} in $\mathbb{R}^2$. Intuitively, the
timestamped location denotes
that at time $t$ an object is at location $(x,y)$.

\begin{definition}[Trajectory] \label{def:traj}
A \emph{trajectory} is an ordered set of timestamped locations
\begin{equation} \label{eq:traj}
T = \{ (t_1,x_1,y_1), \ldots, (t_n,x_n,y_n) \} \enspace,
\end{equation}
where $t_i < t_{i+1}$ for all $1 \leq i < n$.
\end{definition}

\begin{definition}[Sub-trajectory]\label{def:subtraj}
A trajectory $S = \{ (t'_1,x'_1,y'_1), \ldots, (t'_m,x'_m,y'_m) \}$ is a
\emph{sub-trajectory} of $T$ in Expression~\ref{eq:traj}, denoted
$S \preceq T$, if there exist
integers $1 \leq i_1 < \ldots < i_m \leq n$ such that $(t'_j,x'_j,y'_j) =
(t_{i_j},x_{i_j},y_{i_j})$ for all $1 \le j \le m$.
\end{definition}

Hereinafter, we will use {\em triple} as a synonym for
timestamped location.
When there is no risk of ambiguity,
we also say just ``location'' to denote a timestamped location.

We present two heuristic methods for preserving the
privacy of individuals when releasing
trajectories. Both of them exactly preserve original locations
in the sense that the anonymised trajectories contain no fake, perturbed
or generalised trajectories.
The first heuristic is based on microaggregation~\cite{domingo02} of
trajectories and permutation of locations.
Microaggregation has been successfully used in microdata
anonymisation to achieve $k$-anonymity~\cite{samarati98,sweeney02a,domingo05}.
We use it here for trajectory $k$-anonymity (whereby an adversary
cannot decide which of $k$ anonymised trajectories corresponds
to an original trajectory which she partly knows),
first by grouping the trajectories into clusters of size
at least $k$ based on their similarity and then transforming
via location permutation the
trajectories inside each cluster to preserve privacy.
The second heuristic aims no longer at trajectory $k$-anonymity,
but at location $k$-diversity (whereby knowing a sub-trajectory
$S$ of a certain original trajectory $T$ allows an adversary to discover
a location in $T \setminus S$ with probability no greater than $1/k$);
this second heuristic is based on location permutation
and its strong point is that
it takes reachability constraints into account:
movement between locations must follow the edges of an underlying graph
({\em e.g.}, urban pattern) so that not all locations are reachable
from any given location.
Experimental results show that
achieving trajectory $k$-anonymity with reachability constraints
may not be possible without discarding
a substantial fraction of locations, typically those which are rather
isolated. This is the motivation for our second heuristic: it still
considers reachability but it reduces the number of discarded
locations by replacing $k$-anonymity at the trajectory
level by $k$-diversity at the location level.

For clustering purposes, we propose a new distance for trajectories
which naturally considers both spatial and temporal coordinates.
Our distance is able to compare trajectories
that are not defined over the same time span, without resorting
to time generalisation.
Our distance function can compare trajectories
that are timewise overlapping only
partially or not at all. It may seem at first sight that the distance
computation is exponential in terms of all considered trajectories, but
we show that it is in fact computable in polynomial time.

We present empirical results for the two proposed heuristics
using synthetic data and also real-life data.
We theoretically and experimentally compare our first heuristic
with a recent trajectory anonymisation method
called $(k, \delta)$-anonymity~\cite{abul08} also aimed at
trajectory $k$-anonymity without reachability constraints.
Theoretical results show that the privacy
preservation of our first method is the same as that of $(k, \delta)$-anonymity
but dealing with trajectories \emph{not} having the same time span. For the second heuristic involving reachability constraints, no comparable counterparts
seem to exist in the literature.

%
%
%


\section{Trajectory similarity measures}
\label{sec:relworkdistances}

Using microaggregation
for trajectory $k$-anonymisation requires a distance function to measure
the similarity between trajectories.
Such a distance function must consider both space and time.
Although most spatial distances
can be extended into spatio-temporal distances by adding
a time co-ordinate to spatial points,
it is not obvious how to
balance the weight of spatial and temporal
dimensions.
Furthermore, not all similarity measures for trajectories are
suitable for comparing trajectories for anonymisation purposes.
The requirement for anonymisation is not just similarity regarding shape, but
also spatial and temporal closeness. Some typical distances for
trajectories
include the Euclidean distance, the Hausdorff distance~\cite{shonkwiler91},
the Fr\'echet distance~\cite{alt95}, the turning point distance~\cite{arkin91},
and distances based on time series~\cite{liao05}
---{\em e.g.}, dynamic time warping (DTW), short
time series (STS)--- and on edit distance~\cite{chen05}
---{\em e.g}, edit distance
with real penalty (ERP), longest common sub-sequence (LCSS), and the edit
distance on real sequences (EDR) discussed next.

The \emph{edit distance on real sequences} (EDR)~\cite{chen05}
is the number of insert, delete, or replace operations
that are needed to change one sequence into another.
If $P$ and $Q$ are two sequences
of $m$ and $n$ triples, respectively, where each triple $\lambda$ has three
attributes -- x-position $\lambda.x$, y-position $\lambda.y$ and time
$\lambda.t$ --
the distance $EDR(P,Q)$ is defined as
\[\begin{cases} \max\{m,n\} & \text{if } m=0 \text{ or
} n=0 \\
\min \{ match(p_1,q_1) + EDR(Rest(P),Rest(Q)), & \text{otherwise}
\\ \quad 1 + EDR(Rest(P),Q), 1 + EDR(P,Rest(Q)) \} & \\ \end{cases} \enspace\]
 where $p_1$ and $q_1$ are the first elements of a given sequence,
$Rest(\cdot)$ is a function that returns the input sequence without the
first element, and where $match(p,q) := 0$ if $p$ and $q$ are ``close'',
that is, they satisfy either
$|p.x-q.x| \le \epsilon$ and $|p.y-q.y| \le \epsilon$ for some
parameter $\epsilon$~\cite{chen05} or  $|p.x-q.x| \le \Delta.x$,
$|p.y-q.y| \le \Delta.y$, and $|p.t-q.t| \le \Delta.t$ for a triple of
parameters $\Delta$~\cite{abul10}; otherwise, $match(p,q) := 1$. This
definition of $match$ means that the cost for one insert, delete, or replace
operation in EDR is 1 if $p$ and $q$ are not ``close''.

EDR has been used for anonymisation in~\cite{abul10}. However, the
edit distance and variations thereof
are not suitable to guide clustering
for anonymisation purposes. Indeed,
Figure~\ref{fig:edrdistances} shows trajectories with
different degrees of ``closeness'' to trajectory A,
but whose EDR distance from A is the same in all cases.
When time-stamps are considered, the situation is even worse.

In Section~\ref{sec:distance}, we define a
distance measure which is better suited for anonymisation
clustering: it can compare trajectories
defined over different time spans and
even trajectories that are time-wise non-overlapping.

\begin{figure}[!ht]
\centering


\includegraphics[width=3in]{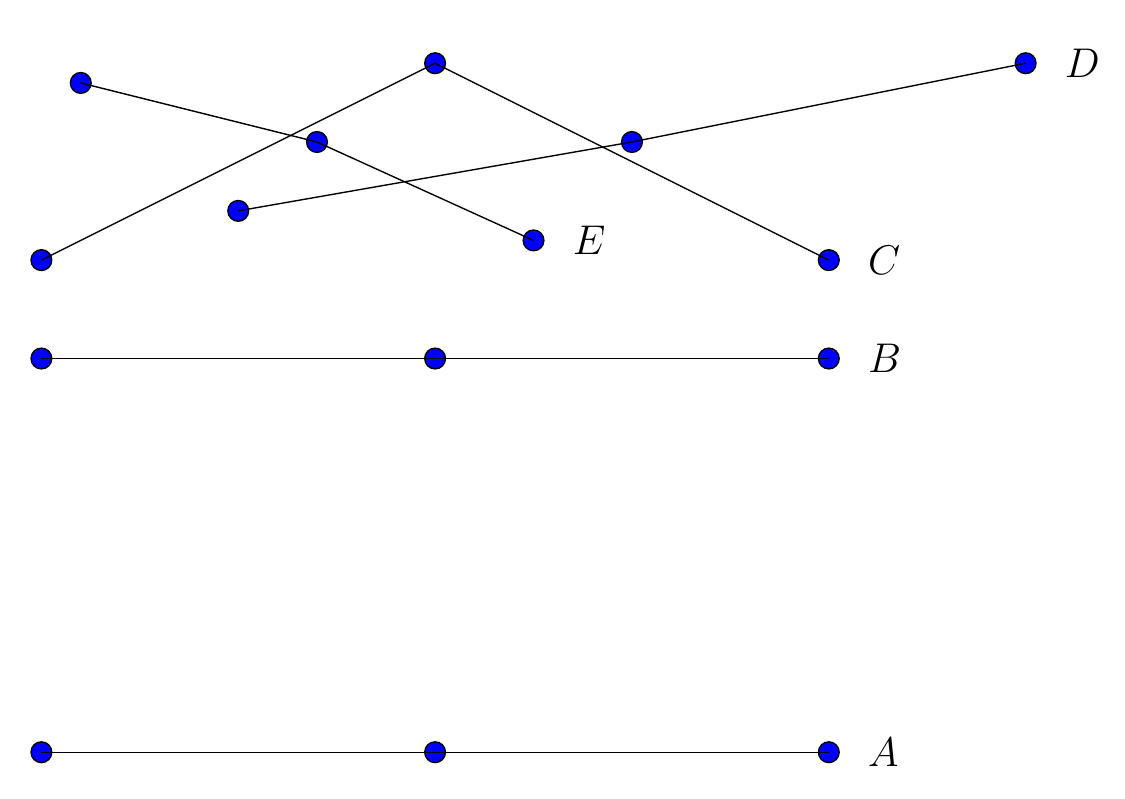}

\caption{Trajectories $B, C, D, E$ are placed at varying
``closeness'' from $A$, yet their EDR distance from $A$
is 3 in all cases. We assume that the first point of $A$
matches the first point of each of $B,C,D,E$; also,
second points are assumed to match each other, and the
same for third points.}
\label{fig:edrdistances}
\end{figure}

\section{Utility and privacy requirements} \label{sec:utilpriv}

Every trajectory anonymisation algorithm must combine utility and privacy.
However, utility and privacy are two largely antagonistic concepts. What is
useful in a set of trajectories is application-dependent, so for each
utility feature probably a different anonymisation algorithm is needed.

\subsection{Desirable utility features}

The utility features that are usually considered in
trajectory anonymisation are: (i)
trajectory length preservation,
(ii) trajectory shape preservation, (iii)
trajectory time preservation, and (iv) minimisation of
the number of discarded locations. We include two additional
utility features that are particularly meaningful in urban scenarios:
\begin{itemize}
\item {\em Location preservation}.
This essentially means that no fake or inaccurate locations
are used to replace original locations; otherwise put,
locations in the anonymised trajectories should
be locations visited by the original trajectories, without
any generalisation or accuracy loss.
Preserving original locations helps answering several queries
that may not be responded by generalisation methods~\cite{monreale10}
or some microaggregation methods~\cite{abul08,abul10}:
(i) what is the ranking of original (non-removed) locations, from most visited
to least visited?; (ii) in which original (non-removed)
locations did two or more mobile
objects meet?, etc.
On the other hand, if trajectory anonymisation rests on replacing true
locations with
fake locations, an adversary can distinguish the latter from the former
and discard fake locations.
Hence, location preservation is desirable for both utility and privacy
reasons.
\item {\em Reachability}. In the second proposed heuristic,
easy reachability between two successive locations in each anonymised
trajectory is enforced. This means that
the distance from the $i$-th location
to the $i+1$-th location on an anonymised location {\em following
the underlying network of streets and/or roads}
should be at most $R^s$, where $R^s$ is a preset
parameter. Like location preservation, this is as good for utility
as it is for privacy: if the adversary sees that reaching the $i+1$-th
location from the $i$-th one takes a long trip across streets and roads,
she will guess that the section between those two locations was not
present in any original trajectory.
\end{itemize}

\subsection{Specific utility measures} \label{subsec:utility}

Basic utility measures are
the number of removed trajectories and the number of
removed locations, whether during pre-processing, clustering or cluster
anonymisation.

The distortion of the trajectory shape is another utility measure,
which can be captured with the
space distortion metric~\cite[Sec.VI.B]{abul08}. This metric
also allows accumulating the total space distortion of all anonymised
trajectories from original ones.

\begin{definition}[Space distortion metric~\cite{abul08}]\label{def:distortion}
The space dis\-tort\-ion of an a\-no\-nym\-ised trajectory
$T^\star$ with respect to its original trajectory $T$ at
time $t$ when $T$ has triple $(t,x,y)$ and $T^\star$ has
possible triple $(t,x^\star,y^\star)$, is
$$SD_t(T,T^\star) = \begin{cases} \Delta((x,y),(x^\star,y^\star)) &
\text{if } (x^\star,y^\star) \text{ is defined at } t \\ \Omega & \text{otherwise}
\end{cases}$$
where $\Delta$ is a distance ({\em e.g.} Euclidean), and $\Omega$ a constant that
penalises for removed locations.
The space distortion of an anonymised trajectory
$T^\star$ from its original $T$ is then
$$SD(T,T^\star) = \sum_{t \in TS} SD_t(T,T^\star) \enspace ,$$
where $TS$ are all the timestamps where $T$ is defined. In particular, if
$T$ is discarded during anonymisation, $T^\star$ is empty, and so
$SD(T,T^\star) = n\Omega$, where $n = |TS|$ is the number of locations of $T$.
In this way, the space distortion of a set of trajectories $\mathcal{T}$
from its anonymised set $\mathcal{T}^\star$ is easily defined as
$$TotalSD(\mathcal{T}, \mathcal{T}^\star) = \sum_{T \in \mathcal{T}}
SD(T,T^\star) \enspace ,$$
where $T^\star \in \mathcal{T}^\star$ (which may be empty) corresponds to $T
\in \mathcal{T}$.
\end{definition}



Another way to measure utility is by comparing the results between queries
performed on both the original data set $\mathcal{T}$ and the
anonymised data set $\mathcal{T}^\star$. Intuitively, when results
on both data sets are similar for a large and diverse number of queries,
the anonymised data set can be regarded as preserving
the utility of the original data set. The challenge of this utility
measure is the selection of queries, which is usually
application-dependent or even user-dependent, {\em i.e.}
two different users are likely
to perform different queries on the same trajectory data set.

In~\cite{Trajcevski:2004:MUM:1016028.1016030} six types of spatio-temporal
range queries were introduced, aimed at evaluating the relative position
of a moving object with respect to a region $R$ in a time
interval $[t_b, t_e]$. We have used these queries
in our experimental work, even though they were designed
for use on uncertain trajectories
(see Definition~\ref{def:motion_curve})
rather than synthetic trajectories.

\begin{definition}[Uncertain trajectory] \label{def:motion_curve}
Given a trajectory $T$ and an uncertainty space threshold $\sigma$,
an \emph{uncertain trajectory} $U(T, \sigma)$ is defined as
the pair $<T, \sigma>$,
where $(t,x,y) \in U(T, \sigma)$ if and only if $\exists x', y'$ such that
$(t, x',y') \in T$ and the Euclidean distance between $(x,y)$ and $(x',y')$
is not greater than $\sigma$.
\end{definition}

\begin{definition}[Possible motion curve]
A \emph{possible motion curve} $PMC^{T}$ of an uncertain trajectory $U(T, \sigma)$ is an ordered set of timestamped locations
\begin{equation}
PMC^{T} = \{ (t_1,x_1,y_1), \ldots, (t_n,x_n,y_n) \} \enspace,
\end{equation}
such that $(t_i,x_i,y_i) \in U(T, \sigma)$ for all $1 \leq i \leq n$.
\end{definition}

In short, a possible motion curve defines one of the possible
trajectories that an object moving along
an uncertain trajectory could follow.
Unlike in~\cite{Trajcevski:2004:MUM:1016028.1016030},
our anonymised trajectories are not uncertain;
hence, we will only
use the two
spatio-temporal range queries proposed in that paper
that can be adapted to non-uncertain trajectories:

\begin{itemize}
\item \emph{Sometime\_Definitely\_Inside($T$, $R$, $t_{b}$, $t_{e}$)}
    is \emph{true} if and only if
    there exists a time $t \in [t_b, t_e]$ at which
    every possible motion curve $PMC^T$ of an uncertain trajectory
    $U(T,\sigma)$ is inside
    region $R$. For a non-uncertain $T$, the previous condition
    can be adapted as: if and only if
    there exists a time $t \in [t_b, t_e]$ at which
      $T$ is inside $R$.
\item \emph{Always\_Definitely\_Inside($T$, $R$, $t_{b}$, $t_{e}$)}
    is \emph{true} if and only if at every time $t \in [t_b, t_e]$, every
    possible motion curve $PMC^T$ of an uncertain trajectory $U(T,\sigma)$
    is inside region $R$. For a non-uncertain $T$, the previous condition
    becomes: if and only if at every time $t \in [t_b, t_e]$, trajectory $T$ is
	    inside $R$.
\end{itemize}

\subsection{Adversarial model and target privacy properties}
\label{adversarial}

In our adversarial model, the adversary has access to the
published anonymised set of trajectories
$\mathcal{T}^\star$. Furthermore, the adversary also knows that every
location $\lambda \in \mathcal{T}^{\star}$ must be in
the original set of trajectories $\mathcal{T}$.
Note that this adversary's knowledge makes an
important difference from
previous adversarial models~\cite{abul08, nergiz09, monreale10, yarovoy09},
because in our model the linkage of some location
with some user reveals the exact location of
this user rather than a generalised or perturbed location.

Further, the method used for
transforming the original set
of trajectories $\mathcal{T}$ into $\mathcal{T}^\star$ is
assumed known by the adversary. However, this does not
include the method parameters or the seeds for pseudo-random number generators, which are considered secret. Indeed, the two methods we
are proposing rely on random permutations of locations and
random selection of trajectories during the clustering process,
and such randomness is in practice implemented using pseudo-random number
generators.
If an adversary knew the seeds of the generators, she could
easily reconstruct the original trajectories from the anonymised
trajectories.

Finally, the adversary also knows a
sub-trajectory $S$ of some original target trajectory
$T \in \mathcal{T}$ ($S \preceq T$) and knows that the
anonymised version of $T$ is in $\mathcal{T}^\star$. As in previous works,
we consider that every location in $\mathcal{T}$ is sensitive, {\em i.e.}
for any location, learning that a specific user
visited it represents useful knowledge
for the adversary.

Then, we identify two attacks:
\begin{enumerate}
    \item Find a trajectory $T^{\star} \in \mathcal{T}^{\star}$
    that is the anonymised version of $T$.
    \item Given a location $\lambda \not\in S$,
    determine whether $\lambda \in T$.
\end{enumerate}

If the adversary succeeds in the first attack of
linking a trajectory $T^{\star}$ with the target $T$, the second
is not trivial, because in general the locations in
$T^{\star}$ will not be those in $T$, but it is indeed easier.
This means that both attacks
are not independent. However, the second attack
can trivially succeed even if the first attack does not:
if all anonymised trajectories cross the same location $\lambda$
and $\lambda \not\in S$, the attacker knows that $\lambda \in T$.
As we show below, both attacks are
related to the two well-known privacy notions of
$k$-anonymity~\cite{samarati98,sweeney02a} and
$\ell$-diversity~\cite{machanavajjhala06}, respectively.

\begin{definition}[Trajectory $p$-privacy] \label{def:trajectory_private}
Let $Pr_{T^{\star}}[T|S]$ denote the probability of the
adversary's correctly linking the anonymised
trajectory $T^{\star} \in \mathcal{T}^{\star}$ with $T$ given
the adversary's knowledge $S \preceq T$. Then, \emph{trajectory $p$-privacy}
is met when $Pr_{T^{\star}}[T|S] \leq p$ for every trajectory $T \in \mathcal{T}$ and every subset $S \preceq T$.
\end{definition}

\begin{definition}[Trajectory $k$-anonymity]\label{def:anonymity}
Trajectory $k$-anonymity is achieved if and
only if trajectory $\frac{1}{k}$-privacy is met.
\end{definition}

\begin{definition}[Location $p$-privacy] \label{def:location_private}
Let $Pr_\lambda[T|S]$ denote the probability of the
adversary's success in correctly determining
a location $\lambda \in T \setminus S$,
given the adversary's knowledge $S \preceq  T$.
Then, \emph{location $p$-privacy} is met
when $Pr_\lambda[T|S] \leq p$ for every triple $(T,S,\lambda)$ such
that $T \in \mathcal{T}$, $S \preceq  T$ and $\lambda \not\in S$.
\end{definition}


\begin{definition}[Location $k$-diversity]\label{def:diversity}
Location $k$-diversity is achieved if and only
if location $\frac{1}{k}$-privacy is met.
\end{definition}

\subsection{Discussion on privacy models}

Achieving straightforward
trajectory $k$-anonymity, where each anonymised trajectory would
be identical to $k-1$ other anonymised trajectories, would in general
cause a huge information loss.
This is why some other trajectory $k$-anonymity definitions
under different assumptions have been proposed.

The $(k, \delta)$-anonymity definition~\cite{abul08, abul10}
relies on the uncertainty inherent to trajectory data recorded by
technologies like GPS. However, it may be hardly applied
when accurate data sets of trajectories are needed. Furthermore,
in order to achieve $(k, \delta)$-anonymity, the $k$ identical
anonymised trajectories should be defined roughly in the
same interval of time and they must contain the same number of locations.
Such constraints are indeed hard to meet.

According to our privacy model, trajectory $k$-anonymity
is achieved when there are at least $k$ anonymised trajectories
in $\mathcal{T}^\star$ having an anonymised version of $T$
as a sub-trajectory. Although this definition ignores the
time dimension, it does not require the length of the $k$
anonymised trajectories to be equal. However, suppose
that the adversary has a trajectory $T$ consisting of only one
location, an individual's home;
whatever the anonymisation method,
the anonymised version of $T$ is likely to be very similar to $T$.
This means that there
will be $k$ anonymised trajectories
containing the single location of $T$.
However, not all of these anonymised
trajectories start at the single location of $T$. Since an individual's home
is likely to be the first location of any individual's original
trajectory, those anonymised trajectories that do not start at the single
location of $T$ (just pass through it) can be filtered out by an adversary
and only the remaining trajectories are considered.
The same filtering process can be performed if the adversary
knows locations where the individual has never been. In this way, using
side knowledge the adversary identifies less than $k$ anonymised
trajectories compatible with the original trajectory $T$.
Hence, this definition may not actually guarantee $k$-anonymity
in the sense of Definition~\ref{def:anonymity}.

In conclusion, different levels of privacy can be provided
according to different assumptions on the
original data, the anonymised data, and the adversary's capabilities.
We defined above trajectory $p$-privacy
(Definition~\ref{def:trajectory_private})
and location $p$-privacy (Definition~\ref{def:location_private})
in order to capture two different privacy notions
when the original locations are preserved.

\section{Distance between trajectories}
\label{sec:distance}

Clustering trajectories requires defining a similarity measure ---a distance
between two trajectories. Because trajectories are distributed over space
and time, a distance that considers both spatial and temporal aspects of
trajectories is needed. Many distance measures have been proposed in the
past for both trajectories of moving objects and for time series but most of them are ill-suited
to compare trajectories for anonymisation purposes.
Therefore we define a new distance which can compare trajectories that are
only partially or not at all timewise overlapping. We believe this is
necessary to cluster trajectories for anonymisation.
We need some preliminary notions.

\subsection{Contemporary and synchronised trajectories}

\begin{definition}[$p$\%-contemporary trajectories]
Two trajectories
\[T_i = \{ (t^i_1,x^i_1,y^i_1), \ldots, (t^i_n,x^i_n,y^i_n)\} \]
and
\[T_j = \{ (t^j_1,x^j_1,y^j_1), \ldots, (t^j_m,x^j_m,y^j_m)\} \]
are said to be $p$\%-contemporary if
\[ p = 100 \cdot \min( \frac{I}{t^i_n-t^i_1}, \frac{I}{t^j_m-t^j_1}) \]
with $I = \max(\min(t^i_n,t^j_m) - \max(t^i_1,t^j_1), 0)$.
\end{definition}

Intuitively, two trajectories are 100\%-contemporary if and only if they
start at the same time and end at the same time; two trajectories are
0\%-contemporary if and only if they occur during
non-overlapping time intervals. Denote the overlap time of two trajectories
$T_i$ and $T_j$ as $ot(T_i,T_j)$.

\begin{definition}[Synchronised trajectories]
Given two $p$\%-contemporary trajectories $T_i$ and $T_j$ for some $p > 0$, both
\emph{trajectories are said to be synchronised} if they have the same number
of locations time-stamped within $ot(T_i,T_j)$ and these correspond to the same
time-stamps. A \emph{set of trajectories is said to be synchronised} if all
pairs of $p$\%-contemporary trajectories in it are
synchronised, where $p>0$ may be different for each pair.
\end{definition}

If we assume that between two locations of a trajectory, the
object is moving along a straight line between the locations at a constant
speed, then interpolating new locations is
straightforward. Trajectories can be then synchronised in the sense that if
one trajectory has a location at time $t$, then other trajectories defined
at that time will also have a (possibly interpolated) location at time $t$.
This transformation
guarantees that the set of new locations interpolated in order to
synchronise trajectories is of minimum cardinality.
Algorithm~\ref{alg:sync} describes
this process.
The time complexity of this algorithm is $O(|TS|^2)$ where $|TS|$
is the number of different time-stamps in the data set.

\begin{algorithm}[!ht]
\caption{Trajectory synchronisation} \label{alg:sync}
\begin{algorithmic}[1]
\STATE \textbf{Require:} $\mathcal{T} = \{ T_1, \ldots, T_N \}$ a set of trajectories to be
synchronised, where each $T_i \in \mathcal{T}$ is of the form:
\[ T_i = \{ (t^i_1,x^i_1,y^i_1), \ldots, (t^i_{n^i},x^i_{n^i},y^i_{n^i})\}; \]
\STATE Let $TS = \{ t^i_j \;|\; (t^i_j,x^i_j,y^i_j) \in T_i \;:\; T_i \in
\mathcal{T} \}$ be all time-stamps from all locations of all trajectories;
\FORALL{$T_i \in \mathcal{T}$}
  \FORALL{$ts \in TS$ with $t^i_1 < ts < t^i_{n^i}$}
    \IF{location having time-stamp $ts$ is not in $T_i$}
      \STATE insert new location in $T_i$ having the time-stamp $ts$ and
      coordinates interpolated from the two timewise-neighboring locations;
    \ENDIF
  \ENDFOR
\ENDFOR
\end{algorithmic}
\end{algorithm}

\subsection{Definition and computation of the distance}
\label{comput}

\begin{definition}[Distance between trajectories] \label{def:distance}
Consider a set of synchronised trajectories $\mathcal{T}= \{ T_1, \ldots,
T_N\}$ where each trajectory is written as
\[T_i = \{ (t^i_1,x^i_1,y^i_1), \ldots, (t^i_{n^i},x^i_{n^i},y^i_{n^i})\}
\enspace . \]
The \emph{distance between trajectories} is defined as follows.
If $T_i, T_j \in \mathcal{T}$ are $p$\%-contemporary with $p >0$, then
\[ d(T_i, T_j) = \frac{1}{p} \sqrt{\sum_{t_\ell \in ot(T_i,T_j)}
\frac{(x^i_\ell - x^j_\ell)^2 + (y^i_\ell - y^j_\ell)^2}{|ot(T_i,T_j)|^2}} \enspace . \]
If $T_i,T_j \in \mathcal{T}$ are $0$\%-contemporary but there is
at least one subset of $\mathcal{T}$
\[ \mathcal{T}^k(ij)=\{T^{ijk}_1, T^{ijk}_2, \ldots,T^{ijk}_{n^{ijk}} \}
\subseteq \mathcal{T} \]
such that $T^{ijk}_1=T_i$, $T^{ijk}_{n^{ijk}}=T_j$
and $T^{ijk}_\ell$ and $T^{ijk}_{\ell+1}$ are $p_\ell$\%-contemporary
with $p_\ell>0$ for $\ell=1$ to $n^{ijk}-1$, then
\[ d(T_i,T_j) = \min_{\mathcal{T}^k(ij)} \left(\sum_{\ell=1}^{n^{ijk}-1}
d(T^{ijk}_\ell, T^{ijk}_{\ell +1}) \right) \]
Otherwise $d(T_i, T_j)$ is not defined.
\end{definition}

The computation of the distance between every pair of trajectories is not
exponential as it could seem from the definition. Polynomial-time
computation of a distance graph containing the distances between all pairs
of trajectories can be done as follows.

\begin{definition}[Distance graph] \label{def:distancegraph}
A \emph{distance graph} is a weighted graph where
\begin{enumerate}
\setlength{\itemsep}{-1mm}
\item[(i)] nodes represent trajectories,
\item[(ii)] two nodes $T_i$ and $T_j$ are adjacent if the corresponding
trajectories are $p$\%-contemporary for some
$p>0$, and
\item[(iii)] the weight of the edge $(T_i, T_j)$ is the distance between
the trajectories $T_i$ and $T_j$.
\end{enumerate}
\end{definition}

Now, given the distance graph for $\mathcal{T} = \{ T_1, \ldots, T_N \}$,
the distance $d(T_i, T_j)$ for two trajectories is easily computed as the
minimum cost path between the nodes $T_i$ and $T_j$, if such path exists.
The inability to compute the distance for all possible trajectories (the
last case of Definition~\ref{def:distance}) naturally splits the distance
graph into connected components. The connected component that has the majority
of the trajectories must be kept, while the remaining components represent
outlier trajectories that are discarded in order to preserve privacy. Finally,
given the connected component of the distance graph having the majority of
the trajectories of $\mathcal{T}$, the distance $d(T_i, T_j)$ for \emph{any
two} trajectories on this connected component is easily computed as the
minimum cost path between the nodes $T_i$ and $T_j$.
The minimum cost path between every pair of nodes can
be computed using the Floyd-Warshall algorithm~\cite{Floyd} with computational cost
$O(N^3)$, {\em i.e.} in polynomial time.

\subsection{Intuition and rationale of the distance}

In order to deal with the time dimension, our distance measure applies
a linear penalty of $\frac{1}{p}$ to those trajectories that
are $p\%$-contemporary. This means that, the closer in time are
two trajectories, the shorter is our distance between both.
It should be remarked that we choose a linear penalty because
the Euclidean distance is also linear in terms of the spatial coordinates
and the Euclidean distance is the spatial distance measure we
consider by default. Other distances
and other penalties might be chosen, {\em e.g.} $\frac{1}{p^2}$.

A problem appears when considering $0\%$-contemporary trajectories.
How can two non-overlapping trajectories be penalised?
A well-known strategy is to give a weight to the time dimension
and another weight to the spatial dimension. By doing so,
the time distance and the spatial distance can be computed separately,
and later be merged using their weights. However, determining
proper values for these weights is a challenging task.

Anyway, the following lemma guarantees that, whenever
we consider two trajectories at minimum distance
for clustering, they do have some overlap.

\begin{lemma}
Any two trajectories in data set $\mathcal{T}$ at minimum distance are
$p\%$-contemporary with $p>0$.
\end{lemma}

\begin{proof} 
Consider a trajectory $T_i \in \mathcal{T}$ and another trajectory
$T_j \in \mathcal{T}$ at minimum distance from $T_i$.
Assume that $T_i$ and $T_j$ are not $p\%$-contemporary with $p>0$.
Then, since the distance
between $T_i$ and $T_j$ is defined, according to
Definition~\ref{def:distance} a subset of distinct trajectories
$\mathcal{T}(ij)=\{T^{ij}_1, T^{ij}_2, \ldots,T^{ij}_{n^{ij}} \} \subseteq \mathcal{T}$
must exist such that $T^{ij}_1=T_i$, $T^{ij}_{n^{ij}}=T_j$
and $T^{ij}_\ell$ and $T^{ij}_{\ell+1}$ are $p_\ell$\%-contemporary
with $p_\ell>0$ for $\ell=1$ to $n^{ij}-1$, and
\[ d(T_i,T_j) =\sum_{\ell=1}^{n^{ij}-1}
d(T^{ij}_\ell, T^{ij}_{\ell +1})\; .  \]
Then $d(T_i,T_j) > d(T^{ij}_\ell, T^{ij}_{\ell +1})$ for all
$\ell$ from 1 to $n^{ij}-1$ (strict inequality holds because
all trajectories in $\mathcal{T}(ij)$ are distinct). Thus,
we reach the contradiction that $d(T_i,T_j)$ is not minimum.
Hence, the lemma must hold.
\end{proof}


\section{Anonymisation methods} \label{sec:anonymisation}

We present two anonymisation methods,
called SwapLocations and ReachLocations, respectively, which yield anonymised trajectories consisting of true original locations.
The first method is
partially based on the microaggregation~\cite{domingo02}
of trajectories and partially based
on the permutation of locations.
The second method is based on the permutation of locations only.
The main difference between
the SwapTriples method~\cite{domingo10springl}
and the two new methods we propose here is that the latter
effectively guarantees trajectory $k$-anonymity (SwapLocations)
or location $k$-diversity (ReachLocations).
To that end, an original
triple is discarded if it cannot be swapped randomly with
another triple drawn from a set of $k-1$ other original triples.

Our two methods differ from each other in several aspects.
The first method assumes an unconstrained environment, while
the second one considers an environment with mobility constraints,
like an underlying street or road network.
SwapLocations effectively achieves trajectory $k$-anonymity.
ReachLocations provides higher utility by design, but regarding
privacy, it offers location $k$-diversity instead of
trajectory $k$-anonymity. A common feature of both methods
is that locations in the resulting anonymised trajectories
are true, fully accurate original locations,
{\em i.e.} no fake, generalised or perturbed locations are
given in the anonymised data set of trajectories.

\subsection{The SwapLocations method}
\label{swaploc}

Algorithm \ref{alg:clustering} describes the process
followed by the SwapLocations method in order to anonymise
a set of trajectories. First, the set of trajectories
is partitioned into several clusters.
Then, each cluster is anonymised
using the SwapLocations function in Algorithm \ref{alg:swap_locations}.
We should remark here that we only consider trajectories
for which the distance to other trajectories can be computed
using the distance in Definition~\ref{def:distance}. Otherwise said,
given the distance graph $G$ (Definition~\ref{def:distancegraph}),
our distance
measure can only be used within one of the connected components of $G$;
obviously, we take the trajectories in the largest connected component
of $G$.
It should also be remarked that Algorithm~\ref{alg:sync} is only used
to compute the distance between trajectories. Once a cluster $C$
is created, the anonymisation algorithm works over the original
triples of the trajectories in $C$, and not over the triples created
during synchronisation.

We limit ourselves to clustering algorithms which
try to minimise the sum of the intra-cluster distances or approximate
the minimum and such that the cardinality of each cluster
is $k$, with $k$ an input parameter; if the number of trajectories
is not a multiple of $k$, one or more clusters must absorb the
up to $k-1$ remaining trajectories, hence those clusters will have
cardinalities between $k+1$ and $2k-1$.
This type of clustering is precisely the one used in
microaggregation~\cite{domingo02}.
The purpose of minimising the sum of the intra-cluster distances
is to obtain clusters as homogeneous as possible, so that the subsequent
independent treatment of clusters does not cause much information loss.
The purpose of setting $k$ as the cluster size
is to fulfill trajectory $k$-anonymity, as shown in
Section~\ref{sec:guaranteesswaploc}.
We might employ any microaggregation heuristic for clustering
purposes (see details in Section~\ref{sec:complex} below).

\begin{algorithm}[!ht]
\caption{Cluster-based trajectory anonymisation($\mathcal{T},R^t, R^s,k$)} \label{alg:clustering}
\begin{algorithmic}[1]
\STATE \textbf{Require:} (i) $\mathcal{T}= \{ T_1, \ldots, T_N\}$ a set of original
trajectories such
that $d(T_i,T_j)$ is defined for all $T_i,T_j \in \mathcal{T}$, (ii) $R^t$ a
time threshold and $R^s$ a space threshold, both of them public;
\STATE Use any clustering algorithm to cluster the trajectories of
$\mathcal{T}$, while minimising the sum of intra-cluster distances measured
with the distance of Definition~\ref{def:distance} and ensuring that
minimum cluster size is $k$;
\STATE Let $C_1, C_2, \dots, C_{n_{\mathcal{T}}}$ be the resulting clusters;
\FORALL{clusters $C_i$}
\STATE $C^\star_i$ = SwapLocations($C_i,R^t,R^s$); \hfill // Algorithm~\ref{alg:swap_locations}
\ENDFOR
\STATE Let $\mathcal{T}^\star=C^\star_1 \cup \cdots \cup
C^\star_{n_{\mathcal{T}}}$
be the set of anonymised trajectories.
\end{algorithmic}
\vspace*{1mm}
\end{algorithm}

The SwapLocations function (Algorithm~\ref{alg:swap_locations})
begins with a random trajectory $T$ in $C$.
The function attempts to cluster
each unswapped triple $\lambda$ in $T$
with another $k-1$
unswapped triples belonging to different trajectories
such that: (i) the time-stamps of these triples differ by
no more than a time threshold $R^t$ from the time-stamp of $\lambda$;
(ii) the spatial coordinates differ by
no more than a space threshold $R^s$.
If no $k-1$ suitable
triples can be found that can be clustered with $\lambda$,
then $\lambda$ is removed; otherwise,
random swaps of triples
are performed within the formed cluster.
Randomly swapping this cluster of triples guarantees that any of these triples
has the same probability of remaining in its original trajectory or becoming
a new triple in any of the other $k-1$ trajectories.
Note that Algorithm~\ref{alg:swap_locations} guarantees that every triple $\lambda$ of every trajectory $T \in C$
will be swapped or removed.


\begin{algorithm}[!ht]
\caption{SwapLocations($C, R^t, R^s$)} \label{alg:swap_locations}
\begin{algorithmic}[1]
\STATE \textbf{Require:} (i) $C$ a cluster of trajectories to be transformed, (ii) $R^t$ a time threshold and $R^s$ a space threshold;
\STATE Mark all triples in trajectories in $C$ as ``unswapped'';
\STATE Let $T$ be a random trajectory in $C$; 
\label{line:fewest}
\FORALL {``unswapped'' triples $\lambda = (t_\lambda, x_\lambda, y_\lambda)$ in $T$} \label{line:for}
    \STATE Let $U = \{\lambda\}$;
    // Initialise $U$ with $\{\lambda\}$
    \FORALL {trajectories $T'$ in $C$ with $T' \neq T$}
        \STATE Look for an ``unswapped''
	triple $\lambda'= (t_{\lambda'}, x_{\lambda'}, y_{\lambda'})$
	in $T'$ minimising the intra-cluster distance in $U \cup \{\lambda'\}$ and such that:
        \[ |t_{\lambda'} - t_{\lambda}| \leq R^t \]
        \[ 0 \leq \sqrt{(x_{\lambda'}-x_{\lambda})^2+(y_{\lambda'}-y_{\lambda})^2} \leq R^s \enspace ; \]
        \IF {$\lambda'$ exists}
            \STATE $U \leftarrow U \cup \{\lambda'\}$;
        \ELSE
        	\STATE Remove $\lambda$ from $T$;
          \STATE Go to line~\ref{line:for} in order to analyse the next triple
	  $\lambda$;
        \ENDIF
    \ENDFOR
    \STATE Randomly swap all triples in $U$;
    \STATE Mark all triples in $U$ as ``swapped'';
\ENDFOR
\STATE Remove all ``unswapped'' triples in $C$;
\STATE \textbf{Return}  $C$
\end{algorithmic}
\end{algorithm}

The SwapLocations function specified by
Algorithm~\ref{alg:swap_locations} swaps entire triples, that is,
time and space coordinates.
The following example illustrates the advantages of swapping
time together with space.

\begin{example}\label{tahrir}{\em
Imagine John attended one day the political protests in Tahrir
Square, Cairo, Egypt, but he would not like his political
views to become broadly known.
Assume John's trajectory is anonymised
and published. Assume further that an adversary knows the precise time
John left his hotel in the morning, say 6:36 AM ({\em e.g.}
because the adversary has bribed the hotel concierge into recording
John's arrival and departure times). Now:
\begin{itemize}
\item If SwapLocations swapped only
spatial coordinates, the adversary could re-identify John's
trajectory as one starting with a triple (6:36 AM, $x'_h$, $y'_h$).
Furthermore, $(x'_h,y'_h)$ must be a location
within a distance $R^s$ from the hotel coordinates $(x_h,y_h)$,
although the adversary does not know the precise value of $R^s$.
The re-identified trajectory would contain all true timestamps
of John's original trajectory (because they would not have been swapped),
and spatial coordinates within distance $R^s$ of John's really visited
spatial coordinates.
Hence, it would be easy to check whether John was near Tahrir
Square during that day. Without swapping times,
privacy protection can only be obtained by taking $R^s$ large enough
so that within distance $R^s$ of the original locations visited by John there
are several semantically different spatial coordinates.
To explain
what we mean by semantic difference, assume $(x,y)$
is Tahrir Square and the trajectory
anonymiser guarantees that he has taken $R^s$ large enough
so that $(x,y)$ could be swapped
with some spatial coordinates $(x',y')$ off Tahrir Square;
even if $(x',y')$ turned out to be still within Tahrir Square, John
could claim to have been off Tahrir Square; the adversary
could not disprove such a claim, because in fact
$(x,y)$ could be at a distance $R^s$ from $(x',y')$ and hence
outside the Square.
However, a large $R^s$ means a large total space distortion.
\item If entire triples are swapped, as actually done by SwapLocations,
the adversary can indeed locate an anonymised
trajectory containing
(not necessarily starting with) triple (6:36 AM, $x_h$, $y_h$).
However, there is only a chance $1/k$ that this triple was not swapped
from another of the $k-1$ original trajectories with which John's
original trajectory was clustered.
Similarly, the other triples in the anonymised
trajectory containing (6:36 AM, $x_h$, $y_h$)
have also most likely ``landed'' in that anonymised trajectory as a result
of a swap with some location in some of the $k-1$ original trajectories
clustered with John's. Hence, John's trajectory is cloaked with $k-1$
other trajectories. We will prove in
Section~\ref{sec:guaranteesswaploc} that
this guarantees trajectory $k$-anonymity in the sense of
Definition~\ref{def:anonymity}. In particular, the triple $(t,x,y)$
corresponding to John at Tahrir Square will appear in one of the
$k$ anonymised trajectories, unless that triple has been removed
by the SwapLocations function because it was unswappable (the smaller
$R^t$ and $R^s$, the more likely it is for the triple to be removed).
\end{itemize}}
\end{example}

\subsection{The ReachLocations method}

The ReachLocations method, described in Algorithm~\ref{alg:reach},
takes reachability constraints into account: from a given
location, only those locations at a distance below a threshold
{\em following a path in an underlying graph} ({\em e.g.},
urban pattern or road network) are considered to be directly
reachable. Enforcing such reachability constraints while requiring
full trajectory $k$-anonymity would result in a lot of original locations
being discarded. To avoid this, trajectory $k$-anonymity is
changed by another useful privacy definition: location 
$k$-diversity.

Computationally, this means that trajectories are {\em not} microaggregated
into clusters of size $k$. Instead,
each location is $k$-anonymised independently using
the entire set of locations of all trajectories. To do so, a
cluster $C_\lambda$ of ``unswapped'' locations is created around a
given location $\lambda$, {\em i.e.}
$\lambda \in C_\lambda$.
The cluster $C_\lambda$ is constrained as follows:
(i) it must have the lowest intra-cluster distance among those clusters
of $k$ ``unswapped'' locations that contain the location $\lambda$;
(ii) it must have locations belonging to $k$ different trajectories; and (iii)
it must contain only locations at a path from $\lambda$
at most $R^s$ long and with time-stamps differing from
$t_\lambda$ at most $R^t$. Then, the spatial coordinates
$(x_\lambda, y_\lambda)$ are swapped with the
spatial coordinates of some random location in $C_\lambda$
and both locations are marked as ``swapped''. If no cluster $C_\lambda$
can be found, the location $\lambda$ is removed from the data set
and will not be considered anymore in the subsequent anonymisation.
This process continues until no more ``unswapped'' locations
appear in the data set.

It should be emphasised that, according to Algorithm~\ref{alg:reach},
two successive locations $\lambda^i_j$
and $\lambda^i_{j+1}$ of
an original trajectory
$T_i$ may be cloaked with respective sets of $k-1$ locations belonging
to different sets of $k-1$ original trajectories; for this reason
we cannot speak of trajectory $k$-anonymity, see the example below.

\begin{example}{\em
Consider $k-1$ trajectories
within city $A$, $k-1$ trajectories within city $B$
and one trajectory $T_{AB}$ crossing from $A$ to $B$.
When applying ReachLocations, the initial locations of
$T_{AB}$ are swapped with locations of trajectories
within $A$, whereas the final locations of $T_{AB}$
are swapped with locations of trajectories within $B$.
Imagine that an adversary knows a sub-trajectory $S$
of $T_{AB}$ containing
one location $\lambda_A$ in $A$ and one
location $\lambda_B$ in $B$.
Assume $\lambda_A$ and $\lambda_B$ are not removed by
ReachLocations anonymisation.
Now, the adversary will know that the anonymised trajectory
$T^\star_{AB}$ corresponding to $T_{AB}$ is the only anonymised
trajectory crossing from $A$ to $B$. Thus, there is no
trajectory $k$-anonymity, even if the adversary will
be unable to determine the exact locations of $T_{AB} \setminus S$,
because each of them has been swapped within a set of $k$ locations.}
\end{example}

\begin{algorithm}[p]
\caption{ReachLocations($\mathcal{T},R^t,R^s,k$)} \label{alg:reach}
\begin{algorithmic}[1]
\fontsize{11}{10}\selectfont
\STATE \textbf{Require:}
(i) $\mathcal{T} = \{ T_1, \ldots, T_N \}$ a set of original trajectories,
(ii) $G$ a graph describing the paths between locations,
(iii) $R^t$ is a time threshold and $R^s$ is a space threshold, both
of them public;
\STATE Let $TL = \{\lambda^i_j \in T_i \;:\; T_i \in
\mathcal{T} \}$ contain all locations from all trajectories, where
$\lambda^i_j=(t^i_j,x^i_j, y^i_j)$ and the spatial coordinates
$(x^i_j,y^i_j)$ are called a point;
\STATE Mark all locations in $TL$ as ``unswapped'';
\STATE Let $\mathcal{T^{\star}} = \emptyset$ be an empty set of anonymised trajectories;
\WHILE {there exist trajectories in $\mathcal{T}$}
    \STATE Let $T_i$ be a trajectory randomly chosen in $\mathcal{T}$;
\FOR {$j = 1$ to $j = |T_i|$}
      \IF {$\lambda^i_j$ is ``unswapped''}
        \STATE Let $C^i_j = \{\lambda_1, \cdots, \lambda_{k-1}\}$ be
	a cluster of locations in $TL$ such that: \label{line:bestCluster}
\begin{enumerate}
\item All locations in $C^i_j$ are ``unswapped'', with points
different from $(x^i_j,y^i_j)$ and no two equal points;
\item Points in $C^i_j$ belong to trajectories
in $\mathcal{T} \setminus \{T_i\}$
and no two points belong to the same trajectory;
\item For any $\lambda \in C^i_j$, it holds that:
\begin{enumerate}
\item $|t_\lambda - t^i_j| \leq R^t$
\item If $j > 1$ there is a path in $G$ between  $(x^i_{j-1},y^i_{j-1})$
and  $(x_\lambda, y_\lambda)$;
\item If $j < |T_i|$ there is a path in $G$
between $(x_\lambda,y_\lambda)$ and $(x^i_{j+1},y^i_{j+1})$;
\item The length of each path above is no more than $R^s$;
\end{enumerate}
\item The sum of intra-cluster distances
in $C^i_j \cup \{\lambda^i_j\}$ is minimum among clusters
of cardinality $k-1$ meeting the previous conditions;
\end{enumerate}
        \IF {such a cluster $C^i_j$ does not exist}
            \STATE Remove $\lambda^i_j$ from $T_i$;
        \ELSE
	    \STATE Mark $\lambda^i_j$ as ``swapped'';
            \STATE With probability $\frac{k-1}{k}$:
	    \begin{enumerate}
	    \item Pick a random
	  location $\lambda \in C^i_j$ and mark it as ``swapped'';
	    \item Swap the spatial
	    coordinates $(x^i_j,y^i_j)$ of $\lambda^i_j$
	    with the spatial coordinates $(x_\lambda,y_\lambda)$
	    of $\lambda$;
	    \end{enumerate}
        \ENDIF
      \ENDIF
    \ENDFOR
    \STATE $\mathcal{T^{\star}} = \mathcal{T^{\star}} \cup \{T_i\}$;
    \STATE Remove $T_i$ from $\mathcal{T}$;
\ENDWHILE
\STATE \textbf{Return} $\mathcal{T^{\star}}$.
\end{algorithmic}
\end{algorithm}

Algorithm~\ref{alg:reach} swaps
only spatial coordinates instead of full triples.
We show in the example below that this is enough
for ReachLocations to achieve location $k$-diversity
(we have shown above that it cannot achieve trajectory
$k$-anonymity anyway). If swapping time coordinates
is not beneficial in terms of privacy guarantees,
they should not be swapped, because the fact that
anonymised trajectories preserve the original sequence
of time-stamps of original trajectories increases their utility.

\begin{example}{\em
Let us resume Example~\ref{tahrir}, but now assume that
ReachLocations is used instead of SwapLocations to anonymise
trajectories. In this case, the adversary will find
an anonymised trajectory starting with (6:36 AM,$x'_h$,$y'_h$).
This anonymised trajectory will contain all true timestamps of John's
original trajectory. However, the spatial coordinates appearing
in any location of this re-identified trajectory are John's
original spatial coordinates with a probability at most $1/k$.
We will prove in Section~\ref{guaranteesreach} below that this guarantees
location $k$-diversity in the sense of
Definition~\ref{def:diversity}. If we want to prevent the adversary
from making sure that John visited Tahrir Square, we should
take $R^s$ large enough (the discussion in Example~\ref{tahrir}
about the protection afforded by a large $R^s$ when time is
not swapped is valid here).
}\end{example}

\subsection{Complexity of SwapLocations and ReachLocations}
\label{sec:complex}

We first give a complexity assessment of SwapLocations and
ReachLocations assuming that the distance graph
mentioned in Section~\ref{comput}
has been precomputed and is available.
This is reasonable, because
the distance graph needs to be computed only once, while the anonymisation
methods may need to be run several times ({\em e.g.} with different parameters).
Regarding SwapLocations, we have:
\begin{itemize}
\item Algorithm~\ref{alg:clustering} can use
any fixed-size microaggregation heuristic for clustering
({\em e.g.} MDAV in~\cite{domingo05}). Most microaggregation heuristics
have quadratic complexity, that is $O(N^2)$, where $N$ is the number
of trajectories.
\item Algorithm~\ref{alg:clustering} calls the procedure
SwapLocations once for each resulting cluster, that is,
$O(N/k)$ times.
\item In the worst case, the complexity of procedure SwapLocations
(Algorithm~\ref{alg:swap_locations}) is proportional
to the number of locations of the longest trajectory in $C$,
say $O(n_{max})$.
For each location, a search of another
location for swapping is performed among
the other $k-1$ trajectories. The number of candidates
for swapping is $O((k-1)n_{max})$. Hence, the complexity
of SwapLocations is $O((k-1)n^2_{max})$.
\item The total complexity of the method is thus
\begin{equation}
\label{compswaplocations}
O(N^2) + O(N/k)\cdot O((k-1)n^2_{max}) = O(N^2) + O(N n^2_{max})
\end{equation}
\end{itemize}

Regarding the complexity of ReachLocations, we have
\begin{itemize}
\item Algorithm~\ref{alg:reach} has an external loop
which is called $N$ times, where $N$ is the number
of trajectories in $\mathcal{T}$. For each trajectory,
a swap is attempted for each of its unswapped locations.
Hence the algorithm performs $O(N n_{max})$ swaps,
where $n_{max}$ is the number of locations in the longest
trajectory.
\item Each swap involves forming a cluster which $k-1$ locations
selected from $TL$,
which takes time
proportional to the total number of locations in $TL$,
that is, $O(N n_{max})$.
\item Hence, the total complexity of the method is $O(N^2 n^2_{max})$.
\end{itemize}

By comparing the last expression and Expression~\ref{compswaplocations},
we see that both SwapLocations and ReachLocations are quadratic in
$N$ and quadratic in $n_{max}$, but ReachLocations is slower.
Such complexity motivates the following two comments
related to scalability:
\begin{itemize}
\item If the number of trajectories $N$ in the original data set is very large,
quadratic complexity may be very time consuming. In this case, a good
strategy is to use some blocking technique to split the original data set
into several subsets of trajectories, each of which should be anonymised
separately.
\item $n_{max}$ being large may be less problematic than $N$ being large,
provided that only a small fraction of trajectories have $n_{max}$ or close
to $n_{max}$ locations. If a lot of trajectories are very long, a good
strategy would be to split each of these into two or more trajectories and anonymise them independently.
\end{itemize}

Finally, in case we add the time complexity of the computation of the
distance graph mentioned in Section~\ref{comput} (which
is $O(N^3)$ using the Floyd-Warshall algorithm), the time complexities
of both SwapLocations and ReachLocations become
$O(N^3) + O(N n^2_{max})$ and $O(N^3) + O(N^2 n^2_{max})$, respectively.

\section{Privacy guarantees}
\label{sec:guarantees}

\subsection{Privacy guarantees of SwapLocations}
\label{sec:guaranteesswaploc}

The main difference between the SwapTriples method
in~\cite{domingo10springl} and the SwapLocations method here
is that, in the latter, no original location
remains unswapped in an anonymised trajectory.

\begin{proposition} \label{prop:prob}
Let $S \preceq T_S$ be the adversary's knowledge of a target
original trajectory $T_S$ and
$\lambda_1, \lambda_2, \cdots, \lambda_{|S|}$ be all triples in $S$.
For every trajectory $T_i$, the probability
that the triple $\lambda$ in $S$ appears in the anonymised version $T_i^{\star}$ of $T_i$ produced by SwapLocations is:
$$
\Pr(\lambda \in T_i^{\star}|\lambda \in S)  = \left\{\begin{array}{l l}
\frac{1}{k} &  \quad \text{if } T_S \text{ and } T_i \mbox{ lie in the same cluster}\\
\\
0 & \quad \mbox{otherwise.}
\end{array}\right.
$$
\end{proposition}

\begin{proof} 
By construction of Algorithm~\ref{alg:swap_locations},
if $T_S$ and $T_i$ do not lie in the same cluster, there is no
possibility of swapping triples between them. Hence, in this
case, $\Pr(\lambda \in T_i^{\star}|\lambda \in S) = 0$.

Let $T_1, T_2, \cdots, T_k \in \mathcal{T}$ be $k$ trajectories that
are anonymised together in the same cluster by the SwapLocations method.
Without loss of generality, let us assume that $T_S = T_1$.
By construction of Algorithm~\ref{alg:swap_locations},
for every $1 \leq i \leq k$, $\Pr(\lambda \in T_i^{\star}| \lambda \in T_1)$
is $0$ if $\lambda$ was removed, $\frac{1}{k}$ otherwise.
Note that a swapping option is to swap a triple with itself,
that is, not to swap it.
Since it does not make sense to consider removed triples in $S$,
we conclude that $\Pr(\lambda_j \in T_i^{\star}| \lambda_j \in T_1) =
\frac{1}{k}, \;\; \forall 1\leq j \leq |S|, 1 \leq i \leq k$ and,
in consequence, $\Pr(\lambda_j \in T_i^{\star}| \lambda_j \in S)
= \frac{1}{k},\;\; \forall 1\leq j \leq |S|, 1 \leq i \leq k$. 
\end{proof}

\begin{theorem} \label{theo:anonymity}
The SwapLocations method achieves trajectory $k$-anonymity.
\end{theorem}

\begin{proof} 
By Proposition~\ref{prop:prob},
any sub-trajectory $S' \preceq S \preceq T_1$ has the same probability
of being a sub-trajectory of $T_1^{\star}$ than of being a sub-trajectory
of any of the $k-1$ trajectories $T_2^{\star}, \cdots, T_k^{\star}$.
Thus, given $S$, an adversary is not able to link $T_1$
with $T_1^{\star}$ with probability higher than $\frac{1}{k}$.
Therefore, SwapLocations satisfies $\frac{1}{k}$-privacy
according to Definition~\ref{def:trajectory_private}; according
to Definition~\ref{def:anonymity}, it also satisfies
trajectory $k$-anonymity.
\end{proof}

\subsection{Privacy guarantees of ReachLocations}
\label{guaranteesreach}

We show below that ReachLocations
provides location $k$-diversity.

\begin{proposition} \label{prop:reach}
Any triple $\lambda$ in an original
trajectory $T$
appears in the anonymised trajectory
$T^\star$ corresponding to $T$ obtained with ReachLocations
if and only if $\lambda$ was not removed and was swapped with itself,
which happens with probability at most $\frac{1}{k}$.
\end{proposition}

\begin{proof}
Let us prove the necessity implication.
By construction of Algorithm~\ref{alg:reach},
any triple
$\lambda$ whose spatial coordinates (point)
cannot be swapped within a cluster
$C \cup \{\lambda\}$ containing $k$ different points
belonging to $k$ different trajectories
is removed and does not appear in the set of anonymised trajectories.
Further, the only way for a non-removed triple $\lambda \in T$ to remain
unaltered in $T^\star$ is precisely that its point is swapped with itself,
which happens with probability $\frac{1}{k}$.
Therefore, to remain unaltered in $T^\star$, a triple
in $T$ needs to avoid removal and to have its point swapped with itself,
which happens with probability at most $\frac{1}{k}$.

Now let us prove the sufficiency implication.
Assume that $\lambda=(t,x,y) \in T$ appears
in $T^\star$ without having been swapped with itself.
Then, by construction
of ReachLocations, $\lambda \in T^\star$ must have been formed as the result
of swapping a triple $(t,x',y') \in T$ with a triple $(t',x,y)$
from another original trajectory, where $(x',y') \neq (x,y)$.
Buth then $T$ would contain two triples with the same time-stamp $t$
and different spatial locations, which is a contradiction. 
\end{proof}

\begin{theorem}
The ReachLocations method achieves location $k$-diversity.
\end{theorem}

\begin{proof}
Assume the adversary knows a sub-trajectory
$S$ of an original trajectory $T$. The sequence of time-stamps
in $S$ allows the adversary to re-identify the anonymised trajectory
$T^\star$ corresponding to $T$ (because the time-stamp
sequence is preserved).
By Proposition~\ref{prop:reach}, any
triple $\lambda \in T^\star \setminus S$ belongs
to $T \setminus S$ with probability at most $\frac{1}{k}$.
Now, consider a triple $\lambda=(t,x,y) \in T^{\star\star} \setminus S$,
where $T^{\star\star}$ is an anonymised trajectory different
from $T^\star$.
The probability that $\lambda$ came to $T^{\star\star} \setminus S$ from
$T \setminus S$ is the probability that $\lambda$ was
swapped and swapping did not alter it. This probability is zero,
because swaps preserve time coordinates
but take place only between triples having different
space coordinates.
Hence, in terms of Definition~\ref{def:trajectory_private},
$Pr_\lambda[T|S] \leq \frac{1}{k}$ for every triple
$(T,S,\lambda)$ such that $T \in \mathcal{T}$, $S \preceq T$
and $\lambda \not\in S$. 
\end{proof}

Note that the previous proof also implies
that, even if a triple $\lambda=(t,x,y) \not\in S$ is shared
by $M>1$ anonymised
trajectories, the probability of $\lambda \in T \setminus S$
remains at most $\frac{1}{k}$.
What can be inferred by the adversary,
however, is that $M$ original trajectories
(in general not the ones corresponding to the $M$ anonymised
trajectories) visited spatial coordinates $(x,y)$
at possibly different times.
Indeed, $(t,x,y)$ can be obtained by swapping $(t',x,y)$
and $(t,x',y')$ for any $t'$ such that $|t'-t| \leq R^t$
and for any $(x',y') \neq (x,y)$ at path distance at most $R^s$.
If $M$ is the total number of anonymised trajectories, then
the adversary can be sure that original
trajectory $T$ visited spatial coordinates
$(x,y)$ at some time $t'$ such that $|t'-t| \leq R^t$.
Such inference by the adversary does
not violate location $k$-diversity: violation
would require guessing {\em both} the spatial {\em and} temporal coordinates
of a triple in $T \setminus S$.
Of course, the time threshold $R^t$ must be taken large enough so that
the time coordinate $t$ is sufficiently protected.

\section{Experimental results and evaluation}

We implemented SwapLocations and ReachLocations. SwapLocations
performs clustering of trajectories using the partitioning step
of the MDAV microaggregation heuristic~\cite{domingo05}.
We used two data sets in our experiments:

\begin{itemize}
\item {\em Synthetic data set}.
We used the Brinkhoff's generator~\cite{brinkhoff03} to generate
1,000 synthetic trajectories which altogether visit
45,505 locations in the German city of
Oldenburg.
Synthetic trajectories generated with the Brinkhoff's generator have also
been used in~\cite{abul08,nergiz08,nergiz09,yarovoy09}.
We used this data set mainly for comparing
our methods with $(k,\delta)$-anonymity~\cite{abul08}.
The number of trajectories
being moderate, we were able to run
in reasonable time the methods to be compared with a large
number of different parameter choices. Another advantage
is that the street graph of Oldenburg was available, which
is necessary to run ReachLocations. The downside of this data set having
a moderate number of trajectories is that these are rather sparse,
which causes the relative distortion in the anonymised data set to
be substantial, no matter the method used. Anyway, this is not
a serious problem to compare methods with each other.
\item {\em Real-life data set}.
We also used a real-life data set of cab mobility traces that were
collected in the city of San Francisco \cite{comsnets09piorkowski}.
This data set consists of 536 files, each of them containing
the GPS coordinates of a cab during a period of time.
After a filtering process, we obtained 4582 trajectories and 94
locations per trajectory on average. The advantage
of this data set over the synthetic one is that it
contains a larger number of trajectories and that these are real ones.
Then, we show through a real example how appropriate is our distance metric for trajectory clustering. Also, we present
utility measures on the SwapLocations method for this real-life data set
using different space thresholds. The weakness of this data set
is that it cannot be used for ReachLocations, because it does
not include the underlying street graph of San Francisco.
\end{itemize}

\subsection{Results on synthetic data}

For the sake of reproducibility, we indicate
the parameters we used in Brinkhoff's generator
to generate our Oldenburg synthetic data set: 6
moving object classes and 3 external object classes; 10 moving objects and
1 external object generated per timestamp; 100 timestamps; speed 250;
and ``probability'' 1,000. This resulted in 1,000 trajectories containing
45,405 locations. The maximum trajectory length was 100 points, the
average length was 45.4 locations, and the median length was 44 locations.

\subsubsection{Implementation details of our methods}

We have introduced a new distance measure between trajectories used by
the SwapLocations proposal during the clustering process.
As mentioned in Section~\ref{swaploc} above,
our distance function can only be used
within one of the connected components of the distance graph $G$.
During the construction of
the distance graph for the synthetic data we found $11$ connected
components, $10$ of them of size $1$. Therefore, we removed these $10$
trajectories in order to obtain a new distance graph with just one connected
component. In this way, we preserved $99\%$ percent of all
trajectories before the anonymisation process. The
removed trajectories were in fact trajectories of length one, {\em i.e.}
with just one location in each one.

The SwapLocations method has been implemented using the
following simple microaggregation method for trajectories:
first, create clusters of  size $k$ with minimum
intra-cluster distance and then disperse the up to $k-1$
unclustered trajectories to existing
clusters while minimising the intra-cluster distance. This algorithm incurs
no additional discarding of trajectories.

On the other hand, the ReachLocations method
does not remove trajectories, unlike the SwapLocations method.
It does, however, remove non-swappable locations, which
causes the removal of any trajectory consisting of
non-swappable locations only.

\subsubsection{Implementing $(k, \delta)$-anonymity for comparison with
our method}

We compared our proposals with $(k, \delta)$-anonymity~\cite{abul08}. Since
$(k, \delta)$-anonymity only works over trajectories having the same
time span, first a pre-processing step to partition the trajectories is
needed. Superimposing the begin and end times of the trajectories
through reduction of the time coordinate modulo a parameter
$\pi$ does not always yield
at least $k$ trajectories having the same time span; it may also happen that a
trajectory disappears because the new reduced end time lies before the new
reduced begin time.

We have used $\pi = 3$ which kept the maximum (and so discarded the minimum)
trajectories. From the 1,000 synthetic trajectories, 40 were discarded because
the end time was less than the begin time and 187 were discarded because
there were at most 4 trajectories having the same time span. In total, 227
(22.7\%) trajectories were discarded just in the pre-processing step. The
remaining 773 trajectories were in 32 sets having the same time span, each
set containing a minimum of 15 trajectories and 24 on average.

We performed $(k, \delta)$-anonymisation for $k = 2$, 4, 6, 8, 10, and
15 and $\delta = 0$, 1000, 2000, 3000, 4000 and 5000. Because of the
pre-processing step, using a higher $k$ was impossible without
causing a
significant number of additional trajectories to be discarded.

\subsubsection{Utility comparison}

The performance of our proposals strongly depends on the values
of the time and space threshold parameters, denoted as $R^t$
and $R^s$, respectively. In practice, these values must be chosen
to maximise utility
while affording sufficient privacy protection. Too large
thresholds reduce utility (large space distortion if $R^s$ is too
high and large time distortion is $R^t$ is too high),
but too small thresholds
reduce utility because of removal of many unswappable locations.
As a rule of thumb, as illustrated in Example~\ref{tahrir},
the space threshold $R^s$ must be sufficiently large so that
within a radius $R^s$ of any spatial location
there are sufficiently distinct locations ({\em e.g.}
if $(x,y)$ lies in Tahrir Square, Cairo, there should
be points outside the Square within a radius $R^s$ of $(x,y)$).

In order to compute the total space distortion, a value
for $\Omega$ must be chosen and this can be a challenging task.
Note that the value of $\Omega$ is application-dependent
({\em e.g.} for applications where the distortion should
measure the accuracy of trajectories
$\Omega$ should be zero so that only non-removed triples
contribute to $TotalSD$, while for
applications that should avoid removing any triple
$\Omega$ should be very high).
For this reason we propose to compare separately the following three utility
properties: (i) total space distortion; (ii) percentage
of removed trajectories; and
(iii) percentage of removed locations. To do so, we set $\Omega = 0$
when computing the total space distortion. Consequently, the
percentage of removed triples as well as the percentage of removed
trajectories are considered separately from the total space distortion.

It should be remarked that the computation of the total space distortion
of the ReachLocations method is done using the Euclidean distance
between locations rather than the distance defined by the
reachability constraints (distance on the underlying network).
Note that reachability constraints should be considered during
the anonymisation process but not necessarily when computing
the total space distortion.

For successive anonymisations aimed at comparing the
SwapLocations and ReachLocations methods with $(k,\delta)$-anonymity,
we set $R^t$ and $R^s$ in a way to obtain roughly the same
total space distortion values as in $(k, \delta)$-anonymity
(cf. Table~\ref{tab:totalsd}) with $\Omega = 0$.
The idea is that, after assuring that the three methods achieve
roughly the same total space distortion,
we will be able to focus on other utility properties
like the percentage of removed trajectories
and the percentage of removed locations. It should be
noted that our comparison is not entirely fair for any of
the three methods because all of them are aimed at achieving
different privacy notions. However, we believe that our
results are indicative of the weaknesses and the strengths of
our proposals.

\begin{table}[!ht]
\centering
\begin{tabular}{|r|r|r|r|r|r|r|}
\hline
$\delta$ $\backslash$ $k$ & 2 & 4 & 6 & 8 & 10 & 15 \\
\hline
0 & 48e6 & 93e6 & 120e6 & 143e6 & 165e6 & 199e6 \\
1,000 & 19e6 & 60e6 & 86e6 & 109e6 & 131e6 & 165e6 \\
2,000 & 4e6 & 32e6 & 56e6 & 78e6 & 99e6 & 133e6 \\
3,000 & .9e6 & 14e6 & 32e6 & 52e6 & 71e6 & 104e6 \\
4,000 & .2e6 & 5e6 & 16e6 & 32e6 & 48e6 & 79e6 \\
5,000 & .03e6 & 2e6 & 7e6 & 18e6 & 31e6 & 58e6 \\
\hline
\end{tabular}
\caption{Total space distortion (TotalSD) of $(k, \delta)$-anonymity
for several parameter values (e6 stands for $\times 10^6$)}
\label{tab:totalsd}
\end{table}

The above principle of equating the space distortions
with $(k,\delta)$-anonymity yields a value for the space threshold
$R^s$ in each of SwapLocations and ReachLocations; however,
it does not constrain the time threshold, which we set at $R^t=100$.
Regarding $R^s$, we set it to achieve the total space
distortions of $(k,\delta)$-anonymity
for cluster size $k = \{2, 4, 6, 8, 10, 15\}$ and
\[ \delta = \{0, 1000, 2000, 3000, 4000, 5000\} \]
(parameter values considered in Table~\ref{tab:totalsd}).
In order to find such space thresholds efficiently, we assume that
the total space distortions of our methods define a monotonically
increasing function of the space threshold, {\em i.e.} the higher
the space threshold, the higher the total space distortion.
Under this assumption, we perform a logarithmic search over the set
of space thresholds defined by the interval $[0, 10^6]$. The reason behind
defining the maximum value for the space threshold as $10^6$ is
that it is high enough to achieve low numbers of removed trajectories.
Indeed, as shown in Figure~\ref{fig:both_method}, for both methods there
exists a value $R^s_{cutoff} < 10^6$ such
that, for every space threshold $R^s > R^s_{cutoff}$, neither the total
space distortion nor the percentage of removed locations and
removed trajectories significantly change. Table~\ref{tab:swap_thresholds}
and Table~\ref{tab:reach_thresholds} show the values of space thresholds
used in each configuration of $(k,\delta)$-anonymity for
SwapLocations and ReachLocations, respectively.

\begin{figure}[!ht]
\centering
  	\includegraphics[width=3in, angle=270]{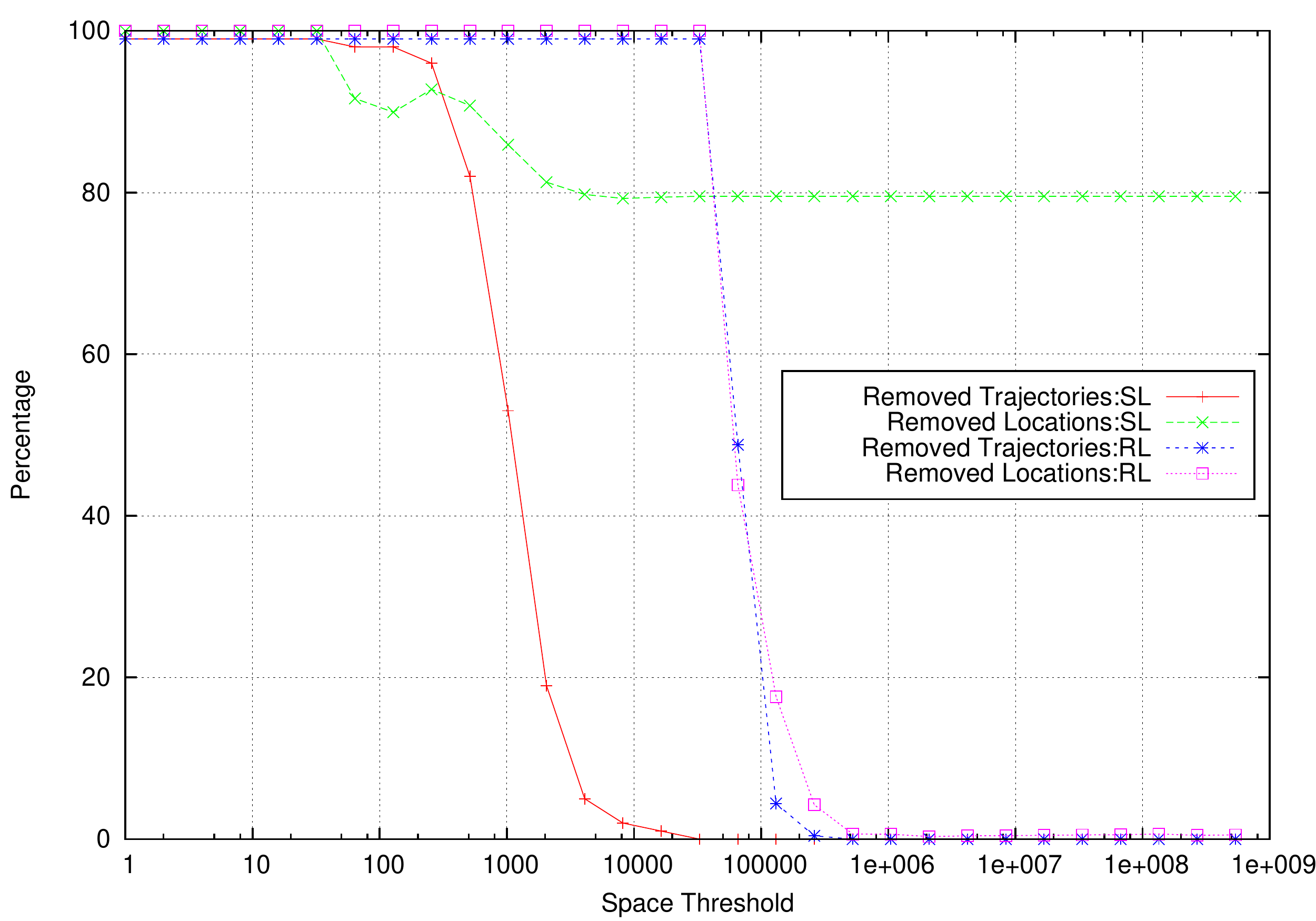}
\includegraphics[width=3in, angle=270]{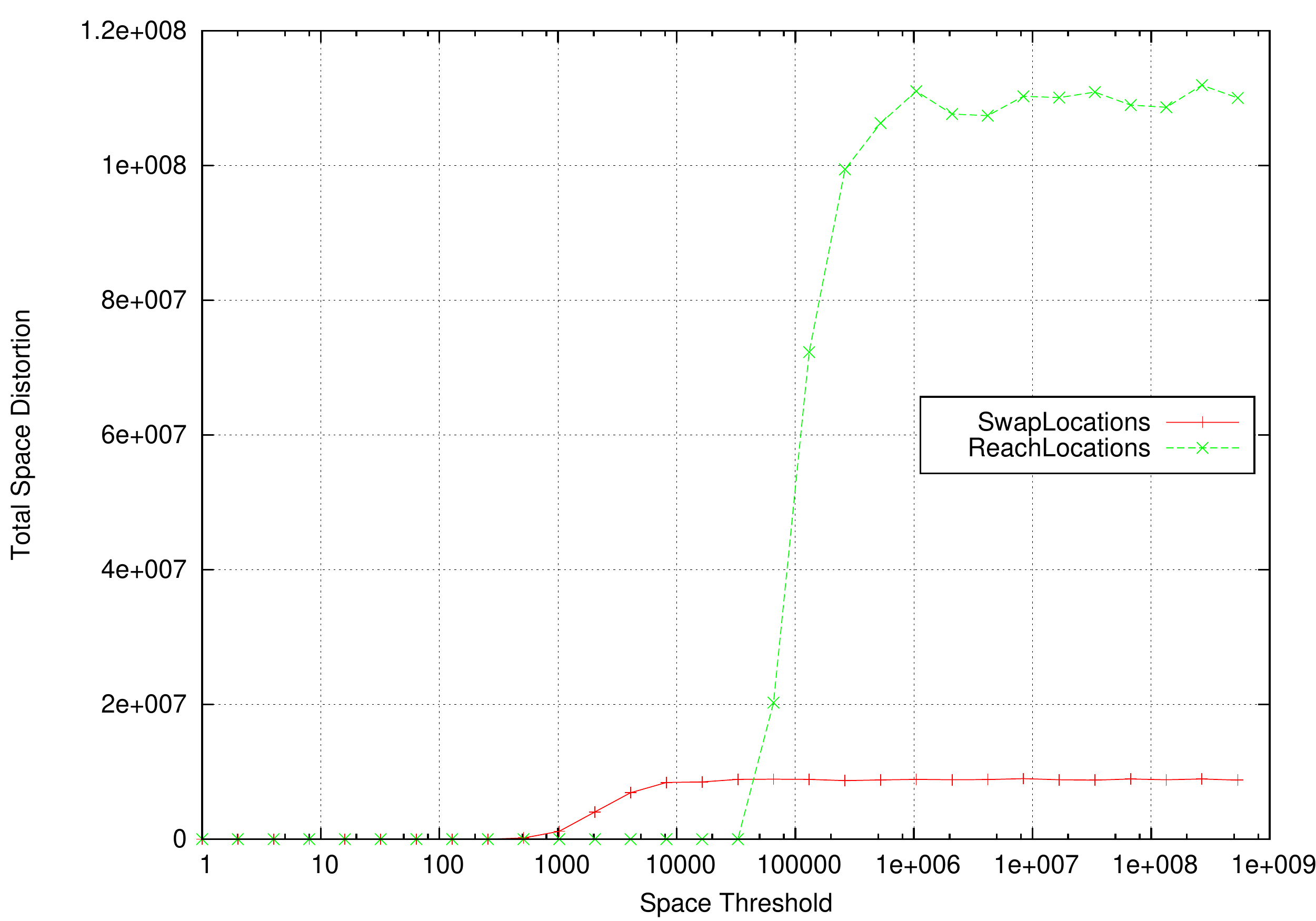}
  \caption{Top, percentage of removed trajectories and locations
with $k=10$, $R^t=100$ and several
values of $R^s$ for SwapLocations (SL) and ReachLocations (RL). Bottom, total
  space distortion with $k=10$, $R^t = 100$ and several
  values of $R^s$ for SwapLocations and ReachLocations}
  \label{fig:both_method}
\end{figure}

\begin{table}[!ht]
\centering
\begin{tabular}{|c|c|c|c|c|c|c|}
\hline
$\delta$ $\backslash$ $k$ & 2 & 4 & 6 & 8 & 10 & 15 \\
\hline
0 & $10^6$ & $10^6$ & $10^6$ & $10^6$ & $10^6$ & $10^6$ \\
1,000 & $10^6$ & $10^6$ & $10^6$ & $10^6$ & $10^6$ & $10^6$ \\
2,000 & 899 & $10^6$ & $10^6$ & $10^6$ & $10^6$ & $10^6$ \\
3,000 & 257 & $10^6$ & $10^6$ & $10^6$ & $10^6$ & $10^6$ \\
4,000 & 86 & 1390 & $10^6$ & $10^6$ & $10^6$ & $10^6$ \\
5,000 & 19 & 681 & 2507 & $10^6$ & $10^6$ & $10^6$ \\
\hline
\end{tabular}
\caption{Space thresholds used in SwapLocations to match the total
space distortion of each
configuration of $(k, \delta)$-anonymity}
\label{tab:swap_thresholds}
\end{table}

\begin{table}[!ht]
\centering
\begin{tabular}{|c|c|c|c|c|c|c|}
\hline
$\delta$ $\backslash$ $k$ & 2 & 4 & 6 & 8 & 10 & 15 \\
\hline
0 & 499875 & $10^6$ & $10^6$ & $10^6$ & $10^6$ & $10^6$ \\
1,000 & 25090 & 106126 & 270157 & $10^6$ & $10^6$ & $10^6$ \\
2,000 & 4780 & 52468 & 93717 & 151915 & 249999 & $10^6$ \\
3,000 & 749 & 37124 & 64801 & 95585 & 132857 & 238884 \\
4,000 & 136 & 25540 & 51089 & 73088 & 94465 & 152862 \\
5,000 & 57 & 18059 & 39061 & 58584 & 79101 & 113280 \\
\hline
\end{tabular}
\caption{Space thresholds used in ReachLocations to match
the total space distortion of each configuration of $(k, \delta)$-anonymity}
\label{tab:reach_thresholds}
\end{table}

As it can be seen in Tables~\ref{tab:swap_thresholds}
and~\ref{tab:reach_thresholds}, we use the maximum value ($10^6$) of
the space threshold for several configurations. This is because in those
configurations
the total space distortion caused by the $(k, \delta)$-anonymity
could not be reached by our methods no matter how much we increased
the space threshold.
Figure~\ref{fig:trashed} explains this behaviour by showing
%
the values of total space distortion SwapLocations
and ReachLocations minus the total space distortion of $(k, \delta)$-anonymity.
With almost every configuration, our methods have a total space
distortion lower than the total space distortion of
$(k, \delta)$-anonymity. In the case of SwapLocations,
the total space distortion is even much lower.

\begin{figure}[!ht]
\centering
  	\includegraphics[width=3in, angle=270]{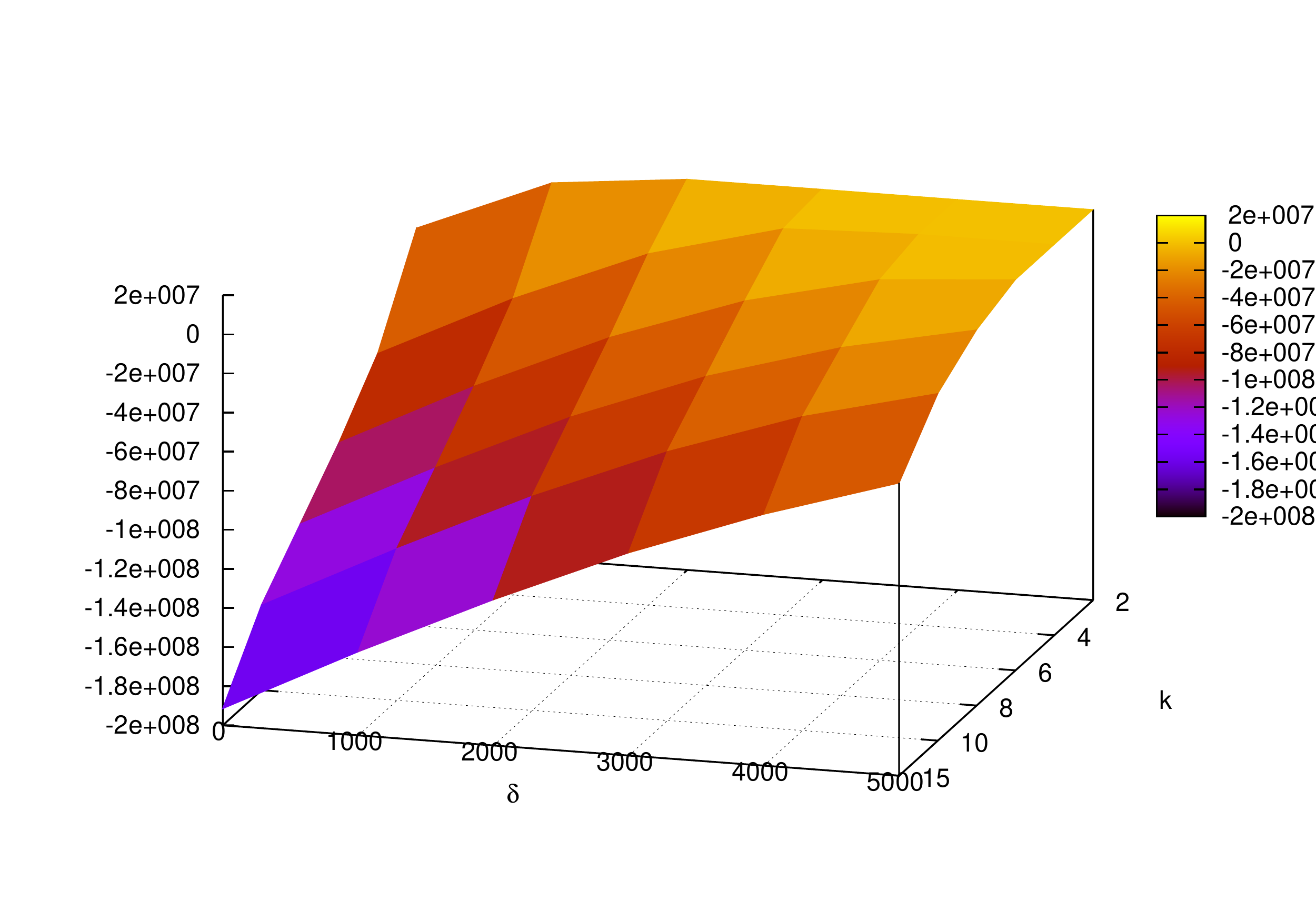}
  	\includegraphics[width=3in, angle=270]{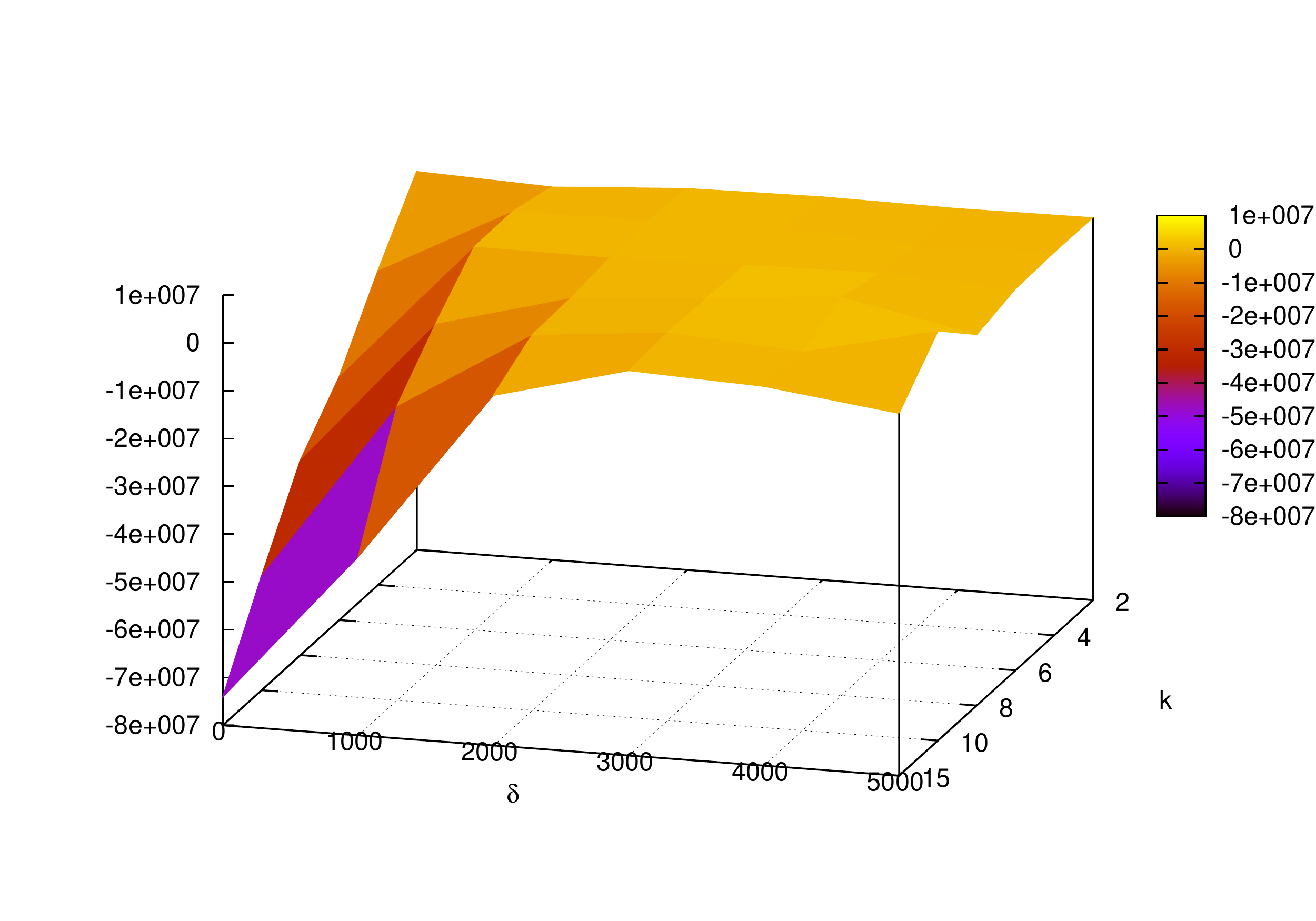}
\caption{Top: total space distortion of SwapLocations minus
total space distortion of $(k, \delta)$-anonymity
for several parameter configurations.
Bottom: total space distortion of ReachLocations minus
total space distortion of $(k, \delta)$-anonymity
for several parameter configurations.
The space thresholds defined in Tables~\ref{tab:swap_thresholds}
  and~\ref{tab:reach_thresholds} have been used, respectively.}
\label{fig:trashed}
\end{figure}

In general, SwapLocations does not reach high values of the
total space distortion because it removes more locations than
ReachLocations in order to achieve trajectory $k$-anonymity. Note that
removing locations does not increase the total space distortion
because we are considering $\Omega = 0$. Tables~\ref{tab:swapLocations}
and~\ref{tab:reachLocations} show in detail the percentage of removed
trajectories and the percentage of removed locations for different
values of $k = \{2, 4, 6, 8, 10, 15\}$ and
$\delta = \{0, 1000, 2000, 3000, 4000, 5000\}$, for SwapLocations and
ReachLocations, respectively.

As it can be seen in Table~\ref{tab:swapLocations}, in general SwapLocations
removes less trajectories than $(k, \delta)$-anonymity because
SwapLocations can cluster non-overlapping trajectories. Indeed, with
$(k, \delta)$-anonymity 227 trajectories were discarded
in the pre-processing step alone
because their time span could not match the time
span of other trajectories, and additional outlier trajectories
were discarded during clustering, up to a total $24\%$ of discarded
trajectories. However, SwapLocations removed up to $84\%$ of all
locations in the worst cases
and thus, it may not be suitable for applications where
preserving the number of locations really matters. SwapLocations
removes any location whose swapping set $U$ contains less than
$k$ locations, which is a relatively frequent event when
$k$ trajectories with
different lengths are clustered together.
As the cluster size $k$ increases, the length diversity tends
to increase and the removal percentage increases.
A simple way around the location removal problem is to
create clusters that contain trajectories with roughly the same length,
even though this may result in a higher total space distortion;
higher space distortion is a natural consequence of clustering
based on the trajectory length rather than the trajectory distance.

Table~\ref{tab:reachLocations} shows that
ReachLocations removes few trajectories when $\delta$ is small
and $k$ is large. The reason is that, for those parameterisations,
$(k,\delta)$-anonymity introduces so much total space distortion
that ReachLocations can afford taking the
maximum space threshold $R^s=10^6$ without reaching
that much distortion. Such a high space threshold allows
ReachLocations to easily swap spatial coordinates, so that
very few locations need to be removed.
Furthermore, the trajectories output by ReachLocations
are consistent with the underlying city topology.
As said above, the only drawback of this method is that in general it
does not provide trajectory $k$-anonymity; rather, it provides
location $k$-diversity.

\begin{table}[!ht]
\renewcommand{\arraystretch}{1.0}
\centering

\begin{tabular}{ @{} l c c c c c c c c c c c c @{}}
\toprule %
\multirow{2}*{ $\delta$ $\backslash$ $k$}
					&	\multicolumn{2}{c}{$2$}	&	\multicolumn{2}{c}{$4$}	&
					\multicolumn{2}{c}{$6$}	& \multicolumn{2}{c}{$8$}	& \multicolumn{2}{c}{$10$}	&
                    \multicolumn{2}{c}{$15$}	\\
					&	\textbf{T}	&	\textbf{L}	&	\textbf{T}	&	\textbf{L}	&
					\textbf{T}	&	\textbf{L}	& \textbf{T}	&	\textbf{L}	&
                    \textbf{T}	&	\textbf{L}	& \textbf{T}	&	\textbf{L}\\\cmidrule(l){2-13}

0	&	$0$	&	$34$	&	$0$	&	$58$	&	$0$	&	$69$	&	$1$	&	$75$	&	$0$	&	 $79$	&	$0$	 &	 $84$	 \\\cmidrule(l){2-13}

1000	&	$0$	&	$34$	&	$0$	&	$58$	&	$0$	&	$69$	&	$1$	&	$75$	&	$0$	&	 $79$	&	$0$	 &	 $84$	\\\cmidrule(l){2-13}

2000	&	$4$	&	$45$	&	$0$	&	$58$	&	$0$	&	$69$	&	$1$	&	$75$	&	$0$	&	 $79$	&	 $0$	 &	 $84$	\\\cmidrule(l){2-13}

3000	&	$11$	&	$62$	&	$0$	&	$58$	&	$0$	&	$69$	&	$1$	&	$75$	&	$0$	&	 $79$	&	 $0$	 &	 $84$	\\\cmidrule(l){2-13}

4000	&	$19$	&	$68$	&	$5$	&	$66$	&	$0$	&	$69$	&	$1$	&	$75$	&	$0$	&	 $79$	&	 $0$	 &	 $84$	\\\cmidrule(l){2-13}
				
5000 &	$32$	&	$78$	&	$20$	&	$73$	&	$4$	&	$72$	&	$1$	&	$75$	&	$0$	&	 $79$	 &	 $0$	 &	$84$	\\\bottomrule
				
\end{tabular}
\caption{Percentage of trajectories (columns labeled with \textbf{T})
and locations
(columns labeled \textbf{L}) removed by SwapLocations
when using time threshold 100, $k = \{2, 4, 6, 8, 10, 15\}$
and space thresholds that match the space distortion
caused by $(k,\delta)$-anonymity with the previous $k$'s and
$\delta = \{0, 1000, 2000, 3000, 4000, 5000\}$. Percentages have
been rounded to integers for compactness.
\label{tab:swapLocations}}
\end{table}

\begin{table}[!ht]
\renewcommand{\arraystretch}{1.0}

\centering

\begin{tabular}{ @{} l c c c c c c c c c c c c @{}}
\toprule %
\multirow{2}*{ $\delta$ $\backslash$ $k$}
					&	\multicolumn{2}{c}{$2$}	&	\multicolumn{2}{c}{$4$}	&
					\multicolumn{2}{c}{$6$}	& \multicolumn{2}{c}{$8$}	& \multicolumn{2}{c}{$10$}	&
                    \multicolumn{2}{c}{$15$}	\\
					&	\textbf{T}	&	\textbf{L}	&	\textbf{T}	&	\textbf{L}	&
					\textbf{T}	&	\textbf{L}	& \textbf{T}	&	\textbf{L}	&
                    \textbf{T}	&	\textbf{L}	& \textbf{T}	&	\textbf{L}\\\cmidrule(l){2-13}

0	&	$0$	&	$1$	&	$0$	&	$3$	&	$0$	&	$3$	&	$0$	&	$4$	&	$0$	&	$4$	&	$0$	&	 $3$	 \\\cmidrule(l){2-13}

1000	&	$0$	&	$2$	&	$0$	&	$3$	&	$0$	&	$3$	&	$0$	&	$4$	&	$0$	&	$5$	&	$0$	&	 $3$	 \\\cmidrule(l){2-13}

2000	&	$36$	&	$27$	&	$9$	&	$18$	&	$3$	&	$11$	&	$0$	&	$5$	&	$0$	&	$6$	 &	$0$	&	 $4$	 \\\cmidrule(l){2-13}

3000	&	$74$	&	$38$	&	$33$	&	$39$	&	$18$	&	$28$	&	$6$	&	$21$	&	$2$	&	 $13$	 &	 $0$	&	 $7$	\\\cmidrule(l){2-13}

4000	&	$82$	&	$43$	&	$65$	&	$49$	&	$41$	&	$40$	&	$20$	&	$34$	&	$10$	&	 $27$	 &	 $2$	&	 $16$	\\\cmidrule(l){2-13}

5000 &	$84$	&	$60$	&	$84$	&	$53$	&	$60$	&	$52$	&	$40$	&	$44$	&	$27$	&	 $35$	 &	 $10$	&	$27$	\\\bottomrule
				
\end{tabular}
\caption{Percentage of trajectories (columns labeled with \textbf{T})
and locations
(columns labeled \textbf{L}) removed by ReachLocations
when using time threshold 100, $k = \{2, 4, 6, 8, 10, 15\}$
and space thresholds that match the space distortion
caused by $(k,\delta)$-anonymity with the previous $k$'s and
$\delta = \{0, 1000, 2000, 3000, 4000, 5000\}$. Percentages have
been rounded to integers for compactness.
\label{tab:reachLocations}}
\end{table}

\subsubsection{Spatio-temporal range queries}

As stated in Section~\ref{subsec:utility}, a typical use
of trajectory data is to perform spatio-temporal range
queries on them. That is why we report
empirical results when performing the two query types
described and motivated in Section~\ref{subsec:utility}:
\emph{Sometime\_Definitely\_Inside} (SI)
and \emph{Always\_Definitely\_Inside} (AI).
We accumulate the number of trajectories in a set of
trajectories $\mathcal{T}$
that satisfy the SI or AI range queries using
the  SQL style code below.

\begin{itemize}
    \item Query $\mathcal{Q}_1(\mathcal{T}, R, t_b, t_e)$:
    \begin{itemize}
        \item[] \texttt{SELECT COUNT (*) FROM }$\mathcal{T}$ \texttt{WHERE SI(}$\mathcal{T}$\texttt{.traj, R, }$\texttt{t}_{\texttt{b}}, \texttt{t}_{\texttt{e}})$
    \end{itemize}
    \item Query $\mathcal{Q}_2(\mathcal{T}, R, t_b, t_e)$:
    \begin{itemize}
        \item[] \texttt{SELECT COUNT (*) FROM }$\mathcal{T}$ \texttt{WHERE AI(}$\mathcal{T}$\texttt{.traj, R, }$\texttt{t}_{\texttt{b}}, \texttt{t}_{\texttt{e}})$
    \end{itemize}
\end{itemize}

Then, we define two different \emph{range query distortions}:

\begin{itemize}
    \item SID$(\mathcal{T}, \mathcal{T}^{\star}) =
    \frac{1}{|\xi|}\sum_{\forall <R, t_b, t_e> \in \xi}\frac{|\mathcal{Q}_1(\mathcal{T}, R, t_b, t_e)-\mathcal{Q}_1(\mathcal{T}^{\star}, R, t_b, t_e)|}{\max{(\mathcal{Q}_1(\mathcal{T}, R, t_b, t_e), \mathcal{Q}_1(\mathcal{T}^{\star}, R, t_b, t_e))}}$
where $\xi$ is a set of SI queries as defined in Section~\ref{subsec:utility}
(definition of SI adapted to non-uncertain trajectories).
    \item AID$(\mathcal{T}, \mathcal{T}^{\star}) =
    \frac{1}{|\xi|}\sum_{\forall <R, t_b, t_e> \in \xi}\frac{|\mathcal{Q}_2(\mathcal{T}, R, t_b, t_e)-\mathcal{Q}_2(\mathcal{T}^{\star}, R, t_b, t_e)|}{\max{(\mathcal{Q}_2(\mathcal{T}, R, t_b, t_e), \mathcal{Q}_2(\mathcal{T}^{\star}, R, t_b, t_e))}}$
where $\xi$ is a set of AI queries as defined in Section~\ref{subsec:utility}
(definition of AI adapted to non-uncertain trajectories).
\end{itemize}

For our experiments with the synthetic data set, we chose random
time intervals $[t_b, t_e]$ such that $0 \leq t_e - t_b \leq 10$.
Also, we chose random uncertain trajectories with a randomly
chosen radius $0 \leq \sigma \leq 750$ as regions $R$.
Actually, $10$ and $750$ are, respectively, roughly a quarter of
the average duration and distance of all trajectories. Note that
we used uncertain trajectories {\em only} as regions $R$; however,
the methods we are considering in this chapter all
release non-uncertain trajectories.

Armed with these settings, we ran $100,000$ times
both queries $\mathcal{Q}_1$ and $\mathcal{Q}_2$ on the original data set
and the anonymised data sets provided by SwapLocations, ReachLocations,
and $(k, \delta)$-anonymity; that is, we took a set $\xi$ with
$|\xi|=100,000$.
The ideal range query distortion would be zero,
which means that query $\mathcal{Q}_{i}$ for $i \in {1,2}$ yields the same
result for both the original and the anonymised data sets; in practice,
zero distortion is hard to obtain. Therefore, in order to compare
our methods against $(k, \delta)$-anonymity, we use the same parameters
of the previous experiments (Tables~\ref{tab:totalsd},~\ref{tab:swap_thresholds}, and~\ref{tab:reach_thresholds}).
We show in Tables~\ref{tab:range_swapLocations}
and~\ref{tab:range_reachLocations} a comparison of SwapLocations,
respectively ReachLocations, against $(k,\delta)$-anonymity in terms
of SID and AID.

\begin{table}[p]
\renewcommand{\arraystretch}{1.0}

\centering
\begin{tabular}{ @{} l c c c c c c c c c c c c @{}}
\toprule %
\multirow{2}*{ $\delta$ $\backslash$ $k$}
					&	\multicolumn{2}{c}{$2$}	&	\multicolumn{2}{c}{$4$}	&
					\multicolumn{2}{c}{$6$}	& \multicolumn{2}{c}{$8$}	& \multicolumn{2}{c}{$10$}	&
                    \multicolumn{2}{c}{$15$}	\\
					&	\textbf{S}	&	\textbf{A}	&	\textbf{S}	&	\textbf{A}	&
					\textbf{S}	&	\textbf{A}	& \textbf{S}	&	\textbf{A}	&
                    \textbf{S}	&	\textbf{A}	& \textbf{S}	&	\textbf{A}\\\cmidrule(l){2-13}

0	&	$34$	&	$29$	&	$31$	&	$14$	&	$36$	&	$16$	&	$36$	&	$13$	&	$37$	&	 $13$	&	$43$	 &	 $14$	 \\\cmidrule(l){2-13}

1000	&	$24$	&	$20$	&	$24$	&	$8$	&	$28$	&	$10$	&	$27$	&	$8$	&	$28$	&	 $9$	 &	 $41$	 &	 $14$	\\\cmidrule(l){2-13}

2000	&	$18$	&	$14$	&	$18$	&	$4$	&	$20$	&	$3$	&	$20$	&	$2$	&	$27$	&	 $6$	&	 $39$	 &	 $10$	\\\cmidrule(l){2-13}

3000	&	$8$	&	$3$	&	$11$	&	$-2$	&	$13$	&	$0$	&	$16$	&	$-1$	&	$21$	&	 $4$	&	 $36$	 &	 $10$	\\\cmidrule(l){2-13}

4000	&	$-6$	&	$-7$	&	$6$	&	$-6$	&	$9$	&	$-5$	&	$11$	&	$-4$	&	$17$	&	 $2$	 &	 $30$	 &	 $5$	\\\cmidrule(l){2-13}
				
5000 &	$-22$	&	$-19$	&	$1$	&	$-9$	&	$3$	&	$-9$	&	$7$	&	$-7$	&	 $14$	 &	 $-2$	 &	$27$	 &	$2$	\\\bottomrule
				
\end{tabular}
\caption{Range query distortion of SwapLocations compared to
$(k,\delta)$-anonymity for SID
(columns labeled with \textbf{S}) and AID (columns labeled with \textbf{A})
when using $k = \{2, 4, 6, 8, 10, 15\}$
and space thresholds that match the space distortion
caused by $(k,\delta)$-anonymity with the previous $k$'s and
$\delta = \{0, 1000, 2000, 3000, 4000, 5000\}$.
In this table,
a range query distortion $x$ obtained with SwapLocations and a
range query distortion $y$ obtained with $(k,\delta)$-anonymity
are represented as the integer rounding of $(y-x)*100$. Hence,
values in the table are positive if and only if SwapLocations outperforms
$(k,\delta)$-anonymity.}
\label{tab:range_swapLocations}
\end{table}

\begin{table}[p]
\renewcommand{\arraystretch}{1.0}
\centering

\begin{tabular}{ @{} l c c c c c c c c c c c c @{}}
\toprule %
\multirow{2}*{ $\delta$ $\backslash$ $k$}
					&	\multicolumn{2}{c}{$2$}	&	\multicolumn{2}{c}{$4$}	&
					\multicolumn{2}{c}{$6$}	& \multicolumn{2}{c}{$8$}	& \multicolumn{2}{c}{$10$}	&
                    \multicolumn{2}{c}{$15$}	\\
					&	\textbf{S}	&	\textbf{A}	&	\textbf{S}	&	\textbf{A}	&
					\textbf{S}	&	\textbf{A}	& \textbf{S}	&	\textbf{A}	&
                    \textbf{S}	&	\textbf{A}	& \textbf{S}	&	\textbf{A}\\\cmidrule(l){2-13}

0	&	$34$	&	$25$	&	$28$	&	$12$	&	$33$	&	$10$	&	$32$	&	$5$	&	$31$	&	$5$	&	 $37$	&	 $6$	 \\\cmidrule(l){2-13}

1000	&	$25$	&	$19$	&	$21$	&	$6$	&	$24$	&	$4$	&	$23$	&	$1$	&	$25$	&	$2$	&	$35$	 &	 $5$	 \\\cmidrule(l){2-13}

2000	&	$10$	&	$10$	&	$8$	&	$-7$	&	$17$	&	$-3$	&	$19$	&	$-3$	&	$23$	&	$-3$	 &	$33$	 &	 $4$	 \\\cmidrule(l){2-13}

3000	&	$-4$	&	$2$	&	$0$	&	$-12$	&	$9$	&	$-12$	&	$13$	&	$-5$	&	$19$	&	 $-4$	 &	 $29$	&	 $1$	\\\cmidrule(l){2-13}

4000	&	$-11$	&	$-6$	&	$-6$	&	$-18$	&	$-2$	&	$-17$	&	$3$	&	$-16$	&	$13$	&	 $-6$	 &	 $26$	&	 $-3$	\\\cmidrule(l){2-13}

5000 &	$-14$	&	$-5$	&	$-10$	&	$-22$	&	$-8$	&	$-25$	&	$-4$	&	$-21$	&	$8$	&	 $-14$	 &	 $20$	&	$-5$	\\\bottomrule
				
\end{tabular}
\caption{Range query distortion of ReachLocations compared to
$(k,\delta)$-anonymity for SID
(columns labeled with \textbf{S}) and AID (columns labeled with \textbf{A})
when using $k = \{2, 4, 6, 8, 10, 15\}$
and space thresholds that match the space distortion
caused by $(k,\delta)$-anonymity with the previous $k$'s and
$\delta = \{0, 1000, 2000, 3000, 4000, 5000\}$.
In this table,
a range query distortion $x$ obtained with ReachLocations and a
range query distortion $y$ obtained with $(k,\delta)$-anonymity
are represented as the integer rounding of $(y-x)*100$. Hence,
values in the table are positive if and only if ReachLocations outperforms
$(k,\delta)$-anonymity.}
\label{tab:range_reachLocations}
\end{table}

It can be seen from Table~\ref{tab:range_swapLocations}
that SwapLocations performs significantly better than $(k, \delta)$-anonymity
for every cluster size and $\delta \leq 3000$. On
the other hand, Table~\ref{tab:range_reachLocations}
shows that ReachLocations outperforms $(k, \delta)$-anonymity only
for $\delta$ up to roughly 2000. Not surprisingly,
SwapLocations offers better performance than ReachLocations,
because the latter must deal with reachability constraints.
It is also remarkable that ReachLocations performs much better in
terms of SID than in terms of AID. The explanation is that,
while $(k, \delta)$-anonymity and SwapLocations operate at the trajectory
level, ReachLocations works at the location level.

We conclude that, according to these experiments,
our methods perform better than $(k, \delta)$-anonymity
regarding range query distortion for values of $\delta$ up to $2000$.
The performance for larger values of $\delta$ is less
and less relevant: indeed,
when $\delta \rightarrow \infty$, $(k, \delta)$-anonymity means
that no trajectory needs to be anonymised and hence the anonymised
trajectories are the same as the original ones.

\subsection{Results on real-life data}


The San Francisco cab data set \cite{comsnets09piorkowski}
we used consists of several files each of them
containing the GPS information of a specific cab during May 2008.
Each line within a file contains the space coordinates (latitude and longitude)
of the cab at a given time. However,
the mobility trace of a cab during an entire month can hardly be
considered a single trajectory.
We used big time gaps between two consecutive locations
in a cab mobility trace to split that trace into several
trajectories. All trajectory visualisations shown in this Section
were obtained using Google Earth.

For our experiments we considered just one
day of the entire month given in the real-life data set, but
the empirical methodology described below could be extended to several days.
In particular, we chose the day between May 25 at 12:04 hours
and May 26 at 12:04 hours because during this 24-hour period
there was the highest
concentration of locations in the data set. We also defined
the maximum time gap in a trajectory as 3 minutes; above 3 minutes,
we assumed that the current trajectory ended and that the next location
belonged to a different trajectory. This choice was based on
the average time gap between consecutive locations in the data set,
which was 88 seconds; hence, 3 minutes was roughly twice the average.
In this way, we obtained 4582 trajectories and 94 locations
per trajectory on average.

The next step was to filter out trajectories with strange features
(outliers). These outliers could be detected
based on several aspects like velocity, city topology, etc.
We focused on velocity and defined 240 km/h as the maximum speed that could
be reached by a cab. Consequently, the distance between
two consecutive locations could not be greater than 12 km because the maximum
within-trajectory time gap was 3 minutes. This allowed us to
detect and remove trajectories containing obviously erroneous locations;
Figure~\ref{fig:outlier} shows one of
these removed outliers where a cab appeared to have
jumped far into the sea probably due
to some error in recording its GPS coordinates.
Altogether, we removed 45 outlier trajectories and we were left
with a data set of 4547 trajectories with an average of
93 locations per trajectory.
Figure~\ref{fig:ten_trajectories} shows the ten longest
trajectories (in number of locations) in the final data set that we used.

\begin{figure}[!ht]
\begin{minipage}[b]{0.5\linewidth}
\centering
\includegraphics[width=0.80\textwidth, bb = 0 0 1020 922]{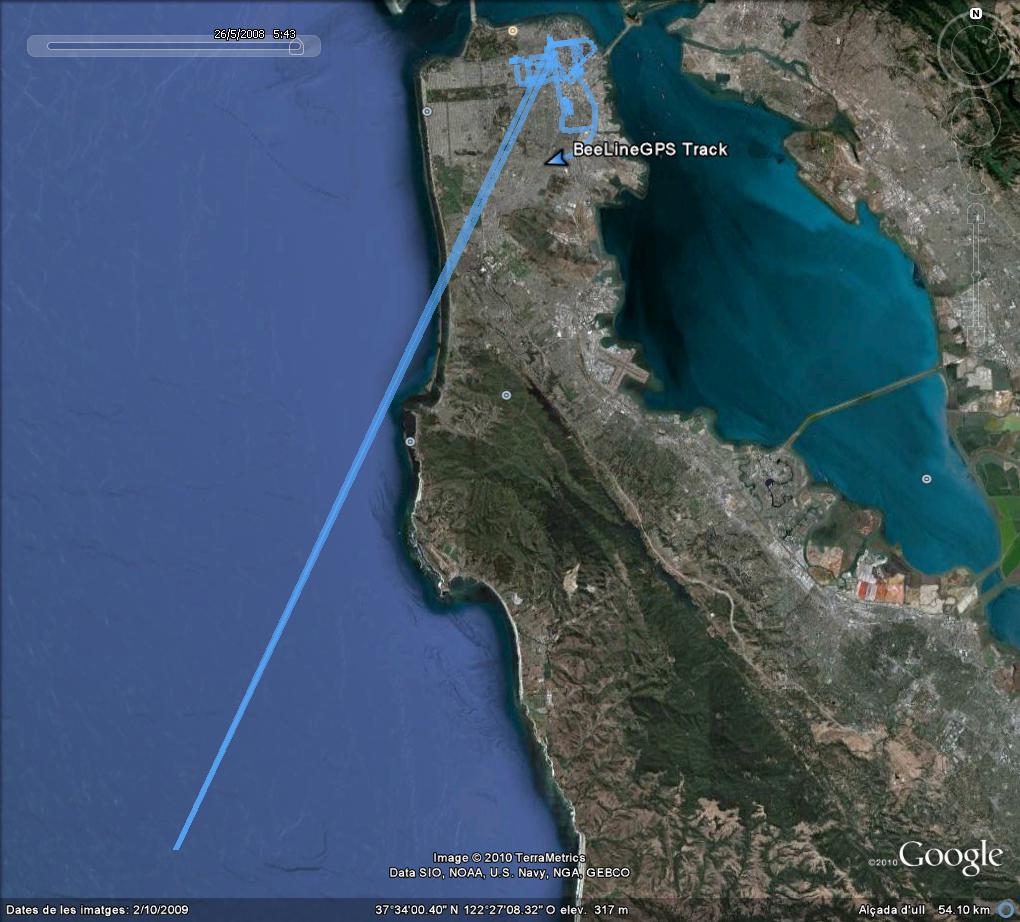}
\caption{Example of an outlier trajectory in the original real-life data set}
\label{fig:outlier}
\end{minipage}
\hspace{0.5cm}
\begin{minipage}[b]{0.5\linewidth}
\centering
\includegraphics[width=0.80\textwidth, bb = 0 0 1020 922]{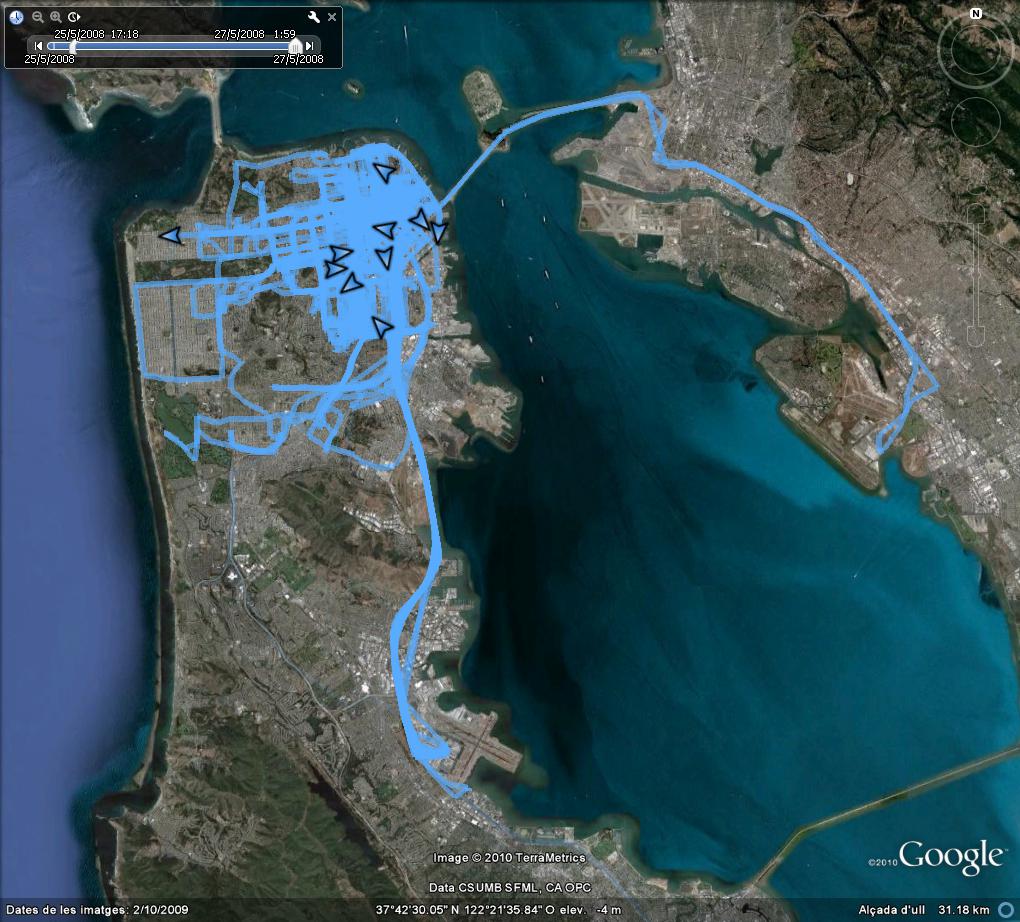}
\caption{Ten longest trajectories in the filtered real-life data set}
\label{fig:ten_trajectories}
\end{minipage}
\end{figure}

\subsubsection{Experiments with the distance metric}

We propose in this chapter
a new distance metric designed specifically for clustering
trajectories. Our distance metric considers both space and time,
dealing even with non-overlapping or partially-overlapping trajectories.
Contrary to the synthetic data where 10
trajectories had to be removed because the distances
to them could not be computed,
in this real-life data set our distance function
could be computed for every pair of trajectories.

Figure~\ref{fig:two_trajectories} shows two trajectories identified by
our distance metric as the two closest ones in the data set.
The two cabs moved around a parking lot and therefore stayed very
close to one another in space. Also in time both trajectories were very close:
one of them was recorded between 12:00:49 hours and 13:50:47
hours, while the other was recorded between 12:00:25 hours and 13:52:30 hours.
Therefore, both trajectories were correctly identified
by our distance metric as being close in time and space; they
could be clustered together with minimum utility loss for anonymisation
purposes.

\begin{figure*}[p]
\centering
\includegraphics[width=0.60\textwidth, bb = 0 0 1020 922]{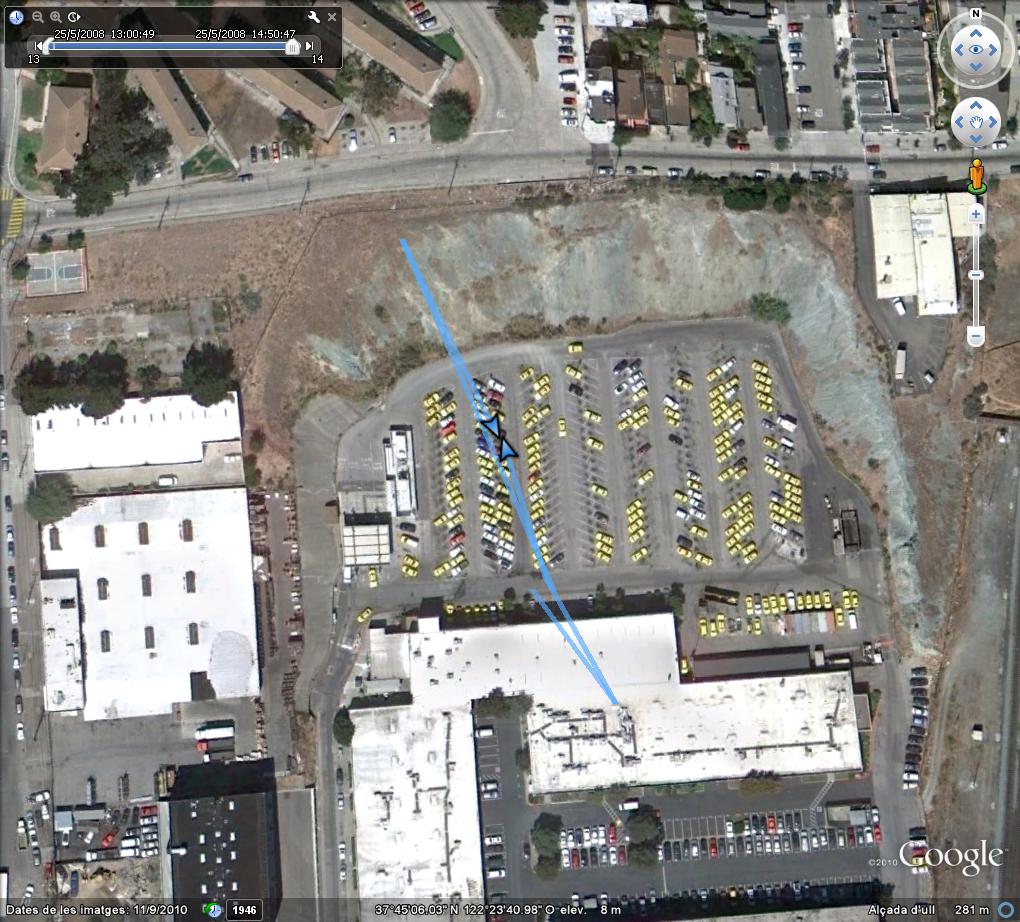}
\caption{The two closest trajectories in the real-life data set
according to our distance metric}
\label{fig:two_trajectories}
\end{figure*}

To compare, Figure~\ref{fig:two_trajectories2} shows two trajectories
identified by the Euclidean distance as the two closest ones in the data set.
These trajectories are located in a parking lot inside
San Francisco Airport and,
spatially, they are closer than the two trajectories shown
in Figure~\ref{fig:two_trajectories}. However, one of these trajectories
was recorded between 24:42:55 hours and 24:55:59 hours, while the other
was recorded between 19:05:29 hours and 19:06:15 hours. Hence, they
should not be
in the same cluster, because an adversary with time knowledge can easily
distinguish them.

\begin{figure*}[p]
\centering
\includegraphics[width=0.60\textwidth, bb = 0 0 1020 922]{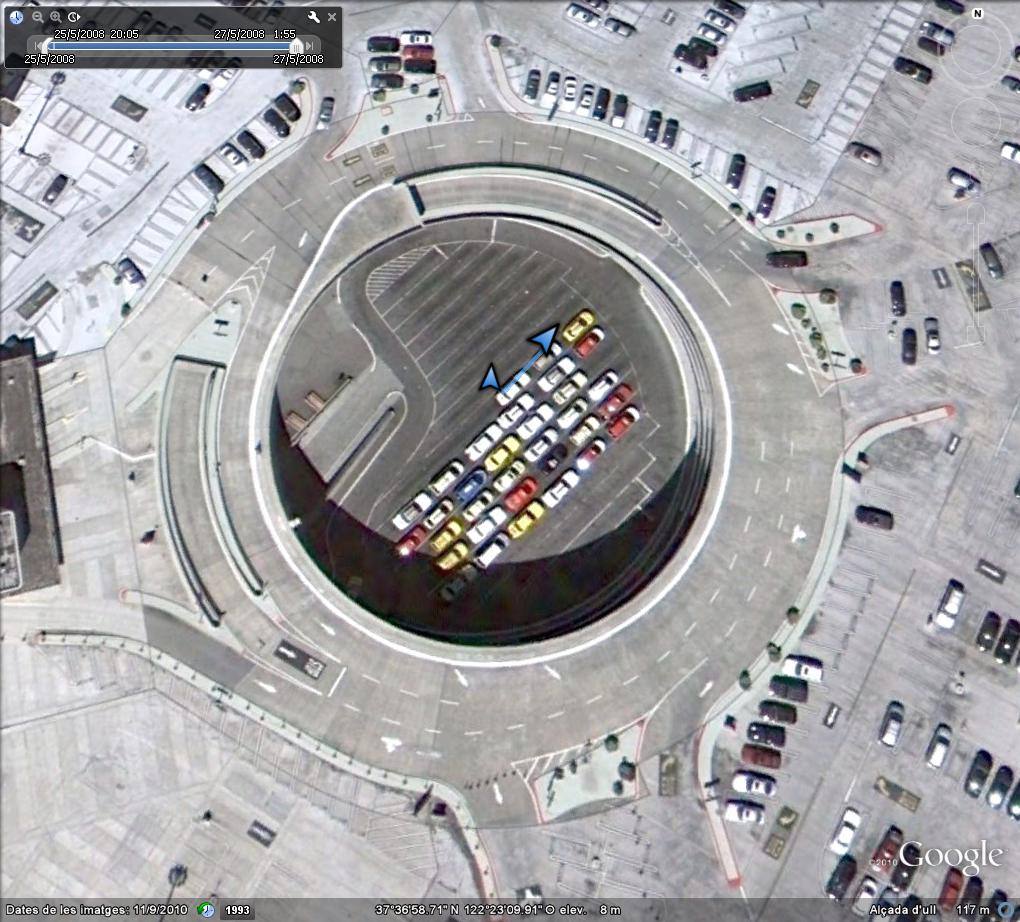}
\caption{The two closest trajectories in the real-life data set
according to the Euclidean distance}
\label{fig:two_trajectories2}
\end{figure*}


\subsubsection{Experiments with the SwapLocations method}

The ReachLocations method cannot be used when the graph of the
city is not provided. Hence, in the experiments with
the San Francisco real data we just considered the SwapLocations method.
As in the experiments with synthetic data, we set $\Omega = 0$
during the computation of the total space distortion.
Figure~\ref{fig:real_distortion} shows the values of total space
distortion given by the SwapLocations for different space thresholds
and different cluster sizes.

\begin{figure*}[p]
\centering
\includegraphics[width=0.6\textwidth, angle=270]{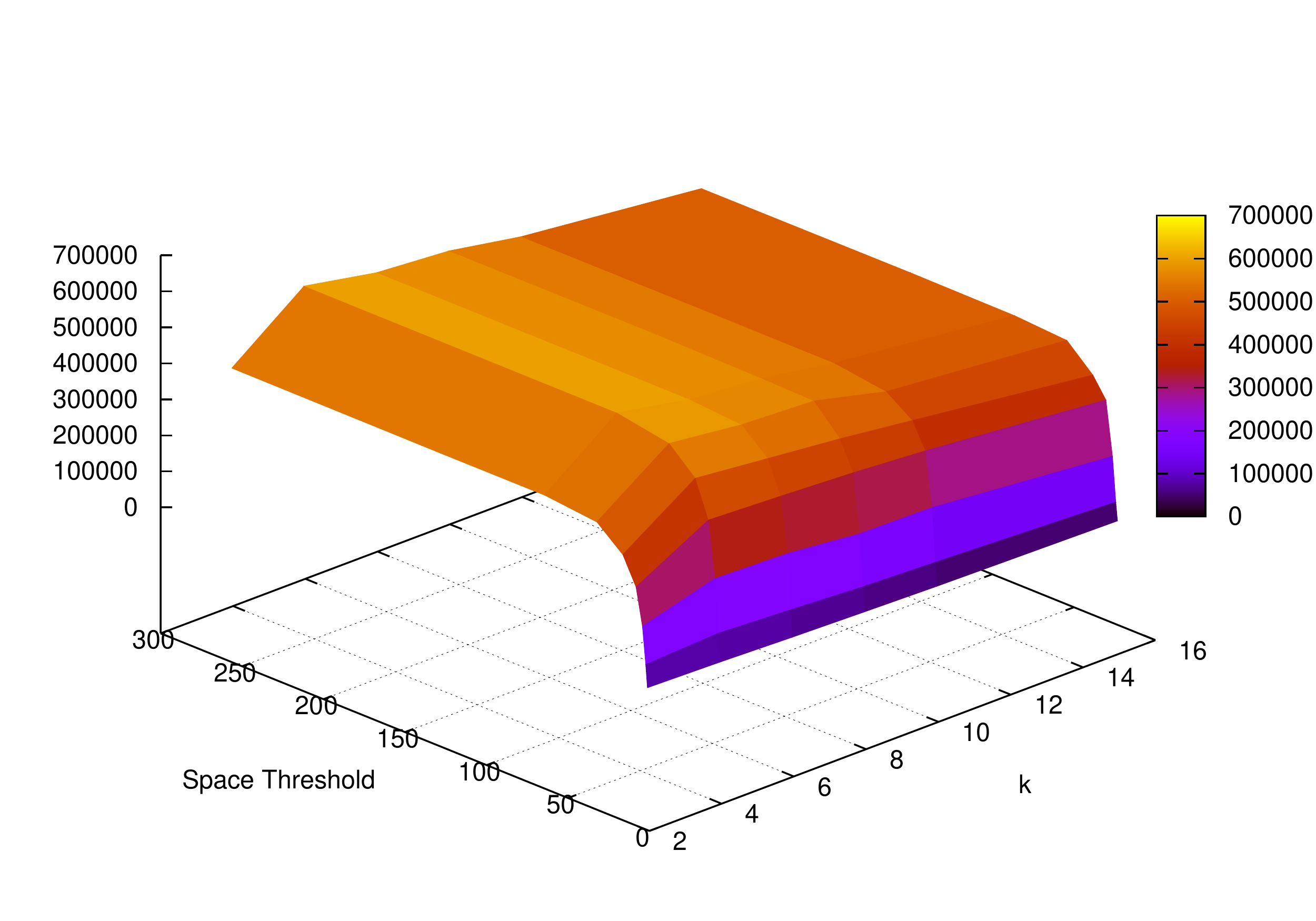}
\caption{Total space distortion (km) for SwapLocations using several
different space thresholds and cluster sizes on the real-life data set}
\label{fig:real_distortion}
\end{figure*}

Two other utility properties we are considering in this work are:
percentage of removed trajectories and percentage of removed locations.
Table~\ref{tab:swapLocations2} shows the values obtained with the
SwapLocations method for both utility properties.

\begin{table}[p]
\renewcommand{\arraystretch}{1.0}
\centering
\begin{tabular}{ @{} l c c c c c c c c c c c c @{}}
\toprule %
\multirow{2}*{ $R^s$ $\backslash$ $k$}
					&	\multicolumn{2}{c}{$2$}	&	\multicolumn{2}{c}{$4$}	&
					\multicolumn{2}{c}{$6$}	& \multicolumn{2}{c}{$8$}	& \multicolumn{2}{c}{$10$}	&
                    \multicolumn{2}{c}{$15$}	\\
					&	\textbf{T}	&	\textbf{L}	&	\textbf{T}	&	\textbf{L}	&
					\textbf{T}	&	\textbf{L}	& \textbf{T}	&	\textbf{L}	&
                    \textbf{T}	&	\textbf{L}	& \textbf{T}	&	\textbf{L}\\\cmidrule(l){2-13}

1	&	$23$	&	$43$	&	$40$	&	$64$	&	$49$	&	$71$	&	$58$	&	$74$	&	$62$	&	 $77$	&	$71$	&	 $81$	 \\\cmidrule(l){2-13}

2	&	$19$	&	$29$	&	$34$	&	$47$	&	$42$	&	$54$	&	$50$	&	$58$	&	$54$	&	 $60$	&	$50$	&	 $66$	 \\\cmidrule(l){2-13}

4	&	$14$	&	$17$	&	$27$	&	$29$	&	$35$	&	$35$	&	$40$	&	$40$	&	$45$	&	 $41$	 &	$54$	&	 $49$	\\\cmidrule(l){2-13}

8	&	$9$	&	$10$	&	$19$	&	$19$	&	$25$	&	$25$	&	$31$	&	$29$	&	$34$	&	 $31$	 &	 $42$	&	 $38$	\\\cmidrule(l){2-13}

16	&	$5$	&	$7$	&	$11$	&	$16$	&	$17$	&	$22$	&	$20$	&	$27$	&	$23$	&	 $30$	 &	 $32$	&	 $38$	\\\cmidrule(l){2-13}
				
32	&	$1$	&	$7$	&	$2$	&	$15$	&	$3$	&	$22$	&	$4$	&	$27$	&	$5$	&	 $30$	 &	 $8$	&	 $38$	\\\cmidrule(l){2-13}
				
64	&	$0$	&	$6$	&	$0$	&	$15$	&	$0$	&	$22$	&	$0$	&	$27$	&	$0$	&	 $30$	 &	 $0$	&	 $38$	\\\cmidrule(l){2-13}
				
 128 &	$0$	&	$6$	&	$0$	&	$15$	&	$0$	&	$22$	&	$0$	&	$27$	&	$0$	&	 $30$	 &	 $0$	&	 $38$	\\\bottomrule
				
\end{tabular}
\caption{Percentage of trajectories (columns labeled with \textbf{T}) and
locations (columns labeled with \textbf{L}) removed by SwapLocations
for several values of $k$ and several space thresholds $R^s$
on the real-life data set.
Percentages have been rounded to integers for compactness.
\label{tab:swapLocations2}}
\end{table}

Finally, Table~\ref{tab:range_swapLocations_real} reports
the performance of SwapLocations regarding spatio-temporal range queries.
We picked random time intervals of length at most $20$ minutes. Also,
random uncertain trajectories with uncertainty threshold of
size at most $7$ km were chosen as the regions.
Analogously to the experiments with the synthetic data set,
$20$ and $7$ are roughly a quarter of the average duration and distance
of all trajectories, respectively. It can be seen
that the SwapLocations method provides low range query distortion
for every value of $k$ when the space threshold is small, \emph{i.e.}
when the total space distortion is also small. However, the smaller
the space threshold, the larger the number of removed trajectories and
locations (see Table~\ref{tab:swapLocations2}). This illustrates
the trade-off between the utility properties considered.

\begin{table}[p]
\renewcommand{\arraystretch}{1.0}
\centering
\begin{tabular}{ @{} l c c c c c c c c c c c c @{}}
\toprule %
\multirow{2}*{ $R^s$ $\backslash$ $k$}
					&	\multicolumn{2}{c}{$2$}	&	\multicolumn{2}{c}{$4$}	&
					\multicolumn{2}{c}{$6$}	& \multicolumn{2}{c}{$8$}	& \multicolumn{2}{c}{$10$}	&
                    \multicolumn{2}{c}{$15$}	\\
					&	\textbf{S}	&	\textbf{A}	&	\textbf{S}	&	\textbf{A}	&
					\textbf{S}	&	\textbf{A}	& \textbf{S}	&	\textbf{A}	&
                    \textbf{S}	&	\textbf{A}	& \textbf{S}	&	\textbf{A}\\\cmidrule(l){2-13}

1	&	$13$	&	$22$	&	$18$	&	$27$	&	$20$	&	$29$	&	$19$	&	$29$	&	$24$	&	 $31$	&	$25$	 &	 $34$	 \\\cmidrule(l){2-13}

2	&	$16$	&	$24$	&	$25$	&	$34$	&	$26$	&	$35$	&	$24$	&	$35$	&	$27$	&	 $37$	 &	 $27$	 &	 $37$	\\\cmidrule(l){2-13}

4	&	$18$	&	$25$	&	$30$	&	$37$	&	$33$	&	$41$	&	$34$	&	$42$	&	$38$	&	 $46$	&	 $38$	 &	 $45$	\\\cmidrule(l){2-13}

8	&	$21$	&	$27$	&	$34$	&	$40$	&	$38$	&	$44$	&	$40$	&	$46$	&	$44$	&	 $50$	&	 $48$	 &	 $54$	\\\cmidrule(l){2-13}

16	&	$20$	&	$26$	&	$36$	&	$42$	&	$42$	&	$47$	&	$45$	&	$50$	&	$50$	&	 $54$	 &	 $53$	 &	 $58$	\\\cmidrule(l){2-13}
				
32	&	$21$	&	$26$	&	$39$	&	$44$	&	$45$	&	$49$	&	$48$	&	$53$	&	$53$	&	 $57$	 &	 $58$	 &	 $62$	\\\cmidrule(l){2-13}
				
64	&	$20$	&	$25$	&	$39$	&	$44$	&	$46$	&	$50$	&	$51$	&	$54$	&	$54$	&	 $57$	 &	 $61$	 &	 $64$	\\\cmidrule(l){2-13}
				
 128 &	$21$	&	$26$	&	$39$	&	$44$	&	$48$	&	$50$	&	$51$	&	$56$	&	 $54$	 &	 $58$	 &	$61$	 &	$64$	\\\bottomrule
				
\end{tabular}
\caption{Range query distortion caused by
SwapLocations on the real-life data set
for SID (columns labeled with \textbf{S}) and AID
(columns labeled with \textbf{A}), for several values of $k$
and several space thresholds $R^s$. In this table,
a range query
distortion $x$ is represented as the integer rounding of $x*100$
for compactness.
\label{tab:range_swapLocations_real}}
\end{table}


\section{Conclusions}

In this chapter, we have presented two permutation-based heuristic
methods to anonymise trajectories
with the common features that: (i)
places and times in the anonymised trajectories are true original places
and times with full accuracy;
(ii) both methods can deal with trajectories with partial or no time
overlap, thanks to a new distance also introduced in this paper.
The first method aims at trajectory $k$-anonymity while the second method takes reachability constraints into account,
that is, it assumes a territory constrained
by a network of streets or roads; to avoid
removing too many locations, the second method changes its privacy
ambitions from trajectory $k$-anonymity to location $k$-diversity.

Both methods use permutation of locations, and the first method
uses also trajectory microaggregation.
There are few counterparts in the literature comparable to the
first method, and virtually none comparable to the
second method.
Experimental results show that, for most parameter choices
and for similar privacy levels,
our methods offer better utility
than $(k,\delta)$-anonymity.

\chapter{Conclusions}
\label{chap:8}

\emph{This chapter summarises the contributions of the present dissertation. In addition, it sketches some lines for future work
that arise from either partially reached goals or expected improvements.}

\minitoc

In this thesis we have focussed on security, privacy, and scalability issues in the RFID technology. We have considered RFID identification protocols based on symmetric key cryptography, which seem to be the most suitable for low-cost RFID tags. We have also dealt with the challenges behind measuring the distance between tags and readers in order to improve the security of any RFID identification protocol. Since the RFID technology is becoming more and more popular, we
noticed that there is an increasing need for new trajectory anonymisation algorithms. For this reason, the last contribution in this
dissertation is devoted to this subject.

\section{Contributions}

In more detail, our contributions are:

\begin{enumerate}
    \item We have presented a communication-efficient
    protocol for collaborative RFID readers
to privately identify RFID tags. With the presented protocol, the centralised management
of tags can be avoided, along with bottlenecks and undesired delays.
    \item We have presented a novel protocol that
uses location and time of arrival predictors to improve the efficiency of the
widely accepted IRHL scheme. We have shown that our protocol outperforms
previous proposals in terms of scalability whilst guaranteeing the same level
of privacy and security.
    \item We have contributed to the design of distance-bounding protocols by: (i) providing a way to compute an upper bound on the distance-fraud probability, which is useful for analysing previous protocols and designing future ones; (ii) re-analysing the mafia fraud probability of the Kim and Avoine protocol~\cite{KimA-2009-cans}; (iii) proposing a new distance-bounding protocol that strikes
    a better balance than all previously published distance-bounding protocols
    between memory consumption, distance fraud resistance,
    and mafia fraud resistance.
    \item We have presented two permutation-based heuristic
methods to anonymise trajectories
with the common features that: (i)
places and times in the anonymised trajectories are true original places
and times with full accuracy;
(ii) both methods can deal with trajectories with partial or no time
overlap, thanks to a new distance also introduced in this dissertation.
The first method aims at trajectory $k$-anonymity while the second method takes reachability constraints into account,
that is, it assumes a territory constrained
by a network of streets or roads; to avoid
removing too many locations, this second method changes its privacy
ambitions from trajectory $k$-anonymity to location $k$-diversity.
\end{enumerate}

\section{Publications}

The main publications supporting the content of this thesis are the following:

\begin{itemize}
    \item Rolando Trujillo-Rasua, Benjamin Martin, and Gildas Avoine.
        \newblock The Poulidor distance-bounding protocol.
        \newblock In {\em The 6th Workshop on RFID Security and Privacy - RFIDSEC 2010}, pages 239--257, 2010.
    \item Josep Domingo-Ferrer, Michal Sramka, and Rolando Trujillo-Rasua.
        \newblock Privacy-preserving publication of trajectories using
          microaggregation.
        \newblock In {\em Proceedings of the SIGSPATIAL ACM GIS 2010 International
          Workshop on Security and Privacy in GIS and LBS, SPRINGL 2010}, San Jose,
          California, USA, 2 November 2010. ACM, pages 26--33, 2010.
    \item Rolando Trujillo-Rasua and Agusti Solanas.
        \newblock Efficient probabilistic communication protocol for the private
          identification of RFID tags by means of collaborative readers.
        \newblock {\em Computer Networks}, 55(15):3211--3223, 2011.
    \item Rolando Trujillo-Rasua and Agusti Solanas.
        \newblock Scalable trajectory-based protocol for RFID tags identification.
        \newblock In {\em The IEEE International Conference on RFID-Technologies and Applications - RFID-TA}, pages 279--285, 2011.
    \item Josep Domingo-Ferrer and Rolando Trujillo-Rasua.
        \newblock Microaggregation- and permutation-based anonymization of movement data.
        \newblock {\em Information Sciences}, \url{http://dx.doi.org/10.1016/j.ins.2012.04.015}.
    \item Rolando Trujillo-Rasua, Agusti Solanas, Pablo A. P\'erez-Mart\'inez and Josep Domingo-Ferrer.
        \newblock Predictive protocol for the scalable identification of RFID tags through collaborative readers.
        \newblock {\em Computers in Industry}, \url{http://dx.doi.org/10.1016/j.compind.2012.03.005}.
    \item Josep Domingo-Ferrer and Rolando Trujillo-Rasua.
        \newblock Anonymization of trajectory data.
        \newblock {\em 7th Joint UNECE/Eurostat Work Session on
	Statistical Data Confidentiality}, Tarragona, Catalonia,
	26-28 October 2011.
	Published at \url{http://www.unece.org/fileadmin/DAM/stats/documents/ece/ces/ge.46/2011/32_Domingo-Trujillo.pdf}.
\end{itemize}

Other publications co-authored by the candidate
and related to RFID systems,
but not included in this thesis, are listed below:

\begin{itemize}
    \item Albert Fernàndez-Mir, Rolando Trujillo-Rasua, Jordi
    Castellà-Roca and Josep Domingo-Ferrer.
        \newblock Scalable RFID authentication protocol supporting ownership transfer and controlled delegation.
        \newblock In {\em The 7th Workshop on RFID Security and Privacy - RFIDSEC 2011}, Amherst, Massachusetts (USA), pages 147--162, Jun 2011.
    \item Rolando Trujillo-Rasua, Antoni Martínez-Ballesté and Agusti Solanas.
        \newblock Revisión de protocolos para la identificación escalable, segura y privada en sistemas RFID.
        \newblock {\em 5as Jornadas Científicas sobre RFID}, Tarragona, Catalonia, 2011. Published at \url{http://crises2-deim.urv.cat/articles/index/type/conferences#672}.
\end{itemize}

\section{Future work}

Next, we sketch possible
lines for future work in the same order in which
we have presented our main contributions.

\begin{enumerate}
    \item  Our first proposal based on collaborative readers (see Chapter~\ref{chap:4}) opens at least the following
  research issues: (i)     study the effect of the number of
neighbours, (ii)  propose methods to dynamically vary $p$ so as to adapt it to
tag
movements, (iii) propose hybrid methods that mix hash-based solutions and tree-based
solutions with collaborative readers.
    \item In Chapter~\ref{chap:4_5} we partially tackle the second issue explained above by proposing some algorithms aimed at location prediction. However, those predictors may work well in some scenarios,
but their performance decreases in others. Although we have provided some
practical implementations for the predictors, the definition of our protocol is
flexible enough to accept the use of any location predictor. Due to the fact
that the efficiency of our proposal highly depends on the accuracy of the
predictors we plan to study and compare a variety of predictors in different
scenarios in the future.
    \item Chapter~\ref{chap:5} introduces the graph-based protocol concept,
    which in turn suggests lines for further work. First of all, an interesting question is to know if there are graph-based protocols that behave still better than the one presented here. In particular, if the number of rounds is not a critical parameter, prover and verifier may be allowed to increase the number of rounds while keeping a $2n$-node graph. This means that some nodes may be used twice. In such a case, the security analysis provided in this paper must be refined. On the other hand, although a bound on the distance fraud success probability is provided, calculating the exact probability of success is still cumbersome.
    \item Regarding trajectory anonymisation, the future work will be directed towards designing trajectory
anonymisation methods aimed at achieving trajectory $p$-privacy
(see Definition~\ref{def:trajectory_private}), but discarding less
locations than the SwapLocations method. Also, finding trajectory
anonymisation methods for constrained territories with better
utility than ReachLocations is an open challenge.
\end{enumerate}

\bibliographystyle{plain}

\end{document}